\documentclass[11pt]{article}

\usepackage{fullpage}
\usepackage{amsfonts,latexsym,amsthm,amssymb,amscd,euscript}
\usepackage{epsfig, slashed}
\usepackage{graphicx,bm}
\usepackage{subfigure, svg}
\usepackage{dsfont}
\usepackage{hyperref, youngtab}
\usepackage{mathrsfs}
\usepackage{xcolor}
    \hypersetup{colorlinks=true,citecolor=blue,urlcolor =black,linkbordercolor={1 0 0}}
\usepackage{amsmath}               
{
	\theoremstyle{plain}
	\newtheorem{theorem}{Theorem}
	\newtheorem{assumption}[theorem]{Assumption}
	\newtheorem{proposition}[theorem]{Proposition}
	\newtheorem{lemma}[theorem]{Lemma}
	\newtheorem{corollary}[theorem]{Corollary}

	\theoremstyle{definition}
	\newtheorem{definition}[theorem]{Definition}
	
}

\usepackage{color}
 \definecolor{green}{rgb}{0.0, 0.5, 0.5}

\allowdisplaybreaks[1]

\theoremstyle{remark}
\newtheorem*{remark}{Remark}
\newtheorem{notation}[theorem]{Notation}

\newcommand{\R}{\mathbb R}
\newcommand{\CF}{\mathbb C}

\newcommand{\N}{\mathbb N}

\newcommand{\HH}{\mathcal H}

\newcommand{\nc}{\newcommand}

\nc{\on}{\operatorname}
\nc{\Spec}{\on{Spec}}

\DeclareMathOperator{\range}{range}

\nc{\T}{\mathcal{T}}
\nc{\ZZ}{\mathcal{Z}}
\nc{\ex}[1]{\left\langle {#1} \right\rangle}
\nc{\e}{\epsilon}
\nc{\vp}{\varphi}
\nc{\eg}{E_\gamma}
\nc{\tn}{\tilde{\nu}}
\nc{\MT}{\tilde{V}}
\nc{\id}{\mathds{1}}

\nc{\gam}{\gamma}
\nc{\beq}{\begin{equation}}
\nc{\eeq}{\end{equation}}

\nc{\op}{\left(}
\nc{\cp}{\right)}

\nc{\os}{\left[}
\nc{\cs}{\right]}

\nc{\oa}{\big{(}}
\nc{\ca}{\big{)}}

\nc{\ob}{\Big{(}}
\nc{\cb}{\Big{)}}

\nc{\oc}{\bigg{(}}
\nc{\cc}{\bigg{)}}

\nc{\od}{\Bigg{(}}
\nc{\cd}{\Bigg{)}}

\nc{\pdev}[2]{\frac{\partial {#1}}{\partial {#2}}}
\nc{\pdevt}[3]{ \op \frac{\partial {#1}}{\partial {#2}} \cp_{#3}}

\nc{\ket}[1]{| {#1} \rangle}
\nc{\bra}[1]{\langle {#1} |}
\nc{\inner}[2]{\langle {#1} | {#2} \rangle}
\nc{\modsq}[1]{| {#1} |^2}

\nc{\mel}[3]{\langle {#1} | {#2} | {#3} \rangle}

\nc{\indp}[2]{ {p_{#1}}^{#2} }

\newcommand{\largehalffigsize}{0.5}

\usepackage{bbm} 
\usepackage[normalem]{ulem} 

\newcommand{\ind}{\mathbbm{1}}
\newcommand{\setof}[2]{\left\{ #1\; : \;#2 \right\}}

\newsavebox\MBox

\numberwithin{equation}{section}
\numberwithin{theorem}{section}

\newcommand{\footremember}[2]{%
    \footnote{#2}
    \newcounter{#1}
    \setcounter{#1}{\value{footnote}}%
}

\newcommand{\slength}[1]{L(#1)}

\begin{document} 

\title{The spin-one Motzkin chain is gapped for any area weight $t<1$} 

\author{Radu Andrei  \footremember{harvard}{Department of Physics, ETH Zurich, Switzerland} \and Marius Lemm \footremember{tuebingen}{Department of Mathematics, University of T\"ubingen, 72076 T\"ubingen, Germany} \and Ramis Movassagh\footremember{ibmcambridge}{IonQ}}

 \graphicspath{ {./images/} }

\date{May 22, 2026}
\maketitle

\begin{abstract}
We consider the spin-one Motzkin chain with area weight $t>0$.
We resolve three open questions from the literature about this model. We prove (i) existence of a uniform spectral gap for all $t<1$ as conjectured by Zhang--Ahmadein--Klich \cite{zhang2017novel} (ii)  an explicit formula for the long-distance limit of the string order parameter, which shows it is non-vanishing at small $t$, confirming a conjecture by Barbiero et al. \cite{barbiero2017haldane}, and (iii) that  gaplessness for  $t>1$ is robust and extends to hard boundary conditions, answering a question of Zhang--Klich \cite{zhang2017entropy}.
Our proof rests on an effective approximate description of the local ground states on finite open Motzkin chains. These ground states can be labeled by Motzkin walks with imbalances between up- and down-steps and we obtain different low-area approximations depending on whether the imbalance is high or low.
\end{abstract}


\renewcommand{\labelenumi}{(\alph{enumi})}

{
  \hypersetup{linkcolor=black}
  \setcounter{tocdepth}{1}
  \tableofcontents
}

\section{Introduction and main results}
Motzkin spin chains \cite{bravyi2012criticality,movassagh2016supercritical,zhang2017novel,levine2017gap,barbiero2017haldane,movassagh2017entanglement,movassagh2017entanglement,udagawa2017finite,sugino2018renyi,sugino2018area,dell2019long,alexander2019exact,tong2021shor,menon2024symmetries,hao2023exact}, have emerged as a new class of quantum lattice Hamiltonians that allow to explore a variety of fundamental phenomena. The models also have fermionic  cousins called Fredkin chains \cite{salberger2017entangled,salberger2017deformed,movassagh2016gap,zhang2017entropy} and higher-dimensional variants \cite{zhang2023coupled,zhang2024quantum,zhang2026sequential}. The initial surge of interest in Motzkin spin chains arose when Bravyi et al.\ \cite{bravyi2012criticality}, Movassagh-Schor \cite{movassagh2016supercritical} and Zhang-Ahmadein-Klich \cite{zhang2017novel} discovered that these display unusually large ground state entanglement entropy. Subsequent investigations revealed detailed intricate structure and symmetries of their  higher ground state correlation functions \cite{movassagh2017entanglement,sugino2018renyi,dell2019long,menon2024symmetries} and critical dynamical exponents  \cite{chen2017gapless,chen2017quantum}. Finally, a rich list of unforeseen connections of these models to other areas of physics has emerged, specifically to fully integrable models \cite{udagawa2017finite,tong2021shor,hao2023exact,zhang2023coupled}, holography \cite{alexander2019exact,alexander2021exact}, and number theory \cite{hao2022can}.
This rapid journey has taken Motzkin spin chains within a little over 10 years from a toy model with curious entanglement behavior to presenting a new paradigmatic model of quantum matter. It has been used in quantum error-correction \cite{brandao2019quantum,movassagh2020constructing} and is nowadays the target of quantum simulation in laboratory experiments \cite{mukherjee2026quantum}. 

The Motzkin spin chains introduced by Zhang-Ahmadein-Klich in \cite{zhang2017novel} come with two  parameters --- spin $s\in\mathbb Z_+$ and area weight $t>0$ --- that allow to explore rich physical behavior, e.g., volume-law entanglement for $s\geq 2$ and $t>1$. It is  of interest to \textit{understand the ground state phase diagram of Motzkin spin chains}.

In the present paper, we present a comprehensive study of the ground state properties of the spin $s=1$ Motzkin chain. We address  three open questions from the literature about this model, as we describe now.

\subsection{First main result: spectral gap}
In the 2017 paper introducing the area-weighted Motzkin spin chains, ZAK conjectured that the area-weighted Motzkin spin chains are \textit{gapped} for area weight $t<1$ and any spin $s$. This conjecture has remained open. Upper bounds on the closing rate of the spectral gap exist for $t=1$ and $s\geq 1$ \cite{bravyi2012criticality,movassagh2016supercritical}, for $t>1$ and $s\geq 2$ \cite{levine2017gap} and for $t>1$ and $s=1$ \cite{zhang2017entropy}. In other words, the only regime that can be gapped is the $t<1$ regime. (We recall that ``gapped'' means that there exists a constant gap independent of the system size.)  One motivation for the ZAK gap conjecture is Hastings' famous result \cite{hastings2007area} that a gap implies the area law for the entanglement entropy that was proved to hold for $t<1$ by ZAK.\\ 

In this paper, we prove the ZAK conjecture for the spin-$1$ Motzkin chain. 
Given area weight $t>0$, let $H_{n}(s,t)$ be the Motzkin Hamiltonian on a chain of length $n$ as introduced in \cite{zhang2017novel} and recalled in Subsection \ref{ssect:Hamiltoniandefn} below. We write $\gam_n(t)$ for its spectral gap. 

\begin{theorem}[The spin-$1$ Motzkin chain is gapped for any area weight $t<1$]
\label{thm:main}
For every $t\in (0,1)$, there exists a constant $c(t)>0$ such that
\beq
\gam_{n}(t)\geq c(t)>0,\qquad \textnormal{for all } n\geq 1.
\eeq
\end{theorem}
We emphasize that the constant $c(t)>0$ does not depend on the system size and thus the lower bound extends to the thermodynamic limit. 
The proof is analytic and rests on an effective approximate description of the ground states on finite open Motzkin chains. These ground states can be labeled by Motzkin walks with imbalances between up- and down-steps and we obtain different low-area approximations depending on whether the imbalance is high or low. 

A key point is that the result is valid for all $t<1$. Indeed, standard finite-size criteria \cite{knabe1988energy,lemm2019spectral,gosset2016local} allow to derive a gap for very small $t\ll 1$ from finite-size calculations, but as $t\uparrow 1$, a finite-size criterion at fixed size becomes progressively weaker, because the area-dampening becomes weaker. Instead, we use a finite-size criterion about ground state projectors that is related in spirit to the martingale method (Theorem \ref{thm:criterion}). This operates on blocks of length $k$ and it is important for us to exploit $k$ as an additional large parameter that we will choose large, but finite depending on how close $t$ is to $1$.

 The proof of the gap appeared in the 2022 version of this preprint \cite{andrei2022spin}. This updated version has an improved presentation and contains two further main results about the spin-$1$ Motzkin chain that we describe in the following. These two additional results are obtained by exploiting and refining the structural understanding and effective descriptions we have gained of the ground states of open Motzkin chains in proving Theorem \ref{thm:main}.

\subsection{Second main result: long-distance limit of string-order parameter}


A numerical DMRG study of the $t<1$ system \cite{barbiero2017haldane} by Barbiero et al.\ found further evidence for a spectral gap as well as a non-vanishing limit of the string order parameter
\[
O_n(i, j)  = \langle S_i^z (-1)^{\sum_{i\le l < j} S_l^z} S_j^z \rangle_n
\]
at large distances $|i-j|$. Here, the expectation is taken with respect to the unique ground state of the Motzkin Hamiltonian on $n$ sites.  We recall that non-vanishing of the string-order parameter is related to symmetry-protected topological order for one-dimensional quantum spin-$1$ chains. It was first observed for the AKLT chain  where it is also known as ``hidden string order'' and  associated with the breaking of a $\mathbb Z_2\times\mathbb Z_2$ symmetry \cite{den1989preroughening,girvin1989hidden,kennedy1992hidden,oshikawa1992hidden}. 
In \cite{barbiero2017haldane}, the authors investigated the gap and string-order numerically using DMRG. They conjectured that the long-distance limit of the string-order parameter is positive uniformly in  the system size.
(They also discussed the possibility that the spin-$1$ Motzkin chain is in a symmetry-protected topological quantum phase. We do not expect this to be the case, because the gap established by Theorem \ref{thm:main} stays open as $t\to 0$ and the $t=0$ model has tensor product ground states, making it topologically trivial. Note that it is  possible for a model to exhibit non-vanishing string-order parameter without being topologically ordered, e.g., the transverse-field Ising model in the paramagnetic phase.)

 Topological or not, the long-distance limit of the string order parameter carries significant structural information about the ground state, especially in a gapped phase where the correlations are known to decay to zero.
 
 We rigorously derive an explicit power series representation for the long-distance limit of the string-order parameter for any $t<1$. We can prove that this formula yields a non-vanishing string-order parameter for $t\leq 0.848$. The power series can be evaluated numerically at all $t$ showing a non-zero order parameter at all $t<1$ (see Figure \ref{fig:order_parameter_numerics}).

\begin{theorem}[String order parameter]\label{thm:main2}
Let $t\in (0,1)$ and let $0< \theta<\theta'< 1$. Then the limits 
    \begin{align}
O_{\mathrm{bulk}}=\lim_{n\to\infty} O_n (\lfloor\theta n\rfloor, \lfloor\theta' n\rfloor),\qquad 
   O_{\mathrm{bdry}}=\lim O_n (1, \lfloor \theta n\rfloor)
    \end{align}
    exist, are independent of $\theta,\theta'$, and are explicitly given by \eqref{eq:Bseries}.
    
    Suppose that  $t<t_{\mathrm{SOP}}$ where $t_{\mathrm{SOP}}$ is the unique zero in $[0,1]$ of  $t^4+t^3+t^2-t-1$. Then, 
    the string order parameters are non-zero in the thermodynamic limit
    \[
    O_{\mathrm{bulk}}>0,\qquad O_{\mathrm{bdry}}>0
    \]
\end{theorem}


We prove Theorem \ref{thm:main2} in Section \ref{sec:stringorder}.

We note the analytical estimate $t_{\mathrm{SOP}}>0.848$. That is, for all $t\leq 0.848$, we give an analytical direct proof that the string-order parameter is non-zero.

 \begin{figure}[t]
	\centering	
    \scalebox{0.35}{\includegraphics{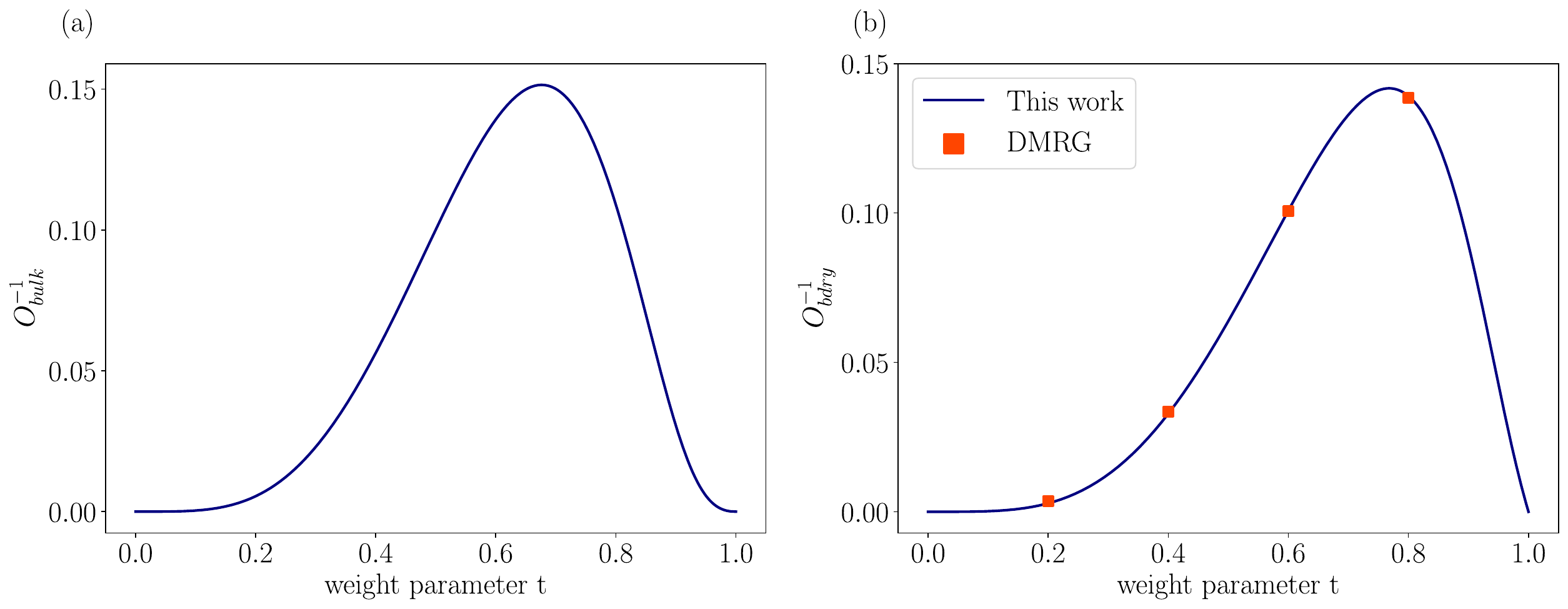}}
	\caption{Calculated values of (a) bulk, and (b) boundary limiting values of the string order parameter, versus the weight parameter $t$.     The points labeled `DMRG' were calculated by Barbiero et al.\ in \cite{barbiero2017haldane}. The curves were calculated using the power series formulae \eqref{eq:Bseries} that explicitly give the value of the string-order parameters in the thermodynamic limit for each value of $t\in(0,1)$.}
	\label{fig:order_parameter_numerics}
\end{figure}

In fact, the full statement that we prove on the string-order parameter is more informative: we derive an \textit{explicit analytical formula} for the string-order parameter in the thermodynamic limit; see  Theorem \ref{thm:sop_values}. The formula for the string-order parameter is \eqref{eq:Bseries}, which represents it as a power series in $t$ whose coefficients satisfy a simple recursion relation and decay exponentially. All of this, as well as the proof of Theorem \ref{thm:main2}, are discussed in detail in Section \ref{sec:stringorder}. Here, we would just like to emphasize that the analytical formula for the string-order parameter can be easily evaluated numerically; see Figure \ref{fig:order_parameter_numerics} and aligns with the DMRG results of \cite{barbiero2017haldane}.

\subsection{Third main result: gapless modes for $t>1$ with hard boundary conditions }
It turns out that the effective description of the unbalanced ground states can be adapted to the gapless $t>1$ phase. This allows us to answer the following open question by Klich-Zhang \cite{zhang2017entropy}:\\

\textit{``It is interesting to note the role of boundary conditions when discussing the gap. (...) However, in the proof for the colorless $s = 1$ case, the excited state we constructed is a superposition of many walks, some of those walks in the superposition have a boundary spin flipped (...) and, our bound on the gap will have a dependence on the strength of the boundary terms, which have both been set to unity. We have not addressed the question of what happens to the gap if the norm boundary terms in (2.5) are set to infinity - i.e. if only excitations consistent with a hard boundary condition are allowed. We may address this interesting question in the future.''} \cite{zhang2017entropy}\\

We answer this question in the affirmative: the gaplessness of the $t>1$ model is robust and also occurs for a hard boundary condition.

\begin{theorem}[Robust gaplessness for $t>1$ ]\label{thm:main3}
  Let $t>1$. Then
\[
\gamma_{2n}(t)\leq C n t^{-n}.
\]  
The same bound holds for the spectral gap of the Motzkin Hamiltonian with a hard boundary condition, i.e., \eqref{eq:original-hamiltonian} restricted to $\mathrm{ker}\Pi_{bdry}$.
\end{theorem}

The upper bound on $\gamma_n(t)$ in \cite{zhang2017entropy} also behaves as $t^{-n}$ to leading order, but it does not apply for the Hamiltonian with a hard boundary condition. 

We prove Theorem \ref{thm:main3} in Section \ref{sec:trial}. Since this is an \textit{upper} bound on the spectral gap, the proof only requires constructing a suitable trial state. The authors of  \cite{zhang2017entropy} construct a trial state by a boundary perturbation, which creates sensitivity to the boundary conditions. Instead, we construct a trial state that is roughly of the ``double-peak'' shape \textnormal{/\textbackslash/\textbackslash}  and thus does not see the boundary condition.  In general, a trial state calculation is  conceptually a lot simpler than the proof of a gap lower bound as in Theorem \ref{thm:main}. Nonetheless, the calculations we do with the trial state (to show approximate orthogonality to the ground state and to bound its energy) are $t>1$ analogs of combinatorial insights we found when studying the $t<1$ local ground states. In this sense, all the results fit closely together.

\subsection{Summary}

To summarize, this work resolves the following three open problems about the spin-$1$ Motzkin chain.

\begin{itemize}
\item Existence of a spectral gap for any $t<1$ as first conjectured in \cite{zhang2017novel} and numerically observed in \cite{barbiero2017haldane}.
\item Non-vanishing long-distance limit of the string order  parameter  as numerically observed and conjectured in \cite{barbiero2017haldane}.
\item Gaplessness for $t>1$ with hard boundary conditions, answering a question from \cite{zhang2017entropy}.
\end{itemize}


\subsection{Proof strategy}

The overarching idea to derive a spectral gap is to use an \textit{analytical finite-size criterion}. In recent years, related finite-size criteria have been successfully applied to other frustration-free Hamiltonians including higher-dimensional ones~\cite{abdul2020class,guo2021nonzero,lemm2019gaplessness,lemm2019gapped,lemm2019aklt,lemm2020existence,nachtergaele1996spectral,pomata2019aklt,pomata2020demonstrating,haferkamp2021improved,warzel2021bulk,warzel2021bulk,lemm2022quantitatively,jauslin2022random,mizel2023renormalization,hunter2025two,rai2026hierarchy}, but we emphasize that the verification of a finite-size criterion always requires good understanding of the Hamiltonian under investigation. Here we use a criterion based on a well-known duality lemma of Fannes-Nachtergaele-Werner \cite{fannes1992finitely} for estimating the angle between local ground spaces. Similar criteria played a central role in the recent works \cite{abdul2020class,guo2021nonzero,pomata2019aklt,pomata2020demonstrating}; see also \cite{spitzer2003improved}. The version we utilize here reduces the spectral gap problem to bounding the ground state overlap $\| G_{[k + 1, 3k]} G_{[1, 2k]} - G_{[1, 3k]}  \| < {1 \over 2}$ which roughly speaking measures the ``delocalization'' of possible excitations. See Theorem \ref{thm:criterion} for the precise statement of the finite-size criterion. The norm $\| G_{[k + 1, 3k]} G_{[1, 2k]} - G_{[1, 3k]}  \|$ can be calculated solely in terms of states on the full chain which are excited (orthogonal to the full-chain ground space), but their components on the first two-thirds of the chain are local ground states. To prove $\| G_{[k + 1, 3k]} G_{[1, 2k]} - G_{[1, 3k]}  \| < {1 \over 2}$, we ask how much these states can overlap with the ground space on the last two-thirds of the chain. If  this overlap is small, the finite-size criterion implies a spectral gap.

The proof of small overlap is the crux and it requires a detailed understanding of the finite-size ground states. In particular, the following two technical challenges arise when we decompose the Motzkin Hamiltonian into subsystems to verify the finite-size criterion and need to be addressed:
\begin{itemize}
\item The Motzkin walks are pinned to have an initial up or flat step and final down or flat step through particular boundary projectors. This leads to a particular breaking of translation-invariance and different kinds of subsystem Hamiltonians at the bulk versus boundaries.
\item The ground space of each subsystem relevant to the finite-size criterion (which naturally comes with open boundary conditions) is highly degenerate; the dimension scales like system size squared. When composing Motzkin subchains, this degeneracy leads to massive combinatorial factors which have to be a posteriori balanced by the area weight. 
\end{itemize}

To address the first point, we develop a scheme to remove the boundary projectors and reduce the derivation of the gap to the gap of the Motzkin Hamiltonian with open boundary conditions. For this, we rely on Kitaev's projection lemma \cite{kempe2006complexity} and an explicit calculation of the boundary energy penalty incurred by superpositions of Motzkin walks not satisfying the boundary conditions. See Subsections \ref{ssect:boundarypenalty} and \ref{ssect:proofmain} for the details.

To address the more difficult second point, we introduce a notion of approximate ground states to ameliorate some of the massive combinatorial issues. Our approximation scheme is based on the observation that a ground state of the open chain with $p$ unbalanced up steps and $q$ unbalanced down steps will tend to have the up-steps accumulating on the left end and the down-steps accumulating on the right end of the chain because of the exponential weighting of area. This leads us to classes of approximate ground state which have very low area. The analysis of their overlaps is central to the verification of the finite-size criterion and is decomposed into a low- and high-imbalance regime that will be explained further in due course. 

Working with these approximate ground states, which are themselves superpositions of unbalanced Motzkin walks, requires carefully tracking of the $p$- and $q$-dependence of their normalization factors. It is paramount that these satisfy an approximate factorization condition at large $k$ and we prove by a rather intricate interplay of elementary induction schemes that are guided by numerics and inspired by heuristically viewing the recursion relation of normalization factors as a spatially inhomogeneous 1D discrete diffusion equation; see Appendix \ref{sec:norm-ratios-appendix}.

The approximate description of finite-volume ground states that we gain in the process can be used to understand the Motzkin chains in additional detail beyond the gap. In particular, we use this approximate description to prove Theorem \ref{thm:main2}, which gives us a power series for the long-distance limit of the string-order parameter at any $t\in(0,1)$. Moreover, they also improve our understanding of the high-area regime $t>1$ to inspire a new choice of trial state proving Theorem \ref{thm:main3}.

For readers familiar with the literature on Motzkin spin chains, we mention that we do not reformulate the gap problem as a gap problem for classical Markov chains as was done in \cite{levine2017gap,movassagh2016gap,movassagh2016supercritical}. This has two reasons: (i) Even with the reformulation, one would still have to derive a new non-trivial classical probability result, namely existence of an order-$1$ gap for the appropriate classical Markov chain. (ii) As the above works show, this approach commonly incurs losses of factors depending on the system size. This makes it well suitable for proving that the gap closes polynomially for $t=1$ regime or exponentially in the $t>1$ regime. However, the present case is more delicate because we aim to derive an order-$1$ lower bound on the gap and such system-size dependent losses, even just logarithmic ones, can no longer be afforded. Nonetheless, it would be interesting to see an alternative derivation of our main result via purely probabilistic techniques. Conversely, our result can be used to derive the spectral gap of the corresponding classical Markov chain, which may be of independent interest.

\subsection{Comparison with other methods} Let us briefly explain why other spectral gap methods we have tried fail to give a spectral gap for the full range $t\in (0,1)$. 

First, Knabe type finite-size criteria \cite{knabe1988energy,anshu2020improved,gosset2016local,lemm2020finite,lemm2019spectral,lemm2022quantitatively} rely on explicitly calculating the gap for a small finite subsystem. We found that these can be used to derive the spectral gap for sufficiently small $t$, specifically, we have successfully verified a similar finite-size criterion \cite{lemm2019spectral} numerically for $t\leq 0.55$. However, it is clear that these can never cover the full range of $t\in (0,1)$, because any fixed finite-size gap must lie below the relevant threshold for $t=1$  and thus also for all $t$ sufficiently close to $1$.

Second, the martingale method \cite{nachtergaele1996spectral} with finite spatial overlap does not seem apply in our model for values of $t$ arbitrarily close to $1$, basically, because the relevant systems entering it have small overlap compared to their size. See the Remark \ref{rmk:MM} for further explanations.  This is different in our finite-size criterion where we can choose $k$ as a large parameter (diverging as $t\to 1$) and the relevant subsystems overlap by approximately  50\%.

\subsection{Open problems}
We close by mentioning some open problems. First, it is natural to aim to extend the result to higher spin while keeping $t\in (0,1)$ arbitrary. We expect that the combinatorial analysis in the main text can be generalized to higher spin. However, the extension will require a  new, more conceptual proof of the convergence of ratios of normalization factors in Appendix A, which is currently done by many nested inductions that would become complicated for higher spin. 

As a second open problem, we mention that the Motzkin spin chains are ``bosonic'' in the sense that the local spin number is an integer. Their ``fermionic'' siblings with half-integer spin are called Fredkin spin chains and have also been studied in detail \cite{movassagh2016gap,salberger2017deformed,udagawa2017finite,zhang2017entropy}. While the local interaction on the Fredkin side is slightly more subtle (it is $3$-local instead of nearest-neighbor as in the Motzkin case), the developments there are essentially parallel and we generally expect these models to be amenable to our technique. 

\subsection{Organization of the paper}
\label{ssect:orga}
In \textit{Section \ref{sect:main}}, we introduce the main model, the Motzkin Hamiltonian with pinning boundary conditions. We also introduce it with open boundary conditions and explain how a spectral gap with open boundary conditions implies the main result via Kitaev's projection lemma \cite{kempe2006complexity}.

In \textit{Section \ref{sect:oc-gs}}, we give a precise characterization of open-chain ground states by extending the analysis of \cite{zhang2017novel} and use this characterization to infer the important boundary energy penalty of raised ground states mentioned above.

In \textit{Section \ref{sec:criterion}}, we formulate our finite-size criterion for general frustration-free open spin chains (Theorem \ref{thm:criterion}) based on the duality projection lemma of \cite{fannes1992finitely}.

\textit{Sections \ref{sect:strategy}-\ref{sec:completing-low-imbalance-proof}} and two appendices contain the technically challenging analytical verification of the finite-size criterion for any $t<1$, i.e., the proof of Theorem \ref{thm:main}. The key point is to give an approximate description of local ground states in terms of suitable classes of dominant walks. We summarize the  proof strategy for these sections in more detail in Section \ref{sect:strategy}.

In Section \ref{sec:stringorder}, we prove Theorem \ref{thm:main2} on the string-order parameter. For this, we heavily rely on the approximations of the local ground states developed in Sections \ref{sect:strategy}-\ref{sec:completing-low-imbalance-proof}.

In Section \ref{sec:trial}, we prove Theorem \ref{thm:main3} by constructing suitable trial states and developing $t>1$ analogs of ground state representations from Sections \ref{sect:strategy}-\ref{sec:completing-low-imbalance-proof}.

\section{The model and the proof of the main result}\label{sect:main}
\subsection{The Motzkin Hamiltonian}\label{ssect:Hamiltoniandefn}
The Motzkin Hamiltonian is defined on a chain of $n$ spin-$1$ particles, so the total Hilbert space is $\bigotimes_{j=1}^n \mathbb C^3$. We label the basis states as up-, down- or null-steps, i.e., $\mathbb C^3=\mathrm{span}\{\ket{u},\ket{d},\ket{0}\}$.
Following ZAK \cite{zhang2017novel}, we define for any area weight $t> 0$, the Motzkin Hamiltonian
\begin{equation}
H_{n}(t)=\Pi_{bdry}+\sum_{j=1}^{n-1}\Pi_{j,j+1}(t) \label{eq:original-hamiltonian}
\end{equation}
with local interactions
\begin{align}
\Pi_\text{bdry}&=\ket{d}\bra{d}_1 + \ket{u} \bra{u}_{n} \label{eq:hamiltonian_bdry_term} \\ 
\Pi_{j,j+1}(t) &=  \ket{U(t)} \bra{ U(t)}_{j,j+1} + \ket{D(t)} \bra{D(t)}_{j,j+1} + \ket{\varphi(t)} \bra {\varphi(t)}_{j,j+1} \label{eq:projectors}
\end{align}
where 
\begin{align}
|U(t)\rangle &=\frac{1}{\sqrt{1+t^2}} \ob t \cdot \ket{0 u} - \ket{u 0} \cb\\
|D(t)\rangle &=\frac{1}{\sqrt{1+t^2}} \ob \ket{0 d} - t \cdot \ket{d 0} \cb\\
|\varphi(t)\rangle &=\frac{1}{\sqrt{1+t^2}} \ob \ket{u d} - t \cdot \ket{0 0} \cb
\end{align}

We recall that $H_{n}(t)$ has a unique frustration-free ground state which is an area-weighted superposition of Motzkin walks.

\begin{theorem}[\cite{zhang2017novel}]\label{thm:zhang2017novelff}
For every $t>0$, the zero eigenspace of $H_n(t)$ is spanned by the normalized ground state vector
$$
\ket {GS_{0,0}}=\frac{1}{\sqrt{N_{0,0}}}\sum_{w\in G_{0,0}} t^{\mathcal A(w)} \ket{w}.
$$
\end{theorem}

Here $G_{0,0}$ denotes the set of Motzkin walks of length $n$ and $\mathcal A$ denotes the total area under the walk.  Recall that a Motzkin walk is a discrete one-dimensional walk comprised of up-, down-, and flat steps, which only takes non-negative values, i.e., it stays above the horizontal axis.

Theorem \ref{thm:zhang2017novelff} implies that $H_{n}(t)$ is frustration-free with zero energy ground state and so its spectral gap is equal to its smallest positive eigenvalue,
$$
\gam_{n}(t)=\inf(\mathrm{spec}\, H(t)\setminus\{0\}).
$$

\subsection{The Motzkin Hamiltonian with open boundary conditions}
The following open-boundary Motzkin Hamiltonian will play a central role in our proof. We define
\begin{equation} \label{eq:translation-invariant-part}
H^0_{n}(t)=\sum_{j=1}^{n-1}\Pi_{j,j+1}(t)
\end{equation}
so that $H_{n}(t) = H_{n}^0(t) + \Pi_\text{bdry}$. We write $\gamma^{0}_n(t)$ for the spectral gap of $H^{0}_n(t)$.

As Theorem \ref{thm:degenerate-pq-gs} below shows, $H^{0}_n(t)$ has frustration-free ground states that can be explicitly described in a similar way as in Theorem \ref{thm:zhang2017novelff}. The main difference is that, without the boundary projectors, the initial and final height of the Motzkin walk are free, leading to a degeneracy of the ground space of $H^{0}_n(t)$ that is quadratic in system size. 

To characterize the ground states of $H_n^0(t)$ exactly, we introduce the following notions. First, notice that there is a one-to-one correspondence between basis states in the Hilbert space $\bigotimes_{j=1}^{n} \mathbb C^3$ and strings $\mathfrak{s}\in\{0,u,d\}^{n}$. General states thus correspond to linear combinations of such strings. On these, we define the local moves 
\begin{equation} \label{eq:local-moves}
(0u)\leftrightarrow t\text{ }(u0)\qquad t\text{ }(0d)\leftrightarrow(d0)\qquad(00)\leftrightarrow t\text{ }(ud).
\end{equation}
Notice that by making the unit of area equal to $1$, each local move changes the area by $1$.

The key idea is to introduce an equivalence relation among (scalar multiples of) strings.
\begin{definition}[Equivalence relation]
    Two strings are equivalent if and only if
	they are related by a sequence of local moves:
	\begin{equation}
	\mathfrak{s_{1}}\sim\mathfrak{s_{2}}\quad\stackrel{\mathrm{def}}{\Longleftrightarrow}\quad t^{{\mathcal A}(\mathfrak{s_{1})}}\mathfrak{s_{1}}\leftrightarrow t^{{\mathcal A}(\mathfrak{s_{2}})}\mathfrak{s_{2}},\label{eq:localMove}
	\end{equation}
	where ${\mathcal A}(\mathfrak{s})$ is the area under the walk encoded by $\mathfrak {s}$.
\end{definition}

Since the string $\mathfrak {s}$ only defines the corresponding walk up to an overall up- or down-shift, we use the convention that the walk corresponding to $\mathfrak {s}$  is the unique non-negative walk of minimal area. See Figure \ref{fig:equivalence-class} for an example and Subsection \ref{ssec:li-split-gs} for further discussion.


The ground states will be defined in terms of the following equivalence classes of the imbalanced walks.

\begin{definition}
For any integers $p,q \ge 0$ with $p + q \le n$, we define the walk
\[
g_{p,q}=( \underbrace{d, \dots, d}_{p},0,\dots ,0,\underbrace{u,\dots, u}_{q} ),
\]
where the starting height is $p$, the ending height is $q$,
and in between the imbalanced steps we have a string of all flat steps
(zeros) at zero height; (see Fig. \ref{fig:equivalence-class}).
We write $G_{p,q}$ for its equivalence class under the equivalence relation $\sim$ from \eqref{eq:localMove}.
\end{definition}

\begin{figure}[t]
\begin{center}
	\scalebox{\largehalffigsize}{\includegraphics{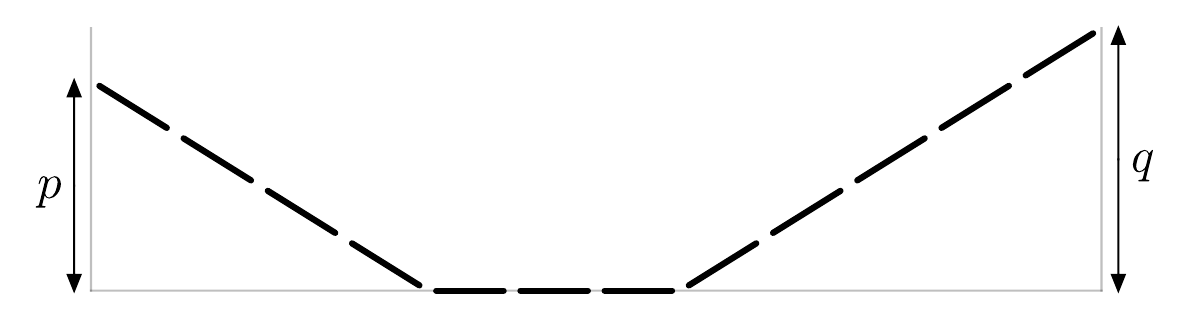}}
    \caption{\label{fig:equivalence-class}Representative walk $g_{p,q}$ in the equivalence class $G_{p,q}$. Through the local moves \eqref{eq:local-moves}, any walk in the $G_{p,q}$ class can be brought (with an appropriate area-dependent prefactor) to the form shown here.}
	\end{center}
\end{figure}




\begin{theorem}[Characterization of open-chain ground space]
    \label{thm:degenerate-pq-gs}
	The ground space of $H_n^0(t)$ is spanned by the collection of orthonormal vectors
	\begin{equation}
	|GS_{p,q}\rangle=\frac{1}{\sqrt{{N}_{p,q}}}\sum_{w\in G_{p,q}}t^{{\mathcal A}(w)}|w\rangle\label{eq:GroundStates},
	\qquad p,q\geq 0,\quad p+q\leq n.
	\end{equation}
	where ${N}_{p,q}$ is a normalization factor. 
\end{theorem} 
This theorem will be proved in Subsection \ref{ssec:ground-state-characterization} by generalizing ideas in \cite{zhang2017novel}.

The key result and main technical work of this paper is to prove that the Motzkin Hamiltonian with open boundary conditions is gapped.

\begin{theorem}[Gap with open boundary conditions] \label{thm:ff-component-gap}
	For every $t\in (0,1)$, there exists a constant $c_1(t)$ such that
	\begin{equation}
	\gamma^0_n(t) \ge c_1(t) > 0,\qquad n\geq 2.
	\end{equation}
\end{theorem}

This result will be proved by verifying a suitable finite-size criterion presented in Section \ref{sec:criterion}. 

\subsection{Boundary penalty of raised ground states}\label{ssect:boundarypenalty}

The fact that the Motzkin Hamiltonian $H(t)$ is equipped with special ``pinning'' boundary projectors makes the model highly non-translation-invariant at the boundary and not well suited for finite-size criteria. Indeed, notice that a finite-size criterion naturally concerns open boundary conditions, since it requires good understanding of subchains. (To our knowledge, the only exception to this general rule is the recent work \cite{warzel2021bulk}.)  

Therefore, we require a post-processing step to reduce the spectral gap problem for $H_{n}(t)$ to that of $H_{n}^0(t)$, i.e., to conclude Theorem \ref{thm:main} from Theorem \ref{thm:ff-component-gap}. The basic idea for this step is as follows: Given that $H_{n}^0(t)$ is gapped, it is relatively clear that the gap of $H_{n}(t)$ is mainly challenged by the possibility that members of the ground state family $\ket {GS_{p,q}}$ (described in Theorem \ref{thm:degenerate-pq-gs} above) could have small excitation energy with respect to the boundary projector, thereby closing the gap. We are able to exclude this possibility by Proposition \ref{pr:projector-min-exp} below.

\begin{definition} \label{def:lifted-states}
	Let $G^{n} \subset \oa \CF^3 \ca^{\otimes n}$ be the collection of raised ground states of $H_{n}^0(t)$, 
	\begin{equation}\label{eq:Gndefn}
	G^{n} = \text{span} \{\ket{GS_{p,q}}: p,q\geq 0,\, 1\leq p+q\}.
	\end{equation}
\end{definition}

\begin{proposition}[Boundary penalty of raised ground states] \label{pr:projector-min-exp}
	For every $t\in (0,1)$, there exists a constant $c_2(t)$ such that
	\begin{equation}
	\lambda_{min} \op \Pi_\text{bdry} |_{G^{n}} \cp \ge c_2(t) > 0,\qquad n\geq 1
	\end{equation}
\end{proposition}

We postpone the proof of Proposition \ref{pr:projector-min-exp} to Subsection \ref{ssec:account-for-boundary}.

\subsection{Derivation of spectral gap assuming spectral gap with open b.c.\ (Theorem \ref{thm:ff-component-gap})}
\label{ssect:proofmain}
Assuming Proposition \ref{pr:projector-min-exp} holds, we can conclude the spectral gap of the Motzkin Hamiltonian with pinning boundary conditions from Theorem \ref{thm:ff-component-gap} via an application of a standard argument known as Kitaev's projection lemma \cite{kempe2006complexity}.

\begin{proof}[Proof of Theorem \ref{thm:main} assuming Theorem \ref{thm:ff-component-gap}]
    Throughout the proof, we always restrict to the subspace $\ket{GS_{0, 0}}^{\perp}\subset (\CF^3)^{\otimes n}$, the orthogonal complement of the ground state $\ket{GS_{0, 0}}$. We suppress this restriction from the notation, i.e., we identify $H_n \equiv H_n \vert_{\ket{GS_{0, 0}}^\perp}$, etc. The claim that $H_n$ has a uniform spectral gap when considered on the whole space now translates to the bound
	\begin{equation}
	\lambda_{min} \op H_{n} \cp \ge c(t) > 0
	\end{equation}
	where $c(t) > 0$ should be independent of the system size $n$. Let $ 0 < \epsilon \leq 1$. We define the operator
	\begin{equation}
	H_{n}^{\epsilon} \equiv H_{n}^0 + \epsilon \Pi_\text{bdry}
	\end{equation}
	Since $H_{n} \ge H_{n}^\epsilon$ and both operators have identical ground space, it suffices to prove
	\begin{equation} \label{eq:rephrased-claim}
	\lambda_{min}\op H_n^{\epsilon} \cp  \ge c(t) > 0
	\end{equation}
	for some $\epsilon > 0$. 
	
	To this end, we apply the projection lemma \cite{kempe2006complexity} to the operator $H_n^\epsilon = H_n^0 + \epsilon \Pi_\text{bdry}$ and the subspace $G^n = \ker H^0_n\setminus \mathrm{span}\{\ket{GS_{0,0}}\}$ defined in \ref{def:lifted-states}. Thanks to Theorem \ref{thm:ff-component-gap}, we know that 
	\begin{equation}
	H^{0}_n\vert_{(G^{n})^\perp} \geq c_1(t).
	\end{equation}
	Since $\|\Pi_{bdry}\|=1$, the projection lemma says that
	$$
	\lambda_{min}(H_n^{\epsilon})\geq \epsilon \; \lambda_{min}(\Pi_{bdry}\vert_{G_{n}}) - \frac{\epsilon^2}{c_1(t) - \epsilon}.
	$$
	By Proposition \ref{pr:projector-min-exp}, we have $\lambda_{min}(\Pi_{bdry}\vert_{G^n})\geq c_2(t)$ and so
	$$
	\lambda_{min}(H^{\epsilon})\geq \epsilon\; c_2(t)-\frac{\epsilon^2}{c_1(t)-\epsilon}.
	$$
	Choosing $\epsilon>0$ sufficiently small yields the claim \eqref{eq:rephrased-claim} and hence Theorem \ref{thm:main}.
\end{proof}

\section{Characterization of open-chain ground states}\label{sect:oc-gs}

This section is structured as follows: first, in Subsection \ref{ssec:ground-state-characterization} we classify the ground states of the Hamiltonian $H_n^0$ and prove Theorem \ref{thm:degenerate-pq-gs}. Afterwards, we use the characterization to infer the boundary penalty thereby proving Proposition \ref{pr:projector-min-exp}. 

This reduces our problem to establishing a gap for $H_n^0$, i.e., to prove Theorem \ref{thm:ff-component-gap} which we shall address via the finite-size criterion presented in the next section. 

\subsection{Ground states of $H_n^0(t)$}
\label{sec:ground-states-ramis}
\label{ssec:ground-state-characterization}

The following arguments generalize the considerations used for proving Theorem 3 in \cite{zhang2017novel}. Accordingly, we will sketch them only briefly and invite the reader to consider \cite{zhang2017novel} for further details.

Recall that we can identify each state by a linear combination of strings $\mathfrak{s}\in\{0,d,u\}^{n}$. Recall also that we say two strings $\mathfrak{s}$ and $\mathfrak{t}$ are equivalent, $\mathfrak{s}\sim\mathfrak{t}$, if they are related by a sequence of local moves \eqref{eq:local-moves} and that $G_{p,q}$ is the equivalence class of the special walk $g_{p,q}$ shown in Figure \ref{fig:equivalence-class}.

\begin{lemma} \label{lm:equiv-classes}
Any string $\mathfrak{s}\in \{0,u,d\}^{n}$ belongs to a unique $ G_{p,q}$.
Any Motzkin walk belongs to $G_{0,0}$.
\end{lemma}

\begin{proof}
	The claim can be rephrased by saying that each string $\mathfrak{s}\in\{0,u,d\}^{n}$ is equivalent to a unique $g_{p,q}$. The special case of a Motzkin path is equivalent to $g_{0,0}=\ket{0\ldots 0}$ \cite{zhang2017novel}.

	 Consider a fixed string $\mathfrak{s}$ which has a total of $p_1$ up-steps and $q_1$ down-steps. These come in two categories: A subset of the up-steps occurs to the left of a down-step; we call the number of such partnered up-steps $\tilde p\leq p_1$. By applying local moves, we can merge all these partnered up- and down-steps, yielding an equivalent state with $q=q_1-\tilde p$ down-steps which are to the left of the $p=p_1-\tilde p$ up-steps. All other steps in the walk are $0$. Since the remaining down-steps have no up-steps to their left, we can move them to the left edge by successively applying the local move of swapping them with a $0$ step only. Similarly, we can move the remaining up-steps to the right edge through local moves. This procedure terminates in a scalar multiple of the walk $g_{p,q}$ which is determined by the various factors of $t$ and $t^{-1}$ obtained by applying the local moves. 
	 
	 It was shown rigorously in \cite[Proof of Theorem 3]{zhang2017novel} that the net effect of local moves is to transform the area weight consistently, i.e., if $\mathfrak s_1$ and $\mathfrak s_2$ are connected by local moves, then $t^{{\mathcal A}(\mathfrak{s_{1})}}\mathfrak{s_{1}}\leftrightarrow t^{{\mathcal A}(\mathfrak{s_{2}})}\mathfrak{s_{2}}$. In the present situation, this implies
	 $$
	 t^{\mathcal A[\mathfrak s]} \mathfrak s \leftrightarrow  t^{\mathcal A[\mathfrak g_{p,q}]}  g_{p,q}
	 $$
	 or, in other words, $\mathfrak s\sim g_{p,q}$. This proves that the string $\mathfrak{s}$ belongs to the equivalence class $G_{p,q}$. Notice that the numbers $p$ and $q$ were uniquely defined by the initial state and they are invariant under local moves. Hence, the different equivalence classes are disjoint.
	
	Finally, if $\mathfrak{s}$ was a Motzkin walk, then by definition $\tilde p=p_1=q_1$ in the beginning, leading to $p=q=0$ at the end, and so $\mathfrak{s}\in G_{0,0}$.
	\end{proof}
	

We are now ready to give the characterization of ground states.

\begin{proof}[Proof of Theorem \ref{thm:degenerate-pq-gs}] 
The key observation due to  \cite[Proof of Theorem 3]{zhang2017novel} is that the local moves \eqref{eq:local-moves} characterize the kernels (zero-energy eigenspaces) of the projectors $\Pi_{j,j+1}$ from \eqref{eq:projectors} constituting the Hamiltonian $H_n^0$.  Fix a pair of $p,q\geq 0$ with $p+q\leq n$. By construction of $\ket{GS_{p,q}}$, it is a uniform superposition of elements of the equivalence class $G_{p,q}$. Here we use that the area weights are transformed consistently by local moves as noted in the proof of Lemma \ref{lm:equiv-classes} above. This implies that $\ket{GS_{p,q}}$ lies in the kernel of all local projectors and
$$
H_n^{0}(t) \ket{GS_{p,q}}=0\;;
$$
that is, every $\ket{GS_{p,q}}$ is a frustration-free ground state of $H_n^0(t)$.

Next, we show that these are all the ground states. By frustration-freeness, any ground state $\ket{\psi}$ must be annihilated by all local projectors in $\Pi(t)_{j,j+1}$.
	Suppose we pick a specific $j$ and a string $\mathfrak{s}_1$ such that $\langle \psi,\mathfrak s_1\rangle\neq 0$. Since $\ket \psi$ is annihilated
	by $\Pi_{j,j+1}(t)$, we must have
	$$
	t^{{\mathcal A}(\mathfrak{s_{1})}}\langle \psi,\mathfrak s_1\rangle=t^{{\mathcal A}(\mathfrak{s_{2})}}\langle \psi,\mathfrak s_2\rangle
	$$
	for any $\mathfrak s_1\sim \mathfrak s_2$. Iterating this, it follows that $\psi$ has the same overlap with any area-weighted member of the equivalence class of $\mathfrak{s}_1$. By Lemma \ref{lm:equiv-classes}, this implies that $\psi$ belongs to the span of the $\ket{GS_{p,q}}$'s.
	
	It remains to prove the orthogonality of different $\ket{GS_{p,q}}$'s. For this, note that any nonzero contribution to the overlap between two states must come from them containing the same string/walk, as individual walks/strings form an orthonormal set. By the disjointness part of Lemma \ref{lm:equiv-classes}, any individual string/walk belongs to only one equivalence class $G_{p,q}$. Hence, it only contributes to a unique $\ket{GS_{p,q}}$. This establishes orthogonality and completes the proof of Theorem \ref{thm:degenerate-pq-gs}.
\end{proof}

\subsection{Boundary penalty of raised ground states}\label{ssec:account-for-boundary}
In this subsection, we prove Proposition \ref{pr:projector-min-exp} by calculating the expectation of the boundary projector $\Pi_{bdry}$ in states from the raised subspace $G^{n}$ from Definition \eqref{eq:Gndefn}.




\begin{proof}[Proof of Proposition \ref{pr:projector-min-exp}]
We begin with some standard reductions.
As stated in the proof of Theorem~\ref{thm:degenerate-pq-gs} an unbalanced space with $p,q$ extra steps is spanned by the set of strings in $G_{p,q}$, where there are $p$ unmatched step-ups and $q$ unmatched step-downs. Under the local moves any extra step-up or step-down can only exchange position with flat steps and otherwise does not participate in the local moves. Hence, we can treat the unbalanced steps on equal footing which implies that the spectrum of $H_n$ restricted to the subspace with $p,q$ extra steps depends only on $p+q$.  We drop the projector $|u\rangle\langle u|_n$, which only decreases the energy. A standard argument from \cite[Supplementary Material]{bravyi2012criticality} and \cite[see Section 3.3.3]{movassagh2016gap} allows to reduce to the case of a single unbalanced up-step, i.e.,  $p=1$ and $q=0$. Since the argument is standard and contained in \cite{bravyi2012criticality,movassagh2016gap}, we do not repeat it here and only summarize it at a high level: One labels the first unbalanced up-step by a parameter $x$ and the remaining unbalanced steps by $y$ and one drops all $y$ projectors, which only decreases the energy. This reduces to the analysis of the single $x$ parameter, which is exactly the case of a single unbalanced up-step, i.e., $p=1$ and $q=0$.

It thus remains to consider  the case $p = 1$ and $q = 0$ in which case we have to bound $\langle GS_{1,0}|\Pi_{bdry}|GS_{1,0}\rangle=|\langle GS_{1,0}|d\rangle_{1}|^{2}$.
States in $GS_{1,0}$ can be divided into $n$
subclasses. The first is $|d\rangle\otimes|{\mathcal M}_{n-1}^{t}\rangle$,
which contributes to the overlap. The remaining $n-1$ subclasses
correspond to embedding a down step at the position $2\le j \le n$
into a Motzkin walk of length $n-1$, which leads to an area weight of at most $t^{j}$. (To see that the area weight can be even smaller, take the length $n-1$
Motzkin walk, $uu...ud...dd$ and embed a $d$ anywhere in the second
half of the walk.) These considerations show that the normalization constant satisfies 
\begin{equation}
N^n_{1,0}=N^{n-1}_{0,0}+\tilde N
\end{equation}
with the bound
\begin{equation}
\tilde N \le   N^{n-1}_{0,0} \op \sum_{j=1}^{n-1}t^{j} \cp \le N^{n-1}_{0,0}\frac{t}{1-t}.
\end{equation}
Hence, we have
\begin{equation}
\lambda_{min} \op \Pi_\text{bdry} |_{G^{n}} \cp=|\langle G_{1,0}|d\rangle_{1}|^{2}=\frac{N^{n-1}_{0,0}}{N^n_{1,0}}\geq 1-t,\label{eq:GAP_Imbalanced}
\end{equation}
which is a constant independent of $n$ for any fixed $t<1$. 
\end{proof}

\section{The finite-size criterion}\label{sec:criterion}

Given the considerations above, our main task is reduced to establishing a spectral gap for any system size for the Motzkin Hamiltonian $H_n^0(t)$ with open boundary conditions.

We will achieve this by verifying a finite-size criterion for the existence of a spectral gap based on the Fannes-Nachtergaele-Werner duality lemma for pairs of projections \cite{fannes1992finitely}. The criterion works for general frustration-free quantum spin chains and is formulated in Theorem \ref{thm:criterion} below. Afterwards, we reformulate the finite-size criterion for the Motzkin spin chain by using special properties of the open-chain ground states.\\


\subsection{The finite-size criterion}
We formulate the finite-size criterion for general frustration-free quantum spin chains for the benefit of readers interested in using it elsewhere.

\begin{assumption}\label{as:criterion-2}
	Consider a one-dimensional spin chain on $n$ sites with open boundary conditions described by the following nearest-neighbor, translation-invariant Hamiltonian
	$$H_n = \sum_{i = 1}^{n - 1} h_{i,i+1}$$
	where $h_{i,i+1}$ is a positive semi-definite operator that only acts on the Hilbert spaces associated with sites $i$ and $i+1$. We assume that $H_n$ is frustration-free and we write $\gam_n$ for its spectral gap.
\end{assumption}

\begin{definition} \label{def:ground-space-projectors}
	Let $G_{[a,b]}$ be the projector onto the ground space of the part of the Hamiltonian acting between sites $a$ and $b$, namely
	\begin{equation}
	\range G_{[a,b]} = \ker \op \sum_{i = a}^{b-1} h_{i, i+1} \cp
	\end{equation}
\end{definition}

\begin{theorem}[Finite-size criterion] \label{thm:criterion}
	Given Assumption \ref{as:criterion-2}, if there exists a fixed positive integer $k$ such that
	\begin{equation}\label{eq:criterion-condition}
	\| G_{[k + 1, 3k]} G_{[1, 2k]} - G_{[1, 3k]}  \| < {1 \over 2}
	\end{equation}		
	then $H_n$ has a spectral gap in the thermodynamic limit, i.e., there exists a constant $c>0$ such that
	\begin{equation}
	\gamma_n \geq c> 0, \quad \forall n\geq 2.
	\end{equation}	
\end{theorem}

This result is inspired by the successful use of a similar finite-size criterion for two-dimensional AKLT-type systems \cite{abdul2020class,pomata2019aklt,pomata2020demonstrating}.

\begin{remark}[Comparison to the martingale method]\label{rmk:MM}
	We compare this criterion with the well-established martingale method \cite{nachtergaele1996spectral} which also requires an upper bound on an expression of the form $\| G_{X} G_{Y} - G_{X \cup Y} \|$. The martingale method is different in two ways: First, the overall sizes of $X$ and $Y$ become arbitrarily large, so it is not a finite-size criterion. Second, $X$ and $Y$ intersect only at a small number of sites, usually the range of the interaction terms within the Hamiltonian. We have found that making the intersection $X \cap Y$ significantly larger is helpful because, roughly speaking, a large overlap between the subspaces where $G_X$ and $G_Y$ act will ensure that the product $G_X G_Y$ is very close to $G_{X\cup Y}$. This will give the desired bound. On a related note, it seems that for the Motzkin Hamiltonian, the martingale method with finite overlap does not seem to apply for arbitrary $t<1$. For example, numerics show that for $t > 0.3$, the relevant condition for the martingale method fails if consecutive subsystems differ in size by 1 site. The reason is that the walks which compose the ground states can fluctuate with less and less penalty per local move as $t \to 1$, and only a large intersection between $X$ and $Y$ will ensure that such fluctuations are penalized enough such that the necessary bound is achieved. 
\end{remark}

We present an alternative to condition \eqref{eq:criterion-condition} in Theorem \ref{thm:criterion}, that will be more useful in the following sections. It involves the projections

	\begin{equation}
	E_k \equiv G_{[1,2k]} - G_{[1,3k]}
	\end{equation}

Using them, we can rephrase condition \eqref{eq:criterion-condition} as follows. 
\begin{lemma}
    \label{lm:alternative-criterion}
	Condition \eqref{eq:criterion-condition} is equivalent to
	\begin{equation}\label{eq:criterion-condition-alt}
	\| G_{[k + 1, 3k]} E_k \| < {1 \over 2}
	\end{equation}
\end{lemma}

\begin{proof}
    Frustration-freeness implies that if $A \subseteq B$, then the projectors obey $G_A G_B = G_B G_A = G_B$. Therefore $G_{[1,3k]} = G_{[k + 1,3k]} G_{[1,3k]}$ and we conclude that

	\begin{equation}
	G_{[k+1, 3k]} G_{[1,2k]} - G_{[1,3k]} = G_{[k+1, 3k]} \oa G_{[1,2k]} - G_{[1,3k]} \ca = G_{[k+1, 3k]} E_k
	\end{equation}
	which proves the lemma.
\end{proof}

\begin{proof}[Proof of Theorem \ref{thm:criterion}] For simplicity write $h_i$ for the quantity $h_{i, i + 1}$.
Given a positive integer $k$, define the sums
\begin{equation}
H_j^{(k)} = \sum_{i = jk + 1}^{(j+2)k - 1} h_i
\end{equation}
so that the $j$th sum contains projectors that involve $2k$ consecutive sites, starting from $jk+ 1$. Note that by translation invariance, all the $H_j^{(k)}$ have identical spectra.

For simplicity of notation, we assume that $n$ is divisible by $k$, say $n = k \cdot a$. (If this is not the case, then the standard trick \cite{lemm2019spectral} is to  introduce a few artificial extra sites at the right edge so that we reach an exact multiple of $k$ and add zero interactions to the Hamiltonian across the newly created edges, which we also denote by $h_i$ in the argument below.)

We investigate the following sum (see Figure \ref{fig:criterion-figure} for a pictorial representation):
\begin{align}
H^{(k)} &= \sum_{j = 0}^{a - 2} H^{(k)}_j
= 2 \cdot H_n - \op  \sum_{i = 1}^{k} h_i + \sum_{i = n - k + 1}^{n - 1} h_i  + \sum_{j = 2}^{a-1} h_{j}\cp 
\end{align}

\begin{figure}[t]
	
	\centering
	
	\scalebox{.55}{\includegraphics{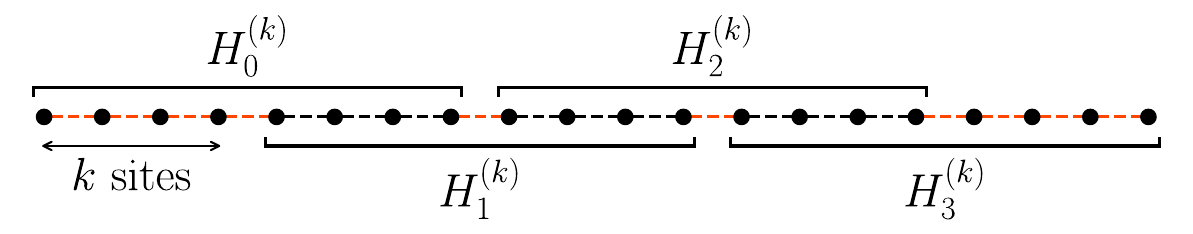}}	
	\caption{Illustration of where each term $H^{(k)}_j$ acts, for $N = 20$ and $k = 4$. The blue dots represent sites, and the dashed lines connecting them are bonds (interactions) $h_i$. The blue bonds are counted twice in the sum $H^{(k)}$, while the orange ones are only included once. This split is used to find the inequality \eqref{eq:criterion-bound2}.}
	\label{fig:criterion-figure}
	
\end{figure}

The residual terms in the parentheses on the last line are non-negative operators and so we have
\begin{equation} \label{eq:criterion-bound1}
2 H_n \ge H^{(k)}.
\end{equation}
At any finite $k$, an individual term $H^{(k)}_j$ has a finite number of eigenvalues, and therefore a finite spectral gap, which we will call $\gamma_k$. Since all $H^{(k)}_j$ have the same spectrum regardless of $j$, this value $\gamma_k$ will not depend on $j$. Moreover, by frustration-freeness $H^{(k)}_j$ has a nontrivial kernel, and its lowest eigenvalue is zero. Therefore we can write 
\begin{equation}
H^{(k)}_j \ge \gamma_k P^{(k)}_j
\end{equation}
where $P^{(k)}_j$ is the projector onto the range of $H^{(k)}_j$. Summing over $j$, we obtain
\begin{equation}
H^{(k)} = \sum_{j = 0}^{a - 2} H^{(k)}_j \ge \gamma_k S^{(k)},\qquad S^{(k)} =\sum_{j = 0}^{a - 2} P^{(k)}_j.
\end{equation}
Combining this with the previous result \eqref{eq:criterion-bound1} we get
\begin{equation}\label{eq:criterion-bound2}
H_n \ge {1 \over 2} H^{(k)} \ge {\gamma_k \over 2} S^{(k)}
\end{equation} 
From frustration-freeness we know that the kernel of a sum of projectors $h_i$ consists precisely of the states that are annihilated simultaneously by all terms. Since $H_n$ and $H^{(k)}$ consist of the same projectors $h_i$ (they only have different prefactors), we see that their kernels must be identical. Moreover, $H^{(k)}$ and $S^{(k)}$ have identical kernels by definition of the terms $P_j^{(k)}$. Together with the inequality \eqref{eq:criterion-bound2}, we find an ordering between gaps:
\begin{equation}
\gamma(H_n) \ge {1 \over 2} \gamma(H^{(k)}) \ge {\gamma_k \over 2}  \gamma(S^{(k)})
\end{equation}
and it suffices to bound the RHS term from below. To do so, we square $S^{(k)}$,
\begin{align}
\ob S^{(k)} \cb^2 
= \sum_{j = 0}^{a - 2} \ob P^{(k)}_j \cb^2 +  \sum_{j = 1}^{a - 2} \sum_{l = 0}^{j-1} \{P^{(k)}_j, P^{(k)}_l \}
\end{align}
where $\{A,B\}$ denotes the anticommutator of operators $A$ and $B$. Note that since each $P^{(k)}_j$ is a projector, we have $\ob P^{(k)}_j \cb^2 = P^{(k)}_j$ and $P^{(k)}_l P^{(k)}_j \ge 0$ if $|j-l|\geq 2$. Hence,
\begin{align}
\ob S^{(k)} \cb^2 
 &\ge S^{(k)} + \sum_{j = 0}^{a - 3} \{P^{(k)}_j, P^{(k)}_{j+1} \} 
\end{align}
For the remaining anticommutators we use \cite[Lemma 6.3]{fannes1992finitely}, which gives
\begin{equation}
\{P^{(k)}_j, P^{(k)}_{j+1} \} \ge -(P^{(k)}_j + P^{(k)}_{j+1}) \cdot  \| \oa P^{(k)}_{j+1} \ca^\perp \oa P^{(k)}_j \ca^\perp  - \oa P^{(k)}_j \ca^\perp \wedge \oa P^{(k)}_{j+1}\ca^\perp  \|
\end{equation}
By translation invariance, the operator norm does not depend on $j$, so we can focus on 
\begin{equation}
z_k \equiv \| \oa P^{(k)}_{1} \ca^\perp \oa P^{(k)}_0 \ca^\perp  - \oa P^{(k)}_0 \ca^\perp \wedge \oa P^{(k)}_{1}\ca^\perp  \|
\end{equation}
Summing over $j$, we get
\begin{equation}\label{eq:criterion-bound3}
\sum_{j = 0}^{a - 3} \{P^{(k)}_{j+1}, P^{(k)}_j \} \ge \op -  z_k \cp \op \sum_{j = 0}^{a - 3} \oa P^{(k)}_j + P^{(k)}_{j+1} \ca \cp 
\end{equation}
The sum on the RHS is almost twice $S^{(k)}$:
\begin{equation}
\sum_{j = 0}^{a - 3} \oa P^{(k)}_j + P^{(k)}_{j+1} \ca = 2 \sum_{j = 0}^{a - 2} P^{(k)}_j - \op P^{(k)}_0 + P^{(k)}_{a-2}\cp \le 2 \sum_{j = 0}^{a - 2} P^{(k)}_j = 2 S^{(k)}
\end{equation}
Together with $z_k \ge 0$, we get from \eqref{eq:criterion-bound3} that
\begin{equation}\label{eq:criterion-bound4}
\sum_{j = 0}^{a - 3} \{P^{(k)}_{j+1}, P^{(k)}_j \} \ge \op -  2 z_k \cp  S^{(k)}
\end{equation}
so that looking back to the square of $S^{(k)}$ we have
\begin{equation}
\ob S^{(k)} \cb^2 \ge (1 - 2 z_k)  S^{(k)}
\end{equation}
As $S^{(k)}$ is a non-negative operator with nontrivial kernel, this gives a lower bound on the gap, $
\gamma(S^{(k)}) \ge (1 - 2 z_k)$,
which translates into an $n-$independent lower bound on the gap of the original Hamiltonian $H_n$:
\begin{equation}
\gamma(H_n) \ge {\gamma_k \over 2}  \gamma(S^{(k)}) \ge \gamma_k  \op {1 \over 2} - z_k \cp
\end{equation}

Since $P^{(k)}_j$ projected onto the range of $H^{(k)}_j$ by definition, we see that $\oa P^{(k)}_{j}\ca^\perp$ is the ground space projector for sites $jk+1$ through $(j+2)k$, which we will denote by $G_{[jk+1,(j+2)k]}$. In this notation we have 
\begin{align}
\oa P^{(k)}_0 \ca^\perp = G_{[1, 2k]}\qquad
\oa P^{(k)}_1 \ca^\perp = G_{[k + 1, 3k]}\qquad
\oa P^{(k)}_0 \ca^\perp \wedge \oa P^{(k)}_{1}\ca^\perp = G_{[1, 3k]}
\end{align}
where the last identity follows from frustration-freeness. The necessary condition, then, is that for some $k$ we have
\begin{equation}
z_k=\| G_{[k + 1, 3k]} G_{[1, 2k]} - G_{[1, 3k]}  \| < {1 \over 2}
\end{equation}
completing the proof of Theorem \ref{thm:criterion}.
\end{proof}

In the following subsections, we consider the finite-size criterion for open Motzkin chains and reformulate in a convenient way for our later purposes.

\newcommand{\arbset}{\mathcal{Z}}
\newcommand{\unnormstate}[1]{{#1}}
\newcommand{\normstate}[1]{\widehat{#1}}
\newcommand{\area}[1]{\mathcal{A}({#1})}
\newcommand{\norm}[1]{\mathcal{N}({#1})}
\newcommand{\gsp}{G}

\subsection{Relations between Motzkin ground states on different subchains}
\label{ssec:ground-state-other-properties}

As mentioned previously, the proof will rely on constructing approximations to ground states on subsegments of our spin chain, by omitting certain classes of walks (with exponentially vanishing weights) from the superpositions \eqref{eq:GroundStates}. For simplicity, we introduce the following notations:
\begin{definition} \label{def:walk_set_definitions}
    Given any set of walks $\arbset$, defined on a subsegment of the spin chain, denote by $\ket{\unnormstate{\arbset}}$ the area-weighted superposition of all the walks in $\arbset$:
    \begin{equation}
        \ket{\unnormstate{\arbset}} \equiv \sum_{w \in \arbset} t^{\area{w}} \ket{w}
    \end{equation}
    Denote the squared norm of this state by $\norm{\arbset}$,
    \begin{equation}
        \norm{\arbset} \equiv \inner{\unnormstate{\arbset}}{\unnormstate{\arbset}} = \sum_{w \in \arbset} t^{2 \area{w}}
    \end{equation}
    Finally, the normalized version of $\ket{\unnormstate{\arbset}}$ will be represented by $\ket{\normstate{\arbset}}$:
    \begin{equation}
        \ket{\normstate{\arbset}} \equiv {\ket{\unnormstate{\arbset}} \over \sqrt{\norm{\arbset}}}
    \end{equation}
\end{definition}

\begin{remark}
    With this convention, the ground states $\ket{GS_{p,q}}$ introduced in \eqref{eq:GroundStates} are just $\ket{\normstate{G_{p,q}}}$. Furthermore, the normalization factors $N_{p,q}$ implicitly defined by the same equation are equivalent to $\norm{G_{p,q}}$. However, due to the special importance of these states and normalization factors, we will employ for them the simpler notations $\ket{GS_{p,q}}$ and respectively $N_{p,q}$.
\end{remark}

\begin{notation}
    In view of Lemma \ref{lm:alternative-criterion}, we are led to consider several different segments of Motzkin spin chains, so it will be useful to have a label to keep track of the segment under discussion. For	any such segment $S$, we let $G^S_{p,q}$ be the equivalence class defined above, i.e., the set of all walks on segment $S$, with $(p,q)$ unbalanced steps. The corresponding normalization factor will be $N_{p,q}^S = \norm{G_{p,q}^S}$, the ground state $\ket{GS^S_{p,q}} = \ket{\normstate{G_{p,q}^S}}$, etc. When the segment $S$ under discussion is clear from the context, we occasionally drop the $S$ label.
\end{notation}

\begin{notation}
	Given a chain segment $S$ and a subset of it $T \subset S$, we denote by $\gsp_T$ the ground space projector on $T$, analogously to definition \ref{def:ground-space-projectors}. If acting on states living in the Hilbert space associated with $S$ (and $T$ is a proper subset of $S$), we will understand that the operator $G_T$ acts as the identity on $\HH_{S \setminus T}$, a shorthand for $G_T \otimes I_{S \setminus T}$.
\end{notation}

Here, we state and prove two other useful properties of ground states, and ground space projectors:

\begin{proposition}[Overlap properties] \label{pr:gs-projector-diagonal-pq}
	For any $S, T$ as above, the ground space projector $G_T$, when viewed as acting on the Hilbert space $\HH_T$, is diagonal in the basis of states with definite numbers of unbalanced steps:
	\begin{equation}
	\forall \ket{a_{p,q}}, \ket{b_{p',q'}} \in \HH_T \quad \quad \quad \quad (p \ne p' \text{ or } q \ne q') \quad \implies \quad \mel{b_{p',q'}}{G_T}{a_{p,q}} = 0
	\end{equation}
	Furthermore, the above holds true even when $G_T$ is seen as acting on the Hilbert space $\HH_S$ associated with the wider segment $S$:
	\begin{equation}
	\forall \ket{a_{p,q}}, \ket{b_{p',q'}} \in \HH_S \quad \quad \quad \quad (p \ne p' \text{ or } q \ne q') \quad \implies \quad \mel{b_{p',q'}}{G_T}{a_{p,q}} = 0
	\end{equation}
	
\end{proposition}

\begin{proof}[Proof of Proposition \ref{pr:gs-projector-diagonal-pq}]
	We will begin by proving the first claim. The proof of Theorem \ref{thm:degenerate-pq-gs} implies that the $\ket{GS_{p,q}^T}$ form an orthonormal basis for the ground space on $T$. Hence, we can write
	\begin{equation}
	G_T = \sum_{r,s} \ket{GS_{r,s}^T} \bra{GS_{r,s}^T}
	\end{equation}
	with the sum running over all possible $r,s $ consistent with $T$ (i.e. $r,s \le |T|$). Since $\ket{a_{p,q}}$ contains only walks in $G^T_{p,q}$, it will only have nonvanishing overlap with $\ket{GS_{p,q}^T}$ (same argument as in the proof of Theorem \ref{thm:degenerate-pq-gs}), and so
	\begin{equation}
	G_T \ket{a_{p,q}} = \op \sum_{r,s} \ket{GS_{r,s}^T} \bra{GS_{r,s}^T} \cp \ket{a_{p,q}} = \ket{GS_{p,q}^T} \; \inner{GS_{p,q}^T}{a_{p,q}}
	\end{equation}
	Because the RHS is proportional to $\ket{GS_{p,q}^T}$, it only contains walks in $G^T_{p,q}$. Therefore the only possibility for nonvanishing overlap with $\ket{b_{p',q'}}$ is to have both $p=p'$ and $q=q'$.
	
	The following characterization will be useful for the second part of the proof: since all walks in the composition of $G_T \ket{a_{p,q}}$ still lie in the equivalence class $G_{p,q}^T$, it means that any of them can be transformed, using only local moves, into any walk contributing to the original state $ \ket{a_{p,q}}$.
	
	For the second claim, working in $\HH_S$, we again investigate how $G_T$ acts on $\ket{a_{p,q}}$. For any walk included in $\ket{a_{p,q}}$, the projector $G_T$ may only change the steps in $T$, since it acts as the identity on $S \setminus T$. We have also seen above that any such change is reversible by local moves. Therefore any walk in $G_T \ket{a_{p,q}}$ (living on the full segment $S$) can still be transformed, by local moves, into any walk from $\ket{a_{p,q}}$. We conclude that $G_T$ does not change the equivalence class of walks even when acting on the full $S$, and the second claim of the Proposition follows.
	\end{proof}

\subsection{Reformulation of the criterion for Motzkin chains}
In this section we reduce the criterion \ref{thm:criterion} to a form better suited for the translation-invariant part $H_n^0$ of the Motzkin Hamiltonian under discussion (eq. \eqref{eq:translation-invariant-part}). We define a collection of states, indexed by $k$:

\begin{definition} \label{def:phi-k-pq}
	By $\ket{\phi^{k}_{p,q}}$ we mean a state with $(p,q)$ unbalanced steps living in the Hilbert space associated with sites $[1,3k]$. When the $k$ value is clear and fixed, the corresponding label may be dropped, leaving $\ket{\phi_{p,q}}$ as the state.
\end{definition}

The main result of this section is the following:

\begin{proposition}\label{pr:reduction-of-criterion}
	With Definition \ref{def:ground-space-projectors} and \eqref{eq:criterion-condition} we have,
	\begin{equation}
	\|G_{[k + 1, 3k]} E_k \| = \sqrt{\sup\limits_{\substack{p,q \ge 0 \\ p+q \le 3k}} \op \sup\limits_{\phi_{p,q}^k \in \range E_k} \mel{\phi_{p,q}^k}{G_{[k+1, 3k]}}{\phi_{p,q}^k} \cp}
	\end{equation}
\end{proposition}

\begin{proof}[Proof of Proposition \ref{pr:reduction-of-criterion}] From the definition of the norm, 
\begin{equation}
\| G_{[k+1, 3k]} E_k \| =   \sup\limits_{\phi^k \in \HH_{[1,3k]}} \| G_{[k+1, 3k]} E_k \ket{\phi^k} \| =  \sup\limits_{\phi^k \in \HH_{[1,3k]}} \sqrt{ \mel{\phi^k}{E_k  G_{[k+1, 3k]} G_{[k+1, 3k]} E_k}{\phi^k}}
\end{equation}
Since $G_{[k+1, 3k]}$ is a projector, it squares to itself, and the above simplifies to
\begin{equation}
\| G_{[k+1, 3k]} E_k \| = \sup\limits_{\phi^k \in \HH_{[1,3k]}} \sqrt{ \mel{\phi^k}{E_k  G_{[k+1, 3k]} E_k}{\phi^k}}
\end{equation}

In the above, $\ket{\phi^k}$ is a priori an arbitrary state in the Hilbert space $\HH_{[1,3k]}$. However, without loss of generality, we can take it to be in the range of $E_k$. This is because $\HH_{[1,3k]}$ can be written as the direct sum of $\range E_k$ and its orthogonal complement. Any part of $\ket{\phi^k}$ in the orthogonal complement gets annihilated by the projector $E_k$, without contributing anything to the norm. This means we can take 
\begin{equation}
E_k \ket{\phi^k} = \ket{\phi^k} \quad \iff \quad G_{[1,2k]} \ket{\phi^k} = \ket{\phi^k}  \quad \text{and} \quad G_{[1,3k]} \ket{\phi^k} = 0
\end{equation}
and obtain the norm as
\begin{equation}
\| G_{[k+1, 3k]} E_k \| = \sup\limits_{\phi^k \in \range E_k} \sqrt{ \mel{\phi^k}{G_{[k+1, 3k]}}{\phi^k}}
\end{equation}
Since the square root is strictly increasing, we can safely take it out of the supremum to obtain 
\begin{equation}
\| G_{[k+1, 3k]} E_k \| = \sqrt{ \sup\limits_{\phi^k \in \range E_k}  \mel{\phi^k}{G_{[k+1, 3k]}}{\phi^k}}
\end{equation}
As seen in the proof of Proposition \ref{pr:gs-projector-diagonal-pq}, the ground space projector on any interval $I$ can be written as 
\begin{equation}\label{eq:ground-space-projector-decomposition}
G^I = \sum_{p,q} \ket{GS^I_{p,q}} \bra{GS^I_{p,q}}
\end{equation}
Each term selects the component of $\ket{\phi^k}$ with the corresponding number of unbalanced steps. Expanding $\ket{\phi^k}$ in terms of such components, we find
\begin{equation}
\ket{\phi^k} = \sum_{p,q} a_{p,q} \ket{\phi_{p,q}^k}
\end{equation}
where we assume each $\ket{\phi_{p,q}^k}$ is in the range of $E_k$, is normalized, and has $(p,q)$ unbalanced steps. Normalization for $\ket{\phi^k}$ requires
\begin{equation}\label{eq:phi-normalization}
1 = \inner{\phi^k}{\phi^k} = \sum_{p,q} |a_{p,q}|^2
\end{equation}
where the second equality above follows from the orthonormality of individual components: for any $p,q,p',q'$ we have $\inner{\phi_{p,q}}{\phi_{p',q'}} = \delta_{p,p'} \delta_{q,q'}$. It follows that 
\begin{equation}
\mel{\phi^k}{G_{[k+1, 3k]}}{\phi^k} = \sum_{p,q,p',q'} a_{p,q} a_{p',q'}^* \mel{\phi_{p',q'}^k}{G_{[k+1, 3k]}}{\phi_{p,q}^k} = \sum_{p,q} |a_{p,q}|^2  \mel{\phi_{p,q}^k}{G_{[k+1, 3k]}}{\phi_{p,q}^k}
\end{equation}
where the second equality follows from $\mel{\phi_{p',q'}^k}{G_{[k+1, 3k]}}{\phi_{p,q}^k} = \mel{\phi_{p,q}^k}{G_{[k+1, 3k]}}{\phi_{p,q}^k} \delta_{p,p'} \delta_{q,q'}$; the ground space projector does not have matrix elements between states with different numbers of unbalanced steps, as shown in Proposition \ref{pr:gs-projector-diagonal-pq}.
From the normalization condition \eqref{eq:phi-normalization}, and the fact that $G_{[k+1, 3k]}$ is a non-negative operator, we see
\begin{equation}
\mel{\phi^k}{G_{[k+1, 3k]}}{\phi^k} = \sum_{p,q} |a_{p,q}|^2 \mel{\phi^k_{p,q}}{G_{[k+1, 3k]}}{\phi^k_{p,q}} \le \sup\limits_{p,q} \mel{\phi^k_{p,q}}{G_{[k+1, 3k]}}{\phi^k_{p,q}}
\end{equation}
and the bound can be attained, since there is no constraint on the $a_{p,q}$ coefficients other than normalization. Choosing all of them to be zero, except for the one corresponding to the largest matrix element $\mel{\phi_{p,q}^k}{G_{[k+1, 3k]}}{\phi_{p,q}^k}$, gives equality in the above. Formally, this means
\begin{equation}
\sup\limits_{\phi^k \in \range E_k} \mel{\phi^k}{G_{[k+1, 3k]}}{\phi^k} = \sup\limits_{\substack{p,q \ge 0 \\ p+q \le 3k}} \op \sup\limits_{\phi_{p,q}^k \in \range E_k} \mel{\phi_{p,q}^k}{G_{[k+1, 3k]}}{\phi_{p,q}^k} \cp
\end{equation}
where we've made explicit the conditions on $p$ and $q$. At fixed $k$, the length of the full chain is $3k$; the number of unbalanced steps must be non-negative, and also can never be more than the total steps, so $p,q \ge 0$ and $p+q \le 3k$. We then find that
\begin{equation}
\| G_{[k+1, 3k]} E_k \| = \sqrt{\sup\limits_{\substack{p,q \ge 0 \\ p+q \le 3k}} \op \sup\limits_{\phi_{p,q}^k \in \range E_k} \mel{\phi_{p,q}^k}{G_{[k+1, 3k]}}{\phi_{p,q}^k} \cp}
\end{equation}
completing the proof.
\end{proof}

\section{Analytical  verification of the finite-size criterion: overview}\label{sect:strategy}

The remaining sections will focus on showing the following key asymptotic.

\begin{theorem}[Key asymptotic]\label{thm:asymptotic}We have
\begin{equation} \label{eq:reduced-main-claim}
\lim\limits_{k \to \infty} \op \sup\limits_{\substack{p,q \ge 0 \\ p+q \le 3k}} \op \sup\limits_{\phi_{p,q}^k \in \range E_k} \mel{\phi_{p,q}^k}{G_{[k+1, 3k]}}{\phi_{p,q}^k} \cp \cp = 0
\end{equation}
\end{theorem}

This asymptotic allows to verify the finite-size criterion and hence implies the main result.

\begin{proof}[Proof of Theorem \ref{thm:ff-component-gap} assuming Theorem \ref{thm:asymptotic}]
  Together with Proposition \ref{pr:reduction-of-criterion}, the asymptotic  \eqref{eq:reduced-main-claim} implies that  
  $$	
  \|G_{[k + 1, 3k]} E_k \|\to 0,\qquad k\to\infty.
  $$
  By Lemma \ref{lm:alternative-criterion}, this implies that there exists a large $k_0$ such that Theorem \ref{thm:criterion} applies and yields Theorem \ref{thm:ff-component-gap} (which as already seen then gives the main result).
 \end{proof}

Therefore, the task that will occupy us in the remainder of this work is to prove Theorem \ref{thm:asymptotic}. This turns out to be technically quite challenging and requires several new ideas in the analysis of the Motzkin spin chain to be completed. A special role is played by a suitable notion of approximate ground state projectors introduced in the next section. We develop a detailed description of the behavior of these approximate ground states under composition and decomposition in different physical regimes (low-imbalance versus high-imbalance). We also derive precise control on their combinatorial normalization coefficients based on rather technical estimates for a spatially inhomogeneous 1D diffusion equation. It would be interesting if this refined understanding of the open-chain ground states could be useful to study other physical properties of Motzkin spin chains.

\paragraph{Proof strategy for Theorem \ref{thm:asymptotic}}
First, notice that we can understand the quantity on the LHS of \eqref{eq:reduced-main-claim} as measuring the `delocalization' of possible excitations in the model: we look at states $\ket \phi$ on the full chain which are excited (orthogonal to the full-chain ground space), but their components on the first two-thirds of the chain are local ground states. We then ask how much these states can overlap with the ground space on the last two-thirds of the chain, and we aim to show that the answer is almost not at all. Intuitively speaking, we have to exclude the possibility that there exists some excited state on the full chain  which is very close to ground states on the first two-thirds and also the last two-thirds of the chain. This latter situation would require some form of non-localized excitation. 

The main technical difficulty arises because of the quadratic ground state degeneracy described in Theorem \ref{thm:degenerate-pq-gs}, which is the price to pay for removing the boundary projectors. This means that we must consider ground states with various $p$- and $q$-values, not only on the initial chain, but also (and this is the crux of the matter) when dividing the chain into subsegments. It is therefore imperative that we develop a simpler-to-work-with effective description, we call these approximate ground states.

Our notion of approximate ground states differs based on two main categories: states with many unbalanced steps or few ones (called high-imbalance and low-imbalance respectively). The simpler case is when there are many unbalanced steps of at least one type (up or down), i.e. more than $(1+c) k$ such steps, for some $c>0$. Then, because of the area weighting in the ground state, the outermost $k$ unbalanced steps will tend to accumulate in the corresponding outermost third of the chain. The reason is that the presence of any balanced step in the outermost third carries, in comparison with the lowest-area walk, an additional-area cost that is of order $k$. Therefore, at large $k$, a ground state on the full chain will overwhelmingly contain only unbalanced steps in one of its thirds, and can therefore be approximated by a convenient product state. This approximation is made precise in Section \ref{sec:high-imb}, and its application to obtain the desired bound is outlined in Section \ref{ssec:high-imb-app}.

On the other hand, for the case with few unbalanced steps on both sides, our approximation will rely on a Schmidt decomposition of the exact ground state, followed by a rigorous proof that most of the terms can be ignored in the large $k$ limit. The intuition is that, if we divide the full chain into three segments, all of which have length on the order of $k$, it is exponentially unlikely to have, in the composition of the full ground state, walks which do not reach the zero-height level in all three such segments separately. This can be understood, again, due to the area cost of such an extraordinary walk being on the order of $k$ larger than the minimum-area walk. At large enough $k$, the exponential weighting will suppress all such extraordinary walks. The specific details and proof of this low-imbalance approximation are presented in Section \ref{sec:low-imb}. Combining it with some technical properties of the normalizations for our states (Section \ref{sec:normalizations}, which relies on the diffusion analysis of Appendix \ref{sec:norm-ratios-appendix}), we find that the bound also holds for few unbalanced steps (Sections \ref{ssec:low-imb-app} and \ref{sec:completing-low-imbalance-proof}).

In the end, we combine all of these results and conclude the central asymptotic formula \eqref{eq:reduced-main-claim}.



\section{Low-imbalance approximations}\label{sec:low-imb}
Throughout this and the following sections, we will approximate ground states and $\ket{\phi_{p,q}}$ states by combinatorially simpler objects.

\begin{definition}\label{def:good-approximation}
	Given two collections of states indexed by $k$, call them $\{\ket{a^{(k)}}\}$ and $\{\ket{b^{(k)}}\}$, we will say that the latter \emph{superpolynomially approximates} the former if
	\begin{equation}
	\forall n \in \N: \qquad \qquad \lim\limits_{k \to \infty} \os k^n \ob 1 - |\inner{a^{(k)}}{b^{(k)}}|^2 \cb \cs = 0
	\end{equation}
\end{definition}

The ground states of our Hamiltonian can be divided into two categories: those with small, and respectively large, numbers of unbalanced steps. This classification is relative to a division of the spin chain into subsegments, and will be made precise below. In this section, we find superpolynomial approximations for the low-imbalance ground states of our Hamiltonian. The next section will similarly treat the high-imbalance states. After describing the approximation schemes for a split of the chain into two segments, we will generalize the results and apply them to the case of division into thirds, which is the relevant situation for our finite-size criterion.\\

To allow for the division of the chain into unequal segments, we make:
\begin{assumption} \label{as:low-imb-half-split-constants}
	Let $f_1$, $f_2 > 0$ be two given constants; fix a small number $b$ with $0 < b < f_1, f_2$.  We also fix two other constants $a_1$, $a_2$ such that $0< a_1 < f_1 - b$ and $0< a_2 < f_2 - b$.
\end{assumption}

	Although not used in this section, the following condition on $b$ will later be essential: we shall impose $b < {1 \over 4 \beta}$, where $\beta > 0$ is the parameter appearing in Theorem \ref{thm:mainec}. Note that $\beta$ depends only on the value of $t$, so we can indeed take it to be fixed once $t$ is specified.

\begin{notation}
    \label{not:half-split-k-segments}
	We will consider a family of spin chain segments indexed by $k$, where at each $k$ the corresponding segment contains $(f_1 + f_2) k$ sites. We will view such a segment as composed from two parts: the left one (L) of length $f_1 k$, and the right (R) one, with length $f_2 k$.
\end{notation}

\newcommand{\atgset}{\mathcal{G}}

\begin{definition} \label{def:atgs-definition}
	For every segment in Notation \ref{not:half-split-k-segments}, we define the following set of walks with $(p,q)$ unbalanced steps, from which to construct our approximation to the corresponding ground state:
    \begin{equation} \label{eq:atgset_def}
        \atgset_{p,q} = \bigcup_{r<b k} G_{p,r}^L \times G_{r,q}^R.
    \end{equation}
\end{definition}
The main result of this section is:
\begin{lemma}\label{lm:low-imbalance-approximation-trunc}
	With the Definitions \ref{def:good-approximation}, \ref{def:atgs-definition}, Assumption  \ref{as:low-imb-half-split-constants}, and Notation \ref{not:half-split-k-segments} from above, we have that, for all $p < a_1 k$ and $q < a_2 k$,	the states constructed from the sets $\atgset_{p,q}$ of eq. \eqref{eq:atgset_def} superpolynomially approximate the true ground states on the corresponding segments:
	\begin{equation}
	\forall n \in \N: \qquad \lim\limits_{k \to \infty} \os k^n \; \sup\limits_{\substack{p < a_1 k \\ q < a_2 k}} \ob 1 - |\inner{GS_{p,q}}{\normstate{\atgset_{p,q}}}|^2 \cb \cs = 0.
	\end{equation} 	
\end{lemma}

\begin{remark}
    Up to normalization, the state constructed from $\atgset_{p,q}$ is a superposition of products of ground states on each segment separately:
    \begin{equation}\label{eq:atgs-remark}
	\ket{\unnormstate{\atgset_{p,q}}} = \sum_{r < bk} \ket{\unnormstate{G^L_{p,r}}} \ket{\unnormstate{G^R_{r,q}}}.
	\end{equation}
    Taking only $r = 0$ would correspond to picking a product state; this would not suffice for a reasonable approximation of the true ground state, as the bipartite entanglement of the latter is known to grow for larger $t$. However, we will argue that picking a small fraction of the terms in the Schmidt decomposition, i.e. $r< bk$, does in fact provide a good approximation.
\end{remark}

We begin with a detailed discussion of the structure of ground states in this low-imbalance regime, followed by a two-part proof of Lemma \ref{lm:low-imbalance-approximation-trunc}.

\subsection{Splitting the ground states}\label{ssec:li-split-gs}
As discussed in Section \ref{sec:ground-states-ramis}, the ground state with ($p,q$) unbalanced steps, corresponding to a chain segment $S$, is the area-weighted superposition of all walks in $G_{p,q}^S$. Since walks are defined to have the minimal area consistent with non-negativity, they must reach zero height in at least one point. If this is the case, the starting and ending heights for a walk $w$ with $(p, q)$ unbalanced steps must be $p$ and $q$ respectively, as can be seen by performing local moves that transform $w$ into $g_{p,q}$. On the other hand, if we are given a walk that is not minimal, we can ``minimize" it by shifting it down by an appropriate amount (Fig. \ref{fig:ground-state-heights}). 

\begin{figure}[t]
	\centering	
	\subfigure[]{\scalebox{.3}{\includegraphics{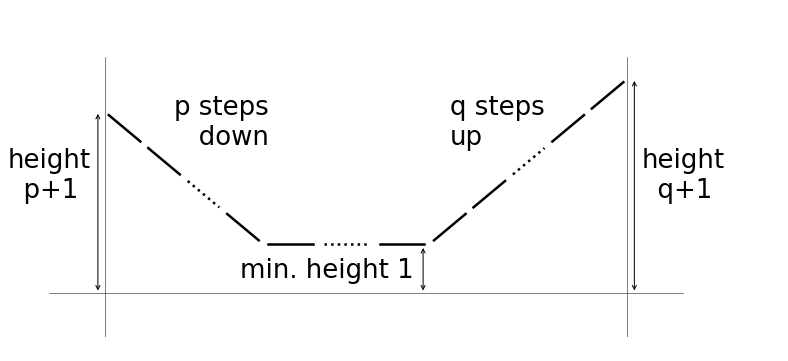}}}%
	\subfigure[]{\scalebox{.3}{\includegraphics{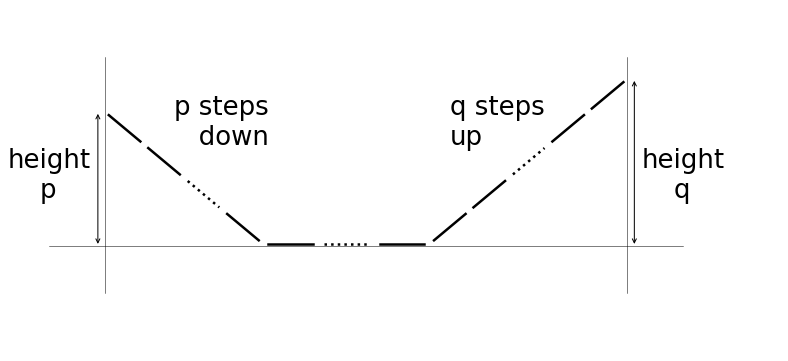}}}%
	\caption{Walks that do not reach zero height at any point (a) can be ``minimized'' by shifting them down, until they do so (b). In this latter form, the starting and ending heights are equal to the numbers of unbalanced down and up steps, respectively.}
	\label{fig:ground-state-heights}
\end{figure}

We will view $S$ as being divided into two parts $L$ and $R$, as in Notation \ref{not:half-split-k-segments}. In this case, a valid walk can either reach zero height in only one of these subsegments, or in both. Formally, we can write $G^S_{p,q}$ as the disjoint union of the following: (also see Fig. \ref{fig:three-gs-collections})
\begin{itemize}
	\item $H^{S; \; LR}_{p,q}$, containing the walks which reach zero height both in the $L$ and $R$ segments.
	\item $H^{S; \; L}_{p,q}$, containing the walks which reach zero height in the $L$ segment, but not in $R$.
	\item $H^{S; \; R}_{p,q}$, containing the walks which reach zero height in the $R$ segment, but not in $L$.
\end{itemize}  

\begin{figure}[t]
	\centering	
	\subfigure[]{\scalebox{.31}{\includegraphics{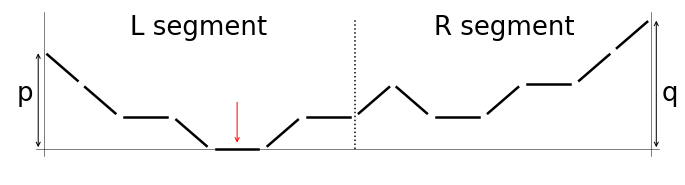}}}%
	\subfigure[]{\scalebox{.31}{\includegraphics{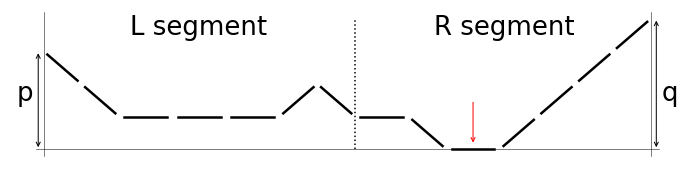}}}%
	
	\subfigure[]{\scalebox{.31}{\includegraphics{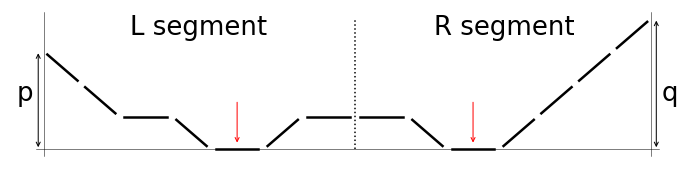}}}%
	\caption{Typical walks belonging to the subcollections $H^{L}_{p,q}$ (a), $H^{R}_{p,q}$ (b), and $H^{LR}_{p,q}$ (c). The red arrows indicate the regions of zero height that determine this classification. The three subcollections are used for the ground state decomposition in eq. \eqref{eq:three-classes-ugs}.}
	\label{fig:three-gs-collections}
\end{figure}

\begin{notation}
	Throughout the rest of the section, the segment $S$ under consideration will be understood to be as in Notation \ref{not:half-split-k-segments}, and so its corresponding label will be omitted for simplicity. The three classes above will be denoted by $H^{LR}_{p,q}$, $H^{L}_{p,q}$, and $H^{R}_{p,q}$ respectively.
\end{notation}

Using Definition \ref{def:walk_set_definitions}, write the (unnormalized) exact ground state as
\begin{equation} \label{eq:three-classes-ugs}
\ket{\unnormstate{G_{p,q}}} = \ket{\unnormstate{H^{LR}_{p,q}}} + \ket{\unnormstate{H^{L}_{p,q}}} + \ket{\unnormstate{H^{R}_{p,q}}},
\end{equation}
where the three terms on the RHS are orthogonal, since the underlying sets are disjoint. The corresponding normalization factor relation is:
\begin{equation} \label{eq:three-classes-ugs-norms}
    N_{p,q} = \norm{H^{LR}_{p,q}} + \norm{H^{L}_{p,q}} + \norm{H^{R}_{p,q}}.
\end{equation}

\subsection{The first approximation} \label{ssec:li-approximation-lemma}
We will first show that, in the low$-p$ and low$-q$ regime (more precisely, we require $p < a_1 k$ and $q < a_2 k$), the first term in the RHS of eq. \eqref{eq:three-classes-ugs-norms} dominates the other two, and we find an approximate ground state which includes only walks in $H^{LR}_{p,q}$:

\begin{lemma} \label{lm:low-imbalance-approximation-1}
	Under Assumption \ref{as:low-imb-half-split-constants}, Notation \ref{not:half-split-k-segments} and with the Definition \ref{def:good-approximation}, we have that the $\ket{\normstate{H^{LR}_{p,q}}}$ superpolynomially approximate the true ground states $\ket{GS_{p,q}}$:
	\begin{equation}
	\forall n \in \N: \qquad \lim\limits_{k \to \infty} \os k^n \; \sup\limits_{\substack{p < a_1 k \\ q < a_2 k}} \ob 1 - |\inner{GS_{p,q}}{\normstate{H^{LR}_{p,q}}}|^2 \cb \cs = 0.
	\end{equation} 	
\end{lemma}

\begin{proof}[Proof of Lemma \ref{lm:low-imbalance-approximation-1}]
From \eqref{eq:three-classes-ugs}, see that the overlap between the approximate and exact ground states is given only by the terms in $H^{LR}_{p,q}$, and it is equal to
\begin{equation}
\inner{GS_{p,q}}{\normstate{H^{LR}_{p,q}}} = {\inner{\unnormstate{G_{p,q}}}{\unnormstate{H^{LR}_{p,q}}} \over \sqrt{N_{p,q} \; \norm{H^{LR}_{p,q}}}} = {\norm{H^{LR}_{p,q}} \over \sqrt{N_{p,q} \; \norm{H^{LR}_{p,q}}}} = \sqrt{\norm{H^{LR}_{p,q}} \over N_{p,q}}.
\end{equation}
Furthermore, from \eqref{eq:three-classes-ugs-norms}, it follows that the above can be expressed as
\begin{equation}\label{eq:app-ovlp}
\inner{GS_{p,q}}{\normstate{H^{LR}_{p,q}}} = \sqrt{N_{p,q} - \norm{H^{L}_{p,q}} - \norm{H^{R}_{p,q}} \over N_{p,q}},
\end{equation}
so that 
\begin{equation}
\quad 1 -|\inner{GS_{p,q}}{\normstate{H^{LR}_{p,q}}}|^2 = {\norm{H^{L}_{p,q}} \over N_{p,q}} + {\norm{H^{R}_{p,q}} \over N_{p,q}}.
\end{equation}
Since both ${\norm{H^{L}_{p,q}} / N_{p,q}}$ and ${\norm{H^{R}_{p,q}} / N_{p,q}}$ are positive quantities by definition, it suffices to show that they separately vanish fast enough, under the given conditions.\\

We will focus on the ${\norm{H^{L}_{p,q}} / N_{p,q}}$ term, and the argument for the other one will be analogous. The normalization factor $\norm{H^{L}_{p,q}}$ contains contributions from walks that only reach zero height in the left interval, but not in the right one. The minimum height they reach within the right interval must be a positive integer, and we can classify the walks by this minimum height.

\begin{definition}
		Let $H^{L,h}_{p,q}$ be the subcollection of walks that reach zero height in $L$, but only reach minimum height $h > 0$ in $R$.
\end{definition}
Since all walks under discussion must end at height $q$ due to the condition on unbalanced steps, we see that their minimum height within $R$ cannot be more than $q$. It follows immediately that $H^{L}_{p,q}$ is the disjoint union
\begin{equation}
H^{L}_{p,q} = \bigcup_{h = 1}^q H^{L,h}_{p,q},
\end{equation}
which implies the relation of normalization factors
\begin{equation}
\norm{H^{L}_{p,q}} = \sum_{h = 1}^q \norm{H^{L, h}_{p,q}}.
\end{equation}

The intuition here is that walks with larger $h$ must enclose correspondingly large areas (for example, at least $h$ times the length of $R$, guaranteed by the minimum height condition). Since the size of $R$ is $f_2 k$ and $t < 1$, this translates into exponentially small normalization factors $t^{2 A(w)}$ when $k$ is large: $\norm{H^{L}_{p,q}} \gg \norm{H^{L, 1}_{p,q}} \gg \norm{H^{L, 2}_{p,q}} \gg \dots$; to formalize this, begin by considering the relation between $\norm{H^{L}_{p,q}}$ and $\norm{H^{L, 1}_{p,q}}$:

\begin{lemma}[Weighted image bound] \label{lm:weighted-image-bound}
	Let $\arbset$ and $\arbset'$ be finite collections of Motzkin walks, and let $\Phi:\arbset \to \arbset'$ be a map. Suppose that there are constants $\Delta>0$ and $M>0$ such that
	\begin{equation}
		\area{\Phi(w)} \le \area{w} - \Delta \qquad \text{for every } w\in \arbset,
	\end{equation}
	and every $w'\in \arbset'$ has at most $M$ preimages under $\Phi$. Then
	\begin{equation}
		\norm{\arbset} \le M t^{2\Delta} \norm{\arbset'}.
	\end{equation}
\end{lemma}

\begin{proof}
	Since $t<1$, the area decrease gives $t^{2\area{w}} \le t^{2\Delta} t^{2\area{\Phi(w)}}$ for every $w\in \arbset$. Therefore
	\begin{equation}
		\norm{\arbset}
		= \sum_{w\in \arbset} t^{2\area{w}}
		\le t^{2\Delta} \sum_{w'\in \Phi(\arbset)} |\Phi^{-1}(w')| t^{2\area{w'}}
		\le M t^{2\Delta} \norm{\arbset'},
	\end{equation}
	as claimed.
\end{proof}

\begin{proposition} \label{pr:low-imbalance-lowering-ratio}
	The ratio of $\norm{H^{L,1}_{p,q}}$ to $\norm{H^{LR}_{p,q}}$ vanishes faster than polynomially in $k$:
	\begin{equation}
	\forall n \in \N: \qquad \qquad \lim\limits_{k \to \infty} \os k^n \cdot \sup\limits_{\substack{p < a_1 k \\ q < a_2 k}} \op {\norm{H^{L,1}_{p,q}} \over \norm{H^{LR}_{p,q}}} \cp \cs = 0.
	\end{equation}
\end{proposition}

\begin{proof}[Proof of Proposition \ref{pr:low-imbalance-lowering-ratio}]
The goal is to formalize the intuition that walks in $H^{L,1}_{p,q}$ will enclose larger areas than those belonging to $H^{LR}_{p,q}$, which gives them exponentially smaller weights, because $t<1$. However, there is not a one-to-one correspondence between walks in $H^{L,1}_{p,q}$ and those in $H^{LR}_{p,q}$, and in fact $\norm{H^{L,1}_{p,q}}$ may contain significantly more distinct terms than $\norm{H^{LR}_{p,q}}$.

The strategy is to construct a mapping $H^{L,1}_{p,q} \to H^{LR}_{p,q}$ with the following two properties:
\begin{itemize}
	\item It maps any walk in $H^{L,1}_{p,q}$ to one in $H^{LR}_{p,q}$, with area smaller by at least a linear function of $k$.
	\item The number of different walks from the domain that get mapped to the same target in $H^{LR}_{p,q}$ is at most polynomial in $k$.
\end{itemize} 
Once such a mapping is constructed, the ratio $\norm{H^{L,1}_{p,q}} / \norm{H^{LR}_{p,q}}$ is bounded above by a polynomial times an exponential in $k$, which will vanish even if multiplied by an additional $k^n$ factor. To construct the map, first establish an important property of unbalanced steps:
\begin{proposition} \label{pr:unbalanced-step-height}
	When counting unbalanced up-steps from left to right in a minimized walk (cf.\ Fig. \ref{fig:ground-state-heights}), the $i^\text{th}$ such step goes from height $i-1$ to $i$. Similarly, the $j^\text{th}$ unbalanced down step goes from height $p-j+1$ to $p-j$.
\end{proposition}
\begin{proof}[Proof of Proposition \ref{pr:unbalanced-step-height}]
	An up-step $y_u$, going from height $z$ to $z+1$ within a walk, is balanced if there exists a down-step to its right, which goes between $z + 1$ and $z$. We will call the nearest such down-step $y_d$ (i.e. the leftmost one that is still to the right of $y_u$) its \emph{balancing partner}. A step is unbalanced if it has no such balancing partner.
	
	It follows that, if we have an \emph{unbalanced} up step $y$ going between $z$ and $z+1$, there is no partner to its right that goes back down to height $z$ or lower. The entire portion of the walk to the right of $y$ is only situated at heights $z + 1$ or higher. Moreover, since the step $y$ is assumed to end at height $z+1$, the portion to its right must start at this height; then, the minimum height of this portion is precisely $z+1$.
	
	As a consequence, for any $z \in \{1,2,\dots q - 1\}$ there is a unique unbalanced up step going between $z$ and $z+1$. (To see why, assume the contrary and take two such distinct steps; the rightmost one has an end at height $z$, contradicting the conclusion of the previous paragraph). The first statement of Proposition \ref{pr:unbalanced-step-height} follows, and an analogous argument also proves the second claim.  
\end{proof}


For a given value of $h$, take any walk $w \in H^{L,h}_{p,q}$, and let $s_i$ be the $i^\text{th}$ unbalanced up-step of $w$. The $s_i$ are constrained to lie in $L$ or $R$ respectively, based on the value of $i$, as follows:

\begin{proposition} \label{pr:unbalanced-steps-location}
	The first $h$ unbalanced steps $\{s_1, s_2, \dots s_h\}$ are located in the $L$ subsegment, and the other ones $\{s_{h+1}, \dots s_q\}$ are in R.
\end{proposition}
\begin{proof}[Proof of Proposition \ref{pr:unbalanced-steps-location}]
	The portion of $w$ that lies in the $R$ segment reaches minimum height $h$ by assumption, while from Proposition \ref{pr:unbalanced-step-height} we know that $s_h$ goes between heights $h-1$ and $h$. Therefore $s_h$ cannot be in $R$, and neither can all the previous unbalanced up steps $\{s_1, s_2, \dots s_{h-1}\}$; all of them must be found in $L$. On the other hand, from the proof of the same Proposition, we find that the portion of $w$ to the right of $s_{h+1}$ only lies at heights $h+1$ and above. This portion cannot contain all the steps in $R$, since by assumption some of them reach height $h$. Therefore $s_{h+1}$ must be contained in $R$. All other unbalanced up-steps $\{s_{h+2}, \dots s_q\}$ are to the right of $s_{h+1}$, so also in $R$. We conclude that $w$ has $h$ unbalanced up steps in $L$, and the other $q - h$ in $R$.	
\end{proof}

\emph{Note:} The result above, with a general value of $h$, is useful when bounding the ratio $\norm{H^{L,h}_{p,q}} / \norm{H^{L,h-1}_{p,q}}$. For the current argument, it suffices to use the $h = 1$ result, which says that walks in $H^{L,1}_{p,q}$ have a single unbalanced up step in $L$, and the others in $R$.\\

To describe the mapping process, consider an arbitrary walk $w \in H^{L,1}_{p,q}$, and let $b_1$ be its rightmost balanced step. We establish that $b_1$ is separated from $s_1$ by a number of steps that grows linearly with $k$:

\begin{proposition} \label{pr:separation-b1-s1}
	The step $b_1$ is located in the $R$ segment, and the distance $d(s_1, b_1)$ between the $s_1$ and $b_1$ steps obeys the following:
	\begin{equation} \label{eq:distance-bound}
	d(s_1, b_1) > f_2 k - q > (f_2 - a_2) k
	\end{equation}
\end{proposition} 
\begin{proof}[Proof of Proposition \ref{pr:separation-b1-s1}]
	There are $f_2 k$ total steps in $R$ but, by Proposition \ref{pr:unbalanced-steps-location}, only $q - 1$ unbalanced ones. Since we have $q < a_2 k < f_2 k$ by assumption, there must also exist balanced steps in $R$. In particular, since it is the rightmost one, $b_1$ is in $R$. As there are only $q-1$ unbalanced up steps in that subsegment, $b_1$ must be at most $q - 1$ positions away from the rightmost end of the chain. Meanwhile, $s_1$ is in $L$, so it is at least the size of $R$ (namely, $f_2 k$) positions away from the right end of the chain. Therefore the distance between $s_1$ and $b_1$ is bounded below by $(f_2 - a_2) k$, as claimed.
\end{proof}

Since any balanced up step has a down partner to its right, which is also balanced itself, the last balanced step ($b_1$) cannot be up; it may only be flat or down. See Fig. \ref{fig:typical-hl1-walk} for an example.
	
\begin{figure}[t]
		
		\centering
		
		\scalebox{.37}{\includegraphics{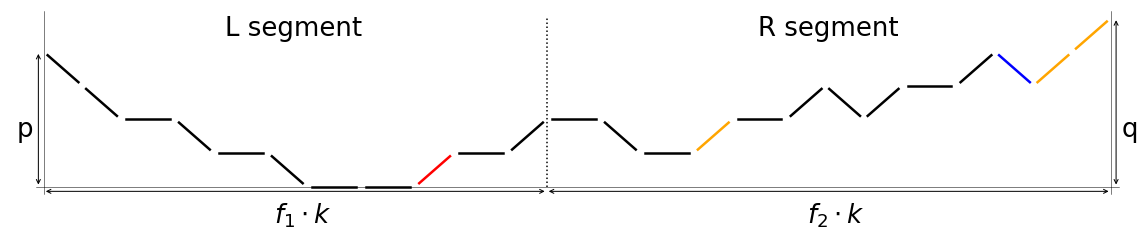}}	
		\caption{Typical walk in $H^{L,1}_{p,q}$. The leftmost unbalanced up step $s_1$ is in the $L$ segment, colored in red. The other unbalanced up-steps are in orange. The rightmost balanced step $b_1$ is blue.}
		\label{fig:typical-hl1-walk}
		
\end{figure}	
	
For the mapping, we will need $b_1$ to be flat. If it is down instead, find its flattening partner $b_2$ (which must be an up-step to its left), and replace them both by flat steps; see Fig.~\ref{fig:h1_single_side_zero} (a). This procedure does not affect the number of unbalanced steps that the walk $w$ has, nor its minimum height in the $R$ segment, and therefore all the previous conclusions are still valid.
	
	
	
	

\begin{figure}[t]
	\centering	
	\scalebox{.4}{\includegraphics{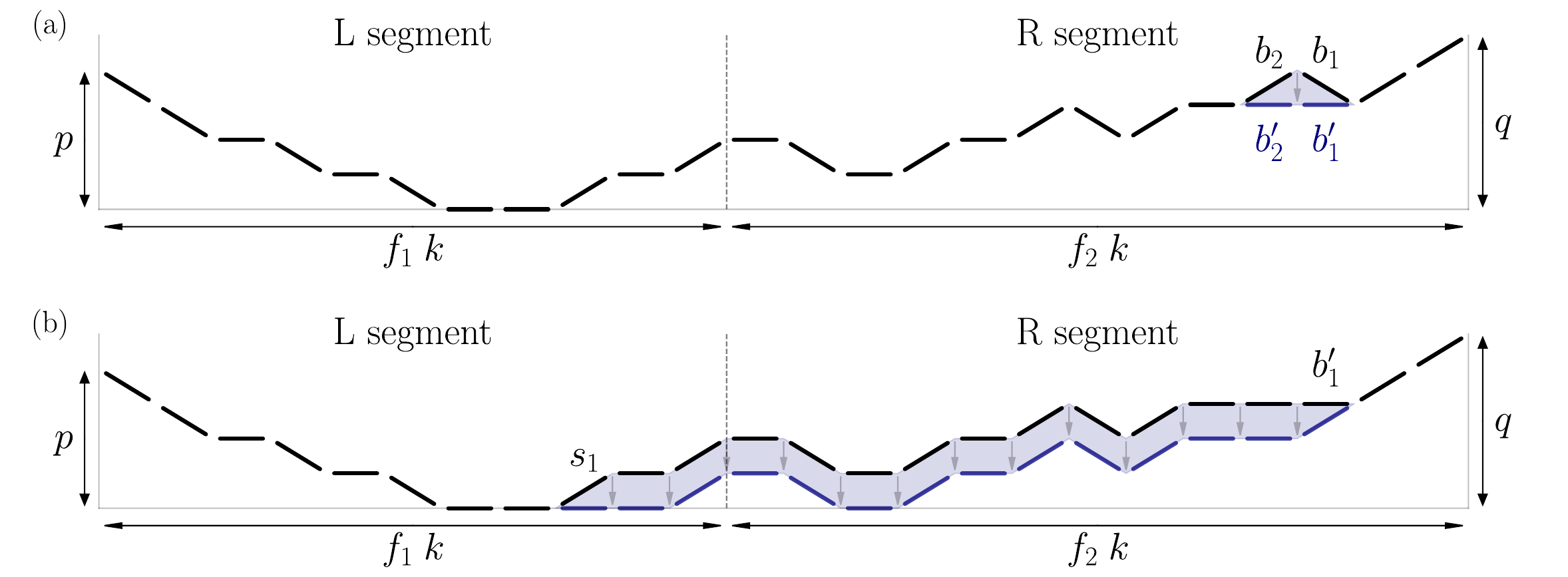}}%
	\caption{Swap argument in the low-imbalance case. In (a), the rightmost balanced step $b_1$ is flattened together with its balanced partner $b_2$ to $b_1'$ and $b_2'$, respectively. This reduces the area by the shaded part. In (b), the unbalanced step $s_1$ on the left segment is swapped with the rightmost flat step $b_1'$, which shifts the shaded region shifts down by one. }
	\label{fig:h1_single_side_zero}
    \end{figure}

Observe that all the points where the walk $w$ reaches height 1 must either belong to $b_1$ or be to its left. This is because all the steps to the right of $b_1$ are ascending by assumption, so they will never return to the height where $b_1$ is located. This height must be at least 1 by the assumption that $w \in H^{L,1}_{p,q}$.
	
The main operation is exchanging the steps $s_1$ and $b_1$. Since $s_1$ was up but $b_1$ was flat, this swap will lower the height of the portion between them by one unit. Everything else will remain at the same levek; see Fig. \ref{fig:h1_single_side_zero} (b). Call the resulting walk $w'$.

		
		
	
\begin{proposition} \label{pr:mapping-is-correct}
	The walk $w'$ obtained through the process described above belongs to the collection $H_{p,q}^{LR}$.
\end{proposition}	
\begin{proof}[Proof of Proposition \ref{pr:mapping-is-correct}]
	First we establish that $w'$ still has the same numbers $(p,q)$ of unbalanced steps. Since for a minimized walk these are equal to the starting and ending heights, and the endpoints of our walk are not affected by the swap, it suffices to argue that $w'$ is still minimized. Namely, we argue that the minimal overall height of $w'$ is still zero. This is true because the section that got shifted down was to the right of $s_1$, so by Proposition \ref{pr:unbalanced-step-height} it had a minimum height of 1 before the shift. After the shift, this minimum height will be reduced by one unit, to zero. The rest of the walk was not changed, and since $w$ was minimized, no part of it went below zero height. Therefore the overall minimal height of $w'$ is also zero, so $w'$ is indeed minimized.
	
	The second property that we must check is that $w'$ reaches zero height within both the $L$ and $R$ segments. It has been argued above that all the points where $w$ reached a height of 1 must have been to the left of $b_1$. At least one such point must have been in $R$ by the assumption $w \in H^{L,1}_{p,q}$, so in particular it was to the right of $s_1$, i.e. in the section that got shifted down. After the swap it is found at zero height, and so $w'$ now reaches zero height within $R$. On the other hand, the left endpoint of $s_1$, which lies in $L$, had been at zero height by Proposition \ref{pr:unbalanced-step-height}. That point is not affected by the swap, so $w'$ also reaches zero height in $L$ and the proof is complete. 
\end{proof}
	
Note that it is straightforward to generalize the above to a mapping $H^{L,h+1}_{p,q} \to H^{L,h}_{p,q}$. Locate the rightmost balanced step, flatten it (along with its partner) if required, and then swap it with $s_1$.

Now we turn to analyzing the area of $w'$. The portion that got lowered by one unit of height was seen in Prop. \ref{pr:separation-b1-s1} to have length larger than $(f_2 - a_2) k$, and so
\begin{equation}\label{eq:area-bounds}
A (w') < A(w) - (f_2 - a_2) k
\end{equation}
Observe that if the extra flattening step is performed before the swap, this only gives a further reduction of the area (Fig. \ref{fig:h1_single_side_zero}) and so the bound above still holds true. In any case, the described swap procedure only changes two (if no flattening is needed) or three (including flattening) steps of the original walk $w$.

\begin{figure}[t]
	\centering	
	\subfigure[]{\scalebox{.37}{\includegraphics{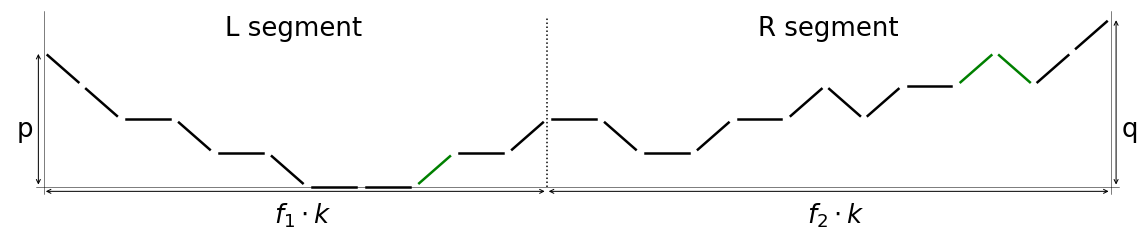}}}%
	
	\subfigure[]{\scalebox{.37}{\includegraphics{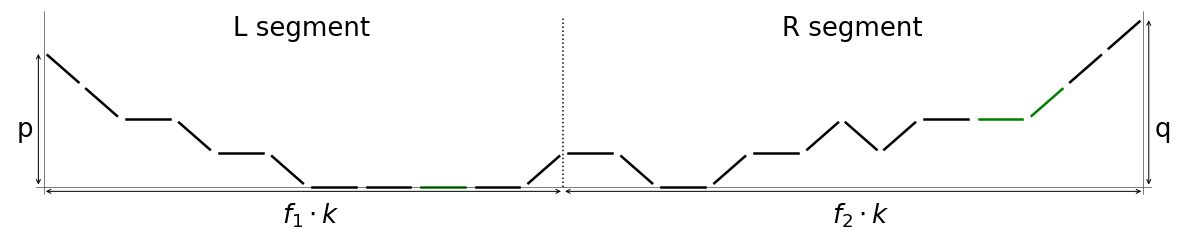}}}%
	
	\caption{The initial walk $w$ (a) and the fully modified one $w'$ (b). Only the three steps shown in green have been changed, but the enclosed area has been greatly reduced. This means that $w$ contributes less than $w'$ to the ground state (eq. \eqref{eq:low-imb-bound-g}).}
	\label{fig:hl1-lowering-summary}
\end{figure}
	
The constructed mapping is not injective, as several different choices of $w$ can lead to the same resulting $w'$. To obtain a valid relation between normalization factors, we must bound the cardinality of the preimage for an arbitrary $w'\in H^{LR}_{p,q}$. We have established that the mapping changes at most three steps in the entire walk. So there are at most ${(f_1 + f_2) k \choose 3}$ choices for the locations at which changes are operated. Each change is uniquely specified, and therefore no more than ${(f_1 + f_2) k \choose 3}$ elements of $H^{L,1}_{p,q}$ get mapped to the same target in $H^{LR}_{p,q}$.

Applying Lemma \ref{lm:weighted-image-bound} with $\arbset=H^{L,1}_{p,q}$, $\arbset'=H^{LR}_{p,q}$, $\Delta=(f_2-a_2)k$, and $M={(f_1+f_2)k\choose 3}$ yields
\begin{equation}\label{eq:low-imb-bound-p}
{\norm{H^{L,1}_{p,q}} \over \norm{H^{LR}_{p,q}}} \le {(f_1 + f_2) k \choose 3} t^{2 k (f_2 - a_2)} < (f_1 + f_2)^3 \cdot k^3 t^{2 k (f_2 - a_2)}.
\end{equation}
This holds uniformly in $p < a_1 k$ and $q < a_2 k$, so
\begin{equation}\label{eq:low-imb-bound-p-sup}
\sup\limits_{\substack{p < a_1 k \\ q < a_2 k}} \op {\norm{H^{L,1}_{p,q}} \over \norm{H^{LR}_{p,q}}} \cp < (f_1 + f_2)^3 \cdot k^3 t^{2 k (f_2 - a_2)}.
\end{equation} 
Since $f_2 - a_2 > 0$ and $t<1$, the RHS of the above goes to zero exponentially fast as $k \to \infty$, completing the proof of Proposition \ref{pr:low-imbalance-lowering-ratio}.
\end{proof}
	
As suggested at various points in the proof, the result generalizes to higher values of $h$:
\begin{corollary} \label{pr:low-imbalance-lowering-ratio-general-h}
	For every $h\ge 1$,
	\begin{equation}\label{eq:low-imb-bound-g}
	\norm{H^{L,h+1}_{p,q}}
	\le (f_1 + f_2)^3 \cdot k^3 t^{2 k (f_2 - a_2)} \norm{H^{L,h}_{p,q}}
	\end{equation}
	uniformly over $p<a_1k$ and $q<a_2k$.
\end{corollary}

\begin{proof}
	The map is the same lowering map as in the proof of Proposition \ref{pr:low-imbalance-lowering-ratio}, applied to a walk whose minimum height in $R$ is $h+1$. It lowers the relevant interval by one unit, thereby decreasing the minimum height in $R$ from $h+1$ to $h$, while changing at most three steps. The area drop is again at least $(f_2-a_2)k$, and the preimage bound is unchanged. Lemma \ref{lm:weighted-image-bound} gives the claim.
\end{proof}

We use Corollary \ref{pr:low-imbalance-lowering-ratio-general-h} to finish the proof of Lemma \ref{lm:low-imbalance-approximation-1}. For $k$ large enough, the coefficient on the RHS of \eqref{eq:low-imb-bound-g} is below $1$, and therefore $\norm{H^{L,h+1}_{p,q}}\le \norm{H^{L,h}_{p,q}}$ for all $h\ge 1$. Since $H^L_{p,q}=\bigsqcup_{h=1}^q H^{L,h}_{p,q}$ and $q<f_2k$, Proposition \ref{pr:low-imbalance-lowering-ratio} implies
\begin{equation}
\norm{H^L_{p,q}}
= \sum_{h=1}^q \norm{H^{L,h}_{p,q}}
\le q\, \norm{H^{L,1}_{p,q}}
< f_2 (f_1+f_2)^3 k^4 t^{2k(f_2-a_2)}\norm{H^{LR}_{p,q}}.
\end{equation}
The analogous right-bound gives
\begin{equation}
\norm{H^R_{p,q}}
< f_1 (f_1+f_2)^3 k^4 t^{2k(f_1-a_1)}\norm{H^{LR}_{p,q}}.
\end{equation}
Since $N_{p,q}\ge \norm{H^{LR}_{p,q}}$, it follows that
\begin{equation}
{\norm{H^L_{p,q}} \over N_{p,q}}
< f_2 (f_1+f_2)^3 k^4 t^{2k(f_2-a_2)},
\qquad
{\norm{H^R_{p,q}} \over N_{p,q}}
< f_1 (f_1+f_2)^3 k^4 t^{2k(f_1-a_1)}.
\end{equation}
Combining these estimates with the overlap identity above gives, uniformly for $p<a_1k$ and $q<a_2k$,
\begin{equation}
1 - |\inner{GS_{p,q}}{\normstate{H^{LR}_{p,q}}}|^2
< f_2 (f_1+f_2)^3 k^4 t^{2k(f_2-a_2)}
+ f_1 (f_1+f_2)^3 k^4 t^{2k(f_1-a_1)}.
\end{equation}
Both terms on the RHS decay exponentially fast in $k$, and hence
\begin{equation}
\forall n \in \N: \qquad
\lim\limits_{k \to \infty} \os k^n \cdot \sup\limits_{\substack{p<a_1k \\ q<a_2k}} \ob 1 - |\inner{GS_{p,q}}{\normstate{H^{LR}_{p,q}}}|^2 \cb \cs = 0.
\end{equation}
\end{proof}

We will now use the result of Lemma \ref{lm:low-imbalance-approximation-1} to prove Lemma \ref{lm:low-imbalance-approximation-trunc}. For the proof, we will require a factorization of the walk sets in $H^{LR}_{p,q}$.

\begin{proposition} \label{pr:decomposition-of-walks}
	Given Assumption \ref{as:low-imb-half-split-constants} and Notation \ref{not:half-split-k-segments}, let $I^r_{p,q}\subset H^{LR}_{p,q}$ be the set of walks whose height at the $L|R$ interface is $r$. Then
	\begin{equation} \label{eq:breakup-of-hlr}
	H^{LR}_{p,q}=\bigsqcup_r I^r_{p,q}.
	\end{equation}
	Moreover, concatenation gives a bijection
	\begin{equation}
	G^L_{p,r}\times G^R_{r,q}\longrightarrow I^r_{p,q}.
	\end{equation}
	Consequently,
	\begin{equation} \label{eq:sum-in-low-approximate-gs}
	\ket{\unnormstate{H^{LR}_{p,q}}}
	= \sum_r \ket{\unnormstate{G^L_{p,r}}}\otimes \ket{\unnormstate{G^R_{r,q}}}
	= \sum_r \ket{\unnormstate{G^L_{p,r}}}\ket{\unnormstate{G^R_{r,q}}},
	\end{equation}
	and
	\begin{equation} \label{eq:normalization_factorization}
	\norm{H^{LR}_{p,q}} = \sum_r N^L_{p,r}\,N^R_{r,q}.
	\end{equation}
\end{proposition}

\begin{proof}[Proof of Proposition \ref{pr:decomposition-of-walks}]
	The decomposition into the disjoint union \eqref{eq:breakup-of-hlr} follows by classifying each walk according to its height $r$ at the separation point between $L$ and $R$. Here $r$ runs from $0$ to $\min(|L|-p,|R|-q)$; since $|L|=f_1k$, $|R|=f_2k$, $p<a_1k<(f_1-b)k$, and $q<a_2k<(f_2-b)k$, this range contains all $r<bk$.

	Every walk in $I^r_{p,q}$ starts from height $p$ on the left, reaches height $r$ at the interface, and ends at height $q$, while also reaching zero height within both $L$ and $R$. It can therefore be viewed as the concatenation of a walk $w_1\in G^L_{p,r}$ and a walk $w_2\in G^R_{r,q}$. Conversely, any pair $(w_1,w_2)\in G^L_{p,r}\times G^R_{r,q}$ concatenates to a walk in $I^r_{p,q}$. This proves the bijection.

	If $w=w_1+w_2$ is such a concatenation, then $\ket{w}=\ket{w_1}\ket{w_2}$ and, because the two walks have matching interface height, the areas add:
	\begin{equation}
	\area{w_1+w_2}=\area{w_1}+\area{w_2}.
	\end{equation}
	Using Definition \ref{def:walk_set_definitions}, this gives \eqref{eq:sum-in-low-approximate-gs}. Taking squared norms of both sides, and using orthogonality of the different interface-height sectors, gives \eqref{eq:normalization_factorization}.
\end{proof}

\begin{definition} \label{def:low-imbalance-truncated-set}
	With $I^r_{p,q}$ as in Proposition \ref{pr:decomposition-of-walks}, define
	\begin{equation}
	H^{LR,<b}_{p,q}:=\bigsqcup_{r<bk} I^r_{p,q}.
	\end{equation}
	Under the concatenation bijection of Proposition \ref{pr:decomposition-of-walks}, this is the same truncated set as $\atgset_{p,q}$ from Definition \ref{def:atgs-definition}.
\end{definition}

\begin{remark}
	The reason why it was imposed in the first place that $p < a_1 k < (f_1 - b) k $ was to allow for the existence of ground states on $L$ with $(p,r)$ unbalanced steps for all $r < bk$. The same goes for $R$ and ground states with $(r,q)$ unbalanced steps.
\end{remark}

\subsection{Proof of the section's main result}

\begin{proof}[Proof of Lemma \ref{lm:low-imbalance-approximation-trunc}]
By Definition \ref{def:low-imbalance-truncated-set}, $H^{LR,<b}_{p,q}\subset H^{LR}_{p,q}$. Hence Definition \ref{def:walk_set_definitions} gives the overlap identity
\begin{equation}
\left|\inner{\normstate{H^{LR}_{p,q}}}{\normstate{H^{LR,<b}_{p,q}}}\right|^2
= {\norm{H^{LR,<b}_{p,q}} \over \norm{H^{LR}_{p,q}}}.
\end{equation}
Equivalently,
\begin{equation}\label{eq:low-imb-trunc-tail}
1-\left|\inner{\normstate{H^{LR}_{p,q}}}{\normstate{H^{LR,<b}_{p,q}}}\right|^2
= {\sum_{r\ge bk} \norm{I^r_{p,q}} \over \sum_r \norm{I^r_{p,q}}}.
\end{equation}

We now estimate the tail in \eqref{eq:low-imb-trunc-tail}. The concatenation bijection in Proposition \ref{pr:decomposition-of-walks} gives
\begin{equation}
\norm{I^r_{p,q}}=N^L_{p,r}N^R_{r,q}.
\end{equation}
There is also a direct lowering map from $I^{r+1}_{p,q}$ to $I^r_{p,q}$. Take a walk $w\in I^{r+1}_{p,q}$, view it as a concatenation of $w_1\in G^L_{p,r+1}$ and $w_2\in G^R_{r+1,q}$, and locate the leftmost unbalanced up step of $w_1$ together with the rightmost unbalanced down step of $w_2$. By Proposition \ref{pr:unbalanced-step-height}, the portion between these two steps lies at heights at least $1$, so flattening those two steps lowers the central peak by one unit. This maps $w$ to a walk in $I^r_{p,q}$, decreases the area by at least $2r-1$, and changes only two steps.

The preimage of any target walk is therefore bounded by ${(f_1+f_2)k\choose 2}$. Applying Lemma \ref{lm:weighted-image-bound} with $\arbset=I^{r+1}_{p,q}$, $\arbset'=I^r_{p,q}$, and $\Delta=2r-1$ gives
\begin{equation} \label{eq:low-imb-central-peak-contribution}
{\norm{I^{r+1}_{p,q}} \over \norm{I^r_{p,q}}}
< {(f_1+f_2)k\choose 2} t^{2(2r-1)}.
\end{equation}
If $r\ge bk$, then the exponent grows linearly with $k$. Bounding the combinatorial factor by $(f_1+f_2)^2k^2$ and summing over at most $(f_1+f_2)k$ possible values of $r$, we find
\begin{equation}
\sup\limits_{\substack{p<a_1k \\ q<a_2k}}
\ob 1-\left|\inner{\normstate{H^{LR}_{p,q}}}{\normstate{H^{LR,<b}_{p,q}}}\right|^2 \cb
< (f_1+f_2)^3 k^3 t^{4bk-2}.
\end{equation}
Due to the exponential factor,
\begin{equation}
\forall n\in \N: \qquad
\lim\limits_{k\to\infty}\os k^n\cdot
\sup\limits_{\substack{p<a_1k \\ q<a_2k}}
\ob 1-\left|\inner{\normstate{H^{LR}_{p,q}}}{\normstate{H^{LR,<b}_{p,q}}}\right|^2 \cb\cs=0.
\end{equation}
Combining this with Lemma \ref{lm:low-imbalance-approximation-1} and Lemma \ref{lm:approx-of-approx}, and using that $H^{LR,<b}_{p,q}$ agrees with $\atgset_{p,q}$, gives
\begin{equation}
\forall n \in \N: \qquad \qquad
\lim\limits_{k \to \infty} \os k^n \cdot \sup\limits_{\substack{p < a_1 k \\ q < a_2 k}} \ob 1 - |\inner{GS_{p,q}}{\normstate{\atgset_{p,q}}}|^2 \cb \cs = 0,
\end{equation}
which is the claim.
\end{proof}

\section{High-imbalance approximations} \label{sec:high-imb}
Having covered the regime of low-$p$ and low-$q$, it remains to find a complementary approximation lemma, which applies when either the high-$p$ or the high-$q$ regime is present. This situation is simpler, since the ground state will be dominated by walks whose unbalanced steps are pushed outwards, and so it approximately factorizes. The setup is the following:

\begin{assumption} \label{as:high-imb-half-split-constants}
	Let $f_1$, $f_2 > 0$ be given constants, and fix some small $c$ with $0 < c < 1/4 - 2b$.
\end{assumption}

Notation \ref{not:half-split-k-segments} remains in the same form.

\begin{definition} \label{def:pgs-definition}
	For every segment in Notation \ref{not:half-split-k-segments} we define an approximate ground state with $(p,q)$ unbalanced steps, which will be useful in the high-$q$ regime, as
	\begin{equation}\label{eq:pgsr-definition}
	\ket{PGS^{(R)}_{p,q}} =  \ket{GS^L_{p, q - f_2 k}} \otimes \oa \ket{u} \ca^{ \otimes f_2 k}
	\end{equation}
	where the $k$ label was suppressed for simplicity in the naming of all states. The main property of this state is that it contains exclusively up steps in the $R$ segment. The analogous state which is useful in the high-$p$ regime is
	\begin{equation}\label{eq:pgsl-definition}
	\ket{PGS^{(L)}_{p,q}} =  \oa \ket{d} \ca^{ \otimes f_1 k} \otimes \ket{GS^R_{p - f_1 k, q}}
	\end{equation}	
	with exclusively down steps in the $L$ segment.
\end{definition}

The result of this section is:

\begin{lemma} \label{lm:approx-high}
	Given Assumption \ref{as:high-imb-half-split-constants} and Notation \ref{not:half-split-k-segments}, and with Definitions \ref{def:good-approximation} and \ref{def:pgs-definition}, the product state in \eqref{eq:pgsr-definition} superpolynomially approximates the true ground state when $q>(f_2 + c)k$:
	\begin{equation} \label{eq:high-q-approximation-lemma}
	\forall n \in \N: \qquad \qquad \lim\limits_{k \to \infty} \os k^n \cdot \sup\limits_{q > (f_2 + c) k} \ob 1 - |\inner{GS_{p,q}}{PGS^{(R)}_{p,q}}|^2 \cb \cs = 0
	\end{equation}
	and similarly the state in \eqref{eq:pgsl-definition} is a good approximation when $p > (f_1 + c)k$:
	\begin{equation} \label{eq:high-p-approximation-lemma}
	\forall n \in \N: \qquad \qquad \lim\limits_{k \to \infty} \os k^n \cdot \sup\limits_{p > (f_1 + c) k} \ob 1 - |\inner{GS_{p,q}}{PGS^{(L)}_{p,q}}|^2 \cb \cs = 0
	\end{equation}	 
\end{lemma}

We begin with a preliminary discussion, and then proceed to prove the lemma. In what follows we will only discuss the high$-q$ case and prove eq. \eqref{eq:high-q-approximation-lemma}, as the argument for \eqref{eq:high-p-approximation-lemma} in the high-$p$ regime will be identical.\\

\subsection{Splitting the ground states}\label{ssec:hi-split-gs}
For any walk $w$ in the high-$q$ regime, all heights reached within the $R$ segment are relatively large (i.e. are bounded below by $ck$). Placing a flat or down step in this segment (as opposed to an up one) will carry a significant additional area cost. For this reason, we expect walks that do not exclusively contain up steps within $R$ to be exponentially suppressed. To make this precise, work with Notation \ref{not:half-split-k-segments}. We classify walks by the number of balanced steps they contain in the $R$ segment:
\begin{definition}
	Let $G^{S;\,z}_{p,q}$ be the collection of walks in $G^S_{p,q}$ that contain exactly $z$ balanced steps in the $R$ segment (see Figure \ref{fig:high-imbalance-subcollections}). Throughout the rest of the section, we will only take the segment $S$ as in Notation \ref{not:half-split-k-segments}, and we omit the corresponding $S$ label, to simplify the notation as $G^{z}_{p,q}\subset G_{p,q}$.
\end{definition}

\begin{figure}[t]
	\centering	
	\subfigure[]{\scalebox{.31}{\includegraphics{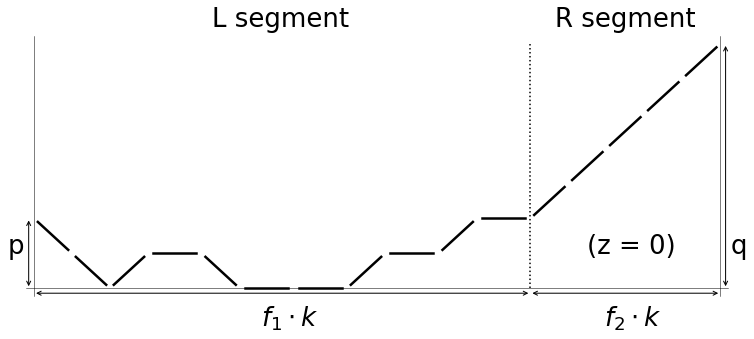}}}%
	\subfigure[]{\scalebox{.31}{\includegraphics{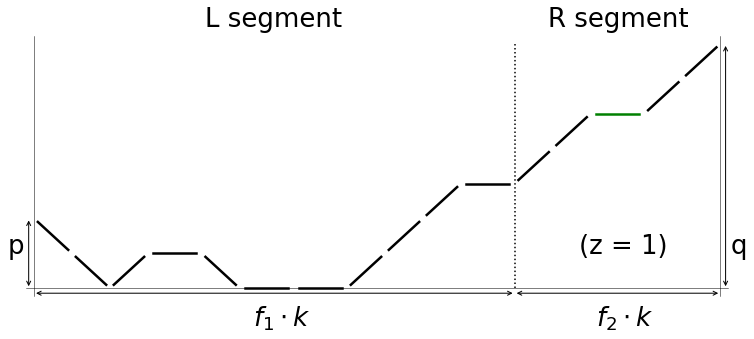}}}%
	
	\subfigure[]{\scalebox{.31}{\includegraphics{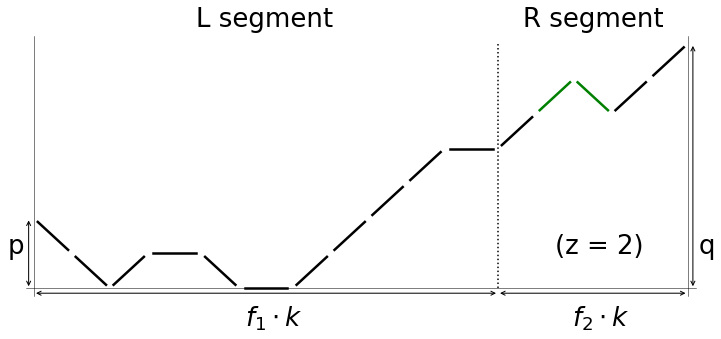}}}%
	\subfigure[]{\scalebox{.31}{\includegraphics{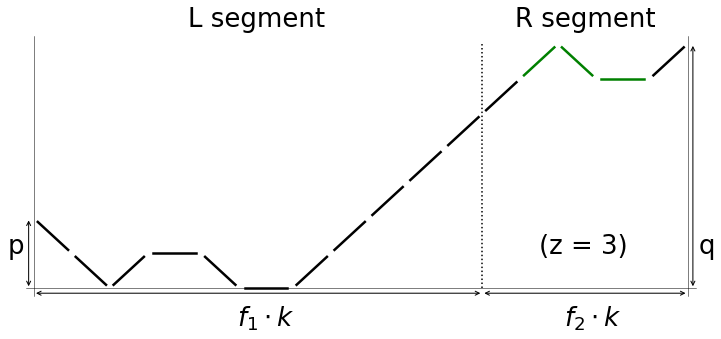}}}%
	
	\caption{Typical walks in $G^{z}_{p,q}$ for $z$ equal to 0, 1, 2, and 3. The balanced steps contributing to the count of $z$ are shown in green. As $z$ increases, the total area under typical walks becomes larger, and so their contribution within a ground state is exponentially suppressed (eq. \eqref{eq:high-imb-successive-ratio}).}
	\label{fig:high-imbalance-subcollections}
\end{figure}

From this definition it follows that $G_{p,q}$ is the disjoint union over $z$ of the $G^{z}_{p,q}$. Using Definition \ref{def:walk_set_definitions}, the unnormalized exact ground state decomposes as
\begin{equation} \label{eq:high-imbalance-exact-ugs}
\ket{\unnormstate{G_{p,q}}} = \sum_z \ket{\unnormstate{G^z_{p,q}}},
\end{equation}
where the $z$ index in the summation goes from zero up to, nominally, $f_2k$ (i.e. when all steps in the $R$ segment are balanced). Note that the subcollections $G^{z}_{p,q}$ with $z$ close to this upper limit of $f_2 k$ will often be empty due to constraints imposed by the total length of the spin chain; however this will not matter further in the argument.

Since the sets $G^z_{p,q}$ are disjoint, the normalization factor splits as
\begin{equation}\label{eq:high-imbalance-normalization-factor-split}
N_{p,q}=\norm{G_{p,q}}=\sum_z \norm{G^z_{p,q}}.
\end{equation}
Thus
\begin{equation}\label{eq:high-imbalance-exact-gs}
\ket{GS_{p,q}} = {1 \over \sqrt{N_{p,q}}} \sum_z \ket{\unnormstate{G^z_{p,q}}}.
\end{equation}
The claim is that, at large enough $q$, the $\norm{G^0_{p,q}}$ term dominates all others in \eqref{eq:high-imbalance-normalization-factor-split}. That enables us to approximate the true ground state by
\begin{equation}\label{eq:high-imbalance-0z-gs}
\ket{\normstate{G^0_{p,q}}}.
\end{equation}

\begin{proposition} \label{pr:equivalence-z0}
	The state $\ket{\normstate{G^0_{p,q}}}$ is identical to the $\ket{PGS^{(R)}_{p,q}}$ of eq. \eqref{eq:pgsr-definition}.
\end{proposition}

\begin{proof}[Proof of Proposition \ref{pr:equivalence-z0}]
	The condition of having $z=0$ balanced steps in the $R$ segment means that all walks in $G^{0}_{p,q}$ have their last $f_2 k$ steps all up. We are placing no restriction on the steps in the $L$ segment, and since we take the area-weighted superposition of all possible walks, we form exactly the ground state on the $L$ segment, with the corresponding numbers $(p, q-f_2 k)$ of unbalanced steps:
	\begin{equation} \label{eq:pgsr-definition-equivalence}
	\ket{\normstate{G^0_{p,q}}}
	= \ket{GS^L_{p,q - f_2 k}} \otimes \oa \ket{u} \ca^{\otimes f_2 k}
	\equiv \ket{PGS^{(R)}_{p,q}}.
	\end{equation}
\end{proof}
The characterization \eqref{eq:pgsr-definition-equivalence} shows that $\ket{PGS^{(R)}_{p,q}}$ consists precisely of the walks in $G^{0}_{p,q}$, with the correct area weights. This additional result will help us to prove Lemma \ref{lm:approx-high}.

\subsection{Proof of the section's main result}

\begin{proof}[Proof of Lemma \ref{lm:approx-high}]
	The proof is very similar to that of Lemma \ref{lm:low-imbalance-approximation-1}. From equations \eqref{eq:high-imbalance-exact-gs} and \eqref{eq:pgsr-definition-equivalence}, Definition \ref{def:walk_set_definitions} gives
	\begin{equation}
	\left|\inner{GS_{p,q}}{PGS^{(R)}_{p,q}}\right|^2
	= {\norm{G^0_{p,q}} \over N_{p,q}}.
	\end{equation}
	Combined with eq. \eqref{eq:high-imbalance-normalization-factor-split}, this yields
	\begin{equation} \label{eq:high-imbalance-error-term}
	1 - \left|\inner{GS_{p,q}}{PGS^{(R)}_{p,q}}\right|^2
	= \sum_{z>0} {\norm{G^z_{p,q}} \over N_{p,q}}.
	\end{equation}
	
	As before, the strategy is to map walks in $G^{z+1}_{p,q}$ to walks in $G^{z}_{p,q}$ and obtain an upper bound on $\norm{G^{z+1}_{p,q}}/\norm{G^{z}_{p,q}}$. For any $z\ge 0$, take any $w \in G^{z+1}_{p,q}$, and let $b_1$ be its rightmost balanced step. Since $z+1>0$, the $R$ segment contains a positive number of balanced steps, and in particular $b_1$ must be in $R$ (Fig. \ref{fig:high_imbalance} (a)).
	
	The step $b_1$ must be flat or down, because any balanced up-step has a partner to its right. The following steps of the proof will require $b_1$ to be flat. If instead it is down, find its balancing partner and replace them both by flat steps to lower the area as shown in Fig.~\ref{fig:high_imbalance} (a). Note that the two flat steps we've introduced are still balanced, so this leaves the walk in $G^{z+1}_{p,q}$.
	
	Next, find the first unbalanced up-step $s_1$. Since there are at least $(f_2 + c) k$ of them (from the condition on $q$), we see that $s_1$ must be at least $(f_2 + c) k$ positions away from the rightmost end of the chain, i.e. at least $ck$ positions to the left of the boundary between $L$ and $R$. As $ b_1$ is in $R$, we obtain that the distance between the two is at least $ck$:
	
	\begin{equation} \label{eq:high-imbalance-distance-bound}
	d(s_1, b_1) > ck.
	\end{equation}
	
	Interchanging $s_1$ and $b_1$ to form a walk $w'$ then lowers the area by at least $ck$ units; see Fig.~\ref{fig:high_imbalance}(b). This swap eliminated a balanced step from the rightmost segment, so $w' \in G^{z}_{p,q}$.

    \begin{figure}[t]
	\centering	
	\scalebox{.4}{\includegraphics{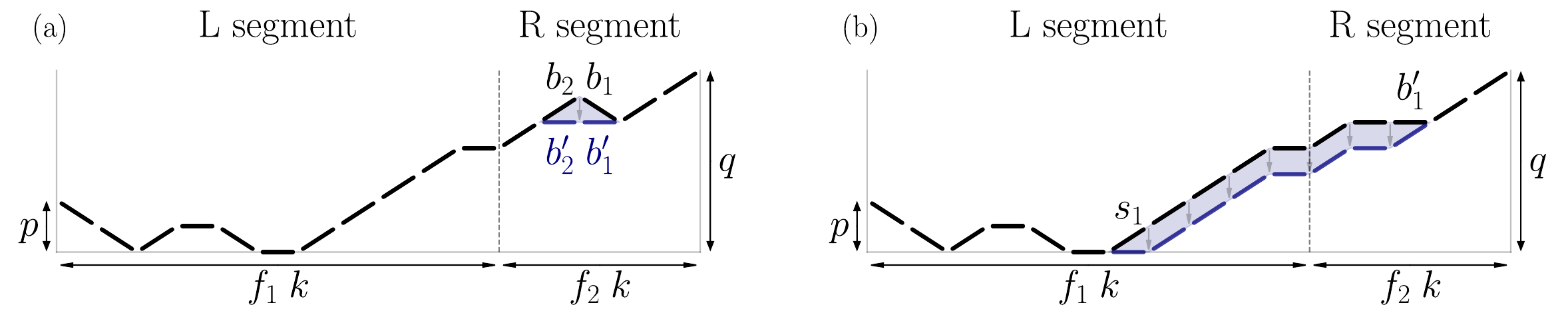}}%
		\caption{Typical walk in $G^{2}_{p,q}$, and an illustration of the swap procedure. In (a), the rightmost balanced step $b_1$ is flattened together with its balanced partner $b_2$. This maps to a path of area that is smaller by the shaded part. In (b), $s_1$ and $b_1'$ are exchanged, and the shaded section shifts down by one unit as a result. In the end, only two steps have been altered. The final result belongs to $G^{1}_{p,q}$, and encloses a smaller area, helping to prove eq. \eqref{eq:high-imb-successive-ratio}.}
	\label{fig:high_imbalance}
    \end{figure}

	Similarly to the previous section, the mapping is not injective, and we need to bound the cardinality of the preimage for any given $w' \in G^{z}_{p,q}$. The treatment of this aspect is identical to that in the proof of Lemma \ref{lm:low-imbalance-approximation-1}, with the result that the desired cardinality is at most ${(f_1 + f_2) k \choose 3}$. Applying Lemma \ref{lm:weighted-image-bound} with $\Delta=ck$ gives
	\begin{equation}\label{eq:high-imb-successive-ratio}
	{\norm{G^{z+1}_{p,q}} \over \norm{G^{z}_{p,q}}}
	<  {(f_1 + f_2) k \choose 3} t^{2 k c}
	< (f_1 + f_2)^3 \cdot k^3 t^{2 k c}.
	\end{equation}
	Since $c > 0$ and $t<1$, the RHS of the above goes to zero exponentially fast as $k \to \infty$. At large enough $k$ this implies monotonicity in $z$, $\norm{G^{z+1}_{p,q}} < \norm{G^{z}_{p,q}}$, and so in particular we can use $\norm{G^{z}_{p,q}} < \norm{G^{1}_{p,q}}$ for all $z\ge 1$. Therefore
	\begin{equation}
	\sum_{z = 1}^{f_2 k} {\norm{G^z_{p,q}} \over \norm{G^0_{p,q}}}
	< f_2 k \cdot {\norm{G^1_{p,q}} \over \norm{G^0_{p,q}}}
	< f_2 (f_1 + f_2)^3 \cdot k^4 t^{2 k c}.
	\end{equation}
	By construction $\norm{G^0_{p,q}} \le N_{p,q}$, so replacing the denominator $\norm{G^0_{p,q}}$ by $N_{p,q}$ in the sum above will only make it smaller:
	\begin{equation}
	\sum_{z >0} {\norm{G^z_{p,q}} \over N_{p,q}}
	< \sum_{z>0} {\norm{G^z_{p,q}} \over \norm{G^0_{p,q}}}
	< f_2 (f_1 + f_2)^3 \cdot k^4 t^{2 k c}.
	\end{equation}
	It follows that
	\begin{equation}
	1 - \left|\inner{GS_{p,q}}{PGS^{(R)}_{p,q}}\right|^2
	< f_2 (f_1 + f_2)^3 \cdot k^4 t^{2 k c}.
	\end{equation}
	The above is valid for all $q > (f_2 + c) k$, but the RHS does not involve $q$. Taking the supremum of the LHS over $q$ in this range gives
	\begin{equation}
	\sup\limits_{q > (f_2 + c) k} \ob 1 - \left|\inner{GS_{p,q}}{PGS^{(R)}_{p,q}}\right|^2 \cb < f_2 (f_1 + f_2)^3 \cdot k^4 t^{2 k c}.
	\end{equation}
	Due to the exponential factor on the right, eq. \eqref{eq:high-q-approximation-lemma} follows. As mentioned previously, an identical argument shows that eq. \eqref{eq:high-p-approximation-lemma} is also true, completing the proof of the lemma.
\end{proof}

\section{Implementing the approximations}\label{sec:application}
\subsection{Imbalance regimes and splitting the chain}\label{ssec:division}
We now use the approximations of sections \ref{sec:low-imb} and \ref{sec:high-imb} to find the limiting behavior in $k$ of the quantity
\begin{equation}
\sup\limits_{\substack{p,q \ge 0 \\ p+q \le 3k}} \op \sup\limits_{\phi_{p,q} \in \range E_k} \mel{\phi_{p,q}}{G_{[k+1, 3k]}}{\phi_{p,q}} \cp
\end{equation}
which appears in the RHS of Proposition \ref{pr:reduction-of-criterion}. Fix a constant $c \in (0, 1/4 - 2b)$, which one may imagine to be very small. We will cover the following distinct regimes:
\begin{itemize}
	\item $p > k \op 1 + {c \over 2}\cp$, $q$ arbitrary, but still consistent with $p$, i.e. $q+p \le 3k$
	\item $q > k \op 1 + {c \over 2}\cp$, $p$ arbitrary (but still consistent with $q$)
	\item $p,q < k \op 1 + c \cp$
\end{itemize}
The three cases are not mutually exclusive, but their union covers all possible values for $p$ and $q$ on a chain segment of length $3k$. For large numbers of unbalanced steps, we directly show that
\begin{proposition}\label{pr:high-imbalance-zero-limits}
	For the large-$p$ regime, $p > k \op 1 + {c \over 2}\cp$, we have
	\begin{equation} \label{eq:high-p-limit}
	\lim\limits_{k \to \infty} \op\sup\limits_{\substack{(1 + c/2) k < p \le 3k \\ 0 \le q \le 3k - p}} \op\sup\limits_{\ket{\phi_{p,q}} \in \range E_k} \mel{\phi_{p,q}}{G_{[k+1, 3k]}}{\phi_{p,q}} \cp \cp = 0 
	\end{equation}
	and similarly in the regime where $q$ is larger than $k \op 1 + {c \over 2}\cp$:
	\begin{equation} \label{eq:high-q-limit}
	\lim\limits_{k \to \infty} \op\sup\limits_{\substack{(1 + c/2) k < q \le 3k \\ 0 \le p \le 3k - q}} \op\sup\limits_{\ket{\phi_{p,q}} \in \range E_k} \mel{\phi_{p,q}}{G_{[k+1, 3k]}}{\phi_{p,q}} \cp \cp = 0 
	\end{equation}
\end{proposition}
For the low-imbalance regime, we find a similar result, albeit through a longer argument:
\begin{proposition}\label{pr:low-imbalance-zero-limits}
	In the $p,q < (1 + c) \cdot k$ regime it is true that
	\begin{equation} \label{eq:low-pq-limit}
	\lim\limits_{k \to \infty} \op\sup\limits_{0 \le p, q < (1 + c) k} \op\sup\limits_{\ket{\phi_{p,q}} \in \range E_k} \mel{\phi_{p,q}}{G_{[k+1, 3k]}}{\phi_{p,q}} \cp \cp = 0 
	\end{equation}
\end{proposition}
We then combine Propositions \ref{pr:high-imbalance-zero-limits} and \ref{pr:low-imbalance-zero-limits} to conclude that the supremum over all possible $p,q$ goes to zero in the limit of large $k$:

\begin{proposition} \label{pr:all-imbalance-zero-limits}
	From Propositions \ref{pr:high-imbalance-zero-limits} and \ref{pr:low-imbalance-zero-limits} it follows that
	\begin{equation}
	\lim\limits_{k \to \infty} \op \sup\limits_{\substack{p,q \ge 0 \\ p+q \le 3k}} \op\sup\limits_{\ket{\phi_{p,q}} \in \range E_k} \mel{\phi_{p,q}}{G_{[k+1, 3k]}}{\phi_{p,q}} \cp \cp = 0 
	\end{equation}
\end{proposition}

\begin{proof} [Proof of Proposition \ref{pr:all-imbalance-zero-limits}]
	As discussed in the beginning of this section, the ranges of $p,q$ covered in Propositions \ref{pr:high-imbalance-zero-limits} and \ref{pr:low-imbalance-zero-limits} can be combined to cover all the possibilities for $p,q\ge 0$ with $p + q \le 3k$. 
\end{proof}

We now turn to proving Propositions \ref{pr:high-imbalance-zero-limits} and \ref{pr:low-imbalance-zero-limits}. The former is a direct application of the results in Section \ref{sec:high-imb}. The latter proof relies on Section \ref{sec:low-imb} in a similar manner, but also requires more in-depth technical discussions, which are presented separately in Section \ref{sec:normalizations} and appendix \ref{sec:norm-ratios-appendix}.

\subsection{High imbalance}\label{ssec:high-imb-app} Since ground states with a large number of unbalanced steps approximately factorize, verification of the criterion is more straightforward in this case.

\begin{proof}[Proof of Proposition \ref{pr:high-imbalance-zero-limits}]
	We will work in the large $p$ regime, and show that eq. \eqref{eq:high-p-limit} holds. By assumption, $\ket{\phi_{p,q}}$ is orthogonal to the ground space on the full chain:
	\begin{equation}
	    G_{[1,3k]} \ket{\phi_{p,q}} = 0
	\end{equation}
	Expanding $G_{[1,3k]}$ in terms of individual ground states, we note that only the state with $(p,q)$ unbalanced steps can contribute. Therefore, the above translates to
	\begin{equation}
	    \inner{GS_{p,q}^{[1,3k]}}{\phi_{p,q}} = 0 \qquad \implies \qquad \inner{\phi_{p,q}}{GS_{p,q}^{[1,3k]}} \inner{GS_{p,q}^{[1,3k]}}{\phi_{p,q}} = 0
	\end{equation}
	With $p > (1+c) k$, we can use the high imbalance approximation Lemma \ref{lm:approx-high}, with $L = [1,k]$ and $R = [k+1,3k]$ to approximate $\ket {GS_{p,q}^{[1,3k]}}$ by
	\begin{equation}
	    \ket{PGS_{p,q}^{[1,3k]}} = \ket{d}^{\otimes k} \otimes \ket{GS_{p - k,q}^{[k + 1,3k]}} 
	\end{equation}
	and it follows by the projector approximation Lemma \ref{lm:proj-approx} that the quantities $|\inner{PGS_{p,q}^{[1,3k]}}{\phi_{p,q}}|^2$ approximate $|\inner{GS_{p,q}^{[1,3k]}}{\phi_{p,q}}|^2$ at large $k$. Since the latter overlaps are, by assumption, identically zero when $\ket{\phi_{p,q}} \in \range E_k$, we find
	\begin{equation} \label{eq:hiap-e1}
	    \lim\limits_{k \to \infty} \op\sup\limits_{\substack{(1 + c/2) k < p \le 3k \\ 0 \le q \le 3k - p}} \op\sup\limits_{\ket{\phi_{p,q}} \in \range E_k} \inner{\phi_{p,q}}{PGS_{p,q}^{[1,3k]}} \inner{PGS_{p,q}^{[1,3k]}}{\phi_{p,q}} \cp \cp = 0 
	\end{equation}
	The projector onto $PGS$ can be decomposed as a tensor product:
	\begin{equation}
	    \ket{PGS_{p,q}^{[1,3k]}} \bra{PGS_{p,q}^{[1,3k]}} = \ket{d}^{\otimes k} \bra{d}^{\otimes k} \otimes \ket{GS_{p - k,q}^{[k + 1,3k]}} \bra{GS_{p - k,q}^{[k + 1,3k]}}
	\end{equation}
	Note that when this acts on $\ket{\phi_{p,q}}$, the second component behaves like $G_{[k+1, 3k]}$. The reason is that $\ket{d}^{\otimes k} \bra{d}^{\otimes k}$ selects only walks with the first $k$ steps down. Since every walk of $\ket{\phi_{p,q}}$ has $(p,q)$ unbalanced steps on $[1,3k]$, those with the first $k$ steps down will have $(p-k,q)$ unbalanced steps on $[k+1, 3k]$. Therefore, the only ground state they need to be compared to is $\ket{GS_{p - k,q}^{[k + 1,3k]}}$. This allows for the rewriting of \eqref{eq:hiap-e1} as
	\begin{equation} \label{eq:hiap-e2}
	    \lim\limits_{k \to \infty} \op\sup\limits_{\substack{(1 + c/2) k < p \le 3k \\ 0 \le q \le 3k - p}} \op\sup\limits_{\ket{\phi_{p,q}} \in \range E_k} \bra{\phi_{p,q}} \ob \ket{d}^{\otimes k} \bra{d}^{\otimes k} \otimes G_{[k+1, 3k]} \cb \ket{\phi_{p,q}} \cp \cp = 0 
	\end{equation}
	This form is close to the desired result, but has the extra $\ket{d}^{\otimes k} \bra{d}^{\otimes k}$ projector. Some walks in the composition of $\ket{\phi_{p,q}}$ will have their first $k$ steps down, and others will not. Equation \eqref{eq:hiap-e2} deals with those that do, and one must separately consider the ones that do not. This will be addressed as follows: since $\ket{\phi_{p,q}}$ is in the range of $E_k$, it consists only of ground states on the first two thirds of the chain $[1,2k]$. Namely, it is an eigenstate of the projector $G_{[1,2k]}$, with eigenvalue 1:
	\begin{equation}
	\ket{\phi_{p,q}} = G_{[1,2k]} \ket{\phi_{p,q}}
	\end{equation}
	This allows for a Schmidt decomposition of $\ket{\phi_{p,q}}$ with respect to subsystems $[1,2k]$ and $[2k+1, 3k]$, in which only ground states will be present on the left side. To determine which numbers of unbalanced steps can appear in this decomposition, we must carefully analyze the minimization of walks. Recall that a walk on $[1,3k]$ with $(p,q)$ unbalanced steps must start at height $p$, end at height $q$, and reach zero somewhere in between. Since we are assuming $p > (1 + c) k$, all the points where the walk reaches zero height must be at least $(1+c)k$ steps away from the start, so in particular none of them can be found in the first third $[1,k]$. We distinguish three cases:
	\begin{itemize}
	    \item Walks that reach zero height both within the middle third $[k+1, 2k]$, and within the last one $[2k+1, 3k]$. When we `cut' them after $2k$ steps, both resulting components will still reach zero height, so they will already be minimized. If the height after $2k$ steps (at the split point) is $v$, then the resulting walks will have $(p,v)$ and $(v,q)$ unbalanced steps respectively.
	    \item Walks that do not reach zero height in the middle third, but rather only within the last one. When splitting, the first component will not be minimized. If $z$ is the minimum height that the original walk reached within the first two thirds $[1,2k]$, and $z+v$ is the height at the split point (with $v \ge 0$), then the resulting walks have $(p-z, v)$ and $(v+z,q)$ unbalanced steps respectively. 
	    \item Walks that only reach zero height within the middle third, and not the last one. With $z$ the minimum height within the last third $[2k+1, 3k]$, and $z + v$ the height at the splitting point, we find $(p, v + z)$ and $(v, q-z)$ unbalanced steps.
	\end{itemize}
	The resulting decomposition will be
	\begin{equation}\label{eq:schmidt}
	    \ket{\phi_{p,q}} = \sum_{v \ge 0} \ket{GS_{p,v}^{[1,2k]}} \otimes \ket{\psi^1_{v,q}} + \sum_{z = 1}^{p} \sum_{v \ge 0} \ket{GS_{p - z, v}^{[1,2k]}} \otimes \ket{\psi^2_{v + z,q}} + \sum_{z = 1}^{q} \sum_{v \ge 0} \ket{GS_{p, v + z}^{[1,2k]}} \otimes \ket{\psi^3_{v, q - z}}
	\end{equation}
	The $\psi^1, \psi^2, \psi^3$ states all live on the last third $[2k+1, 3k]$. By convention we absorb the coefficients from the Schmidt decomposition into their definition, so they are not normalized. This will not pose a problem, since the main focus will be on the properties of the $[1,2k]$ ground states instead. For simplicity of notation, give separate names to the three terms:
	\begin{align*}
	    \ket{\mathrm {I}} &= \sum_{v \ge 0} \ket{GS_{p,v}^{[1,2k]}} \otimes \ket{\psi^1_{v,q}} \\
	    \ket{\mathrm {II}} &= \sum_{z = 1}^{p} \sum_{v \ge 0} \ket{GS_{p - z, v}^{[1,2k]}} \otimes \ket{\psi^2_{v + z,q}} \\
	    \ket{\mathrm {III}} &= \sum_{z = 1}^{q} \sum_{v \ge 0} \ket{GS_{p, v + z}^{[1,2k]}} \otimes \ket{\psi^3_{v, q - z}}
	\end{align*}
	The ground states that appear in $\ket{\mathrm {I}}$ and $\ket{\mathrm {III}}$ have $p$ unbalanced steps on the left, so we can use the high imbalance approximation lemma with $L = [1,k]$ and $R = [k + 1, 2k]$ to argue that:
	\begin{itemize}
	    \item $\ket{GS_{p,v}^{[1,2k]}}$ is superpolynomially approximated by $\ket{d}^{\otimes k} \otimes \ket{GS_{p - k,v}^{[k + 1,2k]}}$.
	    \item $\ket{GS_{p,v + z}^{[1,2k]}}$ is superpolynomially approximated by $\ket{d}^{\otimes k} \otimes \ket{GS_{p - k,v + z}^{[k + 1,2k]}}$.
	\end{itemize}
	Through an argument similar to that of the superposition approximation Lemma \ref{lm:superposition-approx}, we conclude that $\ket{\mathrm {I}}$ is approximated superpolynomially by
	\begin{equation}
	    \ket{\mathrm {I}'} = \sum_{v \ge 0} \ket{d}^{\otimes k} \otimes \ket{GS_{p - k,v}^{[k + 1,2k]}} \otimes \ket{\psi^1_{v,q}} = \ket{d}^{\otimes k} \otimes  \op \sum_{v \ge 0} \ket{GS_{p - k,v}^{[k + 1,2k]}} \otimes \ket{\psi^1_{v,q}} \cp
	\end{equation}
	and similarly, $\ket{\mathrm{III}}$ is approximated by
	\begin{equation}
	    \ket{\mathrm {III}'} = \sum_{z = 1}^{q} \sum_{v \ge 0} \ket{d}^{\otimes k} \otimes \ket{GS_{p - k,v + z}^{[k + 1,2k]}} \otimes \ket{\psi^3_{v, q - z}} = \ket{d}^{\otimes k} \otimes \op \sum_{z = 1}^{q} \sum_{v \ge 0} \ket{GS_{p - k,v + z}^{[k + 1,2k]}} \otimes \ket{\psi^3_{v, q - z}} \cp
	\end{equation}
	The same reasoning doesn't directly work for $\ket{\mathrm{II}}$, since the ground states appearing within it only have $p - z$ unbalanced steps on the left, which is not guaranteed to be above $(1 + c) k$ if $z$ is large. Instead, we will separate $\ket{\mathrm{II}}$ into two terms, one including the walks with the first $k$ steps down, call it $\ket{\mathrm{IIa}}$, and everything else, which will be called $\ket{\mathrm{IIb}}$. Formally,
	\begin{equation}
	    \ket{\mathrm{II}} \equiv \ket{\mathrm{IIa}} + \ket{\mathrm{IIb}} \qquad \qquad \ob \ket{d}^{\otimes k} \bra{d}^{\otimes k} \cb \ket{\mathrm{IIa}} = \ket{\mathrm{IIa}} \qquad \qquad \ob \ket{d}^{\otimes k} \bra{d}^{\otimes k} \cb \ket{\mathrm{IIb}} = 0
	\end{equation}
	With $\ket{\phi_{p,q}} = \ket{\mathrm {I}} + \ket{\mathrm {IIa}} + \ket{\mathrm {IIb}} + \ket{\mathrm {III}}$, we define the approximation
	\begin{equation}
	    \ket{\phi_{p,q}'} = \ket{\mathrm {I}'} + \ket{\mathrm {IIa}} + \ket{\mathrm {IIb}} + \ket{\mathrm {III}'}
	\end{equation}
	This is superpolynomial since $\ket{\mathrm {I}'}$ and $\ket{\mathrm {III}'}$ are. We will show that, at large $k$, the projector $G_{[k+1, 3k]}$ approximately annihilates this. The argument is in two parts:
	\begin{proposition} \label{pr:hiap-p1}
	    \begin{equation*}
	    \lim\limits_{k \to \infty} \op\sup\limits_{\substack{(1 + c/2) k < p \le 3k \\ 0 \le q \le 3k - p}} \op\sup\limits_{\ket{\phi_{p,q}} \in \range E_k} \| G_{[k+1, 3k]} \ob \ket{\mathrm {I}'} + \ket{\mathrm {IIa}} + \ket{\mathrm {III}'} \cb \| \cp \cp = 0 
	    \end{equation*}
	\end{proposition}
	\begin{proposition} \label{pr:hiap-p2}
	    \begin{equation*}
	   \lim\limits_{k \to \infty} \op\sup\limits_{\substack{(1 + c/2) k < p \le 3k \\ 0 \le q \le 3k - p}} \op\sup\limits_{\ket{\phi_{p,q}} \in \range E_k} \| G_{[k+1, 3k]}  \ket{\mathrm {IIb}} \| \cp \cp = 0 
	\end{equation*}
	\end{proposition}
	\noindent Assume these two propositions. From the triangle inequality we obtain
	\begin{align}
	    \| G_{[k+1, 3k]}  \ket{\phi_{p,q}'} \| 
	    &\le \| G_{[k+1, 3k]} \ob \ket{\mathrm {I}'} + \ket{\mathrm {IIa}} + \ket{\mathrm {III}'} \cb \| + \| G_{[k+1, 3k]} \ket{\mathrm {IIb}} \|
	\end{align}
Both terms on the last line vanish at large $k$, and  so $G_{[k+1, 3k]}$ approximately annihilates $\ket{\phi_{p,q}'}$:
	\begin{equation}
	\lim\limits_{k \to \infty} \op\sup\limits_{\substack{(1 + c/2) k < p \le 3k \\ 0 \le q \le 3k - p}} \op\sup\limits_{\ket{\phi_{p,q}} \in \range E_k} \mel{\phi_{p,q}'}{G_{[k+1, 3k]}}{\phi_{p,q}'} \cp \cp = 0 
	\end{equation}
	Using the expectation approximation Lemma \ref{lm:overlaps-and-expectations}, we conclude
	\begin{align}
	    \lim\limits_{k \to \infty} \op\sup\limits_{\substack{(1 + c/2) k < p \le 3k \\ 0 \le q \le 3k - p}} \op\sup\limits_{\ket{\phi_{p,q}} \in \range E_k} \mel{\phi_{p,q}}{G_{[k+1, 3k]}}{\phi_{p,q}} \cp \cp=0
	\end{align}
	completing the proof of eq. \eqref{eq:high-p-limit}. 
	\end{proof}
	The argument for eq. \eqref{eq:high-q-limit} is very similar, so most of its details will be omitted. An outline is sketched in \ref{ssec:high-q-outline}. We turn to the proof of Proposition \ref{pr:hiap-p1}:
	\begin{proof}[Proof of Proposition \ref{pr:hiap-p1}]
	All terms in $\ket{\phi_{p,q}'}$, except for $\ket{\mathrm {IIb}}$, have their first $k$ steps down, so
	\begin{equation}
	    \ob \ket{d}^{\otimes k} \bra{d}^{\otimes k} \cb \ob \ket{\mathrm {I}'} + \ket{\mathrm {IIa}} + \ket{\mathrm {III}'} \cb  = \ket{\mathrm {I}'} + \ket{\mathrm {IIa}} + \ket{\mathrm {III}'} 
	\end{equation}
	and if we add a $G_{[k+1, 3k]}$ projector, it follows that
	\begin{equation}
	    \ob \ket{d}^{\otimes k} \bra{d}^{\otimes k} \otimes G_{[k+1, 3k]} \cb \ob \ket{\mathrm {I}'} + \ket{\mathrm {IIa}} + \ket{\mathrm {III}'} \cb  = G_{[k+1, 3k]} \ob \ket{\mathrm {I}'} + \ket{\mathrm {IIa}} + \ket{\mathrm {III}'} \cb
	\end{equation}
	On the other hand, since $\ket{d}^{\otimes k} \bra{d}^{\otimes k}$ annihilates $\ket{\mathrm{IIb}}$, the LHS of the above is equal to:
	\begin{align}
	    \ob \ket{d}^{\otimes k} \bra{d}^{\otimes k} \otimes G_{[k+1, 3k]} \cb \ob \ket{\mathrm {I}'} + \ket{\mathrm {IIa}} + \ket{\mathrm {III}'} \cb  
	    &= \ob \ket{d}^{\otimes k} \bra{d}^{\otimes k} \otimes G_{[k+1, 3k]} \cb \ket{\phi_{p,q}'}
	\end{align}
    These equalities can be combined, and the norm taken, to yield
    \begin{equation}
        \| G_{[k+1, 3k]} \ob \ket{\mathrm {I}'} + \ket{\mathrm {IIa}} + \ket{\mathrm {III}'} \cb \| = \|\ob \ket{d}^{\otimes k} \bra{d}^{\otimes k} \otimes G_{[k+1, 3k]} \cb \ket{\phi_{p,q}'} \|
    \end{equation}
    From \eqref{eq:hiap-e2} and through the expectation approximation Lemma \ref{lm:overlaps-and-expectations}, we find
	\begin{align*} 
	    0 &= \lim\limits_{k \to \infty} \op\sup\limits_{\substack{(1 + c/2) k < p \le 3k \\ 0 \le q \le 3k - p}} \op\sup\limits_{\ket{\phi_{p,q}} \in \range E_k} \bra{\phi_{p,q}} \ob \ket{d}^{\otimes k} \bra{d}^{\otimes k} \otimes G_{[k+1, 3k]} \cb \ket{\phi_{p,q}} \cp \cp\\
	    &= \lim\limits_{k \to \infty} \op\sup\limits_{\substack{(1 + c/2) k < p \le 3k \\ 0 \le q \le 3k - p}} \op\sup\limits_{\ket{\phi_{p,q}} \in \range E_k} \bra{\phi_{p,q}'} \ob \ket{d}^{\otimes k} \bra{d}^{\otimes k} \otimes G_{[k+1, 3k]} \cb \ket{\phi_{p,q}'} \cp \cp 
	\end{align*}
	On the last line, the matrix element is $\|\ob \ket{d}^{\otimes k} \bra{d}^{\otimes k} \otimes G_{[k+1, 3k]} \cb \ket{\phi_{p,q}'} \|^2$, so the desired result follows: $G_{[k+1, 3k]}$ approximately annihilates $ \ket{\mathrm {I}'} + \ket{\mathrm {IIa}} + \ket{\mathrm {III}'} $.
	\begin{equation}
	   \lim\limits_{k \to \infty} \op\sup\limits_{\substack{(1 + c/2) k < p \le 3k \\ 0 \le q \le 3k - p}} \op\sup\limits_{\ket{\phi_{p,q}} \in \range E_k} \| G_{[k+1, 3k]} \ob \ket{\mathrm {I}'} + \ket{\mathrm {IIa}} + \ket{\mathrm {III}'} \cb \| \cp \cp = 0 
	\end{equation}
	\end{proof}
	\noindent The proof of Proposition \ref{pr:hiap-p2} is deferred to Appendix \ref{ssec:high-p-other-walks}. It consists of analyzing several subcases of walks in $\ket{\mathrm{IIb}}$, and arguing that their total contribution to $G_{[k+1,3k]} \ket{\mathrm{IIb}}$ vanishes superpolynomially.

\subsection{Low imbalance}\label{ssec:low-imb-app} Now we consider the case when both $p$ and $q$ are less than $k(1 + c)$. The goal is to prove Proposition \ref{pr:low-imbalance-zero-limits}, which will require a series of approximations and technical discussions. We divide the full chain $[1, 3k]$ into three intervals (Fig. \ref{fig:division}):
\begin{itemize}
	\item $A$ on the left, of width $(1+2c) k$.
	\item $B$ in the middle, of width $(1 - 4c) k$.
	\item $C$ on the right, again of size $(1+2c) k$.
\end{itemize}
Existence of the middle interval requires that $c < 1/4$. The reason for this construction is that any $p \le (1 + c) k$ is classified as "low-$p$ regime" with respect to interval $A$, since it is smaller than the size of $A$, by at least $ck$ steps. Similarly, all $q \le (1 + c) k$ are in the low-$q$ regime relative to interval $C$. This will allow for the approximation Lemma \ref{lm:low-imbalance-approximation-trunc} to be applied.

\begin{figure}[t]
	\centering	
	\subfigure[]{\scalebox{.31}{\includegraphics{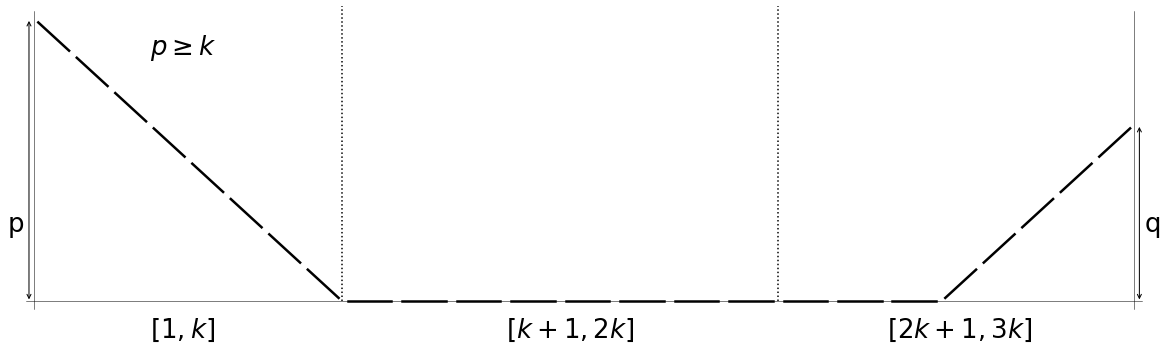}}}%
			
	\subfigure[]{\scalebox{.31}{\includegraphics{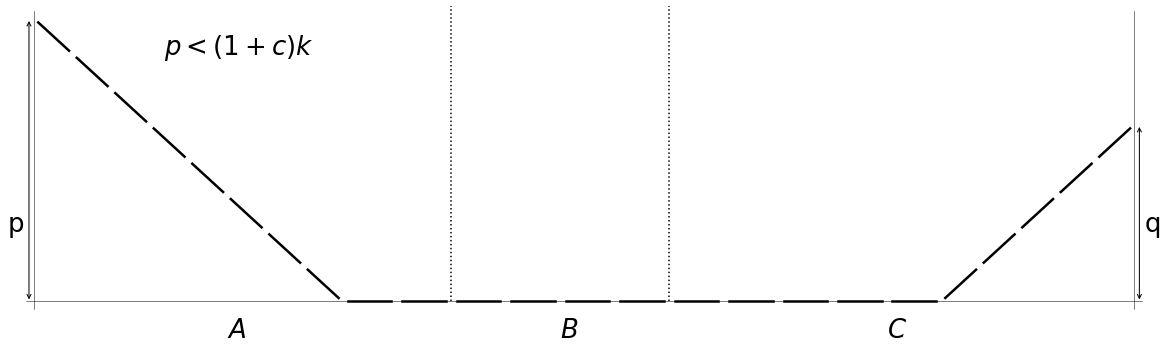}}}%
		
	\caption{Division of the full chain, illustrated here at $k = 8$ and $c = 1/8$. A walk is schematically shown, that has too many unbalanced steps on the left for Lemma \ref{lm:low-imbalance-approximation-trunc} to apply when the chain is split into equal thirds. But when splitting into unequal intervals A, B, and C, as described above, we can safely apply the approximation \ref{lm:low-imbalance-approximation-trunc}.}
	\label{fig:division}
\end{figure}

The $B$ and $C$ intervals are fully contained in the last two thirds: $B \cup C \subset [k+1, 3k]$. Because of that, and the frustration-freeness of the Hamiltonian, we have 
\begin{equation}
\mel{\phi_{p,q}}{G_{[k+1,3k]}}{\phi_{p,q}} \le \mel{\phi_{p,q}}{G_{BC}}{\phi_{p,q}}
\end{equation}
and it suffices to show that the term on the RHS is small. Similarly, since $G_{[1,2k]} \ket{\phi_{p,q}} = \ket{\phi_{p,q}}$ by definition of $\ket{\phi}$, and the left and middle intervals are contained in the first two thirds (i.e. $A \cup B \subset [1, 2k]$), we have
$G_{AB} \ket{\phi_{p,q}} = \ket{\phi_{p,q}}$, so we can expand $\ket{\phi_{p,q}}$ in terms of ground states on $A \cup B$. We will work with unnormalized such ground states, and normalize at the end.

Similarly to Section \ref{ssec:high-imb-app}, we perform a Schmidt decomposition of $\ket{\phi_{p,q}}$ with respect to subsystems $A \cup B$ and $C$, in which we must include three types of terms: walks that reach zero height both within $A \cup B$ and $C$, those that do so only in $C$, and lastly those that do so only in $A \cup B$. The result is
\begin{equation} \label{eq:phi-pq-full-expansion}
\ket{\phi_{p,q}} = {1 \over \sqrt{{N'}^{\phi}_{p,q}}} \op \sum_{v \ge 0} \ket{\unnormstate{G^{AB}_{p,v}}} \ket{\psi^C_{v, q}} + \sum_{z = 1}^p \sum_{v \ge 0} \ket{\unnormstate{G^{AB}_{p-z,v}}} \ket{{\psi'}^C_{v + z,q}} + \sum_{y = 1}^q \sum_{v \ge 0} \ket{\unnormstate{G^{AB}_{p,v+y}}} \ket{{\psi''}^C_{v, q - y}}\cp
\end{equation}
where $\ket{\unnormstate{G^{AB}}}$ denotes the unnormalized ground state constructed from the corresponding walk set on $A\cup B$, while $\ket{\psi_C}$, $\ket{{\psi'}^C}$ and $\ket{{\psi''}^C}$ are also unnormalized (having absorbed the Schmidt coefficients into their definition), and live on the $C$ segment. Although the overall normalization factor ${N'}^{\phi}_{p,q}$ could also be absorbed in the definition of the $\ket{\psi}$, we prefer to keep it explicit. For the purpose of proving that the supremum of $\mel{\phi_{p,q}}{G_{BC}}{\phi_{p,q}}$ vanishes at large $k$, it is enough to only consider the first term above, and furthermore one can restrict the summation over $v$ to only run up to $bk$. Formally, one has:
\begin{definition}
	Given an arbitrary $\ket{\phi_{p,q}} \in \range E_k$, extract from the expansion \eqref{eq:phi-pq-full-expansion} the collection of states $\{\psi^C_{v,q}\}$, indexed by $v$, that contribute to the first sum. We define the following substitute for the original state:
	\begin{equation}
	\ket{\phi_{p,q}'} = {1 \over \sqrt{N^{\phi}_{p,q}}} \sum_{v < bk} \ket{\unnormstate{G^{AB}_{p,v}}} \ket{\psi^C_{v,q}}
	\end{equation}
	where for normalization we need
	\begin{equation}
	N^{\phi}_{p,q} = \sum_{v < bk} \inner{\unnormstate{G^{AB}_{p,v}}}{\unnormstate{G^{AB}_{p,v}}} \cdot \inner{\psi^C_{v,q}}{\psi^C_{v,q}}
	\end{equation}
	Since $\ket{\phi_{p,q}}$ is constrained to be orthogonal to the ground state on the full chain $[1, 3k]$, the collections $\{\psi^C_{v,q}\}$ that can be obtained by the procedure above are not fully arbitrary. Let $P_{k, q}$ be the set of all such collections that can be obtained at specific $k$ and $q$.
	
\end{definition}
Note that $\ket{\phi_{p,q}'}$ may not always provide a `good approximation' of $\ket{\phi_{p,q}}$, in the sense of the previous sections: as $k$ grows, the overlap of the states need not go to 1. Indeed, in eq. \eqref{eq:phi-pq-full-expansion} one may take all the $\ket{\psi^C}$ to vanish, and use only nonzero $\ket{{\psi'}^C}$ and $\ket{{\psi''}^C}$ instead. Then our substitute is orthogonal to the initial state. However, $\ket{\phi_{p,q}'}$ is still useful for our bound as the following proposition shows.

\begin{notation}
    In the remainder of this subsection, we abbreviate
    $$
    \lim\limits_{k \to \infty}\sup\limits_{0 \le p, q < (1 + c) k} \equiv \widetilde{\lim\limits_{k \to \infty}}
    $$
\end{notation}
\begin{proposition} \label{pr:reduction-of-phi-pq} Assuming both limits exist, we have
\begin{equation}
\widetilde{\lim\limits_{k \to \infty}}\op \sup\limits_{\ket{\phi_{p,q}} \in \range E_k} \mel{\phi_{p,q}}{G_{BC}}{\phi_{p,q}} \cp  \le \widetilde{\lim\limits_{k \to \infty}} \op \sup\limits_{\{\psi^C_{s,q}\} \in P_{k, q}} \mel{\phi_{p,q}'}{G_{BC}}{\phi_{p,q}'} \cp 
\end{equation}
	
\end{proposition}

The proof of Proposition \ref{pr:reduction-of-phi-pq} is straightforward, but the details are rather lengthy and tangential to the main argument of this section. Therefore, the discussion is deferred to Section \ref{ssec:proving-reduction-of-phi-pq}. We continue with

\begin{definition} \label{def:approximate-phi-state}
	Since $\ket{\phi_{p,q}'}$ is written in terms of ground states on the $AB$ part of the chain, all of which have a relatively small number of unbalanced steps, to approximate it we consider
	\begin{equation} \label{eq:def-approximate-phi-state}
	\ket{A\phi_{p,q}} = {1 \over \sqrt{\mathcal{N}^{A\phi}_{p,q}}} \sum_{r, v < bk} \ket{\unnormstate{G^A_{p,r}}} \ket{\unnormstate{G^B_{r,v}}} \ket{\psi^C_{v,q}}
	\end{equation} 
	where normalization requires	
	\begin{equation}
	\mathcal{N}^{A\phi}_{p,q} = \sum_{r, v < bk} N^{A}_{p,r} \cdot N^{B}_{r,v} \cdot \inner{\psi^C_{v,q}}{\psi^C_{v,q}}
	\end{equation}
\end{definition}

\begin{proposition} \label{pr:approximation-of-reduction}
	The $\ket{A\phi_{p,q}}$ states of Definition \ref{def:approximate-phi-state} satisfy
	\begin{equation}
\widetilde{\lim\limits_{k \to \infty}} \op \sup\limits_{\{\psi^C_{v,q}\} \in P_{k, q}} \mel{\phi_{p,q}'}{G_{BC}}{\phi_{p,q}'} \cp  = \widetilde{\lim\limits_{k \to \infty}} \op \sup\limits_{\{\psi^C_{v,q}\} \in P_{k, q}} \mel{A\phi_{p,q}}{G_{BC}}{A\phi_{p,q}} \cp 
	\end{equation}
\end{proposition}
\begin{proof}
The assumptions required by Lemma \ref{lm:superposition-approx} clearly hold if we take $L = A \cup B$ and $R = C$. Then we apply Lemma \ref{lm:overlaps-and-expectations} and the proof is complete.
\end{proof}

From the classification of ground states in Section \ref{sec:ground-states-ramis}, we know that the ground space projector on BC is given by
\begin{equation} \label{eq:low-imb-proj-def}
G_{BC} = \sum_{r',q'} \ket{GS_{r',q'}^{BC}} \bra{GS_{r',q'}^{BC}}
\end{equation}
When inserting the form for $\ket{A\phi_{p,q}}$ from eq. \eqref{eq:def-approximate-phi-state} into the matrix element $\mel{A\phi_{p,q}}{G_{BC}}{A\phi_{p,q}}$, only $r$ values below $bk$ contribute, because the overlap $\bra{GS_{r',q'}^{BC}} \; \oa \ket{\unnormstate{G^B_{r,v}}} \ket{\psi^C_{v,q}} \ca$ is only nonzero when $r = r'$ and $q = q'$. For the same reason, only $q'$ values exactly equal to $q$ will matter. 
\begin{definition} Guided by this observation, isolate the part of the projector that actually contributes:
	\begin{equation} \label{eq:low-imb-app-proj-def}
	G_{BC}' = \sum_{r < bk} \ket{GS_{r,q}^{BC}} \bra{GS_{r,q}^{BC}}
	\end{equation}
\end{definition}
From the argument above, we find $\mel{A\phi_{p,q}}{G_{BC}}{A\phi_{p,q}} = \mel{A\phi_{p,q}}{G_{BC}'}{A\phi_{p,q}}$. Every state in eq. \eqref{eq:low-imb-app-proj-def} above has both $r$ and $q$ small enough that approximation Lemma \ref{lm:low-imbalance-approximation-trunc} applies. The following states then superpolynomially approximate the $\ket{GS_{r,q}^{BC}}$ when $r < bk$ and $q < (1 + c) k$:

\begin{definition} \label{def:ags-for-bc-segment}
	Let $I^{BC;v}_{r,q}$ be the set of walks in $G^{BC}_{r,q}$ whose height at the $B|C$ interface is $v$. Define the truncated walk set
	\begin{equation}
	H^{BC,<b}_{r,q}:=\bigsqcup_{v<bk} I^{BC;v}_{r,q}.
	\end{equation}
	Equivalently, by concatenation at the $B|C$ interface,
	\begin{equation}
	\ket{\unnormstate{H^{BC,<b}_{r,q}}} = \sum_{v < bk} \ket{\unnormstate{G^B_{r,v}}} \ket{\unnormstate{G^C_{v,q}}},
	\qquad
	\norm{H^{BC,<b}_{r,q}} = \sum_{v < bk} N^B_{r,v} N^C_{v,q}.
	\end{equation}
\end{definition}

\begin{definition} These states are then used to define the approximate projector:
	\begin{equation}
	AG_{BC} = \sum_{r < bk} \ket{\normstate{H^{BC,<b}_{r,q}}}\bra{\normstate{H^{BC,<b}_{r,q}}}.
	\end{equation}
	Equivalently,
	\begin{equation}
	AG_{BC} = \sum_{r < bk} {1 \over \norm{H^{BC,<b}_{r,q}}} \sum_{v,v' < bk} \ket{\unnormstate{G^B_{r,v}}} \ket{\unnormstate{G^C_{v,q}}} \bra{\unnormstate{G^B_{r,v'}}} \bra{\unnormstate{G^C_{v',q}}}.
	\end{equation}
\end{definition}

According to Lemma \ref{lm:proj-approx}, this approximates $G_{BC}'$, and we have

\begin{equation}
\begin{aligned}
\widetilde{\lim\limits_{k \to \infty}} \op \sup\limits_{\{\psi^C_{v,q}\} \in P_{k, q}} \mel{A\phi_{p,q}}{G_{BC}'}{A\phi_{p,q}} \cp 
= \widetilde{\lim\limits_{k \to \infty}}\op \sup\limits_{\{\psi^C_{v,q}\} \in P_{k, q}} \mel{A\phi_{p,q}}{AG_{BC}}{A\phi_{p,q}} \cp 
\end{aligned}
\end{equation}

Combining Propositions \ref{pr:reduction-of-phi-pq} and \ref{pr:approximation-of-reduction} with the above we find

\begin{equation}
\begin{aligned}
\widetilde{\lim\limits_{k \to \infty}} \op \sup\limits_{\ket{\phi_{p,q}} \in \range E_k} \mel{\phi_{p,q}}{G_{BC}}{\phi_{p,q}} \cp 
= \widetilde{\lim\limits_{k \to \infty}} \op \sup\limits_{\{\psi^C_{v,q}\} \in P_{k, q}} \mel{A\phi_{p,q}}{AG_{BC}}{A\phi_{p,q}} \cp
\end{aligned}
\end{equation}

Recall that the constraint $\{\psi^C_{v,q}\} \in P_{k, q}$ comes from the necessity that $\ket{\phi_{p,q}}$ is always orthogonal to the ground space on the full chain $ABC$. To describe that in a different manner, ground states on the full chain are first approximated as follows:

\begin{definition} \label{def:approximate-full-ground-state}
	Let $I^{ABC;r,v}_{p,q}$ be the set of walks in $G^{ABC}_{p,q}$ whose interface heights at $A|B$ and $B|C$ are respectively $r$ and $v$, and which reach zero height separately in all three segments. Define
	\begin{equation}
	H^{ABC,<b}_{p,q}:=\bigsqcup_{r,v<bk} I^{ABC;r,v}_{p,q}.
	\end{equation}
	For the ground state on the full chain we will use the approximation
	\begin{equation}
	\ket{\normstate{H^{ABC,<b}_{p,q}}} = {1 \over \sqrt{\norm{H^{ABC,<b}_{p,q}}}} \sum_{r,v < bk} \ket{\unnormstate{G^A_{p,r}}} \ket{\unnormstate{G^B_{r,v}}} \ket{\unnormstate{G^C_{v,q}}}
	\end{equation} 
	with the corresponding normalization factor 
	\begin{equation}
	\norm{H^{ABC,<b}_{p,q}} = \sum_{r,v < bk} N^{A}_{p,r} N^{B}_{r,v} N^{C}_{v,q}.
	\end{equation} 
\end{definition}
\begin{proposition} \label{pr:fgs-good-approximation}
	The $\ket{\normstate{H^{ABC,<b}_{p,q}}}$ of Definition \ref{def:approximate-full-ground-state} superpolynomially approximate the true ground states $\ket{GS^{ABC}_{p,q}}$ on the full chain, for all $p,q \le (1+c) k$.
\end{proposition}
\begin{proof}[Proof of Proposition \ref{pr:fgs-good-approximation}]
	It follows directly by Lemma \ref{lm:three-way-split}.
\end{proof}

By construction, the state $\ket{\normstate{H^{ABC,<b}_{p,q}}}$ only contains walks that simultaneously reach zero height in all three segments $A,B,C$. Therefore, when computing an overlap such as $\inner{\normstate{H^{ABC,<b}_{p,q}}}{\phi_{p,q}}$, only walks with the same property will be picked out within $\ket{\phi_{p,q}}$. These are precisely the walks that have been included in the substitute state $\ket{\phi_{p,q}'}$; moreover, these walks have the same relative prefactors in $\ket{\phi_{p,q}'}$ as they did in $\ket{\phi_{p,q}}$. The only possible difference is in the absolute values of the prefactors, which may occur because we imposed that $\ket{\phi_{p,q}'}$ has norm 1, while the corresponding component within $\ket{\phi_{p,q}}$ may have  norm $\leq 1$. However, in eq. \eqref{eq:substitute-fullgs-orthogonality} below we are taking the supremum on both sides, which for $\ket{\phi_{p,q}}$ will occur when no $\ket{{\psi'}^C}$ or $\ket{{\psi''}^C}$ components are present. Then, the absolute values of the prefactors are also identical, and equality between the LHS and RHS below follows:
\begin{equation} \label{eq:substitute-fullgs-orthogonality}
\widetilde{\lim\limits_{k \to \infty}} \op \sup\limits_{\ket{\phi_{p,q}} \in \range E_k} |\inner{\normstate{H^{ABC,<b}_{p,q}}}{\phi_{p,q}}| \cp  =  \widetilde{\lim\limits_{k \to \infty}} \op \sup\limits_{\{\psi^C_{v,q}\} \in P_{k, q}} |\inner{\normstate{H^{ABC,<b}_{p,q}}}{\phi_{p,q}'}| \cp 
\end{equation}

Since the true $\ket{\phi_{p,q}}$ must be orthogonal to the true ground state $\ket{GS_{p,q}^{ABC}}$, we use Lemma \ref{lm:overlaps-and-expectations} twice to find that the overlap of their approximations also vanishes:

\begin{align*}
0 &=\widetilde{\lim\limits_{k \to \infty}} \op \sup\limits_{\ket{\phi_{p,q}} \in \range E_k} |\inner{GS_{p,q}^{ABC}}{\phi_{p,q}}| \cp = \widetilde{\lim\limits_{k \to \infty}} \op \sup\limits_{\ket{\phi_{p,q}} \in \range E_k} |\inner{\normstate{H^{ABC,<b}_{p,q}}}{\phi_{p,q}}| \cp \\
& = \widetilde{\lim\limits_{k \to \infty}} \op \sup\limits_{\{\psi^C_{v,q}\} \in P_{k, q}} |\inner{\normstate{H^{ABC,<b}_{p,q}}}{\phi_{p,q}'}| \cp 
= \widetilde{\lim\limits_{k \to \infty}} \op \sup\limits_{\{\psi^C_{v,q}\} \in P_{k, q}} |\inner{\normstate{H^{ABC,<b}_{p,q}}}{A\phi_{p,q}}| \cp 
\end{align*}
The overlap on the last line can be computed at any particular $(p,q)$:
\begin{equation}
\begin{aligned}
\inner{\normstate{H^{ABC,<b}_{p,q}}}{A\phi_{p,q}} 
= \sum_{r,r',v,v' < b k}\frac{ \ob \bra{\unnormstate{G^A_{p,r'}}} \bra{\unnormstate{G^B_{r',v'}}} \bra{\unnormstate{G^C_{v',q}}}  \cb \ob \ket{\unnormstate{G^A_{p,r}}} \ket{\unnormstate{G^B_{r,v}}} \ket{\psi^C_{v,q}} \cb }{ \sqrt{ \mathcal{N}^{A\phi}_{p,q} \cdot \norm{H^{ABC,<b}_{p,q}}}}
\end{aligned}
\end{equation}
In order to have any overlap, states must have identical numbers of unbalanced steps. Recalling from Theorem \ref{thm:degenerate-pq-gs} and Definition \ref{def:walk_set_definitions} that
\begin{equation}
\inner{\unnormstate{G^i_{a',b'}}}{\unnormstate{G^i_{a,b}}} = N^i_{a,b} \cdot \delta_{a,a'} \cdot \delta_{b,b'} \quad \quad \forall i \in \{A, B, C \}
\end{equation}
one obtains for the above
\begin{align}
\inner{\normstate{H^{ABC,<b}_{p,q}}}{A\phi_{p,q}} 
\
\label{eq:aortho} &= {1 \over \sqrt{ \mathcal{N}^{A\phi}_{p,q} \cdot \norm{H^{ABC,<b}_{p,q}}}} \sum_{r,v < b k} N^{A}_{p,r} N^{B}_{r,v} \cdot \inner{\unnormstate{G^C_{v,q}}}{\psi^C_{v,q}} 
\end{align}

A similar computation can be performed for the matrix element $\mel{A\phi_{p,q}}{AG_{BC}}{A\phi_{p,q}}$:

\begin{align*}
\mel{A\phi_{p,q}}{AG_{BC}}{A\phi_{p,q}} = {1 \over \mathcal{N}^{A\phi}_{p,q}} \sum_{r_1, v_1, r_2, v_2, r, v, v'} &{1 \over \norm{H^{BC,<b}_{r,q}}} \cdot \inner{\unnormstate{G^A_{p,r_1}}}{\unnormstate{G^A_{p,r_2}}} \cdot \inner{\unnormstate{G^B_{r_1,v_1}}}{\unnormstate{G^B_{r,v}}}\\
& \cdot \inner{\unnormstate{G^B_{r,v'}}}{\unnormstate{G^B_{r_2,v_2}}} \cdot \inner{\psi^C_{v_1,q}}{\unnormstate{G^C_{v,q}}} \cdot \inner{\unnormstate{G^C_{v',q}}}{\psi^C_{v_2,q}}
\end{align*}
which, taking care of all delta functions that effectively remove some summation variables, gives
\begin{align}
\mel{A\phi_{p,q}}{AG_{BC}}{A\phi_{p,q}} &= {1 \over \mathcal{N}^{A\phi}_{p,q}} \sum_{v_1, v_2, r} {1 \over \norm{H^{BC,<b}_{r,q}}} \cdot N^A_{p,r} \cdot N^B_{r, v_1} \cdot N^B_{r, v_2} \cdot \inner{\psi^C_{v_1,q}}{\unnormstate{G^C_{v_1,q}}} \cdot \inner{\unnormstate{G^C_{v_2,q}}}{\psi^C_{v_2,q}}\\
\label{eq:aovlp} &= {1 \over \mathcal{N}^{A\phi}_{p,q}} \sum_{r} {N^A_{p,r} \over \norm{H^{BC,<b}_{r,q}}} \cdot \op \sum_{v_1} N^B_{r, v_1} \cdot  \inner{\psi^C_{v_1,q}}{\unnormstate{G^C_{v_1,q}}} \cp \cdot \op \sum_{v_2} N^B_{r, v_2} \cdot \inner{\unnormstate{G^C_{v_2,q}}}{\psi^C_{v_2,q}} \cp
\end{align}
In both \eqref{eq:aortho} and \eqref{eq:aovlp} we can see that any component of $\ket{\psi^C_{v,q}}$ not in the ground space of the Hamiltonian acting on $C$ contributes nothing to the numerator, while making the normalization factor $\mathcal{N}^{A\phi}_{p,q}$ in the denominator larger.  Therefore it is advantageous to separate the component of $\ket{\psi^C_{v,q}}$ in the ground space corresponding to $C$:

\begin{definition}
	For every initial $\ket{\phi_{p,q}}$, let $\{c_v\}$ be the collection, indexed by $v$, of overlaps between $\ket{\psi^C_{v,q}}$ and the relevant (unnormalized) ground state:
	\begin{equation}
	c_v = \inner{\unnormstate{G^C_{v,q}}}{\psi^C_{v,q}}
	\end{equation}
\end{definition}
The above allows us to write
\begin{equation}
\ket{\psi^C_{v,q}} = c_v \cdot \ket{\unnormstate{G^C_{v,q}}} + \op p_v \sqrt{N_{v,q}^C} \cp \cdot  \ket{\chi_{v,q}^C}
\end{equation} 
where $\ket{\chi_{v,q}^C}$ is a normalized state, orthogonal to the ground space on segment $C$, and $p_v$ is a complex number giving the relative amplitude of $\ket{\chi_{v,q}^C}$. The square root of the normalization factor $N_{v,q}^C$ is conveniently chosen such that
\begin{equation}
\inner{\psi^C_{v,q}}{\psi^C_{v,q}} = |c_v|^2 \cdot \inner{\unnormstate{G^C_{v,q}}}{\unnormstate{G^C_{v,q}}} + |p_v|^2 \cdot N_{v,q}^C \cdot 1 = N_{v,q}^C \cdot \op |c_v|^2 + |p_v|^2 \cp
\end{equation} 
The only place where this appears is the approximate normalization factor of $\phi$:
\begin{equation}
\mathcal{N}^{A\phi}_{p,q} = \sum_{r, v < bk} N^{A}_{p,r} \cdot N^{B}_{r,v} \cdot N_{v,q}^C \cdot \op |c_v|^2 + |p_v|^2 \cp
\end{equation}
With this definition, equation \eqref{eq:aortho} becomes
\begin{equation} \label{eq:aorth2}
| \inner{\normstate{H^{ABC,<b}_{p,q}}}{A\phi_{p,q}} | = {1 \over \sqrt{ \mathcal{N}^{A\phi}_{p,q} \cdot \norm{H^{ABC,<b}_{p,q}}}} \cdot \bigg| \sum_{r,v < b k} N^{A}_{p,r} \cdot N^{B}_{r,v} \cdot N^C_{v,q} \cdot c_v \bigg|
\end{equation}
and the matrix element that we're looking to bound in the end is
\begin{equation} \label{eq:aovlp2}
\mel{A\phi_{p,q}}{AG_{BC}}{A\phi_{p,q}} = {1 \over \mathcal{N}^{A\phi}_{p,q}} \sum_{r} {N^A_{p,r} \over \norm{H^{BC,<b}_{r,q}}} \cdot \bigg| \sum_{v} N^B_{r, v} \cdot N^C_{v, q} \cdot c_v \bigg|^2
\end{equation}
The purpose is to bound the supremum of the RHS in eq. \eqref{eq:aovlp2} over all collections of factors $\{c_v, p_v \}$ that could be obtained from an initial $\ket{\phi_{p,q}} \in \range E_k$, as described above. By the (approximate) orthogonality condition, we are guaranteed that the RHS of eq. \eqref{eq:aorth2} vanishes at large $k$, for all such obtainable collections. Therefore, the vanishing of \eqref{eq:aorth2} at large $k$ is a weaker condition than the obtainability of $\{c_v, p_v \}$. In what follows we will prove that the quantity in \eqref{eq:aovlp2} vanishes at large $k$ even only under this weaker condition, which will also prove the desired result. The approximate orthogonality says that 

\begin{equation} \label{eq:cs-constraint-eq}
\widetilde{\lim\limits_{k \to \infty}} \op \sup\limits_{\{c_v, p_v \} \text{ obtainable}} {1 \over \sqrt{ \mathcal{N}^{A\phi}_{p,q} \cdot \norm{H^{ABC,<b}_{p,q}}}} \cdot \bigg| \sum_{r,v < b k} N^{A}_{p,r} \cdot N^{B}_{r,v} \cdot N^C_{v,q} \cdot c_v \bigg| \cp =0
\end{equation}
 
As discussed above, the condition that $\{c_v, p_v\}$ must be obtainable will be relaxed, and replaced by the weaker condition that the above limit vanishes. Going on, we will also abbreviate $\{c_v, p_v\}$ by $\{\tilde c_v\}$ (since the $p_v$ are unimportant). When writing $\sup_{\{\tilde c_v\}}$ we mean the supremum over all collections $\{c_v, p_v\}$ which obey the vanishing condition. We are looking to bound, in the limit of large $k$, the following:

\begin{equation}
\begin{aligned}
&\sup\limits_{0 \le p,q \le (1+c)k} \op \sup\limits_{\{\tilde c_v\}} \ob \mel{A\phi_{p,q}}{AG_{BC}}{A\phi_{p,q}} \cb \cp \\
=& \sup\limits_{0 \le p,q \le (1+c)k} \op \sup\limits_{\{\tilde c_v\}} \op {1 \over \mathcal{N}^{A\phi}_{p,q}} \sum_{r} {N^A_{p,r} \over \norm{H^{BC,<b}_{r,q}}} \cdot \bigg| \sum_{v} N^B_{r, v} \cdot N^C_{v, q} \cdot c_v \bigg|^2 \cp \cp
\end{aligned}
\end{equation}

\begin{definition}\label{def:xr-factors} The following quantities are convenient:
	\begin{equation} \label{eq:xr-factors-definition}
	x_r^{(k, \{\tilde c_v\}) } \equiv {\sum_v N_{r,v}^B \; N_{v,q}^C \; c_v \over \norm{H^{BC,<b}_{r,q}}} \; \sqrt{\norm{H^{ABC,<b}_{p,q}} \over \mathcal{N}^{A\phi}_{p,q}}
	\end{equation}
	where the $p,q$ dependence of the $x_r$ factors is kept implicit. We can simplify the expression by also defining
	\begin{equation} \label{eq:j-factors-definition}
	j^{(k, \{\tilde c_v\}) }_{p, q} = \sqrt{\norm{H^{ABC,<b}_{p,q}} \over \mathcal{N}^{A\phi}_{p,q}} = \sqrt{\sum_{u,v} N_{p,u}^A \; N_{u,v}^B \; N_{v,q}^C \over \sum_{u,v} N_{p,u}^A \; N_{u,v}^B \; N_{v,q}^C \; |c_v|^2 }
	\end{equation}
    The $\pi^k_{r,v}$ ratios, defined in \eqref{eq:rationdefn} and \eqref{eq:rationdefn_infty} will also be useful:
    \begin{align}
        \pi^k_{r,v} &= {N^k_{r,v} \over N^k_{r,0}}\\
        \pi^\infty_{r,v} &= \lim\limits_{k \to \infty} \pi^k_{r,v}
    \end{align}
\end{definition}
\begin{remark}
    By symmetry of the normalization factors, $N^k_{r,v} = N^k_{v,r}$, it follows that
    \begin{eqnarray}
        {N^k_{r,v} \over N^k_{0,v}} = {N^k_{v, r} \over N^k_{v, 0}} = \pi^k_{v, r}.
    \end{eqnarray}
\end{remark}

Writing $\norm{H^{BC,<b}_{r,q}}$ as a sum, we obtain
\begin{equation} \label{eq:xr-factors-reduced}
x_r \equiv {\sum_v \pi_{v, r}^{\slength{B}} \; N_{0,v}^B \; N_{v,q}^C \; c_v \over \sum_v \pi_{v, r}^{\slength{B}} \; N_{0,v}^B \; N_{v,q}^C} \; j_{p, q}
\end{equation}
Note that the $x_r$ are invariant (up to a global phase) under an uniform scaling $c_v \to a \cdot c_v$ and $p_v \to a \cdot p_v$ (fixed $a \in \mathbb C$, for all $v$ simultaneously), i.e. they only depend on the relations among the various $c_v, p_v$ and not their overall magnitudes (as it should be, since the $c_v, p_v$ were defined as part of an unnormalized expression). This definition is convenient because eq. \eqref{eq:aorth2} becomes
\begin{equation} \label{eq:xr-approximate-orthonormality}
| \inner{\normstate{H^{ABC,<b}_{p,q}}}{A\phi_{p,q}} | = \bigg| \sum_{r} {N^{A}_{p,r} \; \norm{H^{BC,<b}_{r,q}} \over \norm{H^{ABC,<b}_{p,q}}} \; x_r \bigg|
\end{equation}
while eq. \eqref{eq:aovlp2} turns into
\begin{equation} \label{eq:xr-approximate-matrix-element}
\mel{A\phi_{p,q}}{AG_{BC}}{A\phi_{p,q}} = \sum_{r} {N^{A}_{p,r} \; \norm{H^{BC,<b}_{r,q}} \over \norm{H^{ABC,<b}_{p,q}}} \; |x_r|^2
\end{equation}
with the property that, in both of the above, the coefficients in front of the $x_r$ sum to 1, by construction:
\begin{equation} 
\sum_{r} {N^{A}_{p,r} \; \norm{H^{BC,<b}_{r,q}} \over \norm{H^{ABC,<b}_{p,q}}} = {\sum_{r} N^{A}_{p,r} \; \norm{H^{BC,<b}_{r,q}} \over \norm{H^{ABC,<b}_{p,q}}} = {\sum_{r} N^{A}_{p,r} \; \norm{H^{BC,<b}_{r,q}} \over \sum_{r} N^{A}_{p,r} \; \norm{H^{BC,<b}_{r,q}}} = 1
\end{equation}
The key must be in how the various $x_r$ relate to each other. The only $r$ dependence is in the $\pi_{v, r}^{\slength{B}}$ factors, present both in the numerator and denominator of eq. \eqref{eq:xr-factors-reduced}.\\ 

As we will see in the following section, at large $k$, the quantity $\pi_{v, r}^{k}$ depends very weakly on $v$, allowing us to show that the $x_r$ factors with various $r$ are very close to each other; and from that, we will prove that prove that the quantity in eq. \eqref{eq:xr-approximate-matrix-element} becomes very small at large $k$.\\

\section{Normalizations and the $x_r$ factors}\label{sec:normalizations}

In this rather technical section we show that the relevant $x_r$ factors approach $r$-independent quantities at large $k$:

\begin{lemma}\label{lm:xr-closeness}
	With Definition \ref{def:xr-factors} and the condition \eqref{eq:cs-constraint-eq} constraining the $\{\tilde c_s\}$ factors, we have that
	\begin{equation}
	\lim\limits_{k \to \infty} \op \sup\limits_{0 \le p,q \le (1+c)k} \op \sup\limits_{\{\tilde c_v\}} \op \sup\limits_{r < b k} |x_{r}^{(k,\{\tilde c_v\})} - x_{0}^{(k,\{\tilde c_v\})}| \cp \cp \cp = 0
	\end{equation}	
\end{lemma}
Before proving Lemma \ref{lm:xr-closeness}, we rephrase the main result of Appendix \ref{sec:norm-ratios-appendix}:

\begin{proposition} \label{pr:convergence-requirement}
	Make the notation $\pi_r \equiv \pi_{0,r}^\infty$. For any $\e > 0$, there exists a $k_0$ such that for all $k > k_0$ the following holds true:
	\begin{equation} \label{eq:convergence-requirement}
	\forall r, v < b k \quad \text{we have}\quad 0 \le 1 - {\pi_{v, r}^k \over \pi_r}  < \e
	\end{equation}
\end{proposition}

\begin{proof}[Proof of Proposition \ref{pr:convergence-requirement}]
	
	We know from Theorem \ref{thm:mainec} that, given the value of $t$, there exist constants $C_*,\alpha,\beta>0$ such that for all $r + v \le k$,
	\begin{equation}\label{eq:mainec0}
	0\leq \pi_r - \pi_{v, r}^{k}\leq C_* \; t^{\alpha(k-\beta r - \beta v)} \; \pi_r \quad \iff \quad 0 \leq 1 - {\pi_{v, r}^{k} \over \pi_r} \leq C_* \; t^{\alpha [k-  \beta (r + v)]}
	\end{equation}
	When imposing $r, v < bk$ we find that 
	\begin{equation}
	k - \beta (r + v) > k - \beta \; 2 b k = k (1 - 2 \beta b)
	\end{equation}
	This is where the condition $b < {1 \over 4 \beta}$ from Assumption \ref{as:low-imb-half-split-constants} plays an important role: with this constraint on $b$, we find that the rightmost term above is greater than or equal to $k/2$. With $\alpha > 0$ and $t < 1$, we find that
	\begin{equation}
	t^{\alpha [k-  \beta (r + v)]} < t^{\alpha k /2}
	\end{equation}
	and so it follows that
	\begin{equation}
	0 \leq 1 - {\pi_{r, v}^{(k)} \over \pi_r} < C_* \; t^{\alpha k /2}
	\end{equation}
	For the desired inequality \eqref{eq:convergence-requirement} to hold, it is clear that it is sufficient to take $k$ large enough such that
	\begin{equation}
	C_* \; t^{\alpha k /2} \stackrel{!}{<} \e
	\end{equation}
	Due to the positivity of all relevant constants, and the fact that $t<1$, this is satisfied if
	\begin{equation}
	k > {2 \over \alpha} \; {\ln \op {\e / C_*} \cp \over \ln t}
	\end{equation}
	which shows that a suitable $k_0$ can indeed be chosen, completing the proof.	
\end{proof}

We now turn to proving the main result of the section:

\begin{proof}[Proof of Lemma \ref{lm:xr-closeness}]
	Fix some small $\e>0$, and using Proposition \ref{pr:convergence-requirement}, take $k$ large enough such that 
	\begin{equation}
	\forall r, v < b k \quad \text{we have}\quad 0 \le 1 - {\pi_{v, r}^{\slength{B}} \over \pi_r}  < \e
	\end{equation}
	This is equivalent to
	\begin{equation}
	\forall r, v < b k \quad \text{we have}\quad 0 \le \pi_r - \pi_{v, r}^{\slength{B}}  < \e \pi_r \quad \iff \quad \pi_r \ge \pi_{v, r}^{\slength{B}}  > \pi_r \; (1 - \e)
	\end{equation}
	Take arbitrary, but specific $p,q$ obeying the condition of the lemma, and a collection of complex $\{\tilde c_v\}$ obeying the constraint \eqref{eq:cs-constraint-eq}. We separate the real and imaginary parts of the $c_v$ factors, so write $c_v = d_v + i e_v$ where $d_v, e_v \in \R$ for all $v$. Some of the $\{d_v\}$ and $\{e_v\}$ may be positive, and some may be negative. We define accordingly
	\begin{definition}
		\begin{align}
		\label{eq:defx} x_r^+ = {\sum_{v; \; d_v > 0} \pi_{v, r}^{\slength{B}} \; N_{0, v}^B \; N_{v, q}^C \; |d_v| \over \sum_v \pi_{v, r}^{\slength{B}} \; N_{0, v}^B \; N_{v, q}^C } \; j_{p,q} \qquad&\qquad
		x_r^- = {\sum_{v; \; d_v < 0} \pi_{v, r}^{\slength{B}} \; N_{0, v}^B \; N_{v, q}^C \; |d_v| \over \sum_v \pi_{v, r}^{\slength{B}} \; N_{0, v}^B \; N_{v, q}^C } \; j_{p,q}\\
		\label{eq:defy} y_r^+ = {\sum_{v; \; e_v > 0} \pi_{v, r}^{\slength{B}} \; N_{0, v}^B \; N_{v, q}^C \; |e_v| \over \sum_v \pi_{v, r}^{\slength{B}} \; N_{0,v}^B \; N_{v,q}^C } \; j_{p,q} \qquad&\qquad
		y_r^- = {\sum_{v; \; e_v < 0} \pi_{v, r}^{\slength{B}} \; N_{0,v}^B \; N_{v,q}^C \; |e_v| \over \sum_v \pi_{v, r}^{\slength{B}} \; N_{0,v}^B \; N_{v,q}^C } \; j_{p,q}
		\end{align}
		the partial sums including only positive or only negative terms in the numerator of the fraction, such that $x_r^+, x_r^-, y_r^+, y_r^- >0$ and $x_r = \oa x_r^+ - x_r^- \ca + i \oa y_r^+ - y_r^- \ca$.
	\end{definition}
	
	We will prove that $x_r^\pm$ is close to $x_0^\pm$, and the same for the $y^{\pm}_r$ and $y^{\pm}_0$ factors. This will then show that $x_r$ is close to $x_0$. To begin, we will only work with the $x_r^\pm$. Since $\pi_{v, r}^{\slength{B}} \le \pi_r$, the denominator of all the fractions above is bounded by 
	\begin{equation}
	\sum_v \pi_{v, r}^{\slength{B}} \; N_{0,v}^B \; N_{v,q}^C \le \pi_r \sum_v N_{0,v}^B \; N_{v,q}^C \quad \implies \quad {1 \over \sum_v \pi_{v, r}^{\slength{B}} \; N_{0,v}^B \; N_{v,q}^C} \ge {1 \over \pi_r} \; {1 \over \sum_v N_{0,v}^B \; N_{v,q}^C}
	\end{equation}
	Moving the $1/\pi_r$ factor up to the numerator of the fraction, we get
	\begin{equation}
	x_r^+ \ge {\sum_{v; \; d_v>0} {\pi_{v, r}^{\slength{B}} \over \pi_r} \; N_{0,v}^B \; N_{v,q}^C \; |d_v| \over \sum_v N_{0,v}^B \; N_{v,q}^C } \; j_{p,q}
	\end{equation}
	Observing that
	\begin{equation}
	x_0^+ = {\sum_{v; \; d_v>0} N_{0,v}^B \; N_{v,q}^C \; |d_v| \over \sum_v N_{0,v}^B \; N_{v,q}^C } \; j_{p,q}
	\end{equation}
	and from the normalization ratio convergence assumption
	\begin{equation}
	{\pi_{v, r}^{\slength{B}} \over \pi_r} > 1 - \e
	\end{equation}
	together with the fact that all the terms $N_{0,v}^B \; N_{v,q}^C \; |d_v|$ in the numerator sum are positive, we find
	\begin{equation} \label{eq:xrp_bound_higher}
	x_r^+ > {\sum_{v; \; d_v>0} {\pi_{v, r}^{\slength{B}} \over \pi_r} \; N_{0,v}^B \; N_{v,q}^C \; |d_v| \over \sum_v N_{0,v}^B \; N_{v,q}^C } \; j_{p, q} > {\sum_{v; \; d_v>0} (1 - \e) \; N_{0,v}^B \; N_{v,q}^C \; |d_v| \over \sum_v N_{0,v}^B \; N_{v,q}^C } \; j_{p,q} = (1 - \e) \; x_0^+
	\end{equation}
	and by an identical argument $x_r^- > (1 - \e) \; x_0^-$. For the corresponding upper bounds we again start with the denominators
	\begin{equation}
	\sum_v \pi_{v, r}^{\slength{B}} \; N_{0,v}^B \; N_{v,q}^C > \pi_r (1 - \e) \sum_v N_{0,v}^B \; N_{v,q}^C \; \implies \; {1 \over \sum_v \pi_{v, r}^{\slength{B}} \; N_{0,v}^B \; N_{v,q}^C} < {1 \over 1- \e} \; {1 \over \pi_r} \; {1 \over \sum_v N_{0,v}^B \; N_{v,q}^C}
	\end{equation}
	and with ${\pi_{v, r}^{\slength{B}} / \pi_r} \le 1$ we see
	\begin{equation} \label{eq:xrp_bound_lower}
	x_r^+ < {1 \over 1 - \e} \; {\sum_{v; \; d_v>0} {\pi_{v, r}^{\slength{B}} \over \pi_r} \; N_{0,v}^B \; N_{v,q}^C \; |d_v| \over \sum_v N_{0,v}^B \; N_{v,q}^C } \; j_{p,q} \le {1 \over 1 - \e} \; {\sum_{v; \; d_v>0} N_{0,v}^B \; N_{v,q}^C \; |d_v| \over \sum_v N_{0,v}^B \; N_{v,q}^C } \; j_{p,q} = {1 \over 1 - \e} \; x_0^+
	\end{equation}
	and similarly $x_r^- < x_0^- / (1 - \e)$. It is now clear from \eqref{eq:xrp_bound_higher} and \eqref{eq:xrp_bound_lower} that $\e \to 0$ implies $x_r^+ \to x_0^+$; analogous properties will be true for $x_r^-$ as well as $y_r^{\pm}$, which will lead to the conclusion $x_r \to x_0$. However, to show uniformity over all the parameters required by Lemma \ref{lm:xr-closeness}, we proceed carefully.\\
    
    By elementary manipulations, we find
	\begin{equation}
	{1 \over 1- \e} = 1 + \e + \e \; {\e \over 1 - \e} 
	\end{equation}
	and taking $\e$ small enough that ${\e \over 1 - \e} < 1$, we find ${1 \over 1- \e} < 1 + 2\e$. This gives
	\begin{equation}
	(1 - \e) \; x_0^\pm < x_r^\pm < (1 + 2\e ) \; x_0^\pm
	\end{equation}
	For symmetry purposes, use $1 - \e > 1 - 2\e$ and write
	\begin{equation}
	(1 - 2 \e) \; x_0^\pm < x_r^\pm < (1 + 2\e ) \; x_0^\pm
	\end{equation}
	which gives a lower bound on the difference $x_r^+ - x_r^-$ of
	\begin{equation}
	x_r^+ - x_r^- > (1 -2 \e) x_0^+-  (1 + 2\e) x_0^- = (x_0^+ - x_0^-) - 2\e (x_0^+ + x_0^-)
	\end{equation}
	Recalling from the definitions above that the real part of $x_r$ is $x_r^+ - x_r^-$, we see that the above reads $\Re (x_r) > \Re (x_0) - 2\e (x_0^+ + x_0^-)$. The upper bound gives similarly $\Re (x_r) < \Re (x_0) + 2\e (x_0^+ + x_0^-)$, so we have
	\begin{equation}
	|\Re(x_r - x_0)| < 2\e (x_0^+ + x_0^-)
	\end{equation}
	To complete the argument, we need to find a bound on $x_0^+ + x_0^-$. We start by examining the square of the quantity $j_{p,q}$ defined in \ref{def:xr-factors} above:
	\begin{equation}
	j_{p,q}^2 = {\sum_{u,v} N_{p,u}^A \; N_{u,v}^B \; N_{v,q}^C \over \sum_{u,v} N_{p,u}^A \; N_{u,v}^B \; N_{v,q}^C \; \op |c_v|^2  + |p_v^2| \cp} = {\sum_{u,v} N_{p,u}^A \; \pi_{v, u}^{\slength{B}} \; N_{0,v}^B \; N_{v,q}^C \over \sum_{u,v} N_{p,u}^A \; \pi_{v, u}^{\slength{B}} \; N_{0,v}^B \; N_{v,q}^C \; \op |c_v|^2  + |p_v^2| \cp }
	\end{equation}
	All the terms in the sums above are positive, so we can safely use the bounds $\pi_{v, u}^{\slength{B}} \le \pi_u$ in the numerator and $\pi_{v, u}^{\slength{B}} > (1 - \e) \pi_u$ in the denominator, to bound $j_{p, q}^2$ from above. Then the $u$ sums cancel and we obtain
	\begin{equation}
	j_{p,q} < \sqrt{1 \over 1 - \e} \; \sqrt{\sum_v N_{0,v}^B \; N_{v,q}^C \over \sum_v N_{0,v}^B \; N_{v,q}^C \; \op |c_v|^2  + |p_v^2| \cp }
	\end{equation}
	The quantity $x_0^+ + x_0^-$ is related to the real parts of all the $c_v$ coefficients, so we will retain only that part of the bound. We also discard the unimportant $p_v$ factors. Namely, by using $|c_v|^2 = |d_v|^2 + |e_v|^2$ and the fact that all the terms in the denominator are positive, we keep
	\begin{equation}
	j_{p,q} < 
	\sqrt{1 \over 1 - \e} \; \sqrt{\sum_v N_{0,v}^B \; N_{v,q}^C \over \sum_v N_{0,v}^B \; N_{v,q}^C \; |d_v|^2 }
	\end{equation}
	Of course, it also holds true by the same token (and would be useful if we were discussing the $y_r^\pm$ factors) that we have a bound involving the imaginary parts of the $\{\tilde c_v\}$ coefficients:
	\begin{equation}
	j_{p,q} < \sqrt{1 \over 1 - \e} \; \sqrt{\sum_v N_{0,v}^B \; N_{v,q}^C \over \sum_v N_{0,v}^B \; N_{v,q}^C \; |e_v|^2 }
	\end{equation}
	From the definitions \eqref{eq:defx}, together with $\pi_{v, 0}^{\slength{B}} \equiv 1$ for all $v$, we see that
	\begin{equation}
	x_0^+ + x_0^- = {\sum_v N_{0,v}^B \; N_{v,q}^C \; |d_v| \over \sum_v N_{0,v}^B \; N_{v,q}^C } \; j_{p,q} < {\sum_v N_{0,v}^B \; N_{v,q}^C \; |d_v| \over \sum_v N_{0,v}^B \; N_{v,q}^C } \; \sqrt{1 \over 1 - \e} \; \sqrt{\sum_v N_{0,v}^B \; N_{v,q}^C \over \sum_v N_{0,v}^B \; N_{v,q}^C \; |d_v|^2 }
	\end{equation}
	where the inequality holds due to the bound on $j_{p,q}$, and the fact that all terms multiplying it are positive. Equivalently,
	\begin{equation}
	x_0^+ + x_0^- < \sqrt{1 \over 1 - \e} \; \sqrt{ \op\sum_v N_{0,v}^B \; N_{v,q}^C \; |d_v|\cp^2 \over \op \sum_v N_{0,v}^B \; N_{v,q}^C \cp \; \op \sum_v N_{0,v}^B \; N_{v,q}^C \; |d_v|^2 \cp  }
	\end{equation}
	Applying Cauchy-Schwarz with terms of the form $\sqrt{N_{0,v}^B \; N_{v,q}^C}$ and $\sqrt{N_{0,v}^B \; N_{v,q}^C} \; |d_v|$ gives that the square root is at most 1. We find
	\begin{equation}
	x_0^+ + x_0^- < \sqrt{1 \over 1 - \e} \quad \implies \quad |\Re (x_r - x_0)| < {2 \e \over \sqrt{1 - \e}}
	\end{equation}
	It is clear that exactly the same argument carries over for the imaginary part if we replace $x_r^\pm$ by $y_r^\pm$ and the real parts of the coefficients, $d_v$, by the imaginary ones, $e_v$. We then obtain the corresponding bound $|\Im (x_r - x_0)| < {2 \e \over \sqrt{1 - \e}}$, which in the end gives
	\begin{equation}
	|x_r - x_0| < {2 \sqrt 2 \; \e \over \sqrt{1 - \e}}
	\end{equation}
	which becomes arbitrarily small when we take $\e \to 0$, as desired.
\end{proof}

\section{Completing the low-imbalance proof} \label{sec:completing-low-imbalance-proof}

The goal of this section is to combine the previous results and conclude that the quantity of eq. \eqref{eq:aovlp2}, or equivalently \eqref{eq:xr-approximate-matrix-element}, goes to zero in the limit of large $k$. The argument consists of two steps: 

\begin{lemma} \label{lm:xr-smallness}
	The absolute values of all $x_r$ factors of interest vanish as $k \to \infty$. Formally, 
	\begin{equation} \label{eq:xr-smallness-eq}
	\lim\limits_{k \to \infty} \op \sup\limits_{0 \le p,q \le (1+c)k} \op \sup\limits_{\{\tilde c_v\}} \op \sup\limits_{r < b k} |x_{r}^{(k,\{\tilde c_v\})}| \cp \cp \cp = 0
	\end{equation}
\end{lemma}

\begin{lemma} \label{lm:xr2-vanishing}
	In the low imbalance limit, the quantity $\mel{A\phi_{p,q}}{AG_{BC}}{A\phi_{p,q}}$ of eq. \eqref{eq:xr-approximate-matrix-element} vanishes at large $k$:
	\begin{equation} \label{eq:xr2-to-bound}
	\lim\limits_{k \to \infty} \op \sup\limits_{0 \le p,q \le (1+c)k} \op \sup\limits_{\{\tilde c_v\}} \op \sum_{r < bk} {N^{A}_{p,r} \cdot \norm{H^{BC,<b}_{r,q}} \over \norm{H^{ABC,<b}_{p,q}}} \cdot |x_r|^2 \cp \cp \cp = 0
	\end{equation}
\end{lemma}

\begin{proof}[Proof of Lemma \ref{lm:xr-smallness}]
	
	Take any $\e > 0$, and using the approximate orthogonality (eq. \eqref{eq:cs-constraint-eq}) pick $k_1$ large enough such that for all $k> k_1$ we have
	\begin{equation} 
	| \inner{\normstate{H^{ABC,<b}_{p,q}}}{A\phi_{p,q}} | = \bigg| \sum_{r} {N^{A}_{p,r} \cdot \norm{H^{BC,<b}_{r,q}} \over \norm{H^{ABC,<b}_{p,q}}} \cdot x_r \bigg| < {\e \over 6}
	\end{equation}
	where the $k$ dependence in various quantities above has been kept implicit for notational simplicity.
	\begin{definition} \label{def:fr-factors}
		For brevity make the following notation:
		\begin{equation}
		f_r \equiv {N^{A}_{p,r} \cdot \norm{H^{BC,<b}_{r,q}} \over \norm{H^{ABC,<b}_{p,q}}}
		\end{equation}
		These $f_r$ factors implicitly depend on $k$, $p$, and $q$, but for now we'll only keep the $r$ dependence explicit. Recall that $\sum_r f_r =1$ by the definition of the denominator $\norm{H^{ABC,<b}_{p,q}}$.
	\end{definition}
	With this notation, the above becomes
	\begin{equation} \label{eq:low-imp-comp-pr1}
	| \inner{\normstate{H^{ABC,<b}_{p,q}}}{A\phi_{p,q}} | = \bigg| \sum_{r} f_r \cdot x_r \bigg| < {\e \over 6}
	\end{equation}
	Using the result of Lemma \ref{lm:xr-closeness}, take $k_2$ large enough such that for all $k > k_2$ it is true that
	\begin{equation} \label{eq:low-imp-comp-pr2}
	\sup\limits_{0 \le p,q \le (1+c)k} \op \sup\limits_{\{\tilde c_v\}} \op \sup\limits_{r < b k} |x_{r}^{(k,\{\tilde c_v\})} - x_{0}^{(k,\{\tilde c_v\})}| \cp \cp < {\e \over 6}
	\end{equation}
	For the rest of this proof work with $k > k_1, k_2$, such that both \eqref{eq:low-imp-comp-pr1} and \eqref{eq:low-imp-comp-pr2} hold. Also take any $r < bk$. We will deal with the real and imaginary parts of $x_r$ separately. Since the absolute value of $x_r - x_0$ is below $\e/6$, its real part will satisfy the same property:
	\begin{equation}
	\left| \Re (x_r - x_0) \right| < {\e \over 6} \quad \implies \quad \Re(x_0) - {\e \over 6} < \Re(x_r) < \Re(x_0) + {\e \over 6}
	\end{equation}
	Since all $f_r$ are positive, we multiply the above by $f_r$ and sum over $r$ (remembering $\sum_r f_r = 1$) to find
	\begin{equation}
	\sum_r f_r \op \Re(x_0) - {\e \over 6} \cp < \sum_r f_r \Re(x_r)  < \sum_r f_r \op \Re(x_0) + {\e \over 6} \cp \quad
		\end{equation} 
		and so
			\begin{equation}
 \Re(x_0) - {\e \over 6} < \sum_r f_r \Re(x_r)  < \Re(x_0) + {\e \over 6}
	\end{equation} 
	There are three possible cases, depending on the signs of the quantities  $\Re(x_0) - {\e \over 6}$ and $\Re(x_0) + {\e \over 6}$ above. If they have the same sign, be it positive or negative, we get to the same conclusion:
	\begin{equation}
	{\e \over 6} > |\sum_r f_r x_r| > |\Re (x_0)| - {\e \over 6} \quad \implies \quad | \Re (x_0)| < {\e \over 3}
	\end{equation}
	while if they have opposite signs we directly get $| \Re (x_0)| < \e/6 < \e / 3$. Combined with the property that $| \Re (x_r - x_0) | < {\e / 6}$, we find through the triangle inequality that $|\Re (x_r)| < \e /2$ for all $r$ values of interest. An identical argument follows for the imaginary part, giving $|\Im (x_r)| < \e / 2$. Together, they show that $|x_r| < \e$, completing the proof of the lemma.
\end{proof}

\begin{proof}[Proof of Lemma \ref{lm:xr2-vanishing}]
	With Definition \ref{def:fr-factors}, the sum in eq. \eqref{eq:xr2-to-bound} becomes $\sum_{r < bk} f_r \cdot |x_r|^2$. Fix any $\e > 0$, and use the result of Lemma \ref{lm:xr-smallness} to pick $k_0$ large enough such that for all $k>k_0$, the inequality $|x_r| < \sqrt \e$ holds true at all $r < bk$. Then it follows that $|x_r|^2 < \e$, and so
	\begin{equation}
	\sum_{r < bk} f_r \cdot |x_r|^2 < \sum_{r < bk} f_r \cdot \e < \e \cdot \sum_{r} f_r = \e \cdot 1 = \e
	\end{equation}
	where for the second inequality we switch from summing over $r$ only up to $bk$, to summing over all possible values. Since the $f_r$ factors are positive, the inequality holds true.
	
	The conclusion that $\sum_{r < bk} f_r \cdot |x_r|^2 < \e$ holds for any relevant collection $\{\tilde c_v\}$ and any $p,q$ in the low imbalance regime ($0 \le p,q \le (1+c)k$), as shown in Lemma \ref{lm:xr-smallness}. Therefore the proof is complete.	
\end{proof}

\section{Asymptotic formula for the string order parameter}
\label{sec:stringorder}

In this section, we prove Theorem \ref{thm:main2}.

\subsection{Setup and statement of the asymptotic formula}


\begin{figure}[t]
	\centering	
    \scalebox{0.35}{\includegraphics{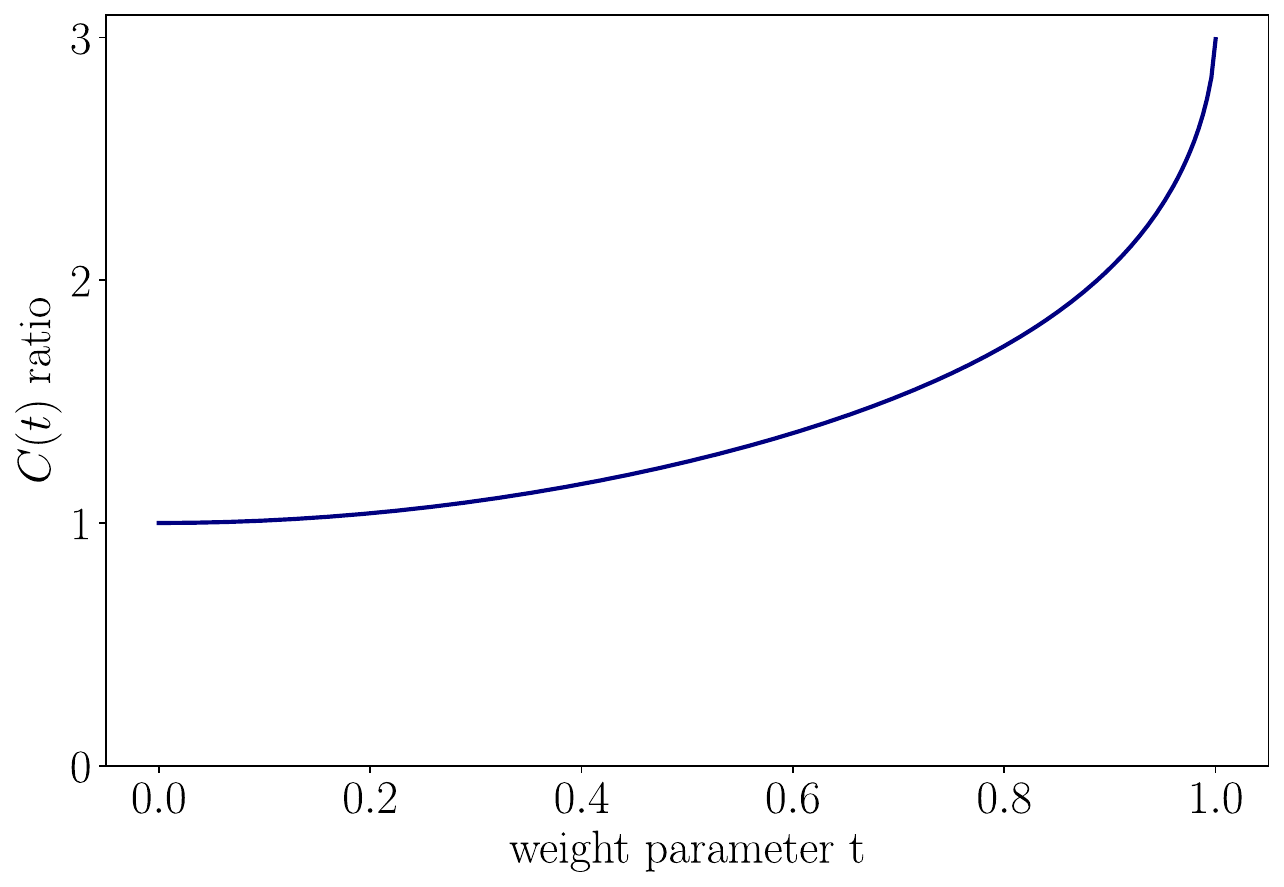}}
	\caption{Calculated dependence of the $\mathcal{C}$ ratio on $t$.}
	\label{fig:c_value_numerics}
\end{figure}

   In order to characterize the quantum phase of our spin chain at zero temperature and $0<t<1$, we introduce the string order parameter.

\begin{definition}
Let $\ex{A}_n$ denote the expectation of the operator $A$, in the unique ground state of a Motzkin Hamiltonian with boundary projectors on $n$ sites. We set
\begin{eqnarray}
    \hat{O}_n^{1,-1}(i, j) & = & S_i^z e^{i\pi \sum_{i\le l < j} S_l^z} S_j^z \nonumber \\ 
    O_n^{1,-1}(i, j) & = & \ex{\hat{O}_n^{1,-1}(i, j)}_n
\end{eqnarray}
\end{definition}

We are interested in understanding their asymptotic behavior in the limit $|i - j| \to \infty$. 
Numerical simulations by Barbiero et al.\ \cite{barbiero2017haldane} indicate a nonzero limiting value for this order parameter, as $|i-j| \to \infty$, at several values of $t \in (0, 1)$. 

On the other hand, the usual spin correlations $\mathcal{C}(i, j) = \ex{S_i^z S_j^z}$ are known to vanish exponentially in $|i-j|$ due to the spectral gap \cite{hastings2006spectral,nachtergaele2006lieb}.
The same methods that we use to understand the string order parameter at large distances will also give a direct proof that the correlations vanish in the large-distance limit, without needing to use \cite{hastings2006spectral,nachtergaele2006lieb}. This is not essential, but it shows the versatility of our representation. We therefore define a single quantity which either represents the string-order parameter ($\sigma=-1$) or the correlation function ($\sigma=1$).

\begin{definition}
Let $\sigma\in\{\pm 1\}$. Set
\begin{eqnarray}
    \label{eq:generalized_string_order_parameter}
    \hat{O}_n^\sigma (i, j) & = & S_i^z \sigma^{\sum_{i\le l < j} S_l^z} S_j^z \\ \nonumber
    O_n^\sigma (i, j) & = & \ex{\hat{O}_n^\sigma (i, j)}_n
\end{eqnarray}
so that picking $\sigma = 1$ recovers $\mathcal{C}(i,j)$, while $\sigma = -1$ corresponds to $O^{1,-1}(i, j)$, the more interesting case of the string-order parameter.
\end{definition}

 The following auxiliary quantity is useful for expressing the string order parameter limits

\begin{definition} \label{def:b_definition}
   Let $\sigma\in\{\pm 1\}$. Set     \begin{equation}\label{eq:Bseries}
        B^\sigma(t) = {\sum_{r = 0}^\infty \sigma^r t^{2r + 1} \pi_r \pi_{r + 1} \over \sum_{r = 0}^\infty \pi_r^2}
    \end{equation}
    where the factors $\pi_r \equiv \pi_{0,r}^\infty$ are defined in Proposition \ref{prop:limit}.
\end{definition}

The main result of this section is the following:
\begin{theorem}[Asymptotic formula] \label{thm:sop_values}
The following limits exist
\begin{itemize}
    \item One site pinned at the chain boundary, the other in the bulk; we refer to this as the boundary order parameter,
\begin{equation} \label{eq:def_ord_par_bdry}
    O^\sigma_\text{bdry} = \lim\limits_{n \to \infty} O_n^\sigma (1, n/2)
\end{equation}
    \item Both sites in the bulk of the chain. Correspondingly, we call this the bulk order parameter:
    \begin{equation} \label{eq:def_ord_par_bulk}
    O^\sigma_\text{bulk} = \lim\limits_{n \to \infty} O_n^\sigma (n/3, 2n/3)
\end{equation}
\end{itemize}  
and obey
    \begin{align}
     O^\sigma_\text{bulk} &= 2 \; (1 - \sigma) \os {B^\sigma(t) \over \mathcal{C}(t)} \cs^2 \label{eq:sop_thm_bulk}\\
        O^\sigma_\text{bdry} &= (1 - \sigma) \; {\mathcal{C}(t) - 1 \over \mathcal{C}^2(t)} \; B^\sigma(t) \label{eq:sop_thm_bdry}
    \end{align}
with $\mathcal{C}(t)$ defined in Corollary \ref{cor:ratios_in_two_directions}, and sketched in Figure \ref{fig:c_value_numerics}.
\end{theorem}

In particular, both expressions vanish for $\sigma=1$ (as they must by \cite{hastings2006spectral,nachtergaele2006lieb} and our proof of the spectral gap). 

\begin{remark}
  Compared to the introduction, we set for the bulk order parameter $\theta=1/3$ and $\theta=2/3$ and for the boundary order parameter $\theta=1/2$. We only do this to avoid carrying around too many parameters: the proof works in the same way for any other choice of $0<
\theta<\theta'<1$.
Moreover, definition \eqref{eq:def_ord_par_bdry} for the boundary order parameter does not rely on site $j$ being located exactly at the middle position $n/2$ of the $n-$site chain. We will find the same limit $O^\sigma_\text{bdry}$ under any procedure that simultaneously takes the distance between $j$ and either end of the chain to $\infty$. Similarly, in \eqref{eq:def_ord_par_bulk} it only matters that the distances between sites, as well as between each site and edge, are simultaneously taken to infinity. In the following, we will assume $n$ to be a multiple of 6, such that $n/2$ and $n/3$ are both integers. As argued above, this serves only for clarity of presentation, and is not a crucial element of the proof. 
\end{remark}

\begin{proof}[Proof of Theorem \ref{thm:main2} assuming Theorem \ref{thm:sop_values}]
    Existence of the limit is implied by Theorem \ref{thm:sop_values}. For the positivity, we note that  Corollary \ref{cor:sodifferencebounds} implies that for all $t<t_{\mathrm{SOP}}$ with $t_{\mathrm{SOP}}$ the unique root of the polynomial $ t^4+t^3+t^2-t-1$ on $(0,1)$, we have
    \[
\pi_{q}>t^2 \; \pi_{q+2},\qquad \textnormal{ for } q\geq 0.
\] 
Rewriting the numerator of \eqref{eq:Bseries} as
\[
\sum_{r = 0}^\infty \sigma^r t^{2r + 1} \pi_r \pi_{r + 1}=\sum_{\substack{1\leq q\leq \infty:\\ q\textnormal{ odd}}} t^q \pi_q (\pi_{q-1}-t^2\pi_{q+1}),
\]
we see that Theorem \ref{thm:main2} follows.
\end{proof}

In the remainder of this section, we prove Theorem \ref{thm:sop_values}.

\subsection{Approximate ground state representation}
We first prove the bulk statement \eqref{eq:sop_thm_bulk}, and afterwards sketch the argument for the boundary one \eqref{eq:sop_thm_bdry}, as the reasoning is similar. \\

The first step is to replace the full ground state by one of the approximate expressions introduced earlier. Recall that once the boundary projectors are included, the unique ground state is fully balanced: $p = q = 0$. Therefore, the low-imbalance approximation scheme will be appropriate.

\begin{assumption} \label{as:bulk_ord_par_split}
    We consider a chain of $n$ sites, which is to be split into three segments as follows: $A = [1, n/3)$, $B = [n/3, 2n/3]$, and $C = (2n/3, n]$. This is chosen such that the middle segment extends between sites $i = n/3$ and $j = 2n/3$, inclusively.
\end{assumption}

\begin{lemma}\label{lm:bulk_ord_par_approximation_state}
    Under the conditions of Assumption \ref{as:bulk_ord_par_split}, the bulk order parameter can be calculated using the three-segment approximate ground state of \ref{def:approximate-full-ground-state}, with no imbalance:
    \begin{equation}
	\ket{\normstate{H^{ABC,<b,(n)}_{0,0}}} = {1 \over \sqrt{\norm{H^{ABC,<b,(n)}_{0,0}}}} \sum_{r,v < b n} \ket{\unnormstate{G^A_{0,r}}} \ket{\unnormstate{G^B_{r,v}}} \ket{\unnormstate{G^C_{v,0}}}
	\end{equation} 
    Namely, we have
    \begin{equation} \label{eq:approximate_expectation_ord_par_bulk}
        O^\sigma_\text{bulk} = \lim\limits_{n \to \infty} \bra{\normstate{H^{ABC,<b,(n)}_{0,0}}} \; \hat{O}_n^\sigma (n/3, 2n/3) \; \ket{\normstate{H^{ABC,<b,(n)}_{0,0}}}
    \end{equation}
\end{lemma}
\begin{proof}[Proof of Lemma \ref{lm:bulk_ord_par_approximation_state}]
The state $\ket{\normstate{H^{ABC,<b,(n)}_{0,0}}}$ approximates the true ground state $\ket{GS_{0,0}^{(n)}}$ superpolynomially, by Proposition \ref{pr:fgs-good-approximation}. Since the operators $\hat{O}_n^\sigma$ have norm $=1$ for all $n$, we then use Lemma \ref{lm:overlaps-and-expectations}, more specifically the form \eqref{eq:oal-res3}, to find
\begin{equation}
    \lim\limits_{n \to \infty} \bra{GS_{0,0}^{(n)}} \; \hat{O}_n^\sigma (n/3, 2n/3) \; \ket{GS_{0,0}^{(n)}} = \lim\limits_{n \to \infty} \bra{\normstate{H^{ABC,<b,(n)}_{0,0}}} \; \hat{O}_n^\sigma (n/3, 2n/3) \; \ket{\normstate{H^{ABC,<b,(n)}_{0,0}}}
\end{equation}
and the conclusion follows, since the LHS above is the definition of $O^\sigma_\text{bulk}$.
\end{proof}

\subsection{Action on approximate ground states}
The main part of the proof is to calculate action of the observable $\hat{O}_n^\sigma$ on approximate ground states.

\begin{definition} \label{def:ord_par_middle_expectation}
The expression \eqref{eq:approximate_expectation_ord_par_bulk} can be simplified, since the operator $\hat{O}_n^\sigma (n/3, 2n/3)$ only acts on the middle segment $B$. To this end, we define
\begin{equation} \label{eq:top_middle_expect}
    \alpha_{r,v}^{(n)} \equiv {1 \over N_{0, 0}^B} \; \mel{\unnormstate{G^B_{r,v}}}{\hat{O}_n^\sigma (n/3, 2n/3)}{\unnormstate{G^B_{r,v}}}
\end{equation}
Note that the expression in the RHS above implicitly depends on the system size $n$, through the definition of segment $B$ given in Assumption \ref{as:bulk_ord_par_split}. We will argue in Lemma \ref{lm:sop_computation_lemma} below that the $\alpha_{r,v}^{(n)}$ have a finite limit as $n \to \infty$.

\end{definition}

\begin{lemma} \label{lm:sop_computation_lemma}
Let $r,v<bn$. The limit 
    \begin{equation} \label{eq:top_middle_expect_limit}
    \alpha_{r,v} \equiv \lim\limits_{n \to \infty} \alpha_{r,v}^{(n)}
\end{equation}
exists and is equal to
\begin{equation} \label{eq:sop_computaton_conclusion}
    \alpha_{r,v} = {1 \over \mathcal{C}^2} \; \sigma^{v-r+1} \; t^{2(r+v)} \Big( \pi_{r-1} \pi_{v+1} + \pi_{r+1} \pi_{v-1} - t^{-2} \pi_{r-1} \pi_{v-1} - t^{2} \pi_{r+1} \pi_{v+1} \Big)
\end{equation}
where the factors $\pi_r \equiv \pi_{0,r}$ are those introduced in Proposition \ref{prop:limit}. 
\end{lemma}

We recall the convention that $\pi_{q} \equiv 0$ for any $q < 0$. 

\begin{proof} [Proof of Lemma \ref{lm:sop_computation_lemma}]
We first assume $r,v \ge 1$ in the following, such that walks with imbalances $r-1$ and $v-1$ are well-defined. The situation of $r=0$ or $v=0$ is discussed at the end of the proof.\\

The description of $\ket{\unnormstate{G^B_{r,v}}}$ as a sum of walks is convenient to work with, since each such walk is an eigenstate of the generalized order parameter operator $\hat{O}_n^\sigma (n/3, 2n/3)$. The string sum in the exponent can be related to the imbalances $r,v$ on either side of the $B$ segment: 
\begin{equation} \label{eq:sop_string_imbalances}
    \sum_{i\le l \le j} S_l^z = v - r \qquad \implies \qquad \sum_{i\le l < j} S_l^z = v - r - S_j^z 
\end{equation}
so we only need to group walks in $\ket{\unnormstate{G^B_{r,v}}}$ by their starting and ending steps. With 3 possibilities for each step, this will give 9 terms total. However, the ones with a flat step $S^z = 0$ at either end will be annihilated by $\hat{O}_n^\sigma$, so we are left with only 4 relevant terms. We let $B'$ denote the collection of sites strictly between $i = n/3$ and $j = 2n/3$ (excluding the ends). Therefore, $B'$ consists of two less sites compared to $B$. We have
\begin{align}
    \ket{\unnormstate{G^B_{r,v}}} &= t^{r+v-1} \ket{d} \ket{\unnormstate{G^{B'}_{r-1,v-1}}} \ket{u} + t^{r+v+1} \ket{u} \ket{\unnormstate{G^{B'}_{r+1,v+1}}} \ket{d} \nonumber \\
    &\quad + t^{r+v} \ket{d} \ket{\unnormstate{G^{B'}_{r-1,v+1}}} \ket{d} + t^{r+v} \ket{u} \ket{\unnormstate{G^{B'}_{r+1,v-1}}} \ket{u} + \ket{\ker} \label{eq:sop_ugs_relevant terms}
\end{align}
where the $\ket{\ker}$ denote the other 5 terms, which are annihilated by the order parameter operator.\\

The four non-$\ket{\ker}$ terms in the RHS of \eqref{eq:sop_ugs_relevant terms} are individually eigenstates of the order parameter operator. Using \eqref{eq:sop_string_imbalances}, we find 
\begin{align}
    \hat{O}_n^\sigma (n/3, 2n/3)  \; \ket{d} \ket{\unnormstate{G^{B'}_{r-1,v-1}}} \ket{u} &= (-1) \times \sigma^{v-r-1} \times (+1) \; \ket{d} \ket{\unnormstate{G^{B'}_{r-1,v-1}}} \ket{u} \nonumber \\ 
    \hat{O}_n^\sigma (n/3, 2n/3)  \; \ket{u} \ket{\unnormstate{G^{B'}_{r+1,v+1}}} \ket{d} &= (+1) \times \sigma^{v-r+1} \times (-1) \; \ket{u} \ket{\unnormstate{G^{B'}_{r+1,v+1}}} \ket{d} \nonumber \\ 
    \hat{O}_n^\sigma (n/3, 2n/3)  \; \ket{d} \ket{\unnormstate{G^{B'}_{r-1,v+1}}} \ket{d} &= (-1) \times \sigma^{v-r+1} \times (-1) \; \ket{d} \ket{\unnormstate{G^{B'}_{r-1,v+1}}} \ket{d} \nonumber \\ 
    \hat{O}_n^\sigma (n/3, 2n/3)  \; \ket{u} \ket{\unnormstate{G^{B'}_{r+1,v-1}}} \ket{u} &= (+1) \times \sigma^{v-r-1} \times (+1) \; \ket{u} \ket{\unnormstate{G^{B'}_{r+1,v-1}}} \ket{u}
\end{align}
Since $\sigma^2 = 1$, we can use $\sigma^{v-r-1} = \sigma^{v-r+1}$. Recalling the norm of the unnormalized ground states,
\begin{equation}
    \inner{\unnormstate{G^{B'}_{p,q}}}{\unnormstate{G^{B'}_{p,q}}} \equiv N_{p,q}^{B'}
\end{equation}
the expectation \eqref{eq:top_middle_expect} becomes
\begin{equation}
    \alpha_{r,v}^{(n)} = {1 \over N_{0, 0}^B} \; \sigma^{v-r+1} \; t^{2(r+v)} \Big( N_{r-1, v+1}^{B'} + N_{r+1, v-1}^{B'} - t^{-2} N_{r-1, v-1}^{B'} - t^{2} N_{r+1, v+1}^{B'}\Big).
\end{equation}
We multiply and divide by $N_{0, 0}^{B'}$ to find
\begin{equation}
    \alpha_{r,v}^{(n)} = {N_{0, 0}^{B'} \over N_{0, 0}^B} \; \sigma^{v-r+1} \; t^{2(r+v)} \Bigg( {N_{r-1, v+1}^{B'} \over N_{0, 0}^{B'}} + {N_{r+1, v-1}^{B'} \over N_{0, 0}^{B'}} - t^{-2} {N_{r-1, v-1}^{B'} \over N_{0, 0}^{B'}} - t^{2} \; {N_{r+1, v+1}^{B'} \over N_{0, 0}^{B'}} \Bigg).
\end{equation}
The ratios inside the parentheses are taken between normalization factors on the same segment $B'$, so they can be expressed in terms of the $\pi^{k}_{p,q}$ factors defined in \eqref{eq:rationdefn}, where $k = n/3-1$ is the length of the $B'$ segment:
\begin{equation} \label{eq:sop_comp_l1}
    {N_{p,q}^{B'} \over N_{0, 0}^{B'}} = \pi^{n/3-1}_{p,q} \pi^{n/3-1}_{0,p}
\end{equation}
According to Corollary \ref{cor:ratios_in_two_directions}, the above converges to $\pi_q \pi_p$ in the limit $n \to \infty$. On the other hand, the ratio $N_{0, 0}^{B'} / N_{0, 0}^B$ relates balanced normalization factors on segments of different lengths, and will be given by
\begin{equation} \label{eq:sop_comp_l2}
    {N_{0, 0}^{B'} \over N_{0, 0}^B} = {1 \over \mathcal{C}^{(n/3)} \; \mathcal{C}^{(n/3 - 1)}}
\end{equation}
which, again due to Corrolary \ref{cor:ratios_in_two_directions}, has the well-defined value $\mathcal{C}^{-2}$ as $n \to \infty$. We take this limit in both \eqref{eq:sop_comp_l1} and \eqref{eq:sop_comp_l2} to find
\begin{equation}
    \lim\limits_{n \to \infty} \alpha_{r,v}^{(n)} = {1 \over \mathcal{C}^2} \; \sigma^{v-r+1} \; t^{2(r+v)} \Big( \pi_{r-1} \pi_{v+1} + \pi_{r+1} \pi_{v-1} - t^{-2} \pi_{r-1} \pi_{v-1} - t^{2} \pi_{r+1} \pi_{v+1} \Big)
\end{equation}
which yields the conclusion \eqref{eq:sop_computaton_conclusion}. 

If we have $r = 0$, then the first step of any walk in $\ket{\unnormstate{G^B_{r,v}}}$ cannot be down; the assumption of $r=0$ implies that the walks already begin at their minimal height. In turn, this will result in two terms of the expansion \eqref{eq:sop_ugs_relevant terms} being absent: namely, those involving an imbalance of $r-1$ on the left side of $B'$. The proof follows in the same way as presented above, but without the contributions of those terms. This is consistent with the expression \eqref{eq:sop_computaton_conclusion} and the definition $\pi_{-1} \equiv 0$. An identical argument deals with the case $v = 0$.
\end{proof}

\begin{proposition} \label{pr:sop_main_simplification_prop}
    The bulk expectation of the generalized order parameter, calculated in the approximate ground state $\ket{\normstate{H^{ABC,<b,(n)}_{0,0}}}$ under the conditions of Assumption \ref{as:bulk_ord_par_split}, obeys
    \begin{equation} \label{eq:sop_bulk_expectation_fgs}
    \lim\limits_{n \to \infty} \bra{\normstate{H^{ABC,<b,(n)}_{0,0}}} \; \hat{O}_n^\sigma (n/3, 2n/3) \; \ket{\normstate{H^{ABC,<b,(n)}_{0,0}}} = {\sum_{r,v = 0}^{\infty} \pi_r \; \alpha_{r,v} \; \pi_v \over \sum_{r, v = 0}^{\infty} \pi_r^2 \; \pi_v^2}
\end{equation}
with $\alpha_{r,v}$ given in Definition \ref{def:ord_par_middle_expectation}.
\end{proposition}

\begin{proof} [Proof of Proposition \ref{pr:sop_main_simplification_prop}]
Expanding the expression for the approximate ground state $\ket{\normstate{H^{ABC,<b,(n)}_{0,0}}}$, and recalling that the segments $A,B,C$ were defined such that the order parameter operator only acts on $B$, we obtain
\begin{align}
    \bra{\normstate{H^{ABC,<b,(n)}_{0,0}}} \; \hat{O}_n^\sigma (n/3, 2n/3) \; \ket{\normstate{H^{ABC,<b,(n)}_{0,0}}} = {1 \over \norm{H^{ABC,<b,(n)}_{0,0}}} & \sum_{r,v,r',v' < b n} \mel{\unnormstate{G^B_{r',v'}}}{\hat{O}_n^\sigma (n/3, 2n/3)}{\unnormstate{G^B_{r,v}}} \times \nonumber \\
    & \inner{\unnormstate{G^A_{0,r'}}}{\unnormstate{G^A_{0,r}}} \inner{\unnormstate{G^C_{v',0}}}{\unnormstate{G^C_{v,0}}}
\end{align}
The latter overlaps impose the identifications $r = r'$ and $v = v'$:
\begin{align}
    \inner{\unnormstate{G^A_{0,r'}}}{\unnormstate{G^A_{0,r}}} & = \delta_{r, r'} N^A_{0, r}\\
    \inner{\unnormstate{G^C_{v',0}}}{\unnormstate{G^C_{v,0}}} &= \delta_{v, v'} N^C_{v, 0}
\end{align}
We also recall the definition of the approximate normalization factor:
\begin{equation}
    \norm{H^{ABC,<b,(n)}_{0,0}} = \sum_{r,v < b n} N^A_{0, r} N^B_{r, v} N^C_{v, 0}
\end{equation}
Combining all the above, we find
\begin{equation}
    \bra{\normstate{H^{ABC,<b,(n)}_{0,0}}} \; \hat{O}_n^\sigma (n/3, 2n/3) \; \ket{\normstate{H^{ABC,<b,(n)}_{0,0}}} = {\sum_{r,v < b n} N^A_{0, r} \;(N_{0,0}^B \; \alpha_{r,v}^{(n)}) \; N^C_{v, 0} \over \sum_{r,v < b n} N^A_{0, r} N^B_{r, v} N^C_{v, 0}}
\end{equation}
It only remains to divide both the numerator and denominator by $N^A_{0, 0} N^B_{0, 0} N^C_{0, 0}$:
\begin{equation}
    \bra{\normstate{H^{ABC,<b,(n)}_{0,0}}} \; \hat{O}_n^\sigma (n/3, 2n/3) \; \ket{\normstate{H^{ABC,<b,(n)}_{0,0}}} = {\sum_{r,v < b n} {N^A_{0, r} \over N^A_{0, 0}} \; \alpha_{r,v}^{(n)} \; {N^C_{v, 0} \over N^C_{0, 0}} \over \sum_{r,v < b n} {N^A_{0, r} \over N^A_{0, 0}} \; {N^B_{r, v} \over N^B_{0, 0}} \; {N^C_{v, 0} \over N^C_{0, 0}}}
\end{equation}
We observe that each individual ratio is the definition of a $\pi$ factor, all of which have well-defined, finite limits as $n \to \infty$. Taking this limit also moves the summation limits to $\infty$, yielding  
\begin{equation}
    \lim\limits_{n \to \infty} \bra{\normstate{H^{ABC,<b,(n)}_{0,0}}} \; \hat{O}_n^\sigma (n/3, 2n/3) \; \ket{\normstate{H^{ABC,<b,(n)}_{0,0}}} = {\sum_{r,v  = 0}^\infty \pi_r \; \alpha_{r,v} \; \pi_v \over \sum_{r,v  = 0}^\infty \pi_r \; (\pi_{r} \pi_{v}) \; \pi_v}
\end{equation}
which completes the proof.    
\end{proof}

\subsection{Concluding Theorem \ref{thm:sop_values}}

\begin{proof} [Proof of Theorem \ref{thm:sop_values}]
We may now combine the previous results to arrive at the expression \eqref{eq:sop_thm_bulk} for the bulk order parameter. From Lemma \ref{lm:bulk_ord_par_approximation_state} and Proposition \ref{pr:sop_main_simplification_prop}, together with the expression for $\alpha_{r,v}$ given in Lemma \ref{lm:sop_computation_lemma}, we find
\begin{equation}
    O^\sigma_\text{bulk} = {\sum_{r,v = 0}^{\infty} \pi_r \; \sigma^{v-r+1} \; t^{2(r+v)} \Big( \pi_{r-1} \pi_{v+1} + \pi_{r+1} \pi_{v-1} - t^{-2} \pi_{r-1} \pi_{v-1} - t^{2} \pi_{r+1} \pi_{v+1} \Big) \; \pi_v \over \mathcal{C}^2 \; \sum_{r, v = 0}^{\infty} \pi_r^2 \; \pi_v^2}
\end{equation}
The summation variables in the denominator are independent, so we can rearrange it as
\begin{equation}
    \sum_{r, v = 0}^{\infty} \pi_r^2 \; \pi_v^2 = \Bigg( \sum_{r = 0}^{\infty} \pi_r^2 \Bigg)^2
\end{equation}
We consider the contribution of the first term in the numerator:
\begin{align}
    \sum_{r,v = 0}^\infty \pi_r \pi_v \; \sigma^{v-r+1}  t^{2(r+v)} \; \pi_{r-1} \pi_{v+1} &= \sum_{r,v = 0}^\infty t^{2r} \sigma^{r-1} \pi_{r-1} \pi_r \; t^{2v} \sigma^{v} \pi_v \pi_{v+1}\\
    &= t^2 \os \sum_{r = 0}^\infty t^{2r} \sigma^{r} \pi_{r} \pi_{r+1} \cs^2  
\end{align}
where in the first line we used $\sigma = \sigma^{-1}$, and in the second we shifted $r$ by 1. This shift is permitted since it only excludes the $r = -1$ term, which has been argued above to vanish. The other three contributions in the numerator give similar expressions:
\begin{align}
    \sum_{r,v = 0}^\infty \pi_r \pi_v \; \sigma^{v-r+1}  t^{2(r+v)} \; \pi_{r+1} \pi_{v-1} &= t^2 \os \sum_{r = 0}^\infty t^{2r} \sigma^{r} \pi_{r} \pi_{r+1} \cs^2  \\
    -t^{-2} \sum_{r,v = 0}^\infty \pi_r \pi_v \; \sigma^{v-r+1}  t^{2(r+v)} \; \pi_{r-1} \pi_{v-1} &= - \sigma \; t^2 \os \sum_{r = 0}^\infty t^{2r} \sigma^{r} \pi_{r} \pi_{r+1} \cs^2\\
    - t^{2} \sum_{r,v = 0}^\infty \pi_r \pi_v \; \sigma^{v-r+1}  t^{2(r+v)} \; \pi_{r+1} \pi_{v+1} &= - \sigma \; t^2 \os \sum_{r = 0}^\infty t^{2r} \sigma^{r} \pi_{r} \pi_{r+1} \cs^2
\end{align}
By combining them, we obtain
\begin{equation}
    O^\sigma_\text{bulk} = {2 (1 - \sigma) \over \mathcal{C}^2} \; {t^2 \os \sum_{r = 0}^\infty t^{2r} \sigma^{r} \pi_{r} \pi_{r+1} \cs^2 \over \os \sum_{r = 0}^{\infty} \pi_r^2 \cs^2 }
\end{equation}
The second fraction is seen to correspond to $(B^\sigma)^2$, with $B^\sigma$ introduced in Definition \ref{def:b_definition}; the bulk result \eqref{eq:sop_thm_bulk} directly follows.\\

Finally, we sketch a similar proof for the boundary expression \ref{eq:sop_thm_bdry}. Since site $i$ is fixed at the system edge in this case, we only need to consider arbitrary heights of walks at position $j$ (the remaining bulk site). We can use a simpler approximation scheme, which only splits the chain into two segments, with the boundary given by site $j$. We will therefore only sum over the possible values of walk heights in the middle of the chain. The analogue to Proposition \ref{pr:sop_main_simplification_prop} reads
\begin{equation}
    O^\sigma_\text{bdry} = {\sum_{r = 0}^{\infty} \alpha_{0,r} \; \pi_r \over \sum_{r = 0}^{\infty} \pi_r^2}
\end{equation}
with the same $\alpha$ factors as before, introduced in Definition \ref{def:ord_par_middle_expectation} and calculated in Lemma \ref{lm:sop_computation_lemma}. We expand the resulting expression in the numerator, and shift the remaining $r$ indices. Since there's a single sum in both numerator and denominator, we obtain a single factor of $B^\sigma$ rather than $(B^\sigma)^2$. In the end, we recover
\begin{equation}
    O^\sigma_\text{bdry} = (1 - \sigma) \; {t \pi_1 \over \mathcal{C}^2} \; B^\sigma
\end{equation}
and together with the identity $\mathcal{C} = 1 + t \pi_1$, we indeed recover the result \ref{eq:sop_thm_bdry}.
\end{proof}

\section{Trial states proving robust gaplessness for  $t>1$}
\label{sec:trial}
In this section, we prove Theorem \ref{thm:main3}.

 We will discuss a class of trial states which yield an upper bound on the $t>1$ gap, which vanishes as $n^2 t^{-n}$, and does not rely on the boundary terms at all. The intuition is that higher-area walks are preferred when $t>1$, so most of the weight in the ground state is held by `mountain-shaped' walks, which contain predominantly up-steps in the first half of the chain, respectively down-steps in the second one. The trial states instead have ``double-peak'' shape \textnormal{/\textbackslash/\textbackslash}, so one such mountain on each half of the chain. On the one hand, the ``double-peak''  encloses a vastly smaller area compared to the typical walk of the full ground state, so they should be almost orthogonal. On the other hand, it will satisfy all projectors in the Hamiltonian, except for the term linking the two halves of the chain. With regard to that term, we argue that most of the walks will contain a $\ket{du}$ contribution for the two middle sites, which is in fact annihilated by the Hamiltonian.\\

\begin{assumption} \label{as:trial_states_basic_assumption}
    With $t>1$, consider a spin chain consisting of $2n$ sites, and governed by the following Hamiltonian:
    \begin{equation}
    H_{2n}(t)= \lambda \; \Pi_{bdry} + \sum_{j=1}^{2n-1}\Pi_{j,j+1}(t) \label{eq:hamiltonian_variable_bdry}
    \end{equation}
    where the boundary and bulk terms are defined in eqns. \eqref{eq:hamiltonian_bdry_term} and \eqref{eq:projectors} respectively. The relative strength $\lambda \ge 0$ of the boundary projectors is left arbitrary, in order to emphasize the independence of our gap bound on the edge terms.
\end{assumption}

\begin{definition} \label{def:trial_state_starting_point}
    We split the system introduced in Assumption \ref{as:trial_states_basic_assumption} into halves: let $A = [1, n]$ and $B = [n + 1, 2n]$. The full chain $[1, 2n]$ will be denoted by $AB$. Later, we will also discuss subsegments with the two middle sites excluded, namely $A' = [1, n-1]$ and $B' = [n+2, 2n]$. The starting point for our trial states will be the product
    \begin{equation} \label{eq:trial_state_starting_point}
        \ket{\Psi_{n}} = \ket{GS_{0,0}^A} \otimes \ket{GS_{0,0}^B}
    \end{equation}
    where $\ket{GS_{0,0}^A}$ denotes the balanced ground state on the chain segment $A$, and analogously for $B$.
\end{definition}

A proper trial state, used for bounding the gap, needs to be exactly orthogonal to the ground space. We next perform this projection:

\begin{definition} \label{def:trial_state_definition}
    From the states $\ket{\Psi_{n}}$ introduced in Definition \ref{def:trial_state_starting_point}, we explicitly separate the components lying within the full-chain ground space, and respectively orthogonal to it:
    \begin{equation} \label{eq:trial_state_definition}
        \ket{\Psi_{n}} = \delta_n \; \ket{GS_{0,0}^{AB}} + \sqrt{1 - \delta_n^2} \; \ket{T_n}
    \end{equation}
    where $\ket{GS_{0,0}^{AB}}$ is the balanced ground state on the full chain, while $T_n$ is by definition orthogonal to it: $\inner{GS_{0,0}^{AB}}{T_n} = 0$. The latter will serve as our trial states.
\end{definition}

Although we will use $\ket{T_n}$ to bound the gap, the $\ket{\Psi_n}$ are still a more convenient object to work with. We therefore argue that the two are close to each other, i.e. $\delta_n$ is small:

\begin{lemma} \label{lm:trial_delta_small}
    The $\delta_n$ introduced in Definition \ref{def:trial_state_definition} vanish in the thermodynamic limit, with an exponential factor. Namely, there exists an $n_1$ such that
    \begin{equation} \label{eq:trial_delta_small_statement}
        \forall n > n_1: \quad  \delta_n < c_1 \; n \; t^{-n}
    \end{equation}
\end{lemma}

Next, we argue that balanced ground states overwhelmingly start with an up-step. Begin by separating walks based on their first step:

\begin{definition} \label{def:trial_state_first_step_split}
    Focusing on segment $B$, we separate walks in the ground state $\ket{GS_{0,0}^B}$ based on their first step:
    \begin{equation} \label{eq:trial_state_first_step_split}
        \ket{GS_{0,0}^B} = \sqrt{1 - \e_n^2} \; \ket{u} \otimes \ket{GS_{1,0}^{B'}} + \e_n \; \ket{0} \otimes \ket{GS_{0,0}^{B'}}
    \end{equation}
    Note that the first step cannot be $\ket{d}$ by the assumption that $\ket{GS_{0,0}^B}$ is balanced; the corresponding term is therefore absent from the expansion above. Moreover, once the first step on $B$ is chosen, the remaining walks on $B'$ all have a fixed imbalance and correct area weighting; therefore, they combine into the ground states $\ket{GS_{1,0}^{B'}}$ and respectively $\ket{GS_{0,0}^{B'}}$.
\end{definition}
\begin{remark}
    By symmetry, we also find on the $A$ segment
    \begin{equation} \label{eq:trial_state_last_step_split}
        \ket{GS_{0,0}^A} = \sqrt{1 - \e_n^2} \;  \ket{GS_{0,1}^{A'}} \otimes \ket{d} + \e_n \; \ket{GS_{0,0}^{A'}} \otimes \ket{0} 
    \end{equation}
    with the same value of $\e_n$.
\end{remark}

\begin{lemma} \label{lm:trial_epsilon_small}
    The RHS in eqs. \eqref{eq:trial_state_first_step_split} and \eqref{eq:trial_state_last_step_split} can be approximated by their first terms, in the thermodynamic limit. Specifically, there exists an $n_2$ for which
    \begin{equation} \label{eq:lm_trial_epsilon_small}
        \forall n > n_2: \quad  \e_n < c_2 \; \sqrt{n} \; t^{-n/2}
    \end{equation}
\end{lemma}

We may now formalize the intuition that the middle steps of $\ket{\Psi_n}$ will mostly be $\ket{du}$:

\begin{definition} \label{def:trial_middle_steps_du}
    Combining the definition \eqref{eq:trial_state_starting_point} with the expansions \eqref{eq:trial_state_first_step_split} and \eqref{eq:trial_state_last_step_split}, we obtain
    \begin{equation} \label{eq:trial_middle_steps_expansion}
        \ket{\Psi_{n}} = (1 - \e_n^2) \; \ket{GS_{0,1}^{A'}} \otimes \ket{du} \otimes \ket{GS_{1,0}^{B'}} + \sqrt{1 - (1 - \e_n^2)^2} \; \ket{E_n}
    \end{equation}
    Note that the first term of the RHS above is obtained by combining the dominant contributions from \eqref{eq:trial_state_first_step_split} and \eqref{eq:trial_state_last_step_split}. There will be three other terms in the full expansion, which we combine into the definition of $\ket{E_n}$, which we take to be a normalized state. From the normalization of $\ket{\Psi_{n}}$, it immediately follows that the amplitude of $\ket{E_n}$ in \eqref{eq:trial_middle_steps_expansion} must be $\sqrt{1 - (1 - \e_n^2)^2}$.    
\end{definition}

\subsection{Proof of Theorem \ref{thm:main3}}
We defer the more technical proofs of Lemmas \ref{lm:trial_delta_small} and \ref{lm:trial_epsilon_small} to the end of the section, and turn to showing the main result:

\begin{proof} [Proof of Theorem \ref{thm:main3}]
    The states $\ket{T_n}$ introduced in Definition \ref{def:trial_state_definition} were defined as orthogonal to the ground space of the full chain. It follows that their energy expectation directly provides an upper bound on the spectral gap:
    \begin{eqnarray} \label{eq:high_t_proof_1}
        \gamma_{2n} (t) \leq \mel{T_n}{H_{2n}}{T_n}
    \end{eqnarray}
    On the other hand, the ground state $\ket{GS_{0,0}^{AB}}$ is defined by $H_{2n} \ket{GS_{0,0}^{AB}} = 0$, so we can replace $\ket{T_n}$ by $\ket{\Psi_n}$ at only a small cost, quantified by $\delta_n$. Using \eqref{eq:trial_state_definition}, we expand:
    \begin{align}
        \mel{\Psi_n}{H_{2n}}{\Psi_n} &= \delta_n^2 \; \mel{GS_{0,0}^{AB}}{H_{2n}}{GS_{0,0}^{AB}} + \delta_n \; \sqrt{1 - \delta_n^2} \; \os \mel{T_n}{H_{2n}}{GS_{0,0}^{AB}} + \mel{GS_{0,0}^{AB}}{H_{2n}}{T_n} \cs \nonumber \\
        &\quad + (1 - \delta_n^2) \; \mel{T_n}{H_{2n}}{T_n}
    \end{align}
    and only the term on the second line survives. It follows that
    \begin{equation} \label{eq:high_t_proof_2}
        \mel{T_n}{H_{2n}}{T_n} = {1 \over 1 - \delta_n^2} \; \mel{\Psi_n}{H_{2n}}{\Psi_n}
    \end{equation}
    We rewrite the definition \eqref{eq:original-hamiltonian} of the full-chain Motzkin Hamiltonian $H_{2n}$ in the following way:
    \begin{equation} \label{eq:full_hamiltonian_split_half}
        H_{2n} = \lambda \left( \ket{d}\bra{d}_1 + \ket{u} \bra{u}_{2 n} \right) + \sum_{j = 1}^{n-1}\Pi_{j,j+1} + \Pi_{n, n+1} + \sum_{j = n + 1}^{2n-1}\Pi_{j,j+1}
    \end{equation}
    Observe that the boundary projectors annihilate $\ket{\Psi_n}$, as the latter is a balanced state by construction. Moreover, the two sums on the RHS of \eqref{eq:full_hamiltonian_split_half} are the open-boundary Hamiltonians on segments $A$ and $B$; in consequence, they annihilate the states $\ket{GS_{0,0}^{A}}$ and $\ket{GS_{0,0}^{B}}$ respectively, by definition. It follows that the only term with a non-trivial action on $\ket{\Psi_n}$ is $\Pi_{n, n+1}$, residing on the bond which links segments $A$ and $B$. We have
    \begin{equation} \label{eq:high_t_proof_3}
        \mel{\Psi_n}{H_{2n}}{\Psi_n} = \mel{\Psi_n}{\Pi_{n, n+1}}{\Psi_n}
    \end{equation}
    and therefore only the middle two steps of $\ket{\Psi_n}$ matter. But we argued (Definition \ref{def:trial_middle_steps_du} and Lemma \ref{lm:trial_epsilon_small}) that these will predominantly be $\ket{du}$. By inspecting the projector definition \eqref{eq:projectors}, one can observe that such a step sequence lies in the kernel: $\Pi_{n, n+1} \ket{du}_{n, n+1} = 0$. Equation \eqref{eq:trial_middle_steps_expansion} then yields
    \begin{eqnarray} \label{eq:high_t_proof_4}
        \mel{\Psi_n}{\Pi_{n, n+1}}{\Psi_n} = \os 1 - (1 - \e_n^2)^2 \cs \; \mel{E_n}{\Pi_{n, n+1}}{E_n}
    \end{eqnarray}
    The matrix element on the RHS is bounded above by 1, since $\Pi_{n, n+1}$ is a projection operator and $\ket{E_n}$ is normalized. Combining \eqref{eq:high_t_proof_1}, \eqref{eq:high_t_proof_2}, \eqref{eq:high_t_proof_3},  and \eqref{eq:high_t_proof_4}, we obtain the upper bound on the gap
    \begin{equation}
        \gamma_{2n} (t) \leq {1 - (1 - \e_n^2)^2 \over 1 - \delta_n^2} = {\e_n^2 (2 - \e_n^2) \over 1 - \delta_n^2}
    \end{equation}
    Using the results of Lemmas \ref{lm:trial_delta_small} and \ref{lm:trial_epsilon_small}, the RHS above is bounded above, in the $n \to \infty$ limit, by $c_0 \; n \; t^{-n}$, for some constant $c_0$, completing the proof.

\end{proof}

\begin{proof} [Proof of Lemma \ref{lm:trial_delta_small}]
   Using the notation of Section \ref{sec:low-imb}, we rewrite:
   \begin{align}
       \ket{GS_{0,0}^{AB}} & = {1 \over \sqrt{N_{0,0}^{AB}}} \ket{\unnormstate{G^{AB}_{0,0}}} = {1 \over \sqrt{N_{0,0}^{AB}}} \sum_{w \in G_{0,0}^{AB}} t^{A(w)} \ket{w}\\
       \ket{\Psi_{n}} & = \ket{GS_{0,0}^A} \otimes \ket{GS_{0,0}^B} = {1 \over \sqrt{N_{0,0}^{A} N_{0,0}^{B}}} \sum_{w \in G_{0,0}^{A} \times G_{0,0}^{B}} t^{A(w)} \ket{w}
   \end{align}
   Recalling that every walk which contributes to $\ket{\Psi_{n}}$ also appears within the full ground state ${GS_{0,0}^{AB}}$, i.e. $G_{0,0}^{A} \times G_{0,0}^{B} \subset G_{0,0}^{AB}$, we find
   \begin{equation}
       \delta_n \equiv \inner{\Psi_{n}}{GS_{0,0}^{AB}} = {1 \over \sqrt{N_{0,0}^{AB}}} {1 \over \sqrt{N_{0,0}^{A} N_{0,0}^{B}}} \sum_{w \in G_{0,0}^{A} \times G_{0,0}^{B}} t^{2A(w)} = \sqrt{N_{0,0}^{A} N_{0,0}^{B} \over N_{0,0}^{AB}}
   \end{equation}
   and therefore we must compare the normalization contributions of walks in $G_{0,0}^{A} \times G_{0,0}^{B}$ to those in $G_{0,0}^{AB}$. Recall from eq. \eqref{eq:normalization_factorization} that
   \begin{equation}
       N_{0,0}^{AB} = \sum_{r = 0}^n N_{0,r}^{A} N_{r,0}^{B}
   \end{equation}
    where no truncation is involved: in the fully balanced case $p=q=0$, the full normalization is the sum over all interface heights. The above yields
   \begin{equation} \label{eq:trial_delta_norm_expression}
       \delta_n = \sqrt{N_{0,0}^{A} N_{0,0}^{B} \over \sum_{r = 0}^n N_{0,r}^{A} N_{r,0}^{B}}
   \end{equation}
   It suffices to compare the numerator with the $r=1$ term in the denominator. We will employ a reverse procedure to that of previous sections: that is, steps are now exchanged to increase the total area of a walk, as this is favorable for weighting parameters $t>1$. Since similar mapping procedures have been discussed at length in previous sections, we will only sketch the argument here.\\

    \begin{figure}[t]
	\centering	
	\scalebox{.4}{\includegraphics{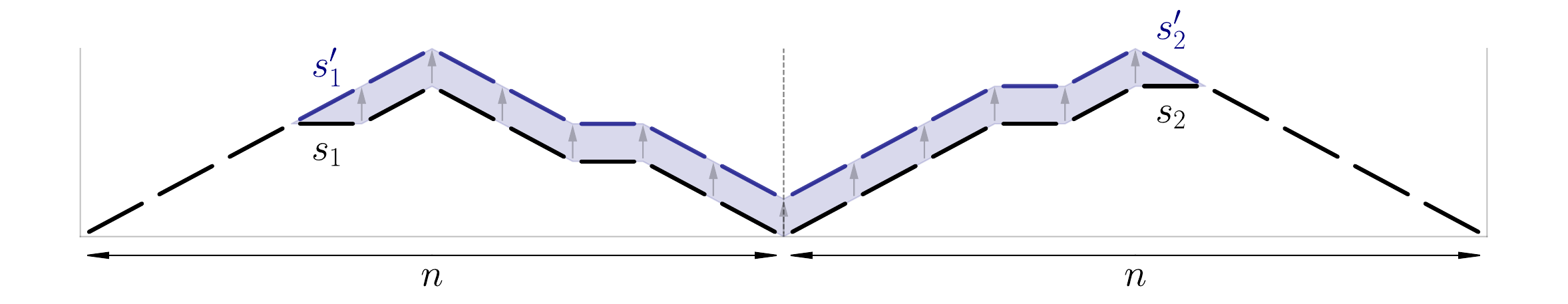}}%
	\caption{Illustration of mapping procedure, with $n=10$. A typical walk $w \in G_{0,0}^{A} \times G_{0,0}^{B}$ is shown in black. The first step not up ($s_1$) and the last one not down ($s_2$) are modified such that the walk portion between them is raised by one unit (blue steps and shaded area). The resulting walk $w'$ belongs to $G_{0,1}^{A} \times G_{1,0}^{B}$. Only two steps have been changed, and the area gain scales linearly with $n$.}
	\label{fig:trial_delta_small}
    \end{figure}

   Start with a walk $w \in G_{0,0}^{A} \times G_{0,0}^{B}$. Identify the leftmost step in segment $A$ which is not up, and call it $s_1$. Similarly, let $s_2$ be the rightmost step inside $B$ which is not down. As $w$ must return to zero height after $n$ steps, there can be no more than $n/2$ up-steps in $A$, or respectively down-steps in $B$. It follows that the distance between $s_1$ and $s_2$ obeys
   \begin{equation} \label{eq:trial_distance_condition_1}
       d(s_1, s_2) \ge n - 1
   \end{equation}
   We lift the portion of $w$ between the two steps under consideration, by raising $s_1$ and lowering $s_2$, to obtain a new walk $w'$. Call this map $\phi$. From eq. \eqref{eq:trial_distance_condition_1}, the area gain is at least $n$, and it follows that
   \begin{equation}
       A(w') \ge A(w) + \left(n - {1 \over 2} \right) \quad \implies \quad t^{2 A(w')} \ge t^{2n - 1} \; t^{2 A(w)}
   \end{equation}
   The walk $w' = \phi(w)$ lies in $G_{0,1}^{A} \times G_{1,0}^{B}$, as illustrated in Figure \ref{fig:trial_delta_small}. The mapping $\phi$ is not injective, but we can bound the cardinality of any preimage $\phi^{-1}(w')$: since only one step was modified in each of the segments $A, B$, and in turn each segment has length $n$, there are no more than $n^2$ possible choices. It follows that
   \begin{equation}
       N_{0,1}^{A} N_{1,0}^{B} \geq \sum_{w' \in \phi(G_{0,0}^{A} \times G_{0,0}^{B})} t^{2 A(w')} \geq n^{-2} \sum_{w \in G_{0,0}^{A} \times G_{0,0}^{B}} t^{2n - 1} \; t^{2 A(w)} = n^{-2} t^{2n - 1} N_{0,0}^{A} N_{0,0}^{B}
   \end{equation}
   Together with eq. \eqref{eq:trial_delta_norm_expression}, from which we keep only the $r\in \{0,1\}$ terms, this gives
   \begin{equation}
       \delta_n < \sqrt{N_{0,0}^{A} N_{0,0}^{B} \over N_{0,0}^{A} N_{0,0}^{B} + N_{0,1}^{A} N_{1,0}^{B}} \leq \sqrt{1 \over 1 + n^{-2} t^{2n - 1}}
   \end{equation}
from which the bound \eqref{eq:trial_delta_small_statement} follows at large $n$.   
\end{proof}

\begin{proof} [Proof of Lemma \ref{lm:trial_epsilon_small}]
    From eq. \eqref{eq:trial_state_first_step_split} we obtain
    \begin{align}
        \e_n &= \bra{GS_{0,0}^B} \left( \ket{0} \otimes \ket{GS_{0,0}^{B'}} \right) \nonumber \\
        &= {1 \over \sqrt{N_{0,0}^B N_{0,0}^{B'}}} \op \sum_{w_1 \in G_{0,0}^B} t^{A(w_1)} \bra{w_1} \cp \op \sum_{w_2 \in \{0\} \times G_{0,0}^{B'}} t^{A(w_2)} \ket{w_2} \cp       
\end {align}
Since $\{0\} \times G_{0,0}^{B'} \subset G_{0,0}^B$, the above can be simplified to
\begin{equation} \label{eq:trial_epsilon_small_rewrite}
    \e_n = {1 \over \sqrt{N_{0,0}^B N_{0,0}^{B'}}} \sum_{w \in \{0\} \times G_{0,0}^{B'}} t^{2 A(w)} = \sqrt{N_{0,0}^{B'} \over N_{0,0}^B}
\end{equation}
and therefore we must compare weights of walks in $\{0\} \times G_{0,0}^{B'}$ to those of $G_{0,0}^B$. The strategy, analogous to the previous proof, is to start with a walk $w \in \{0\} \times G_{0,0}^{B'}$ and raise its area. To this end, identify the rightmost step $s$ which is not down. This must be in the right half of $B$, since no more than half of the steps in $w$ can be `down' by the balanced condition. The distance between the first step of $w$, which is assumed flat, and $s$ therefore obeys
\begin{equation}
    d(1, s) > {n - 1 \over 2}
\end{equation}

\begin{figure}[t]
	\centering	
	\scalebox{.4}{\includegraphics{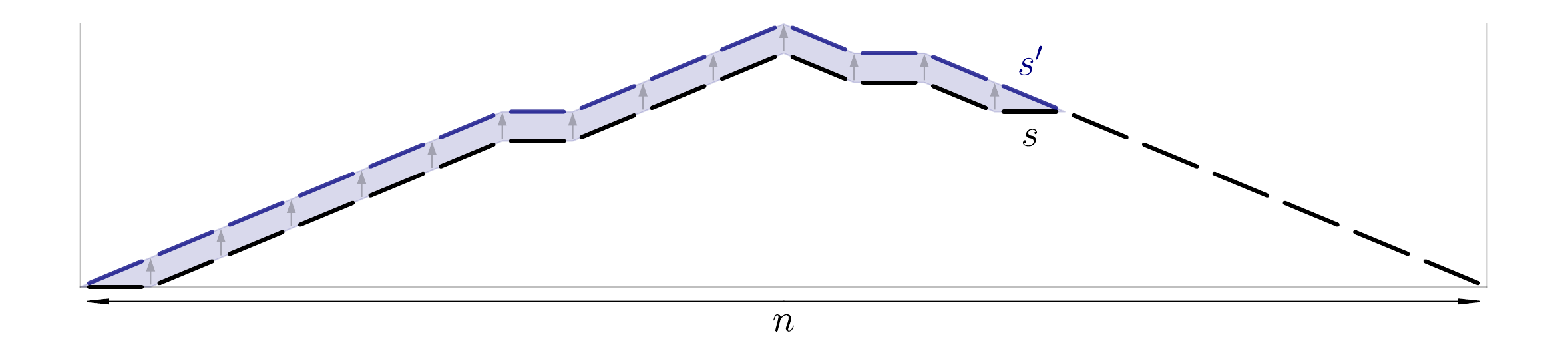}}%
	\caption{Example of mapping $\{0\} \times G_{0,0}^{B'} \to G_{0,0}^B$. Original walk $w$ is shown in black. The first step was raised, and the rightmost step $s$ which is not down was in turn lowered; the resulting walk is depicted in blue. The area gain, shown with blue shading, is at least $n/2$.}
	\label{fig:trial_epsilon_small}
    \end{figure}

As before, raise the first step from `flat' to `up', and lower $s$ to a new step $s'$. The resulting walk $w'$ starts with an up-step and is still balanced, i.e. it belongs to $\{u\} \times G_{1,0}^{B'}$. Its area obeys
\begin{equation}
    A(w') \ge A(w) + \left({n - 1 \over 2} \right) \quad \implies \quad t^{2 A(w')} \ge t^{n - 1} \; t^{2 A(w)}
\end{equation}
The preimage of $w'$ under this mapping will have size at most $n/2$, since the position of $s$ within the walk $w$ is the only variable element. We obtain a bound on the normalization factors of \eqref{eq:trial_epsilon_small_rewrite}:
\begin{align}
    N_{0,0}^B &= \sum_{w \in \{0\} \times G_{0,0}^{B'}} t^{2 A(w)} + \sum_{w' \in \{u\} \times G_{1,0}^{B'}} t^{2 A(w')} \nonumber \\
    &= N_{0,0}^{B'} + \sum_{w' \in \{u\} \times G_{1,0}^{B'}} t^{2 A(w')} \nonumber \\
    &\geq N_{0,0}^{B'} + {2 \over n} \; t^{n-1} \sum_{w \in \{0\} \times G_{0,0}^{B'}} t^{2 A(w)} \nonumber \\
    &= \op 1 + 2 n^{-1} t^{n-1} \cp N_{0,0}^{B'}
\end{align}
From \eqref{eq:trial_epsilon_small_rewrite} we therefore obtain
\begin{equation}
    \e_n \leq \sqrt{1 \over 1 + 2 n^{-1} t^{n-1}}
\end{equation}
Omitting the constant factors, we see that the above behaves as $\sqrt n \; t^{-n/2}$ at large $n$, yielding the bound of \eqref{eq:lm_trial_epsilon_small}.
\end{proof}

\appendix

\section*{Acknowledgments} The proof of Theorem \ref{thm:main} subsumed in this work was done while RA was affiliated with Harvard College and RM was affiliated with Google Quantum AI.  
ML wishes to thank Bruno Nachtergaele and Norbert Schuch for helpful conversations about the string order parameter in Motzkin chains.
The research of ML is  supported by the DFG
through the grant TRR 352 – Project-ID 470903074 and by the European Union (ERC Starting Grant MathQuantProp, Grant Agreement 101163620).\footnote{Views and opinions expressed are however those of the authors only and do not necessarily reflect those of the European Union or the European Research Council Executive Agency. Neither the European Union nor the granting authority can be held responsible for them.}

\section{Analysis of the ratios of normalization factors} \label{sec:norm-ratios-appendix}

\subsection{Setup}
Let $t\in (0,1)$. We will consider the normalization factors
\begin{equation}\label{eq:Nkpqdefn}
N^k_{p,q}=\sum_{w\in G_{p,q}^k} t^{2\mathcal A(w)}
\end{equation}
for $k\geq 1$ and $p,q\geq 0$ such that $p+q\leq k$. Here $G_{p,q}^k$ denotes the ground space $G_{p,q}$ on a segment of length $k$. We adopt the zero boundary conditions
\begin{equation}\label{eq:Nbc}
N^k_{p,q}=0,\qquad \text{for either } p<0,\; q<0 \text{ or } p+q>k.
\end{equation}
For $p,q\geq 0$, we then have the following recursion relations,
\begin{equation}\label{eq:Nrecursion}
\begin{aligned}
N^{k+1}_{p,0}=&t N^k_{p,1}+N^k_{p,0}+t^{2k+1}N^k_{p-1,0},\\
N^{k+1}_{p,q}=&t^{2q+1} N^k_{p,q+1}+t^{2q}N^k_{p,q}+t^{2q-1}N^k_{p,q-1}\qquad \textnormal{for } q\geq 1,\, p+q\leq k+1,
\end{aligned}
\end{equation}
and the symmetry relation 
\begin{equation}\label{eq:Nksymm}
N^k_{p,q}=N_{q,p}^k.
\end{equation}
The initial data for the recursion can be computed from \eqref{eq:Nkpqdefn}, e.g.,
\begin{equation}\label{eq:Ninitialdata}
\begin{aligned}
N^1_{0,0}=&1,\quad N^1_{0,1}=t,\\
N^2_{0,0}=&1+t^2,\quad N^2_{0,1}=t+t^{3},\quad N^2_{0,2}=t^4,\quad N^2_{1,1}=t^2,\\
N^3_{0,0}=&1+2t^2+t^4,\quad N^3_{0,1}=t+2t^{3}+t^5+t^7,\quad N^3_{0,2}=t^4+t^6+t^8,\quad N^3_{0,3}=t^9,\\
N^3_{1,1}=&t^2+2t^4
\quad N^3_{1,2}=t^5.
\end{aligned}
\end{equation}

The recursion relation \eqref{eq:Nrecursion} can be seen as a discrete diffusion equation on the half-line with spatially varying diffusivity by interpreting $k$ as time and $q$ as space variable. We are interested in the limiting behavior of the ratios $\pi^k_{p,q}=\frac{N^k_{p,q}}{N^k_{p,0}}$ for large $k$ and $p,q\lesssim bk$ with $b$ a small constant. The rigorous analysis in this regime is technically moderately challenging, partly because the spatially varying coefficients preclude the use of Fourier theory and related methods to obtain exact formulas for the solution. Instead, we rely on a number of hands-on nested induction arguments and various analytical estimates which are heavily motivated by extensive numerical experiments (see, e.g., Figure \ref{fig:exponential-convergence-numerics}) and physical intuition about one-dimensional diffusion processes. The main technical difficulties that need to be overcome here are due to the boundary behavior at the edges of the conical domain  $0\leq q\leq k$ and the $k$-dependence of the maximal $p$ and $q$-values. It might be interesting to generalize the hands-on approach we develop here to equilibration problems in other discrete 1D diffusion equations.

Let $p,q\geq 0$. We are interested in the large-$k$ behavior of the following ratios, 
\begin{equation}
    \label{eq:rationdefn}
\pi^k_{p,q}=\frac{N^k_{p,q}}{N^k_{p,0}},\qquad \textnormal{for }  p+q\leq k
\end{equation}
Note that $\pi^k_{p,0}=1$. In accordance with \eqref{eq:Nbc}, we adopt the boundary condition that $\pi^k_{p,q}=0$ if $q<0$ or $q>k-p$.

Let $p\geq 0$ and $q\geq 1$. From \eqref{eq:Nrecursion} and \eqref{eq:Nksymm}, we see that the $\pi^k_{p,q}$ satisfy the following recursion relation.

\begin{equation}\label{eq:rkrecursion}
\begin{aligned}
\pi^{k+1}_{p,q}=& t^{2q}  \frac{t^2 \pi^k_{p,q+1}+ t \pi^k_{p,q}+\pi^k_{p,q-1}}
{t^2 \pi^k_{p,1}+t+t^{2k+2} \frac{\pi^k_{0,p-1}}{\pi^k_{0,p}}},
\qquad \textnormal{for } p+q\leq k+1
\end{aligned}
\end{equation}
This is to be understood with the above boundary condition that $\pi^k_{p,q}=0$ if either $p<0$, $q<0$ or $p+q>k$.\\

\textbf{Convention.} We generally suppress the dependence of constants on the parameter $t$, which will be considered fixed in $(0,1)$ except where we want to emphasize it.

\subsubsection{Monotonicity properties and existence of the limit }

We begin with some useful properties of the ratios $\pi^k_{p,q}$. In particular, these ensure the existence of the $k\to\infty$ limit. Define the constant
	\begin{equation}\label{eq:C0defn}
    C_0\equiv C_0(t)=\max\left\{t,\frac{t+t^2-1}{t-t^5
    }\right\}.
    	\end{equation}

\begin{proposition}[Existence of the $k$-limit]\label{prop:limit} Let $t\in (0,1)$. For every $p,q\geq 0$, the following limit exists
	\begin{equation} \label{eq:rationdefn_infty}
	    \pi^{\infty}_{p,q}=\lim_{k\to\infty}  \pi^k_{p,q}.
	\end{equation}
	and satisfies the bounds
	\begin{align}\label{eq:limitmon}
	\pi^k_{p,q}\leq& \pi_{p,q}^{\infty},\qquad\textnormal{for } p+q\leq k,\\
	\label{eq:rlimitdecayq}
	\pi^{\infty}_{p,q+1}\leq& C_0 t^{2q} \pi^{\infty}_{p,q}.
	\end{align}
\end{proposition}


	\begin{corollary}[of Proposition \ref{prop:limit}]\label{cor:sodifferencebounds}
	Let $t\in (0,1)$. We have
	\begin{align}\label{eq:corpipiq1}
	\pi^{\infty}_{p,q}-t^2\pi^{\infty}_{p,q+2}\geq& (1-C_0^2 t^{4q+4})\pi^{\infty}_{p,q},\qquad q\geq 0.
	\end{align}
    In particular, for $t<t_{\mathrm{SOP}}$ with $t_{\mathrm{SOP}}$ the unique root of the polynomial $ t^4+t^3+t^2-t-1$ on $(0,1)$, we have
    \[
\pi_{0,q}^{\infty}>t^2\pi_{0,q+2}^{\infty},\qquad \textnormal{ for } q\geq 0
    \]
	\end{corollary}

 The last fact is used to prove non-vanishing of the string order parameter for $t<t_{\mathrm{SOP}}$.  

 \begin{remark}
      We note that since $C_0(t)\sim t$ as $t\to 0$, it is immediately clear that $\pi_{0,q}^{\infty}>\pi_{0,q+2}^{\infty}$ holds for all $ q\geq 0$ for all sufficiently small $t\approx 0$.
 \end{remark}

\subsubsection{Main convergence result}
It is important in the main text that the convergence to the limit happens in a uniform way at an exponential rate. Since we interpret the variable $k$ as discrete time, we call this a result about exponential equilibration (of the solution to the spatially inhomogeneous discrete diffusion equation). This can be seen as a significant refinement of Proposition \ref{prop:limit}.

\begin{figure}[t]
	\centering	
	\scalebox{\largehalffigsize}{\includegraphics{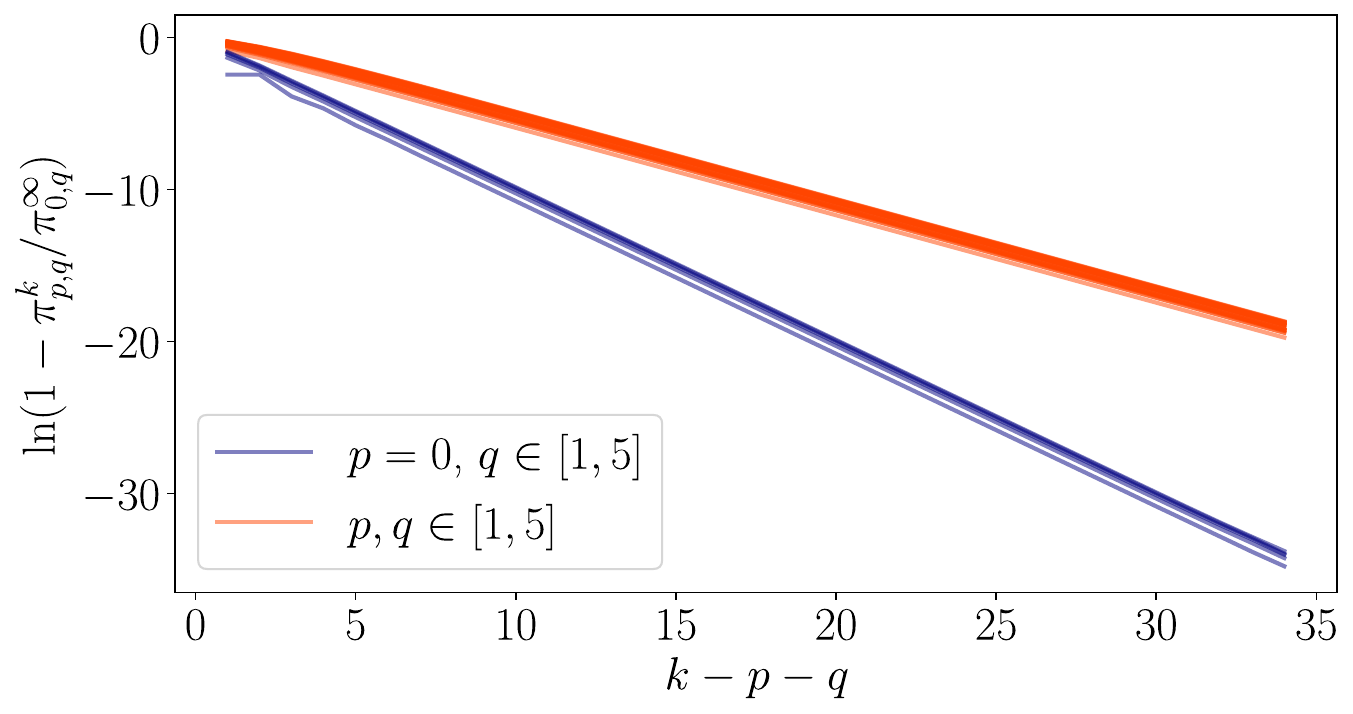}}	
	\caption{Log-linear plots of relative error $1 - \pi^k_{p,q} / \pi_{0,q}^{\infty}$ as a function of $k-p-q$, obtained from numerical simulations, at $t = 0.75$, $p \in [0, 5]$, and $q \in [1,5]$. The results strongly suggest that convergence of the $\pi_{p,q}^k$ ratios is exponential in the quantity $(k-p-q)$, which is slightly stronger than what we prove in Theorem \ref{thm:mainec}. Furthermore, the convergence rate is significantly faster for $p=0$ compared to $p>0$. Similar behavior is reproduced across the range of $t$ accessible through numerics.}
	\label{fig:exponential-convergence-numerics}
\end{figure}

In Figure \ref{fig:exponential-convergence-numerics}, it is shown that the error terms $1 - \pi^{k}_{p,q} / \pi_{0,q}^{\infty}$ converge exponentially fast in $(k-p-q)$. Plotting $\log \op 1 - \pi^{k}_{p,q} / \pi_{0,q}^{\infty} \cp $ versus $k-p-q$ gives a family of straight lines, which overlap at all $p,q \ge 1$. Therefore, the rate appears independent of $p$ when $p \ge 1$, whereas the $p = 0$ ratios converge faster.

\begin{theorem}[Exponential convergence to equilibirium]\label{thm:mainec}
	For every $t\in (0,1)$, there exist constants $C_*,\alpha,\beta>0$ such that
	\begin{equation}\label{eq:mainec}
	0\leq \pi_{0,q}^{\infty}-\pi_{p,q}^k\leq C_* t^{\alpha(k-\beta p-\beta q)}\pi_{0,q}^{\infty}\qquad  \textnormal{for } p+q\leq k.
	\end{equation}
\end{theorem}

We remark that $\alpha$ is a small positive number and $\beta$ is a large positive number.

Owing to the discreteness of the problem, Theorem \ref{thm:mainec} is proved by elementary analytical tools, but is nonetheless surprisingly delicate. In particular, the main technical result we prove (Theorem \ref{thm:expconv}), which then implies Theorem \ref{thm:mainec} as a corollary, involves four auxiliary parameters $A_0, \alpha,\beta$ and $\nu$ and we have found all of them to be essential to make its inductive proof work. 

In Section \ref{sec:stringorder}, where we study the string order parameter, we will also encounter different ratios of normalization factors. Their exponential convergence also follows from Theorem \ref{thm:expconv} by simple manipulations, as we note in this corollary.

\begin{corollary} \label{cor:ratios_in_two_directions}
Let $t\in (0,1)$. We also have the convergence of the following related ratios of normalization factors
\[
\frac{N^k_{p,q}}{N^k_{0,0}}=\pi^k_{p,q}\pi^k_{0,p}\to \pi_{0,q}^{\infty} \pi_{0,p}^{\infty}
\]
and of 
\[
\mathcal C^{(k)}=\frac{N^{k+1}_{0,0}}{N^k_{0,0}}=1+t\pi_{0,1}^k
\to 1+t\pi_{0,1}^{\infty}
\]
with corresponding exponential error bounds that can be read off from \eqref{eq:mainec}.
\end{corollary}

\subsection{Monotonicity and exponential decay in $q$}

The proof of Proposition \ref{prop:limit} rests on the following two lemmas.

\begin{lemma}[Monotonicity in $k$ and $p$]\label{lm:monotonicity}
	Let $t\in(0,1)$ and $p,q\geq 0$. Then we have
	\begin{equation}\label{eq:rmonotone}
	\pi^k_{p,q}\leq \pi^{k+1}_{p,q},\qquad \textnormal{for } p+q\leq k+1,
	\end{equation}
	and
	\begin{equation}\label{eq:rmonotonep}
	\pi^k_{p+1,q}\leq \pi^{k}_{p,q},\qquad \textnormal{for } p+q+1\leq k.
	\end{equation}
\end{lemma}

\begin{lemma}[Exponential decay in $q$]\label{lm:qdecay}
	Let $t\in(0,1)$. For all $p,q\geq 0$,
	\begin{equation}\label{eq:rdecayq}
	\pi^{k}_{p,q+1}\leq C_0 t^{2q} \pi^k_{p,q},\qquad \textnormal{for } p+q\leq k
	\end{equation}

\end{lemma}

\begin{remark} 
	See also the converse bound in Lemma \ref{lm:rlb}. The decay rate $t^{2q}$ is displayed clearly in the numerics, which played an important role in formulating these statements. 
\end{remark}

\begin{proof}[Proof of Proposition \ref{prop:limit}] By Lemma \ref{lm:monotonicity}, the sequence $\{\pi^k_{p,q}\}_{k\geq p+q}$ is monotonically increasing and it follows from Lemma \ref{lm:qdecay} that it is bounded. Hence, the limit exists and satisfies \eqref{eq:limitmon}. The estimate \eqref{eq:rlimitdecayq} follows by taking $k\to\infty$ in \eqref{eq:rdecayq}. 
\end{proof}


%

\subsubsection{Proof of monotonicity (Lemma \ref{lm:monotonicity})}
\textbf{Recursion of auxiliary ratios.} Through the proof, we fix $p\geq 0$. We begin with some reductions. We introduce the auxiliary ratio
\begin{equation}\label{eq:rhodefn}
\rho_{p,q}^k =t\frac{N^k_{p,q+1}}{N^k_{p,q}}=t\frac{\pi^k_{p,q+1}}{\pi^k_{p,q}},
\end{equation}
for which we adopt the boundary condition that $\rho^k_{p,q}=\infty$ for $q< 0$ and $\rho^k_{p,q}=0$ for $q\geq k-p$, in accordance with \eqref{eq:Nbc}. The prefactor $t$ in \eqref{eq:rhodefn} is included for convenience only. Note that
\begin{equation}\label{eq:rrho}
\pi^k_{p,q}=t^{-q}\prod_{q'=0}^{q-1} \rho^k_{p,q'}.
\end{equation}

We define the function $F:\R_+^3\to \R_+$ by
\begin{equation}\label{eq:Fdefn}
F(x,y,z)=t^2 y\frac{x+1+\frac{1}{y}}{y+1+\frac{1}{z}}
\end{equation}
For $q>0$,
$$
\begin{aligned}
&\rho^{k+1}_{p,q}
=t\frac{N^{k+1}_{p,q+1}}{N^{k+1}_{p,q}}
=t\frac{t^{3} N^{k}_{p,q+2}+t^{2}N^{k}_{p,q+1}+tN^{k}_{p,q}}{t N^{k}_{p,q+1}+N^{k}_{p,q}+t^{-1}N^{k}_{p,q-1}}=F(\rho^k_{p,q+1},\rho^k_{p,q},\rho^k_{p,q-1}),
\end{aligned}
$$
where the last step is immediate for $1\leq q<k-p-1$, while for $q\in \{k-p-1,k-p\}$ it rests on the conventions \eqref{eq:Nbc} and $\rho^k_{0}=\infty$ and $\rho^k_{p,q}=0$ for $q\geq k-p$, as well as interpreting $\frac{1}{\infty}$ as $0$. The case $q=0$ can be treated similarly.\\

To summarize, the auxiliary ratios $\rho^k_{p,q}$ satisfy the following recursion relations.
\begin{equation}\label{eq:rhorecursion}
\begin{aligned}
\rho^{k+1}_{p,q}=&F(\rho^k_{p,q+1},\rho^k_{p,q},\rho^k_{p,q-1}),\qquad  \textnormal{for } 1\leq q\leq k-p,\\
\rho^{k+1}_{p,0}=&F(\rho^k_{p,1},\rho^k_{p,0},t^{-2k-2}\rho^k_{0,p-1})
\end{aligned}
\end{equation}

\subsubsection{Proof of monotonicity in $k$}
By \eqref{eq:rrho}, it suffices to prove that $\rho^k_{p,q}\leq \rho^{k+1}_{p,q}$ for $0\leq q\leq k-p-1$. 
Recall that $p\geq 0$ is fixed. We claim that we have the monotonicity formula
\begin{equation}\label{eq:rhomonotone}
\rho^k_{p,q}\leq \rho^{k+1}_{p,q},\qquad  \textnormal{for }  0\leq q\leq k+1-p.
\end{equation}
The proof of \eqref{eq:rhomonotone} is done by induction in $k$. The induction base case occurs at $k=p-1$ and reduces to $0=\rho^{p-1}_{p,0}\leq \rho^{p}_{p,0}$.\\

For the induction step, consider any $k\geq p$ and suppose that  \eqref{eq:rhomonotone} holds up to $k-1$. In view of the boundary conditions, then have that $\rho^{k-1}_{p,q}\leq \rho^{k}_{p,q}$ holds for all $q\geq 0$.\\

\uline{\textit{Case (i): $1\leq q\leq  k-p-1$.}} Since the first recursion relation in \eqref{eq:rhorecursion} preserves the value of $p$, we may suppress $p$ from the notation, i.e., we denote $\rho_{p,q}^k\equiv \rho_{q}^k$, etc. We first consider the subcase $1\leq q\leq k-p-3$, since then \eqref{eq:rhorecursion} does not involve any boundary terms. We have
\begin{equation}\label{eq:ratioofrho}
\frac{\rho^{k+1}_q}{\rho^{k}_q}
=\frac{\rho^{k}_{q-1}}{\rho^{k-1}_{q-1}} \quad \frac{\rho^{k}_{q+1}\rho^{k}_q+\rho^{k}_q+1}{\rho^{k}_{q}\rho^{k}_{q-1}+\rho^{k}_{q-1}+1} 
\quad\frac{\rho^{k-1}_{q}\rho^{k-1}_{q-1}+\rho^{k-1}_{q-1}+1}{\rho^{k-1}_{q+1}\rho^{k-1}_{q}+\rho^{k-1}_{q}+1}.
\end{equation}
To verify that this expression is $\geq 1$, we order terms in the numerator and denominator lexicographically, first by $q$-values and then by $k$-value and then we apply the induction hypothesis only as needed to compare term by term. For example, considering the highest-order terms in numerator and denominator, the induction hypothesis gives $\rho^{k}_{q+1}\geq \rho^{k-1}_{q+1}$ and so
$$
\rho^{k}_{q+1}\rho^{k}_{q}\rho^{k-1}_{q}\rho_{q-1}^{k} \rho_{q-1}^{k-1} 
\geq
\rho^{k-1}_{q+1}\rho^{k}_{q}\rho^{k-1}_{q}\rho_{q-1}^{k} \rho_{q-1}^{k-1}.
$$
After grouping the lower order terms as described above, the induction hypothesis yields a term-by-term comparison in a straightforward manner.\\

Next, we consider the subcase $q=k-p-2$. Here the boundary conditions do not change the derivation of \eqref{eq:ratioofrho} and their sole effect is to set $\rho^{k-1}_{q+1}=0$ in the denominator. Since the induction hypothesis in fact applies for all $q\geq 0$ as noted before, the same term-by-term comparison from above can be used.\\

For the final subcase $q=k-p-1$, the boundary conditions lead to the slimmer expression
$$
\frac{\rho^{k+1}_{k-1}}{\rho^{k}_{k-1}}=
\frac{\rho^{k}_{k-2}}{\rho^{k-1}_{k-2}}
\quad \frac{(\rho^{k}_{k-1}+1)(\rho^{k-1}_{k-2}+1)}{\rho^{k}_{k-1}\rho^{k}_{k-2}+\rho^{k}_{k-2}+1}
\geq \frac{\rho^{k}_{k-2}}{\rho^{k-1}_{k-2}}\geq 1,
$$
where the last step uses the induction hypothesis.\\

\uline{\textit{Case (ii): $q=0$.}} 
Using the second relation in \eqref{eq:rhorecursion} and $t\in (0,1)$, we obtain
$$
\begin{aligned}
\frac{\rho^{k+1}_{p,0}}{\rho^{k}_{p,0}}
=&\frac{\rho^{k}_{0,p-1}}{\rho^{k-1}_{0,p-1}} \quad 
\frac{\rho^{k}_{p,1}\rho^{k}_{p,0}+\rho^{k}_{p,0}+1}
{\rho^{k-1}_{p,1}\rho^{k-1}_{p,0}+\rho^{k-1}_{p,0}+1}\quad
\frac{\rho^{k-1}_{0,p}\rho^{k-1}_{0,p-1}+\rho^{k-1}_{0,p-1}+t^{2k}}
{\rho^{k}_{0,p}\rho^{k}_{0,p-1}+\rho^{k}_{0,p-1}+t^{2k+2}}\\
\geq&
\frac{\rho^{k}_{0,p-1}}{\rho^{k-1}_{0,p-1}} \quad 
\frac{\rho^{k}_{p,1}\rho^{k}_{p,0}+\rho^{k}_{p,0}+1}
{\rho^{k-1}_{p,1}\rho^{k-1}_{p,0}+\rho^{k-1}_{p,0}+1}\quad
\frac{\rho^{k-1}_{0,p}\rho^{k-1}_{0,p-1}+\rho^{k-1}_{0,p-1}+t^{2k+2}}
{\rho^{k}_{0,p}\rho^{k}_{0,p-1}+\rho^{k}_{0,p-1}+t^{2k+2}}\\
\geq&1
\end{aligned}
$$
where the last estimate follows as in the case $q>0$ by appropriate grouping and the induction hypothesis. This proves the claim \eqref{eq:rhomonotone} and thanks to \eqref{eq:rrho} also \eqref{eq:rmonotone}.

\subsubsection{Proof of monotonicity in $p$}
We shall prove that
\begin{equation}\label{eq:rhomonotonep}
\rho^k_{p+1,q}\leq \rho^k_{p,q},\qquad \textnormal{for } 0\leq q\leq k-p-1.
\end{equation}
and use \eqref{eq:rrho} to conclude \eqref{eq:rmonotonep}.\\

The proof of \eqref{eq:rhomonotonep} is by induction in $k$. The induction base occurs at $k=p+1$, $q=0$, and holds because $\rho^{p+1}_{p+1,0}=0$. For the induction step, suppose that \eqref{eq:rhomonotonep} holds up to some $k$. \\ 

\uline{\textit{Case (i):} $1\leq q\leq  k-p-1$.}
First, suppose the subcase $1\leq q\leq  k-p-3$. Recall the recursion \eqref{eq:rhorecursion}. By the induction hypothesis and the fact that $F$ defined in \eqref{eq:Fdefn} is monotonically increasing in each of its arguments, we have
$$
\rho^k_{p+1,q}=F(\rho^{k-1}_{p+1,q+1},\rho^{k-1}_{p+1,q},\rho^{k-1}_{p+1,q-1})
\leq F(\rho^{k-1}_{p,q+1},\rho^{k-1}_{p,q},\rho^{k-1}_{p,q-1})=\rho^k_{p,q}.
$$
as desired.
Next we consider the subcase $q\in \{k-p-2,\; k-p-1\}$. Then we either have $\rho^k_{p+1,q+1}=0$ or $\rho^k_{p+1,q+1}=\rho^k_{p+1,q}=\rho^k_{p,q+1}=0$ and the induction step follows again by monotonicity of $F$.\\

\uline{\textit{Case (ii):} $q=0$.} This case is more subtle due to the appearance of $\rho_{0,p-1}^k$ in \eqref{eq:rhorecursion}. For this, we observe the additional monotonicity
\begin{equation}\label{eq:rhoqmonotonic}
\rho^k_{0,q'+1}\leq \rho^k_{0,q'},\qquad \textnormal{for } q'+1\leq k.
\end{equation}
This inequality can be proved by a separate induction in $k$ thanks to the monotonicity of $F$; we skip the details.
Then by the induction hypothesis and \eqref{eq:rhoqmonotonic},
$$
\rho^{k+1}_{p+1,0}=F(\rho^k_{p+1,1},\rho^k_{p+1,0},t^{-2k-2}\rho^k_{0,p})
\leq F(\rho^k_{p,1},\rho^k_{p,0},t^{-2k-2}\rho^k_{0,p})
=\rho^{k+1}_{p,0}
$$
This completes the induction step and hence the proof of Lemma \ref{lm:monotonicity}.
\qed

\subsection{Proof of exponential decay in $q$ (Lemma \ref{lm:qdecay})}
Let $t\in (0,1)$ and $p\geq 0$. By \eqref{eq:rrho} and Lemma \ref{lm:monotonicity}, we have
$$
\frac{\pi^{k}_{p,q+1}}{\pi^{k}_{p,q}}=t^{-1}\rho^k_{p,q}\leq t^{-1}\rho^k_{0,q}=\frac{\pi^{k}_{0,q+1}}{\pi^{k}_{0,q}}
$$
so it suffices to prove the claim for $p=0$, i.e.,
\begin{equation}\label{eq:p=0claimqdecay}
\pi^{k}_{0,q+1}\leq C_0 t^{2q} \pi^k_{0,q},\qquad \textnormal{for } 1\leq q\leq k
\end{equation}
We note that \eqref{eq:p=0claimqdecay} extends to $q>k$ due to the boundary condition $\pi^k_{p,q}=0$ for $q>k$.\\

We proceed by induction in $k$. The induction base is $k=1$. The  $q=1$ statement holds trivially (with any choice of $C_0$) because $\pi^1_{0,2}=0$ by the boundary condition. The $q=0$ statement follows from \eqref{eq:Ninitialdata} under the condition that $C_0\geq t$, which our choice of $C_0$ satisfies.

For the induction step, suppose that \eqref{eq:p=0claimqdecay} holds up to some $k\geq 1$.\\

\uline{Case (i): $2\leq  q\leq k-p$.} Here \eqref{eq:rkrecursion} and the induction hypothesis give
$$
\frac{\pi^{k+1}_{0,q+1}}{\pi^{k+1}_{0,q}}=\frac{t^{3}\pi^k_{0,q+2}+t^{2}\pi^k_{0,q+1}+t \pi^k_{0,q}}{t\pi^k_{0,q+1}+\pi^k_{0,q}+t^{-1}\pi^k_{0,q-1}}
\leq C_0 t^{2q} \frac{t^{5}\pi^k_{0,q+1}+t^{2}\pi^k_{0,q}+t^{-1}\pi^k_{0,q-1}}{t\pi^k_{0,q+1}+\pi^k_{0,q}+t^{-1}\pi^k_{0,q-1}}
\leq C_0 t^{2q}
$$
as required.\\

\uline{Case (ii):  $q=1$.} We use \eqref{eq:rkrecursion}, $\pi_{0,0}^k=1$, and the induction hypothesis to obtain
\[
\frac{\pi^{k+1}_{0,2}}{\pi^{k+1}_{0,1}}
=\frac{t^{3}\pi^k_{0,3}+t^{2}\pi^k_{0,2}+t \pi^k_{0,1}}{t\pi^k_{0,2}+\pi^k_{0,1}+t^{-1}}
\leq C_0t^{2} \frac{t^{5}\pi^k_{0,2}+t^{4}\pi^k_{0,1}+t^{-1}}{t\pi^k_{0,2}+\pi^k_{0,1}+t^{-1}}
\leq C_0t^{2}.
\]

\uline{Case (iii): $q=0$.} 
By $\pi_{0,0}^k=1$, and the induction hypothesis,
\[
\pi^{k+1}_{0,1}\leq  t  \frac{t^2 \pi^k_{0,2}+ t \pi^k_{0,1}+1}
{t \pi^k_{0,1}+1}
\leq C_0 t  \frac{(t^4 +\tfrac{t}{C_0}) \pi^k_{0,1}+1}
{t \pi^k_{0,1}+1}.
\]
Thus, the induction step can be completed under the assumption that
\[
(t^5 +\tfrac{t^2}{C_0}-t) \pi^k_{0,1}\leq 1-t.
\]
The prefactor in front of $\pi^k_{0,1}$ is positive because our choice $C_0\geq \frac{t}{1-t^4}$ as one easily checks with elementary inequalities. Using the induction hypothesis $\pi^k_{0,1}\leq C_0$ once again, it suffices to have
\[
C_0(t^5-t) \leq 1-t-t^2
\]
which our choice of $C_0\geq \frac{t+t^2-1}{t-t^5}$ satisfies.
This completes the induction step and proves Lemma \ref{lm:qdecay}.
\qed

\subsection{Proof of Theorem \ref{thm:mainec}}

The lower bound in \eqref{eq:mainec} follows from \eqref{eq:rmonotonep} and Proposition \ref{prop:limit}. It thus suffices to prove the upper bound in \eqref{eq:mainec}. We shall occasionally denote
$$
\pi^k_{0,q}\equiv \pi_q^k.
$$

\subsubsection{Preliminaries}
Let $p,q\geq 0$ with $p+q\leq k$. Recall \eqref{eq:rkrecursion}. We aim to derive a similar formula for $\pi_{p,q}$ by taking the limit $k\to\infty$ in \eqref{eq:rkrecursion}. To show that the last term in the denominator does not contribute in the limit, we use the following lemma.

\begin{lemma}\label{lm:rlb}
For every $t\in(0,1)$, there exists another constant $C_1>0$ such that
\begin{equation}\label{eq:lemmaconsequence}
	\pi^k_{q+1}\geq C_1 t^{2q} \pi_{q}^k
\end{equation}
\end{lemma}

This is a converse to Lemma \ref{lm:qdecay}. For $p\geq 1$, Lemma \ref{lm:rlb} implies that
$$
\limsup_{k\to\infty} t^{2k+3}\frac{\pi^k_{0,p-1}}{\pi^k_{0,p}}\leq  \frac{1}{C_1 t^{2(p-1)}} \limsup_{k\to\infty} t^{2k+3} =0
$$
and so
\begin{equation}\label{eq:limitrecursion}
\begin{aligned}
\pi_{p,q}^{\infty}
=\lim_{k\to\infty}\pi^k_{p,q}= t^{2q}\frac{t^{2}\pi_{p,q+1}+t\pi_{p,q}+\pi_{p,q-1}}
{t^2 \pi_{p,1}+t}.
\end{aligned}
\end{equation}

\begin{proof}[Proof of Lemma \ref{lm:rlb}]
Similarly to the proof of the main result, Theorem \ref{thm:expconv}, a direct induction does not work, but instead a suitable inductive argument can be constructed by slightly strengthening the claim through additional parameters. Fix $t\in(0,1)$. Define the function
$\eta:\mathbb Z_+\cup\{0\}\to \mathbb R_+$ by
$$
\eta(q)=\sum_{r=0}^q\frac{\log(1+C_2 t^{2r})}{\log (t^{-1})}
$$
for an appropriate $t$-dependent constant $C_2>0$ to be determined later. 

We claim that there exists a constant $C_3>0$ such that 
\begin{equation}\label{eq:etainduction}
	\pi^k_{q+1}\geq C_3 t^{(2q+\eta(q))} \pi_{q}^k,\qquad \textnormal{for } 0\leq q\leq k-1.
\end{equation}
Note that $\eta(q)$ is bounded from above by the convergent infinite series $\eta_0 = \sum_{r=0}^\infty\log(1+C_2 t^{2r})$ and so it suffices to establish the strengthened version \eqref{eq:etainduction}.

We prove \eqref{eq:etainduction} by an induction in $k$. The induction base at $k=1$ and $q=0$ is by \eqref{eq:Ninitialdata} equivalent to the inequality $\pi^1_{0,1}=\frac{N^1_{0,1}}{N^1_{0,0}}=t\geq C_1$. For the induction step, we suppose the claim holds up to some $k\geq 1$.\\
	
	\uline{\textit{Case (i):} $1\leq q\leq k-2$.} For $q=1$, we use the convention that $\pi_{0}^k=1$. Then \eqref{eq:rkrecursion} and the induction hypothesis give
	$$
	\begin{aligned}
	\frac{\pi^{k+1}_{q+1}}{\pi^{k+1}_{q}}
	=&t^2\frac{t^{2}\pi^k_{q+2}+t\pi^k_{q+1}+ \pi^k_{q}}{t^2\pi^k_{q+1}+t\pi^k_{q}+\pi^k_{q-1}}\\
		\geq &C_3 t^{2q+\eta(q)}\frac{t^{6+\eta(q+1)-\eta(q)}\pi^k_{q+1}+t^3\pi^k_{q}+ t^{\eta(q-1)-\eta(q)}\pi^k_{q-1}}{t^2\pi^k_{q+1}+t\pi^k_{q}+\pi^k_{q-1}}\\
				\geq &C_3 t^{2q+\eta(q)}\frac{t^3\pi^k_{q}+ t^{\eta(q-1)-\eta(q)}\pi^k_{q-1}}{t^2\pi^k_{q+1}+t\pi^k_{q}+\pi^k_{q-1}}.
		\end{aligned}
	$$
 To conclude, we need to prove that the last fraction is $\geq 1$. From Lemma \ref{lm:qdecay} and $q\geq2$, we obtain the sufficient condition
	$$
C_0^2 t^{4q}+C_0 (1-t^2)t^{2q-1}\leq t^{\eta(q-1)-\eta(q)}-1.
	$$
We obtain the simpler sufficient condition
	$$
t^{\eta(q-1)-\eta(q)}\geq 1+ C t^{2 q}
	$$
	for an appropriate constant $C$ that depends only on $t$. From the definition of $\eta$, we obtain $t^{\eta(q-1)-\eta(q)}=1+C_2 t^{2q}$ and so the condition holds for $C_2$ sufficiently large. This completes the induction step for $2\leq q\leq k-2$.\\
	
	\uline{\textit{Case (ii):} $q=k-1$.} The recurrence relation \eqref{eq:rkrecursion} and the induction hypothesis \eqref{eq:etainduction} give
	$$
	\frac{\pi^{k+1}_{q+1}}{\pi^{k+1}_{q}}=t^2\frac{t\pi^k_{q+1}+ \pi^k_{q}}{t^2\pi^k_{q+1}+t\pi^k_{q}+\pi^k_{q-1}}
	\geq C_3 t^{2q+\eta(q)} \frac{t^{3}\pi^k_{q}+t^{\eta(q-1)-\eta(q)}  t^{2-a_q}\pi^k_{q-1}}{t^2\pi^k_{q+1}+t\pi^k_{q}+\pi^k_{q-1}}.
	$$
	The last fraction is $\geq 1$ by the same argument as above and this completes the induction step in Case (ii).\\
	
	\uline{\textit{Case (iii):} $q=k$.} The recurrence relation \eqref{eq:rkrecursion} and the induction hypothesis \eqref{eq:etainduction} give
	$$
	\frac{\pi^{k+1}_{q+1}}{\pi^{k+1}_{q}}=t^2\frac{ \pi^k_{q}}{t\pi^k_{q}+\pi^k_{q-1}}
	\geq C_3 t^{2q+\eta(q)} \frac{t^{\eta(q-1)-\eta(q)} \pi^k_{q-1}}{t\pi^k_{q}+\pi^k_{q-1}}.
	$$
	By Lemma \ref{lm:qdecay}, the last fraction is $\geq 1$, if the following sufficient condition is met 
	$$
	1+C_0t^{2q-1}\leq t^{\eta(q-1)-\eta(q)}=1+C_2 t^{2q}
	$$
	This can be ensured by choosing $C_2$ sufficiently large. This completes the induction step in Case (iii).\\

	\uline{\textit{Case (iv):} $q=0$.} The recurrence relation \eqref{eq:rkrecursion} gives
	$$
	\frac{\pi^{k+1}_{1}}{\pi^{k+1}_{0}}
	=\pi^{k+1}_{1}
	=t^2\frac{t^{2}\pi^k_{2}+t\pi^k_{1}}{t^2\pi^k_{1}+t}
	$$
	Now we use the monotonicity in $k$ established in Lemma \ref{lm:monotonicity} to bound
	$$
	t^2\frac{t^{2}\pi^k_{2}+t\pi^k_{1}}{t^2\pi^k_{1}+t}\geq
	t^2\frac{t\pi^k_{1}}{t^2\pi^k_{1}+t}
	\geq \frac{t^2\pi_1^{p+1}}{t\pi_1^{p+1}+1}.
	$$
Recall that $\pi_1^{p+1}=\frac{N^{p+1}_{p,1}}{N_{p,0}^{p+1}}$. This ratio is bounded from below by a constant depending only on $t$, which we then call $C_3$. To see this, observe that the sum for $N^{p+1}_{p,1}$ contains only a single walk of all down steps and so $N^{p+1}_{p,1}=t^{(p+1)^2}$ and that $N^{p+1}_{p,0}$ contains all walks obtained by placing a single flat step anywhere among $p$ down steps. Such a walk has area $t^{p^2+2\ell}$ if $\ell$ denotes the position of the flat step. Hence we have
$$
N^{p+1}_{p,0}\leq \sum_{\ell=0}^{p} t^{p^2+2\ell}\leq \frac{t^{p^2}}{1-t^2}
$$
and we obtain $\pi_1^{p+1}\geq \frac{t^2}{1-t^2}$. This completes the induction step in Case (iv) and proves Lemma \ref{lm:rlb}.
\end{proof}

\subsubsection{Main induction argument}
We will not directly prove Theorem \ref{thm:mainec}, but start with the following refined version instead. The refinement is introduced because more quantitative control is required to make an inductive argument work.\\

Given  some $A_0>0$ and $\nu\geq 2$, we define the sequence 
\begin{equation}\label{eq:Aqdefn}
A_q=A_0 \sum_{q'=0}^{q} \nu^{-q'},\qquad \forall q\geq 1.
\end{equation}

\begin{theorem}[Refined exponential convergence]\label{thm:expconv}
	For every $t\in (0,1)$, there exist constants $A_0,\beta>0$, $\alpha\in (0,1)$ and $\nu\geq 2$ so that for the function
	\begin{equation}\label{eq:curlyEdefn}
	\mathcal{E}(k,p,q)=\left(\prod_{j=1}^{k-1} \left( 1+ t^{(\alpha^{1/2}-\alpha)j} \right) \right) A_q t^{\alpha (k-1-\beta p-\beta q)}
	\end{equation}
	it holds that
	\begin{equation} \label{eq:convspeed}
	\pi_{0,q}^{\infty}-\pi_{p,q}^k\leq \mathcal{E}(k,p,q) \pi_{0,q}^{\infty},\qquad  \textnormal{for } p+q\leq k.
	\end{equation}
\end{theorem}

\begin{proof}[Proof of Theorem \ref{thm:mainec} from Theorem \ref{thm:expconv}]
	It suffices to note that the prefactors in $\mathcal{E}(k,p,q)$ are harmless, i.e., using $1+x\leq \exp x$,
	$$
	\begin{aligned}
	A_q \leq& 2A_0,\\
	\prod_{j=1}^{k-1} \left( 1+ t^{(\alpha^{1/2}-\alpha)j} \right) \leq & \exp \left( \frac{1} {1-t^{\alpha^{1/2} - \alpha}}  \right) = C_{\alpha}
	\end{aligned}
	$$
	where the geometric series is summable because $\alpha\in (0,1)$.
\end{proof}

We are now ready to prove the upper bound in \eqref{eq:convspeed} by induction in $k$. The induction base case occurs at $k=1$ and $(p,q)\in\{(0,0),(1,0),(0,1)\}$ and its validity is ensured by choosing $A_0$ sufficiently large in a way that only depends on $t$. The precise condition can be calculated from \eqref{eq:Ninitialdata}.\\

For the induction step, we suppose that \eqref{eq:convspeed} holds up to some fixed $k\geq 2$. Thanks to the boundary condition $\pi^k_{p,q}=0$ for $p+q>k$, we may assume without loss of generality that \eqref{eq:convspeed} holds for all $p,q\geq 0$. From now on, we fix a $p\geq 0$.

\subsubsection{The case $q\geq 2$}
We treat the more challenging case of $q\geq 2$ first and discuss the case $q=1$ at the end. By \eqref{eq:rkrecursion}, \eqref{eq:rmonotonep}, \eqref{eq:limitrecursion}, and Proposition \ref{prop:limit}, we have
$$
\begin{aligned}
1-\frac{\pi_{p,q}^{k+1}}{\pi_{q}^{\infty}}
=&1- \frac{t^2\pi_{1}^{\infty}+t}{t^2\pi^k_{p,1}+t+\mathbbm 1_{p\geq 1}t^{2k+2}\frac{\pi^k_{0,p-1}}{\pi^k_{0,p}}}\;\; \frac{t^2 \pi_{p,q+1}^k+t \pi_{p,q}^k+\pi_{p,q-1}^k}{t^2 \pi_{q+1}^{\infty}+t \pi_{q}^{\infty}+\pi_{q-1}^{\infty}}
\\
\leq & 1-\frac{t^2\pi_{1}+t}{t^2\pi_{1}^{\infty}+t+\mathbbm 1_{p\geq 1}t^{2k+2} \;\;\frac{\pi^k_{0,p-1}}{\pi^k_{0,p}}}\frac{t^2 \pi_{p,q+1}^k+t \pi_{p,q}^k+\pi_{p,q-1}^k}{t^2 \pi_{q+1}^{\infty}+t \pi_{q}^{\infty}+\pi_{q-1}^{\infty}} \\
\leq &\mathrm{(I)}+ \mathrm{(II)},
\end{aligned}
$$
where we introduced
$$
\begin{aligned}
\mathrm{(I)}=& \ind_{p\geq 1}\frac{t^{2k+2}\frac{\pi^k_{0,p-1}}{\pi^k_{0,p}}}{t^2\pi_{1}+t+\mathbbm 1_{p\geq 1}t^{2k+2}\frac{\pi^k_{0,p-1}}{\pi^k_{0,p}}},\\
\mathrm{(II)}=
&\frac{t^2 (\pi^{\infty}_{q+1}-\pi_{p,q+1}^k)+t (\pi_{q}^{\infty}-\pi^k_{p,q})+(\pi_{q-1}^{\infty}-\pi^k_{p,q-1})}{t^2 \pi_{q+1}^{\infty}+t \pi_{q}^{\infty}+\pi_{q-1}^{\infty}}.
\end{aligned}
$$

To estimate term $(\mathrm{I})$, we use Lemma \ref{lm:rlb}. It gives
\begin{equation}\label{eq:secondparest}
\mathrm{(I)}\leq t^{2k+1}\frac{\pi^k_{0,p-1}}{\pi^k_{0,p}}
\leq C_1^{-1} t^{2k+1-2(p-1)},\qquad \textnormal{for } p\geq1.
\end{equation}
We add the assumptions that
\begin{equation}\label{eq:A0constraint}
A_0\geq 2C_1^{-1},\qquad \beta=\alpha^{-1/2}.
\end{equation}
This implies, for sufficiently small $\alpha>0$,
\begin{equation}\label{eq:Iestimate}
\begin{aligned}
\mathrm{(I)}
\leq& \frac{1}{2}t^{(2-\alpha)k+(\alpha^{1/2}-2)p} \frac{\mathcal E(k+1,p,q)} {x_{k,p}}\\
\leq&\frac{1}{2} t^{(\alpha^{1/2}-\alpha) k } \frac{\mathcal E(k+1,p,q)}{x_{k,p}}.
\end{aligned}
\end{equation}
where we introduced
$$
x_{k,p}=\max\left\{1, t^{(\alpha^{1/2}-\alpha)k} \right\}.
$$

To estimate term (II), we first use the induction hypothesis.
$$
\begin{aligned}
\mathrm{(II)}&\leq \mathcal{E}(k+1,p,q)\frac{t^2 \frac{\mathcal{E}(k,p,q+1)}{\mathcal{E}(k+1,p,q)} \pi_{q+1}^{\infty}
	+t  \frac{\mathcal{E}(k,p,q)}{\mathcal{E}(k+1,p,q)}  \pi_{q}^{\infty}+ \frac{\mathcal{E}(k,p,q-1)}{\mathcal{E}(k+1,p,q)}\pi_{q-1}^{\infty}}{t^2 \pi_{q+1}^{\infty}+t \pi_{q}^{\infty}+\pi_{q-1}^{\infty}}.\\
\end{aligned}
$$
Then we apply the following lemma. 

\begin{lemma}\label{lm:fraction1} Let $t\in (0,1)$ and let  $C_0$ be given by Lemma \ref{lm:qdecay}. Suppose that
	\begin{equation}\label{eq:nu*constraint}
	\nu\in (\nu_*,2\nu_*),\qquad \nu_*>2\max\{C_0^2,2\}
	\end{equation}
	and define
	\begin{equation}\label{eq:Qdefn}
	Q=\inf\setof{q\geq 2}{C_0^2 t^{4q-2} <\frac{1}{2}}.
	\end{equation}
	There exists a constant $C_2$ such that for
	\begin{equation}\label{eq:alphaconstraint}
	\alpha\leq C_2 \nu_*^{-2Q}
	\end{equation}
	it holds that
	$$
	\frac{t^2 \frac{\mathcal{E}(k,p,q+1)}{\mathcal{E}(k+1,p,q)} \pi_{q+1}^{\infty}
		+t  \frac{\mathcal{E}(k,p,q)}{\mathcal{E}(k+1,p,q)}  \pi_{q}^{\infty}+ \frac{\mathcal{E}(k,p,q-1)}{\mathcal{E}(k+1,p,q)}\pi_{q-1}^{\infty}}{t^2 \pi_{q+1}^{\infty}+t \pi_{q}^{\infty}+\pi_{q-1}^{\infty}}\leq \frac{1}{1+t^{(\alpha^{1/2}-\alpha)k} }.
	$$
\end{lemma}

The proof of this lemma is postponed to Section \ref{sect:twolemmas}. Assuming it for the moment, we have shown that 
$$
\begin{aligned}
1-\frac{\pi_{p,q}^{k+1}}{\pi_{q}^{\infty}}
\leq&\mathrm{(I)}+\mathrm{(II)}\\
\leq& \frac{1}{2}t^{(\alpha^{1/2}-\alpha)k} \frac{\mathcal{E}(k+1,p,q)}{x_{k,p}}+\frac{\mathcal{E}(k+1,p,q)}{1+t^{(\alpha^{1/2}-\alpha)k} }\\
\leq& \mathcal{E}(k+1,p,q),
\end{aligned}
$$
where the last estimate is equivalent to
\begin{equation}\label{eq:tremainsmodif}
 \frac{1}{2}t^{(\alpha^{1/2}-\alpha)k} \frac{1}{x_{k,p}}+\frac{1}{1+t^{(\alpha^{1/2}-\alpha)k} }
\leq 1.
\end{equation}
This can be seen by distinguishing cases as follows. First, assume that $x_{k,p}=1$, or equivalently, $t^{(\alpha^{1/2}-\alpha)k }\leq 1$. Then \eqref{eq:tremainsmodif} follows from elementary estimates. If, conversely, $x_{k,p}>1$, then $ t^{(\alpha^{1/2}-\alpha)k}> 1$ and the left-hand side of \eqref{eq:tremainsmodif} equals $$\frac{1}{2}+\frac{1}{1+t^{(\alpha^{1/2}-\alpha)k}} < 1,$$ which implies \eqref{eq:tremainsmodif}. This completes the induction step for $q\geq 2$ modulo Lemma \ref{lm:fraction1}.

\subsubsection{Proof of Lemma \ref{lm:fraction1} }
\label{sect:twolemmas}
\begin{proof}[Proof of Lemma \ref{lm:fraction1}]
	By \eqref{eq:curlyEdefn} the claim can be written as the $k$-independent condition
	$$
	\begin{aligned}
	t^2 (A_q t^\alpha-A_{q+1} t^{-\alpha\beta}) \pi_{q+1}^{\infty}
	+t A_q (t^\alpha-1) \pi_{q}^{\infty}
	+(A_q t^\alpha-A_{q-1}t^{\alpha\beta})\pi_{q-1}^{\infty}
	\stackrel{!}{\geq} 0.
	\end{aligned}
	$$
	
	For $\alpha\to 0$ we can expand $t^\alpha=1+\alpha\log t+O(\alpha^2)$ and since $\beta=\alpha^{-1/2}$ we can expand $t^{\alpha\beta}$ analogously. Together with \eqref{eq:Aqdefn} this reveals the sufficient condition
	\begin{equation}\label{eq:smallalphaasymptotics}
	\begin{aligned}
	& A_0(\nu^{-q}\pi_{q-1}^{\infty}-t^2 \nu^{-q-1}\pi_{q+1}^{\infty})\\
	&+\alpha\log t \ob (A_q+\beta A_{q+1})t^2\pi_{q+1}+t A_q  \pi_{q} +(A_q -A_{q-1}\beta) \pi_{q-1}^{\infty} \cb \stackrel{!}{\geq} O(\alpha^2 \beta^2) A_0 R_{q} 
	,
	\end{aligned}
	\end{equation}
	where we used $\beta\geq 1$ and introduced 
	$$
	R_{q} =\max\{\pi_{q+1}^{\infty},\pi_{q}^{\infty} ,\pi_{q-1}^{\infty}\}.
	$$
	We remark that the implicit constant in the $O(\alpha^2\beta^2)=O(\alpha)$ term in \eqref{eq:smallalphaasymptotics} depends only on the parameter $t$ thanks to $A_0\leq A_q\leq 2A_0$. If we only keep the leading terms in \eqref{eq:smallalphaasymptotics} as $\alpha\to 0$, we arrive at the sufficient condition
	\begin{equation}\label{eq:smallqcond}
	\begin{aligned}
	&\pi_{q-1}^{\infty}
	\stackrel{!}{\geq} 
	\frac{t^2}{\nu} \pi_{q+1}^{\infty}+\nu^q  R_{q} O(\alpha^{1/2}).
	\end{aligned}
	\end{equation}
	Next, we shall distinguish cases for $q$. For sufficiently small $q$, it suffices to consider the leading-order terms and so we verify the sufficient condition \eqref{eq:smallqcond}. However, for large $q$, the subleading terms in $\alpha$ involve $\nu^q$ and are thus 
	eventually dominant, so we verify the finer condition \eqref{eq:smallalphaasymptotics} instead. In both cases, we will use the following key estimate from Proposition \ref{prop:limit}.
	\begin{equation}\label{eq:rdecayquse}
	\pi_{q+1}^{\infty}\leq C_0 t^{2q} \pi_{q}^{\infty}\leq C_0^2 t^{4q-2}\pi_{q-1}^{\infty}.
	\end{equation}
	
	We use $Q$ from \eqref{eq:Qdefn} as the cutoff value.\\
	
	\uline{\textit{Case (i): $2\leq q\leq Q$.}}
	We aim to verify \eqref{eq:smallqcond}. On the one hand, \eqref{eq:rdecayquse} gives
	$$
	\frac{t^2}{\nu} \pi_{q+1}^{\infty}
	\leq \frac{C_0^2 t^{4q}}{\nu}\pi_{q-1}^{\infty}
	\leq \frac{C_0^2 }{\nu}\pi_{q-1}^{\infty}
	<\frac{1}{2}\pi_{q-1}^{\infty}
	$$
	where the last estimate holds by \eqref{eq:nu*constraint}. Using $\nu^q\leq (2\nu_*)^{Q}$, we have
	$$
	\nu^{-q}\left(\pi_{q-1}^{\infty}-
	\frac{t^2}{\nu} \pi_{q+1}^{\infty}\right)
	\geq (2\nu_*)^{-Q}\frac{\pi_{q-1}^{\infty}}{2}.
	$$
	On the other hand, \eqref{eq:rdecayquse} yields
	\begin{equation}\label{eq:Rqbound}
	R_q=\max\{\pi_{q+1}^{\infty},\pi_{q}^{\infty},\pi_{q-1}^{\infty}\}\leq \max\{1,C_0^2\} \pi_{q-1}^{\infty}.
	\end{equation}
	We see that \eqref{eq:smallqcond} is implied by \eqref{eq:alphaconstraint}.\\

	\uline{\textit{Case (ii): $q> Q$.}}
	First, we drop the $\alpha$-independent terms from \eqref{eq:smallalphaasymptotics} which are no longer so useful for large $q$. We can do this because \eqref{eq:rdecayquse} and $q>Q$ imply that
	$$
	\pi_{q-1}^{\infty}-t^2 \nu^{-1}\pi_{q+1}^{\infty}
	\geq \pi_{q-1}^{\infty} \left( 1-\frac{t^2}{2\nu} \right)   \geq \pi_{q-1}^{\infty} \left( 1-\frac{t^2}{4} \right) \geq 0.
	$$
	Thus, for $q>Q$, \eqref{eq:smallalphaasymptotics} is implied by the condition
	\begin{equation}\label{eq:dropfirst}
	\frac{1}{6A_0}(A_{q-1}\pi_{q-1}^{\infty}-A_{q+1}\pi_{q+1}^{\infty}) \stackrel{!}{\geq} C R_q\alpha^{1/2},
	\end{equation}
	where we estimated $A_{q'}\leq 2A_0$ and absorbed $\frac{1}{6\log( t^{-1})}>0$ into the constant $C$.\\
	
	We recall Definition \eqref{eq:Aqdefn} of $A_q$. By \eqref{eq:rdecayquse} and $q>Q$, we have
	$$
	\begin{aligned}
	\frac{1}{6A_0}(A_{q-1}\pi_{q-1}^{\infty}-A_{q+1}\pi_{q+1}^{\infty})
	\geq&   \frac{1}{6A_0} \pi_{q-1}^{\infty} \left( A_{q-1}-\frac{A_{q+1}}{2} \right) \\
	\geq &   \frac{ \pi_{q-1}^{\infty}}{12} \left( 1-\frac{\nu^{-q}+\nu^{-q-1}}{2A_{q-1}}\right)\\
	\geq &   \frac{ \pi_{q-1}^{\infty}}{12} \left(1-\frac{\nu_*^{-2}}{A_0}\right)\\
	> &   \frac{ \pi_{q-1}^{\infty}}{24},
	\end{aligned}
	$$
	where the last estimate holds under the harmless assumption that $  A_0\nu_*^2 \geq 2.$ Using \eqref{eq:Rqbound}, we see that \eqref{eq:dropfirst} is implied by
	$$
	\alpha\leq CC_0^{-4}
	$$
	for an appropriate constant $C$ that depends only on $t$. This is implied by \eqref{eq:alphaconstraint} and Lemma \ref{lm:fraction1} is proved.
\end{proof}

\subsubsection{The case $q=1$.} The proof is very similar to the case $q\geq 2$ with some simplifcations thanks to $\pi^k_{p,0}=\pi_0=1$. The recursion \eqref{eq:rkrecursion} and \eqref{eq:limitrecursion} give
$$
\begin{aligned}
1-\frac{\pi_{p,1}^{k+1}}{\pi_{1}^{\infty}}
=&1- \frac{t^2\pi_{1}^{\infty}+t}{t^2\pi^k_{p,1}+t+\mathbbm 1_{p\geq 1}t^{2k+2}\frac{\pi^k_{0,p-1}}{\pi^k_{0,p}}}\;\; \frac{t^2 \pi_{p,2}^k+t \pi_{p,1}^k+1}{t^2 \pi_{2}^{\infty}+t \pi_{1}^{\infty}+1}
\leq \mathrm{(III)}+ \mathrm{(IV)},
\end{aligned}
$$
with
$$
\begin{aligned}
\mathrm{(III)}
=\frac{\mathbbm 1_{p\geq 1}t^{2k+2}\frac{\pi^k_{0,p-1}}{\pi^k_{0,p}}}{t^2\pi^k_{p,1}+t},\qquad 
\mathrm{(IV)}=\frac{t^2 (\pi_2^{\infty}-\pi_{p,2}^k)+t (\pi_1^{\infty}-\pi_{p,1}^k)}{t^2 \pi_{2}^{\infty}+t \pi_{1}^{\infty}+1}
\end{aligned}
$$

By \eqref{eq:Iestimate}, we have
$$
\mathrm{(III)}\leq \mathrm{(I)}\leq \frac{1}{2}t^{(\alpha^{1/2}-\alpha) k} \frac{\mathcal E(k+1,p,1)}{x_{k,p}}.
$$

For term (IV), we first use the induction hypothesis to obtain
$$
\begin{aligned}
\mathrm{(IV)}
&= \frac{t^2 (\pi_{2}^{\infty}-\pi_{p,2}^k)+t (\pi_1^{\infty}-\pi^k_{p,1})}{t^2 \pi_{2}^{\infty}+t \pi_1^{\infty}+1}\leq \mathcal{E}(k+1,p,1)\frac{t^2 \frac{\mathcal{E}(k,p,2)}{\mathcal{E}(k+1,p,1)} \pi_{2}^{\infty}
	+t  \frac{\mathcal{E}(k,p,1)}{\mathcal{E}(k+1,p,1)}\pi_1^{\infty} }{t^2 \pi_{2}^{\infty}+t \pi_1^{\infty}+1}.
\end{aligned}
$$

We have the following analog of Lemma \ref{lm:fraction1}:

\begin{lemma}\label{lm:fraction11} Let $t\in (0,1)$ and let  $C_0$ be given by Lemma \ref{lm:qdecay}. Suppose that
	\begin{equation}\label{eq:nu*constraint2}
	\nu\in (\nu_*,2\nu_*),\qquad \nu_*\geq C_0.
	\end{equation}
	There exists a constant $C_3$ such that for
	\begin{equation}\label{eq:alphaconstraint2}
	\alpha\leq C_3,\qquad \beta=\alpha^{-1/2},
	\end{equation}
	it holds that
	$$
	\frac{t^2 \frac{\mathcal{E}(k,p,2)}{\mathcal{E}(k+1,p,1)} \pi_{2}^{\infty}
		+t  \frac{\mathcal{E}(k,p,1)}{\mathcal{E}(k+1,p,1)}\pi_1^{\infty} }{t^2 \pi_{2}^{\infty}+t \pi_1^{\infty}+1}\leq  \frac{1}{1+t^{(\alpha^{1/2}-\alpha)k} }.
	$$
\end{lemma}

Assuming this lemma, we argue exactly as in the case $q=2$ above to obtain
$$
1-\frac{\pi_{p,1}^{k+1}}{\pi_{1}^{\infty}}\leq \mathrm{(III)}+ \mathrm{(IV)}\leq  \mathcal E(k+1,p,1)
$$
and hence the induction step for $q=1$. It thus suffices to prove this lemma.\\

\begin{proof}[Proof of Lemma \ref{lm:fraction11}]
	By \eqref{eq:curlyEdefn}, the claim reduces to
	$$
	t^2 \pi_2(A_1 t^\alpha-A_2 t^{-\alpha\beta})+t\pi_1^{\infty} A_1 (t^\alpha-1)+A_1 t^\alpha\stackrel{!}{\geq}0.
	$$
	Next, as $\alpha\to 0$, we can expand $t^\alpha=1+\alpha\log t+O(\alpha^2)$ and similarly $t^{\alpha\beta}$ with $\beta=\alpha^{-1/2}$ to obtain the sufficient condition
	\begin{equation}\label{eq:q1verify}
	t^2\pi_2^{\infty} (A_1-A_2)+t\pi_1^{\infty} A_1 +A_0 O(\alpha^{1/2})= A_0 \op \pi_1^{\infty}-\pi_2^{\infty}\frac{t^2}{\nu}\cp + A_0 O(\alpha^{1/2})\stackrel{!}{\geq} 0,
	\end{equation}
	where the implicit constant only depends on $t$. (Here we used $A_0=A_1<A_2< 2A_0$.) By Lemma \ref{lm:qdecay}, $\pi_2^{\infty}\leq C_0 t^2 \pi_1^{\infty}$. 	Since $t<1$ and $\nu>C_0$, we see that condition \eqref{eq:q1verify} holds for all $\alpha \leq C_3$ with $C_3$ depending only on the parameter $t$. 
\end{proof}



\section{Auxiliary Results}\label{sec:auxiliary-results}

Here are some simple technical lemmas that help formalize the approximations made in the main body of the proof.\\

\subsection{Transitive approximations lemma}

\begin{notation} \label{not:index-set}
	Let there be, for every $k$, an index set $I^{(k)}$, of size that can depend on $k$. In practice we will mostly use $I^{(k)} = \{(p,q) \; : \; 0 \le p,q \le (1 + c) \cdot k\}$.
\end{notation}

\begin{lemma} \label{lm:approx-of-approx}
	With Notation \ref{not:index-set}, consider three collections of normalized states, $\{\ket{a_{\alpha}^{(k)}}\}$, $\{\ket{b_{\alpha}^{(k)}}\}$, and $\{\ket{c_{\alpha}^{(k)}}\}$, with $\alpha\in I^{(k)}$ such that
	\begin{equation} \label{eq:tal-c1}
	\forall n \in \N: \quad \quad \quad \quad \lim\limits_{k \to \infty} \os k^n \cdot \op \sup_{\alpha \in I^{(k)}} \ob 1 - |\inner{a_{\alpha}^{(k)}}{b_{\alpha}^{(k)}}|^2 \cb \cp \cs = 0
	\end{equation}
and
	\begin{equation} \label{eq:tal-c2}
	\forall n \in \N: \quad \quad \quad \quad \lim\limits_{k \to \infty} \os k^n \cdot \op \sup_{\alpha \in I^{(k)}} \ob 1 - |\inner{b_{\alpha}^{(k)}}{c_{\alpha}^{(k)}}|^2 \cb \cp \cs = 0
	\end{equation}
Then we have
	\begin{equation} \label{eq:tal-claim}
	\forall n \in \N: \quad \quad \quad \quad \lim\limits_{k \to \infty} \os k^n \cdot \op \sup_{\alpha \in I^{(k)}} \ob 1 - |\inner{a_{\alpha}^{(k)}}{c_{\alpha}^{(k)}}|^2 \cb \cp \cs = 0
	\end{equation}
\end{lemma}
\begin{proof}
We introduce the following notations:
	\begin{equation}
	\e_{\alpha, k} \equiv 1 - |\inner{a_{\alpha}^{(k)}}{b_{\alpha}^{(k)}}|^2 \qquad \text{and} \qquad \delta_{\alpha, k} \equiv 1 - |\inner{b_{\alpha}^{(k)}}{c_{\alpha}^{(k)}}|^2
	\end{equation}
	We expand the $\ket{a_{\alpha}^{(k)}}$ and the $\ket{c_{\alpha}^{(k)}}$ in terms of their projection along the $\ket{b_{\alpha}^{(k)}}$, and a small residue.
	
	Note that, since states in our Hilbert space are only defined up to a global phase, any transformation of the type $\ket{a_{\alpha}^{(k)}} \to e^{i \phi} \ket{a_{\alpha}^{(k)}}$, for a real phase $\phi$ (which might even depend on $\alpha, k$), does not in any way affect the conditions \eqref{eq:tal-c1}, \eqref{eq:tal-c2} or the claim \eqref{eq:tal-claim}. Up to global phases, then, it is true that
	\begin{equation}
	\ket{a_{\alpha}^{(k)}} = \sqrt{1 - \e_{\alpha, k}} \cdot \ket{b_\alpha^{(k)}} + \sqrt{\e_{\alpha, k}} \cdot \ket{y_\alpha^{(k)}} \quad \quad \quad \quad \ket{c_\alpha^{(k)}} = \sqrt{1 - \delta_{\alpha, k}} \cdot \ket{b_\alpha^{(k)}} + \sqrt{\delta_{\alpha, k}} \cdot \ket{z_\alpha^{(k)}}
	\end{equation}
	where the $\ket{y_\alpha^{(k)}}$ and $\ket{z_\alpha^{(k)}}$ are normalized and orthogonal to $\ket{b_\alpha^{(k)}}$. Then the overlap between $a$ and $c$ will only contain $\inner{b}{b}$ and $\inner{y}{z}$ contributions, as the other two terms vanish by orthogonality:
	\begin{equation}
	\inner{a_\alpha^{(k)}}{c_\alpha^{(k)}} = \sqrt{(1 - \e_{\alpha, k})(1 - \delta_{\alpha, k})} \cdot \inner{b_\alpha^{(k)}}{b_\alpha^{(k)}} + \sqrt{\e_{\alpha, k} \cdot \delta_{\alpha, k}} \cdot \inner{y_\alpha^{(k)}}{z_\alpha^{(k)}}
	\end{equation}
	The first inner product on the RHS is exactly 1 by normalization of the $\ket{b_\alpha^{(k)}}$. It follows that
	\begin{align}
	|\inner{a_\alpha^{(k)}}{c_\alpha^{(k)}}|^2 &= (1 - \e_{\alpha, k})(1 - \delta_{\alpha, k}) + 2 \sqrt{(1 - \e_{\alpha, k})(1 - \delta_{\alpha, k}) \cdot \e_{\alpha, k} \cdot \delta_{\alpha, k}} \cdot |\inner{y_\alpha^{(k)}}{z_\alpha^{(k)}}|\\
	&\quad + \e_{\alpha, k} \cdot \delta_{\alpha, k} \cdot |\inner{y_\alpha^{(k)}}{z_\alpha^{(k)}}|^2 \nonumber\\
	1 - |\inner{a_\alpha^{(k)}}{c_\alpha^{(k)}}|^2 &= \e_{\alpha, k} + \delta_{\alpha, k} - \e_{\alpha, k} \cdot \delta_{\alpha, k} - 2 \sqrt{(1 - \e_{\alpha, k})(1 - \delta_{\alpha, k}) \cdot \e_{\alpha, k} \cdot \delta_{\alpha, k}} \cdot |\inner{y_\alpha^{(k)}}{z_\alpha^{(k)}}|\\
	&\quad - \e_{\alpha, k} \cdot \delta_{\alpha, k} \cdot |\inner{y_\alpha^{(k)}}{z_\alpha^{(k)}}|^2 \nonumber
	\end{align}
	By normalization, the inner product $\inner{y_\alpha^{(k)}}{z_\alpha^{(k)}}$ is at most 1 in absolute value. We also take the absolute value of the equation above and use the triangle inequality to obtain
	\begin{equation}
	\Big| 1 - |\inner{a_\alpha^{(k)}}{c_\alpha^{(k)}}|^2 \Big| \le \e_{\alpha, k} + \delta_{\alpha, k} +2 \cdot \e_{\alpha, k} \cdot \delta_{\alpha, k} + 2 \sqrt{(1 - \e_{\alpha, k})(1 - \delta_{\alpha, k}) \cdot \e_{\alpha, k} \cdot \delta_{\alpha, k}}
	\end{equation}
	All the terms on the RHS vanish fast enough as $k \to \infty$, and therefore so will the LHS. Specifically we have $1 - \e_{\alpha, k} \le 1$ and $1 - \delta_{\alpha, k} \le 1$, while also $\e_{\alpha, k}, \delta _k \le \max (\e_{\alpha, k}, \delta _k) \le 1$. Then
	\begin{equation}
	\Big| 1 - |\inner{a_\alpha^{(k)}}{c_\alpha^{(k)}}|^2 \Big| \le \e_{\alpha, k} + \delta_{\alpha, k} +2 \cdot \e_{\alpha, k} \cdot \delta_{\alpha, k} + 2 \sqrt{(1 - \e_{\alpha, k})(1 - \delta_{\alpha, k}) \cdot \e_{\alpha, k} \cdot \delta_{\alpha, k}} \le 6 \max (\e_{\alpha, k}, \delta_{\alpha,k})
	\end{equation}
	and, taking the supremum over $\alpha \in I^{(k)}$,
	\begin{equation}
	\sup_{\alpha \in I^{(k)}} \od \Big| 1 - |\inner{a_\alpha^{(k)}}{c_\alpha^{(k)}}|^2 \Big| \cd \le 6 \cdot  \sup_{\alpha \in I^{(k)}} \ob \max (\e_{\alpha, k}, \delta_{\alpha,k}) \cb
	\end{equation}
	and so the LHS will vanish in the limit $k \to \infty$, even when multiplied by $k^n$, since the RHS does by assumption.\\
\end{proof}

\subsection{Superposition approximations lemma}
In the following Lemma \ref{lm:superposition-approx}, we formalize the intuition that, given a superposition (sum) of states, and a superpolynomial approximation for each term in the sum, we naturally get a superpolynomial approximation for the superposition state. This is useful when splitting a spin chain into more than two pieces, since it allows us to do it stepwise (see Lemma \ref{lm:three-way-split}).\\

\begin{assumption} \label{as:sal-given-superposition}
	Suppose we have a consistent method of splitting the full chain into multiple subsegments, at various system sizes, such as the ABC split in Section \ref{sec:application}. Take a collection of states indexed by system size $k$ and unbalanced steps $(p,q)$, each of which is expressed as a superposition of products between an unnormalized ground state on one segment (called $L$) and an arbitrary state on the other segment (called $R$). Naming these states $\{\ket{s_{p,q}^{(k)}}\}$, we want
	\begin{equation}
	\ket{s_{p,q}^{(k)}} = {1 \over \sqrt{m_{p,q}^{(k)}}} \cdot \sum_{v \in I^{(k)}} \ket{\unnormstate{G^L_{p,v}}}\ket{z_{v,q}^{R}}.
	\end{equation}
	Here $I^{(k)}$ is an index set as in Notation \ref{not:index-set}. Also, the vectors $\ket{\unnormstate{G^L_{p,v}}}$ and $\ket{z_{v,q}^{R}}$ are not normalized, and the normalization factor in front is therefore
	\begin{equation}
	m_{p,q}^{(k)} = \sum_{v \in I^{(k)}} N^{L}_{p,v}\, \inner{z_{v,q}^R}{z_{v,q}^R}.
	\end{equation}
\end{assumption}

\begin{assumption} \label{as:sal-further-split}
	In the context of Assumption \ref{as:sal-given-superposition}, further split the segment $L$ into two parts, call them $A$ and $B$. Suppose that the choice of segments $A, B$, and the index set $I^{(k)}$ (i.e. the possible values of $v$), are all such that the assumptions of Lemmas \ref{lm:low-imbalance-approximation-1} and \ref{lm:low-imbalance-approximation-trunc} hold for all $0 \le p,q \le (1 + c) \cdot k$.
\end{assumption}

We want to show that we can approximate the given states $\ket{s^{(k)}}$ by replacing each unnormalized ground state $\ket{\unnormstate{G^L_{p,v}}}$ with the truncated walk-set state obtained after splitting $L=A\cup B$.

\begin{definition} \label{def:sal-approximate-superposition}
	Let $H^{L,<b}_{p,v}$ denote the truncated walk set from Definition \ref{def:low-imbalance-truncated-set}, applied to the segment $L=A\cup B$:
	\begin{equation}
	\ket{\unnormstate{H^{L,<b}_{p,v}}}=
	\sum_{r<bk}\ket{\unnormstate{G^A_{p,r}}}\ket{\unnormstate{G^B_{r,v}}},
	\qquad
	\norm{H^{L,<b}_{p,v}}=\sum_{r<bk}N^A_{p,r}N^B_{r,v}.
	\end{equation}
	Consider the approximate states
	\begin{equation}
	\ket{As_{p,q}^{(k)}} = {1 \over \sqrt{M_{p,q}^{(k)}}} \cdot \sum_{v \in I^{(k)}} \ket{\unnormstate{H^{L,<b}_{p,v}}}\ket{z_{v,q}^{R}}.
	\end{equation}
	Equivalently, expanding the truncated walk-set state,
	\begin{equation}
	\ket{As_{p,q}^{(k)}} = {1 \over \sqrt{M_{p,q}^{(k)}}} \cdot \sum_{v \in I^{(k)}} \op \sum_{r < bk} \ket{\unnormstate{G^A_{p,r}}} \ket{\unnormstate{G^B_{r,v}}} \cp \ket{z_{v,q}^{R}}.
	\end{equation}
	The normalization factor is
	\begin{equation}
	M_{p,q}^{(k)} = \sum_{v \in I^{(k)}} \norm{H^{L,<b}_{p,v}}\, \inner{z_{v,q}^R}{z_{v,q}^R}.
	\end{equation}
\end{definition}
The formal claim of this subsection is, then, the following:
\begin{lemma} \label{lm:superposition-approx}
	With the Assumptions \ref{as:sal-given-superposition}, \ref{as:sal-further-split} and Definition \ref{def:sal-approximate-superposition}, we have	
	\begin{equation}
	\forall n \in \N: \quad \quad \quad \quad \lim\limits_{k \to \infty} \os k^n \cdot \op \sup\limits_{0 \le p,q \le (1+c) \cdot k} \ob 1 - |\inner{s_{p,q}^{(k)}}{As_{p,q}^{(k)}}|^2 \cb \cp \cs = 0.
	\end{equation}
\end{lemma}
\begin{proof}
	At any particular $v$, Definition \ref{def:walk_set_definitions} and the disjoint walk-set decomposition give
	\begin{equation}
	\bra{\unnormstate{G^L_{p,v}}}\ket{\unnormstate{H^{L,<b}_{p,v}}}
	=\norm{H^{L,<b}_{p,v}}.
	\end{equation}
	If the assumptions of Lemma \ref{lm:low-imbalance-approximation-trunc} hold (with $a_1 = a_2 = 1 + c$ and for all $v \in I^{(k)}$), then we have that 
	\begin{equation}
	\forall n \in \N: \quad \quad \quad \quad \lim\limits_{k \to \infty} \os k^n \cdot \sup\limits_{ p \le (1 + c) \cdot k}\op \sup_{v \in I^{(k)}} \op 1 - {\norm{H^{L,<b}_{p,v}} \over N^L_{p,v}} \cp \cp \cs = 0.
	\end{equation}
	The above is independent of $q$, so we can harmlessly introduce it in the first supremum:
	\begin{equation} \label{eq:sal-supdif-vanish}
	\forall n \in \N: \quad \quad \quad \quad \lim\limits_{k \to \infty} \os k^n \cdot \sup\limits_{0 \le p,q \le (1 + c) \cdot k}\op \sup_{v \in I^{(k)}} \op 1 - {\norm{H^{L,<b}_{p,v}} \over N^L_{p,v}} \cp \cp \cs = 0.
	\end{equation}
	Meanwhile, the overlap between the given and approximate states is 
	\begin{align}
	\inner{s_{p,q}^{(k)}}{As_{p,q}^{(k)}} &= {1 \over \sqrt{m_{p,q}^{(k)} \cdot M_{p,q}^{(k)} }} \sum_{v \in I^{(k)}} \bra{\unnormstate{G^L_{p,v}}}\ket{\unnormstate{H^{L,<b}_{p,v}}} \cdot \inner{z_{v,q}^R}{z_{v,q}^R}\\
	&= {1 \over \sqrt{m_{p,q}^{(k)} \cdot M_{p,q}^{(k)} }} \sum_{v \in I^{(k)}} \norm{H^{L,<b}_{p,v}} \cdot \inner{z_{v,q}^R}{z_{v,q}^R}\\
	& = \sqrt{M_{p,q}^{(k)} \over m_{p,q}^{(k)}}.
	\end{align}
	So, with the notation $z_{v,q} \equiv \inner{z_{v,q}^R}{z_{v,q}^R}$, the quantity we're looking to bound is
	\begin{align*}
	1 - |\inner{s_{p,q}^{(k)}}{As_{p,q}^{(k)}}|^2 &= 1 - {M_{p,q}^{(k)} \over m_{p,q}^{(k)}} \\
	&= {\sum_{v \in I^{(k)}} N^L_{p,v} \cdot \op 1 - {\norm{H^{L,<b}_{p,v}} \over N^L_{p,v}} \cp \cdot z_{v,q} \over \sum_{v \in I^{(k)}} N^L_{p,v} \cdot z_{v,q}}\\
	&= \sum_{v \in I^{(k)}} \op { N^L_{p,v} \cdot z_{v,q} \over \sum_{v \in I^{(k)}} N^L_{p,v} \cdot z_{v,q}} \cp \cdot \op 1 - {\norm{H^{L,<b}_{p,v}} \over N^L_{p,v}} \cp.
	\end{align*}
	The expression on the last line makes it explicit that $1 - |\inner{s_{p,q}^{(k)}}{As_{p,q}^{(k)}}|^2$ is equal to the weighted average of the quantities $1 - \norm{H^{L,<b}_{p,v}}/N^L_{p,v}$ at various $v \in I^{(k)}$, since the prefactors of these quantities sum to 1. Since all terms are positive, such a weighted average will be bounded from above by the supremum over $I^{(k)}$ of the quantities:
	\begin{equation}
	\sum_{v \in I^{(k)}} \op { N^L_{p,v} \cdot z_{v,q} \over \sum_{v \in I^{(k)}} N^L_{p,v} \cdot z_{v,q}} \cp \cdot \op 1 - {\norm{H^{L,<b}_{p,v}} \over N^L_{p,v}} \cp \le \sup_{v \in I^{(k)}} \op 1 - {\norm{H^{L,<b}_{p,v}} \over N^L_{p,v}} \cp.
	\end{equation}
	Taking the supremum over $p,q$ gives
	\begin{equation}
	\sup\limits_{0 \le p,q \le (1 + c) \cdot k} \op 1 - |\inner{s_{p,q}^{(k)}}{As_{p,q}^{(k)}}|^2 \cp \le \sup\limits_{0 \le p,q \le (1 + c) \cdot k} \op \sup_{v \in I^{(k)}} \op 1 - {\norm{H^{L,<b}_{p,v}} \over N^L_{p,v}} \cp \cp.
	\end{equation}
	But the RHS vanishes as $k \to \infty$, even when multiplied by any polynomial in $k$, as seen in eq. \eqref{eq:sal-supdif-vanish}. Therefore so does the LHS, and the proof is complete.	
\end{proof}

\subsection{Three-way split lemma}
Here, we combine several previous approximation lemmas to rigorously show how ground states can be approximated when the chain is divided into three parts.

\begin{assumption} \label{as:twsl-conditions}
	Let the spin chain be divided into the $A,B,C$ segments as described in sec. \ref{ssec:low-imb-app}, and also assume the low-imbalance condition $p,q \le (1+c) \cdot k$ holds true.
\end{assumption} 

\begin{definition} \label{def:twsl-approx-def}
		Let $I^{ABC;r,v}_{p,q}$ be the set of walks in $G^{ABC}_{p,q}$ whose interface heights at $A|B$ and $B|C$ are respectively $r$ and $v$, and which reach zero height separately in all three segments. Define
		\begin{equation}
		H^{ABC,<b}_{p,q}:=\bigsqcup_{r,v<bk} I^{ABC;r,v}_{p,q}.
		\end{equation}
		Consider the following approximate ground states on the full chain:
		\begin{equation}
		\ket{\normstate{H^{ABC,<b}_{p,q}}} = {1 \over \sqrt{\norm{H^{ABC,<b}_{p,q}}}} \sum_{r,v < bk} \ket{\unnormstate{G^A_{p,r}}} \ket{\unnormstate{G^B_{r,v}}} \ket{\unnormstate{G^C_{v,q}}} 
		\end{equation}
		where the prefactor is chosen to ensure proper normalization:
		\begin{equation}
		\norm{H^{ABC,<b}_{p,q}} = \sum_{r,v < bk} N^{A}_{p,r} N^{B}_{r,v} N^{C}_{v,q}
		\end{equation}
	\end{definition}

\begin{lemma} \label{lm:three-way-split}
	 With Assumption \ref{as:twsl-conditions}, we have that the states defined in \ref{def:twsl-approx-def} approximate the true ground states superpolynomially:
	 \begin{equation} \label{eq:twsl-claim}
	 \forall n \in \N: \quad \quad \quad \quad \lim\limits_{k \to \infty} \os k^n \cdot \op \sup\limits_{0 \le p,q \le (1 + c) \cdot k} \ob 1 - |\inner{GS_{p,q}^{ABC}}{\normstate{H^{ABC,<b}_{p,q}}}|^2 \cb \cp \cs = 0
	 \end{equation}
\end{lemma}

\begin{proof}
	Consider first splitting the full chain $ABC$ into two parts, namely $AB$ and $C$. Recalling that the size of $C$ was chosen such that the condition $q < (1+c) \cdot k$ is enough for Lemma \ref{lm:low-imbalance-approximation-trunc} to apply (and the same with $p < (1+c) \cdot k$ and the size of $A$, which is of course below the size of $AB$), we find that the states
	\begin{equation} \label{eq:twsl-intermediates}
	\ket{a_{p,q}^{ABC}} = {1 \over \sqrt{n_{p,q}^{ABC}}} \sum_{v < bk} \ket{\unnormstate{G^{AB}_{p,v}}} \ket{\unnormstate{G^C_{v,q}}} 
	\end{equation}
	will superpolynomially approximate the true ground states $\ket{GS_{p,q}^{ABC}}$:
	\begin{equation}
	 \forall n \in \N: \quad \quad \quad \quad \lim\limits_{k \to \infty} \os k^n \cdot \op \sup\limits_{0 \le p,q \le (1 + c) \cdot k} \ob 1 - |\inner{GS_{p,q}^{ABC}}{a_{p,q}^{ABC}}|^2 \cb \cp \cs = 0
	\end{equation}
	In the above, of course, $n_{p,q}^{ABC}$ is chosen to ensure proper normalization.
	
	We now use Lemma \ref{lm:superposition-approx} to approximate each $\ket{\unnormstate{G^{AB}_{p,v}}}$ term inside the sum. Note that the set of $v$ that we are summing over, which is $\{0, 1, \dots, bk \}$, fulfills the conditions of Lemma \ref{lm:superposition-approx}. Most importantly, it ensures that $v < bk$, which in turn is much smaller than the size of $B$, allowing the application of Lemma \ref{lm:low-imbalance-approximation-trunc} when splitting $A$ from $B$. Therefore, the $\ket{\normstate{H^{ABC,<b}_{p,q}}}$ states (which are exactly what comes out of the application of Lemma \ref{lm:superposition-approx}) superpolynomially approximate the $\ket{a_{p,q}^{ABC}}$:
	\begin{equation}
	\forall n \in \N: \quad \quad \quad \quad \lim\limits_{k \to \infty} \os k^n \cdot \op \sup\limits_{0 \le p,q \le (1 + c) \cdot k} \ob 1 - |\inner{a_{p,q}^{ABC}}{\normstate{H^{ABC,<b}_{p,q}}}|^2 \cb \cp \cs = 0
	\end{equation}

	Since the $\ket{a_{p,q}^{ABC}}$ approximate the $\ket{GS_{p,q}^{ABC}}$, and the $\ket{\normstate{H^{ABC,<b}_{p,q}}}$ in turn approximate the $\ket{a_{p,q}^{ABC}}$, we invoke the result of Lemma \ref{lm:approx-of-approx}, with the index set $I^{(k)} = \{(p,q) \; : \; 0 \le p,q \le (1 + c) \cdot k\}$, and the proof is complete.
	
\end{proof}

\subsection{Approximation of expectations and overlaps} 

\begin{assumption} \label{as:oal-conditions-1}
	With Notation \ref{not:index-set}, consider collections of normalized states $\{\ket{a_\alpha^{(k)}}\},\{\ket{b_\alpha^{(k)}}\}$ with $\alpha\in I^{(k)}$ such that
	\begin{equation} \label{eq:oal-cond1}
	\forall n \in \N: \qquad \lim\limits_{k \to \infty} \os  \sup_{\alpha \in I^{(k)}} \ob 1 - |\inner{a_\alpha^{(k)}}{b_\alpha^{(k)}}|^2 \cb \cs = 0
	\end{equation}
\end{assumption}
\begin{assumption} \label{as:oal-conditions-2}
	Also let $O_\alpha^{(k)}$ be a collection of Hermitian, bounded operators, indexed by $k$ and $\alpha \in I^{(k)}$, all with norm less than some constant $c_0 \in \R$:
	\begin{equation} \label{eq:oal-cond2}
	\forall k \in \N \quad \forall \alpha \in I^{(k)}: \qquad \qquad \| O_\alpha^{(k)} \| \le c_0
	\end{equation}
\end{assumption}

\begin{lemma} \label{lm:overlaps-and-expectations}
	Under the conditions of Assumptions \ref{as:oal-conditions-1} and \ref{as:oal-conditions-2}, we have that, in the $k \to \infty$ limit, the expectation of the $O$ operators in the $a$ states can be recovered by replacing $a$ with its approximation $b$. This estimation is also uniform over $\alpha \in I^{(k)}$:
	\begin{equation} \label{eq:oal-res1}
	\lim\limits_{k\to\infty} \os \sup_{\alpha \in I^{(k)}} \Big| \mel{a_\alpha^{(k)}}{O_\alpha^{(k)}}{a_\alpha^{(k)}} - \mel{b_\alpha^{(k)}}{O_\alpha^{(k)}}{b_\alpha^{(k)}} \Big| \cs = 0
	\end{equation} 
	Furthermore, we can distribute the supremum to conclude in particular that
	\begin{equation} \label{eq:oal-res2}
	\lim\limits_{k\to\infty} \os \sup_{\alpha \in I^{(k)}} \oc \mel{a_\alpha^{(k)}}{O_\alpha^{(k)}}{a_\alpha^{(k)}} \cc - \sup_{\alpha \in I^{(k)}} \oc \mel{b_\alpha^{(k)}}{O_\alpha^{(k)}}{b_\alpha^{(k)}} \cc \cs = 0
	\end{equation} 
	or equivalently, if either limit is known to exist, then
	\begin{equation} \label{eq:oal-res3}
	\lim\limits_{k\to\infty} \os \sup_{\alpha \in I^{(k)}} \oc \mel{a_\alpha^{(k)}}{O_\alpha^{(k)}}{a_\alpha^{(k)}} \cc \cs = \lim\limits_{k\to\infty} \os \sup_{\alpha \in I^{(k)}} \oc \mel{b_\alpha^{(k)}}{O_\alpha^{(k)}}{b_\alpha^{(k)}} \cc \cs 
	\end{equation}
\end{lemma}

\begin{corollary} \label{cor:overlaps}
	Given Assumption \ref{as:oal-conditions-1}, and any collection of normalized states $\{\ket{z_\alpha^{(k)}}\}$, it holds true that
	\begin{equation} \label{eq:oalc-res1}
	\lim\limits_{k\to\infty} \os \sup_{\alpha \in I^{(k)}} \big| \inner{z_\alpha^{(k)}}{a_\alpha^{(k)}} \big| - \sup_{\alpha \in I^{(k)}} \big| \inner{z_\alpha^{(k)}}{b_\alpha^{(k)}} \big| \cs = 0
	\end{equation} 
	or equivalently, if either limit is known to exist, that
	\begin{equation} \label{eq:oalc-res2}
	\lim\limits_{k\to\infty} \os \sup_{\alpha \in I^{(k)}} \big| \inner{z_\alpha^{(k)}}{a_\alpha^{(k)}} \big| \cs = \lim\limits_{k\to\infty} \os \sup_{\alpha \in I^{(k)}} \big| \inner{z_\alpha^{(k)}}{b_\alpha^{(k)}} \big| \cs 
	\end{equation}
\end{corollary}

\begin{proof}[Proof of Corollary \ref{cor:overlaps} from Lemma \ref{lm:overlaps-and-expectations}]
	Define the operators which project onto the $\ket{z_\alpha^{(k)}}$ state:
	\begin{equation}
	O_\alpha^{(k)} \equiv \ket{z_\alpha^{(k)}} \bra{z_\alpha^{(k)}}
	\end{equation}
	They are manifestly Hermitian, and due to the normalization of $\ket{z_\alpha^{(k)}}$, the norm of any $O_\alpha^{(k)}$ operator is 1. Since they fulfill the Assumption \ref{as:oal-conditions-2}, Lemma \ref{lm:overlaps-and-expectations} applies. The expectation $\mel{a}{O}{a}$ is just $|\inner{z}{a}|^2$, and similarly for the $\ket b$ states:
	\begin{equation} 
	\lim\limits_{k\to\infty} \os \sup_{\alpha \in I^{(k)}} |\inner{z_\alpha^{(k)}}{a_\alpha^{(k)}}|^2 - \sup_{\alpha \in I^{(k)}} |\inner{z_\alpha^{(k)}}{b_\alpha^{(k)}}|^2 \cs = 0
	\end{equation}
	Since the $|\inner{z_\alpha^{(k)}}{a_\alpha^{(k)}}|$ and $|\inner{z_\alpha^{(k)}}{b_\alpha^{(k)}}|$ are real and positive quantities, the squaring can be omitted:
	\begin{equation} 
	\lim\limits_{k\to\infty} \os \sup_{\alpha \in I^{(k)}} |\inner{z_\alpha^{(k)}}{a_\alpha^{(k)}}| - \sup_{\alpha \in I^{(k)}} |\inner{z_\alpha^{(k)}}{b_\alpha^{(k)}}| \cs = 0
	\end{equation}
	and the proof is complete.	
\end{proof}

\begin{proof} [Proof of Lemma \ref{lm:overlaps-and-expectations}]
	As in the proof of Lemma \ref{lm:approx-of-approx}, denote the approximation error by $\e_{\alpha, k}$:
	\begin{equation}
	\e_{\alpha, k} \equiv 1 - |\inner{a_{\alpha}^{(k)}}{b_{\alpha}^{(k)}}|^2
	\end{equation}
	and expand the approximations in terms of the exact states:
	\begin{equation}
	\ket{b_{\alpha}^{(k)}} = \sqrt{1 - \e_{\alpha, k}} \cdot \ket{a_\alpha^{(k)}} + \sqrt{\e_{\alpha, k}} \cdot \ket{c_\alpha^{(k)}}
	\end{equation}
	where $\ket{c_\alpha^{(k)}}$ is normalized and orthogonal to $\ket{a_\alpha^{(k)}}$. We find the matrix element
	\begin{align*}
	\mel{b_\alpha^{(k)}}{O_\alpha^{(k)}}{b_\alpha^{(k)}} &= (1 - \e_{\alpha, k}) \cdot \mel{a_\alpha^{(k)}}{O_\alpha^{(k)}}{a_\alpha^{(k)}} + \sqrt{(1 - \e_{\alpha, k})\e_{\alpha, k}} \cdot \mel{a_\alpha^{(k)}}{O_\alpha^{(k)}}{c_\alpha^{(k)}} \\
	&\quad + \sqrt{\e_{\alpha, k}(1 - \e_{\alpha, k})} \cdot \mel{c_\alpha^{(k)}}{O_\alpha^{(k)}}{a_\alpha^{(k)}}  + \e_{\alpha, k} \cdot \mel{c_\alpha^{(k)}}{O_\alpha^{(k)}}{c_\alpha^{(k)}} 
	\end{align*}
	Take the absolute value of the difference between the matrix element involving the approximate states and that including the original ones:
	\begin{align*}
	\Big| \mel{b_\alpha^{(k)}}{O_\alpha^{(k)}}{b_\alpha^{(k)}} - \mel{a_\alpha^{(k)}}{O_\alpha^{(k)}}{a_\alpha^{(k)}}\Big| &= \bigg | - \e_{\alpha, k} \cdot \mel{a_\alpha^{(k)}}{O_\alpha^{(k)}}{a_\alpha^{(k)}} + \sqrt{(1 - \e_{\alpha, k})\e_{\alpha, k}} \cdot \mel{a_\alpha^{(k)}}{O_\alpha^{(k)}}{c_\alpha^{(k)}} \\
	&\quad + \sqrt{\e_{\alpha, k}(1 - \e_{\alpha, k})} \cdot \mel{c_\alpha^{(k)}}{O_\alpha^{(k)}}{a_\alpha^{(k)}}  + \e_{\alpha, k} \cdot \mel{c_\alpha^{(k)}}{O_\alpha^{(k)}}{c_\alpha^{(k)}} \bigg |
	\end{align*}
	All the matrix elements on the RHS involve the operators $O_\alpha^{(k)}$ and normalized states, so in absolute value they are bounded above by $c_0$, since the norm of $O_\alpha^{(k)}$ is assumed to have that bound. Each prefactor of an RHS matrix element will clearly vanish at large $k$, and furthermore it will do so uniformly over $\alpha \in I^{(k)}$; that is because $\e_{\alpha, k}$ has this property. (Again, this is similar to the proof of Lemma \ref{lm:approx-of-approx}.) After invoking the triangle inequality, we arrive at the conclusion presented in eq. \eqref{eq:oal-res1}:
	\begin{equation}
	\lim\limits_{k \to \infty } \op \sup_{\alpha \in I^{(k)}} | \mel{a_\alpha^{(k)}}{O_\alpha^{(k)}}{a_\alpha^{(k)}} - \mel{b_\alpha^{(k)}}{O_\alpha^{(k)}}{b_\alpha^{(k)}} | \cp = 0
	\end{equation}
	The proof of the "furthermore" part of the lemma (eqs. \eqref{eq:oal-res2}, \eqref{eq:oal-res3}) is a straightforward exercise in limits. For brevity, make the notations:
	\begin{equation}
	x_{k, \alpha} \equiv \mel{a_\alpha^{(k)}}{O_\alpha^{(k)}}{a_\alpha^{(k)}} \qquad \qquad y_{k,\alpha} \equiv \mel{b_\alpha^{(k)}}{O_\alpha^{(k)}}{b_\alpha^{(k)}}
	\end{equation}
	These are real, since the $O_\alpha^{(k)}$ are Hermitian. This allows us to look at the suprema over $I^{(k)}$:
	\begin{equation}
	x_k \equiv \sup_{\alpha \in I^{(k)}} x_{k, \alpha} = \sup_{\alpha \in I^{(k)}} \mel{a_\alpha^{(k)}}{O_\alpha^{(k)}}{a_\alpha^{(k)}} \qquad \qquad y_k \equiv \sup_{\alpha \in I^{(k)}} y_{k, \alpha} = \sup_{\alpha \in I^{(k)}} \mel{b_\alpha^{(k)}}{O_\alpha^{(k)}}{b_\alpha^{(k)}}
	\end{equation}
	Since all states are normalized and the operators are bounded, the suprema $x_k$ and $y_k$ are finite.
	
	Fix $\e>0$, and use result \eqref{eq:oal-res1} to find $k_0$ large enough such that, at all $k > k_0$,
	\begin{equation}
	\sup_{\alpha \in I^{(k)}} \Big| \mel{a_\alpha^{(k)}}{O_\alpha^{(k)}}{a_\alpha^{(k)}} - \mel{b_\alpha^{(k)}}{O_\alpha^{(k)}}{b_\alpha^{(k)}} \Big| =  \sup_{\alpha \in I^{(k)}} | x_{k, \alpha} - y_{k, \alpha} | < {\e \over 2}
	\end{equation} 
	For an arbitrary but fixed $k>k_0$, consider without loss of generality the case $x_k > y_k$. By the assumption that $x_k$ is the supremum over $\alpha$ of $x_{k,\alpha}$, we know we can pick an $\alpha \in I^{(k)}$ such that
	\begin{equation}
	0 \le x_k - x_{k,\alpha} < {\e \over 2} \quad \implies \quad x_{k,\alpha} > x_k - {\e \over 2}
	\end{equation} 
	Then we use the bound on $| x_{k, \alpha} - y_{k, \alpha} |$ to conclude that
	\begin{equation}
	x_{k,\alpha} + {\e \over 2} > y_{k,\alpha} > x_{k,\alpha} - {\e \over 2} \quad \implies \quad y_{k,\alpha} > x_k - \e
	\end{equation}
	But $y_k \ge y_{k,\alpha}$ by definition, so $y_k > x_k - \e$. Together with the assumption $x_k > y_k$ we see $|x_k - y_k| < \e$. The argument is analogous if $x_k \le y_k$. We conclude
	\begin{equation}
	\lim\limits_{k \to \infty} (x_k - y_k) = 0
	\end{equation}
	which is exactly eq. \eqref{eq:oal-res2}. The truth of eq. \eqref{eq:oal-res3} follows if we assume that either limit involved exists.
\end{proof}

\subsection{Projector approximation lemma} In this subsection we formalize the idea of approximating a projector, when given good estimations for basis states of the operator's range.

\begin{assumption} \label{as:proj-approx-cond}
	Within a collection of Hilbert spaces $\HH^{(k)}$ indexed by $k \in \N$ (e.g. corresponding to spin systems of different sizes), let there be subspaces $S^{(k)}$ of dimension at most polynomial in $k$. Namely, there should exist $c_0 \in \R$ and $n_0 \in \N$ such that
	\begin{equation}
	\forall k : \qquad \dim S^{(k)} \le c_0 \cdot k^{n_0}
	\end{equation}
	Let $P^{(k)}$ be the projectors onto the subspaces $S^{(k)}$. Pick an orthonormal basis $\ket{\psi_\alpha^{(k)}}$ for $S^{(k)}$, where $\alpha$ takes values in an index set $I^{(k)}$, of cardinality $\dim S^{(k)}$. This gives
	\begin{equation}
	P^{(k)} = \sum_{\alpha \in I^{(k)}} \ket{\psi_\alpha^{(k)}} \bra{\psi_\alpha^{(k)}}
	\end{equation}
	Assume that we have good approximations $\{ \ket{A \psi_\alpha^{(k)}} \}$ for the basis states, in the sense that for all $n \in \N$:
	\begin{equation}
	\lim\limits_{k \to \infty} \od k^n \cdot \sup\limits_{\alpha \in I^{(k)}} \ob 1 - |\inner{A \psi_\alpha^{(k)}}{\psi_\alpha^{(k)}}|^2 \cb \cd = 0
	\end{equation}
\end{assumption}
\begin{definition} \label{def:proj-approx}
	Define the approximate projectors corresponding to the $P^{(k)}$ above as:
	$$AP^{(k)} = \sum_{\alpha \in I^{(k)}} \ket{A\psi_\alpha^{(k)}} \bra{A\psi_\alpha^{(k)}}$$
\end{definition}

\begin{lemma} \label{lm:proj-approx}
	
	Under Assumption \ref{as:proj-approx-cond} and with Definition \ref{def:proj-approx}, take collections of normalized states $\{\ket{a^{(k)}}\}_{k \in \N}$ and $\{\ket{b^{(k)}}\}_{k \in \N}$ such that at least one of the limits 
	$$\lim\limits_{k \to \infty} \mel{a^{(k)}}{P^{(k)}}{b^{(k)}} \quad \text{or} \quad \lim\limits_{k \to \infty} \mel{a^{(k)}}{AP^{(k)}}{b^{(k)}}$$ 
	exists. Then we have that the other limit also exists, and moreover they are equal:
	\begin{equation}
	\lim\limits_{k \to \infty} \mel{a^{(k)}}{P^{(k)}}{b^{(k)}}= \lim\limits_{k \to \infty} \mel{a^{(k)}}{AP^{(k)}}{b^{(k)}}
	\end{equation}
\end{lemma}
 
\begin{proof}
Begin by working at a specific $k$ and suppressing the $k$ index for simplicity. Letting 
\begin{equation}
\e_\alpha \equiv 1 - |\inner{A \psi_\alpha}{\psi_\alpha}|^2
\end{equation}
write the true state in terms of the approximate one, up to a global phase:
\begin{equation}
\ket{\psi_\alpha} = \sqrt{1 - \e_\alpha} \cdot \ket{A\psi_\alpha} + \sqrt{\e_\alpha} \cdot  \ket{\psi_\alpha'}
\end{equation}
with $\ket{\psi'_{\alpha}}$ being some normalized error term that is orthogonal to $\ket{A\psi_\alpha}$. The projector onto $\ket{\psi_\alpha}$ is expanded as
\begin{align}
	\ket{\psi_\alpha} \bra{\psi_\alpha} &= \ob \sqrt{1 - \e_\alpha} \cdot \ket{A\psi_\alpha} + \sqrt{\e_\alpha} \cdot  \ket{\psi_\alpha'} \cb \ob \sqrt{1 - \e_\alpha} \cdot \bra{A\psi_\alpha} + \sqrt{\e_\alpha} \cdot  \bra{\psi_\alpha'} \cb\\
	&= (1 - \e_\alpha) \cdot \ket{A\psi_\alpha}\bra{A\psi_\alpha} + \sqrt{\e_\alpha (1 - \e_\alpha)} \cdot \ob \ket{A\psi_\alpha} \bra{\psi_\alpha'} + \ket{\psi_\alpha'} \bra{A\psi_\alpha} \cb + \e_\alpha \cdot \ket{\psi_\alpha'} \bra{\psi_\alpha'}\\
	&= \ket{A\psi_\alpha}\bra{A\psi_\alpha} + \sqrt{\e_\alpha (1 - \e_\alpha)} \cdot \ob \ket{A\psi_\alpha} \bra{\psi_\alpha'} + \ket{\psi_\alpha'} \bra{A\psi_\alpha} \cb + \e_\alpha \cdot \ob \ket{\psi_\alpha'} \bra{\psi_\alpha'} - \ket{A\psi_\alpha}\bra{A\psi_\alpha}\cb
\end{align}
 Moving the approximate projector to the left and taking the matrix element between arbitrary $\bra a$ and $\ket b$, we find 
\begin{align}
	\bra a \ob \ket{\psi_\alpha} \bra{\psi_\alpha} - \ket{A \psi_\alpha} \bra{A \psi_\alpha} \cb \ket b &= \sqrt{\e_\alpha (1 - \e_\alpha)} \cdot \ob \inner{a}{A\psi_\alpha} \inner{\psi_\alpha'}{b} + \inner{a}{\psi_\alpha'} \inner{A\psi_\alpha}{b} \cb\\
	&\quad + \e_\alpha \cdot \ob \inner{a}{\psi_\alpha'} \inner{\psi_\alpha'}{b} - \inner{a}{A\psi_\alpha} \inner{A\psi_\alpha}{b} \cb \nonumber
\end{align}
Due to normalization, every overlap on the RHS is between 1 and -1, so we find
\begin{equation}
\left| \bra a \ob \ket{\psi_\alpha} \bra{\psi_\alpha} - \ket{A \psi_\alpha} \bra{A \psi_\alpha} \cb \ket b \right| \le 2 \sqrt{\e_\alpha} \cdot \ob \sqrt{\e_\alpha} + \sqrt{1 - \e_\alpha} \cb < 4 \sqrt{\e_\alpha}
\end{equation}
where the last inequality follows from $\e_\alpha \in [0,1]$. If we let \begin{equation}
\e^{(k)} \equiv \sup\limits_{\alpha \in I^{(k)}} \e_{\alpha}^{(k)}
\end{equation}
where the $k$ index was momentarily restored for clarity, then it follows that 
\begin{equation}
\left| \bra{a} \ob \ket{\psi_\alpha} \bra{\psi_\alpha} - \ket{A \psi_\alpha} \bra{A \psi_\alpha} \cb \ket{b} \right| < 4 \sqrt{\e} \quad \quad \quad \quad \quad \quad \forall \alpha
\end{equation}
and summing over all $\alpha$
\begin{equation}
\left| \bra a \sum_\alpha \ob \ket{\psi_\alpha} \bra{\psi_\alpha} - \ket{A \psi_\alpha} \bra{A \psi_\alpha} \cb \ket b \right| \le \sum_\alpha\left| \bra a  \ob \ket{\psi_\alpha} \bra{\psi_\alpha} - \ket{A \psi_\alpha} \bra{A \psi_\alpha} \cb \ket b \right|  < \sum_\alpha 4 \sqrt{\e}
\end{equation}
Now consider what happens as we take $k$ to be large. If the set of possible values of $\alpha$ has size polynomial in $k$ and we're assuming $\lim\limits_{k \to \infty} k^n \cdot \e^{(k)} = 0$ for all $n$, then the LHS from above also vanishes:
\begin{equation}
\lim\limits_{k \to \infty} \left| \bra{a^{(k)}} \sum_\alpha \ob \ket{\psi_\alpha^{(k)}} \bra{\psi_\alpha^{(k)}} - \ket{A \psi_\alpha^{(k)}} \bra{A \psi_\alpha^{(k)}} \cb \ket{b^{(k)}} \right| = 0
\end{equation}
Since it goes to zero, we can drop the absolute value. As we assume that the limit of the matrix element exists for at least one projector, we can move them on different sides to get
\begin{equation}
\lim\limits_{k \to \infty} \bra{a^{(k)}} \oc \sum_\alpha \ket{\psi_\alpha} \bra{\psi_\alpha} \cc \ket{b^{(k)}} = \lim\limits_{k \to \infty}  \bra{a^{(k)}} \oc \sum_\alpha  \ket{A \psi_\alpha} \bra{A \psi_\alpha} \cc \ket{b^{(k)}}
\end{equation}
as expected.
\end{proof}

\subsection{Proof of Proposition \ref{pr:hiap-p2}} \label{ssec:high-p-other-walks} The aim is to argue that $\| G_{[k+1, 3k]}  \ket{\mathrm {IIb}} \|$ vanishes at large $k$. Recall that $\ket{\mathrm {IIb}}$ is obtained by selecting, from the sum $\ket{\mathrm {II}}$, only the walks that do not have all their first $k$ steps down. The definition of $\ket{\mathrm {II}}$ was:
\begin{equation}
    \ket{\mathrm {II}} = \sum_{z = 1}^{p} \sum_{v \ge 0} \ket{GS_{p - z, v}^{[1,2k]}} \otimes \ket{\psi^2_{v + z,q}}
\end{equation}
Let $w$ be a walk that appears in the sum $\ket{\mathrm {IIb}}$, and let $r$ be the height that $w$ reaches after the first $k$ steps. Note that, by the definition of $\ket{\mathrm{II}}$, all walks reach zero height only within the last third of the chain, $[2k+1, 3k]$. Since they end at height $q$, it must be that $q \le k$. Also, the number of unbalanced steps of the walk's $[k+1, 3k]$ component is $(r,q)$.

First consider the case $r \le (1 - c) k$. We divide the last two thirds $[k+1, 3k]$ into two unequal segments, $L = [k+1, (2 - c/2) k]$ and $R = [(2-c/2)k, 3k]$. Note that the length of $L$ is $(1-c/2)k$, and it holds true that $r \le (1-c)k < (1 - c/2) k$. Similarly, the length of $R$ is $(1+c/2) k$ and so $q \le k \le (1+c/2) k$. The low-imbalance approximation lemma says that the ground state $\ket{GS_{r,q}^{[k+1, 3k]}}$ can be approximated using only walks that reach zero height within both $L$ and $R$. Since our walk $w$ is assumed to not reach zero height in $[k+1, 2k]$, which includes $L$, we see that it will pick up an exponentially small prefactor when compared to ground states on $[k+1, 3k]$.

On the other hand, let $r > (1 - c) k$. Then we use a different division of $[k+1, 3k]$, into $L = [k+1, (2 - 2c) k]$ and $R = [(2 - 2c)k, 3k]$. The length of $L$ is now $(1-2c)k$, and $r$ is assumed larger than that by at least $ck$. By the high-imbalance approximation lemma, the ground state $\ket{GS_{r,q}^{[k+1, 3k]}}$ can be approximated using only walks whose first $(1-2c k)$ steps (counting from position $k+1$ on) are down. If the walk $w$ does not have that property, then again it picks up an exponentially small factor when acted on by $G_{[k+1, 3k]}$. 

If on the other hand $w$ does have that property, then its contribution to the $[1,2k]$ ground state that it came from must be negligible (recall that $\ket{\mathrm{II}}$ contains only ground states on $[1,2k]$, so $w$ must have come from one of them). This holds true because:
\begin{itemize}
    \item If $p - r > 3ck$, then the component of $w$ on $[1,2k]$ has at least $(1+c)k$ unbalanced down steps, because we know that at least $(1-2c)k$ are found in the middle third. A high-imbalance approximation with $[1,2k]$ divided into $L = [1,k]$ and $R = [k+1, 2k]$ shows that ground states must (approximately) have their first $k$ steps down, which $w$ does not.
    \item If $p - r \le 3ck$, then use a low-imbalance approximation with a slightly wider $L = [1, (1 + 4c) k]$ and narrower $R = [(1 + 4c) k, 2k]$. The low-imbalance approximation must hold because the number of unbalanced down-steps of $w$ in the first third is $p-r \le 3ck$, and in the middle third it is at most $k$ (at most all steps). So the component of $w$ on $[1,2k]$ cannot have more than $(1 + 3c) k$ unbalanced down-steps. On the other hand, it cannot have more than $ck$ unbalanced up-steps either, since $r > (1 - c) k$. So the low imbalance regime applies, and ground states can be approximated by walks that reach zero height within both $L$ and $R$. But since $w$ has the property that the $(1-2c k)$ steps that follow the $k$th one are all down, it cannot satisfy the desired property.
\end{itemize}

Therefore every walk in $G_{[k+1, 3k]}  \ket{\mathrm {IIb}}$ gives a negligible contribution, and Proposition \ref{pr:hiap-p2} follows.

\subsection{High $q$ regime (Proof sketch for \eqref{eq:high-q-limit})}\label{ssec:high-q-outline} 
As mentioned after the derivation of \eqref{eq:high-p-limit} in the high-$p$ regime, the proof of \eqref{eq:high-q-limit} in the high-$q$-regime is very similar. Here we provide a sketch of the argument.  We can again start with the orthogonality $\inner{GS_{p,q}^{[1,3k]}}{\phi_{p,q}} = 0$, and this time we approximate the ground state on the full chain by 
$$\ket{GS_{p,q - k}^{[1,2k]}} \otimes \ket{u}^{\otimes k}.
$$
Comparison with $\ket{GS_{p,q - k}^{[1,2k]}}$ is, in this case, equivalent to acting with $G_{[1,2k]}$, as can be seen by counting unbalanced steps. Since $\ket {\phi_{p,q}}$ is a +1 eigenstate of this latter operator, it must be that the $\ket{u}^{\otimes k} \otimes \bra{u}^{\otimes k}$ approximately annihilates it. That is, for any walk that makes nonvanishing contributions to $\ket{\phi_{p,q}}$, not all the last $k$ steps are up. 

Now we compare $\ket{\phi_{p,q}}$ with ground states on $[k+1, 3k]$. Analogously to \eqref{eq:schmidt}, we can perform a Schmidt decomposition of $\ket{\phi_{p,q}}$ about subsystems $[1,k]$ and $[k+1,3k]$:
\begin{equation}
\ket{\phi_{p,q}} = \ket{\mathrm{I}} + \ket{\mathrm{II}} + \ket{\mathrm{III}}
\end{equation}
with the three terms corresponding to initial walks that reach zero height with both the left and middle thirds, only the left one $[1,k]$, or only the middle $[k+1, 2k]$:
\begin{align*}
    \ket{\mathrm{I}} &=\sum_{r \ge 0} \ket{GS_{p,r}^{[1, k]}} \otimes \ket{\psi^1_{r,q}}\\
    \ket{\mathrm{II}} &= \sum_{z = 1}^{q} \sum_{r \ge 0} \ket{GS_{p, r + z}^{[1, k]}} \otimes \ket{\psi^2_{r, q - z}}\\
    \ket{\mathrm{III}} &= \sum_{z = 1}^{p} \sum_{r \ge 0} \ket{GS_{p - z, r}^{[1,k]}} \otimes \ket{\psi^3_{r + z, q}}
\end{align*}
where the $\ket{\psi^i}$ are unnormalized (having absorbed the Schmidt coefficients) and live on $[k+1, 3k]$. When comparing the $\ket{\psi^i}$ with ground states, they must have the same numbers of unbalanced steps for the overlap to be nonzero. Since $q > (1 + c)k$, we use the high imbalance lemma to argue that ground states on $[k+1, 3k]$ with $q$ unbalanced up-steps must, to a good approximation, have all their last $k$ steps up. This means that their overlap with the terms in $\ket{\mathrm{I}}$ and $\ket{\mathrm{III}}$ vanishes, since the latter are known to be approximately annihilated by $\ket{u}^{\otimes k} \otimes \bra{u}^{\otimes k}$ acting on the last $k$ sites.

For the remaining term $\ket{\mathrm{II}}$ we observe that a walk which, on the entire $[1,3k]$, reaches the ground only within the first third, but does not have all its last $k$ steps up, must make an exponentially vanishing contribution to the relevant ground state on $[1,2k]$, or to that on $[k+1,3k]$, or both. The reason is as follows: if $q$ is very high, then the last $k$ steps must be all up in order for the walk to contribute to ground states on $[k+1, 3k]$. If $q$ is not very high, then consider the height $v$ after $2k$ steps. With high $v$, we have too many balanced steps in the last third to contribute to ground states on $[k+1, 3k]$. With low $v$, the low imbalance approximation on $[1,2k]$ says that ground states should reach zero height within both $[1,k]$ and $[k+1, 2k]$. So $G_{[k+1, 3k]}$ approximately annihilates $\ket{\phi_{p,q}}$, and the argument is complete.

\subsection{Proving Proposition \ref{pr:reduction-of-phi-pq}} \label{ssec:proving-reduction-of-phi-pq}

When expanding the state $\ket{\phi_{p,q}}$ in terms of ground states on $AB$, and some other states on $C$, we need to consider three possibilities:
\begin{itemize}
	\item (i) Walks that reach zero height both in AB and in C;
	\item (ii) Walks that reach zero height only in C, but not in AB;
	\item (iii) Walks that reach zero height only in AB, but not in C.
\end{itemize} 
As in eq. \eqref{eq:phi-pq-full-expansion}, we write
\begin{equation}
\ket{\phi_{p,q}} = \ket{\mathrm{I}} + \ket{\mathrm{II}} + \ket{\mathrm{III}}
\end{equation}
with the three separate terms corresponding to the three cases above:
\begin{align*}
\ket{\mathrm{I}} &= {1 \over \sqrt{{N'}^{\phi}_{p,q}}} \cdot \sum_{s \ge 0} \ket{\unnormstate{G^{AB}_{p,s}}} \ket{\psi^C_{s,q}}\\
\ket{\mathrm{II}} &= {1 \over \sqrt{{N'}^{\phi}_{p,q}}} \cdot \sum_{z = 1}^p \sum_{s \ge 0} \ket{\unnormstate{G^{AB}_{p-z,s}}} \ket{{\psi'}^C_{s + z,q}}\\
\ket{\mathrm{III}} &= {1 \over \sqrt{{N'}^{\phi}_{p,q}}} \cdot \sum_{y = 1}^q \sum_{s \ge 0} \ket{\unnormstate{G^{AB}_{p,s+y}}} \ket{{\psi''}^C_{s, q - y}}
\end{align*}
The normalization factor ${N'}^{\phi}_{p,q}$ is of course chosen such that $\inner{\phi_{p,q}}{\phi_{p,q}} = 1$. First note that, from the characterization above regarding heights, the three categories are fully disjoint. Therefore terms such as $\inner{\mathrm{I}}{\mathrm{II}}$, $\inner{\mathrm{I}}{\mathrm{III}}$ and $\inner{\mathrm{II}}{\mathrm{III}}$ vanish exactly. This gives
\begin{equation}
1 = \inner{\phi_{p,q}}{\phi_{p,q}} = \sum_{x \in \{\mathrm{I}, \mathrm{II}, \mathrm{III} \}} \inner{x}{x} \quad \implies \quad \inner{x}{x} \le 1 \quad \forall x \in \{\mathrm{I}, \mathrm{II}, \mathrm{III}\}
\end{equation}
When computing the matrix element $\mel{\phi_{p,q}}{G_{BC}}{\phi_{p,q}}$ we will get nine terms:
\begin{equation}
\mel{\phi_{p,q}}{G_{BC}}{\phi_{p,q}} = \sum_{x,y \in \{\mathrm{I}, \mathrm{II}, \mathrm{III}\}} \mel{x}{G_{BC}}{y}
\end{equation}
Note that, given a state $\ket{\phi_{p,q}}$, the decomposition into terms I, II, III is unique. We want to find the supremum of $\mel{\phi_{p,q}}{G_{BC}}{\phi_{p,q}}$ under the known conditions on $p, q$ and the state $\phi_{p,q}$. It is clear by the properties of the supremum that 
\begin{equation}
\sup_{\substack{p,q \le (1 + c) k \\ \phi_{p,q} \in \range E_k}} \mel{\phi_{p,q}}{G_{BC}}{\phi_{p,q}} = \sup_{\substack{p,q \le (1 + c) k \\ \phi_{p,q} \in \range E_k}} \sum_{x,y \in \{\mathrm{I}, \mathrm{II}, \mathrm{III}\}} \mel{x}{G_{BC}}{y} \le  \sum_{x,y \in \{\mathrm{I}, \mathrm{II}, \mathrm{III}\}} \sup_{\substack{p,q \le (1 + c) k \\ \phi_{p,q} \in \range E_k}} \mel{x}{G_{BC}}{y}
\end{equation}
i.e. the supremum of the sum must be bounded above by the sum of suprema for individual terms. The aim is to show that all terms in the rightmost sum vanish, except possibly the one with $x = y = \mathrm{I}$. Formally, we have
\begin{proposition} \label{pr:phi-red-others-vanish}
	Under the conditions of Section \ref{ssec:low-imb-app}, one has
	\begin{equation}
	\limsup\limits_{k\to\infty} \op \sum_{x,y \in \{\mathrm{I}, \mathrm{II}, \mathrm{III}\}} \sup_{\substack{p,q \le (1 + c) k \\ \phi_{p,q} \in \range E_k}} \mel{x}{G_{BC}}{y} \cp = \limsup\limits_{k\to\infty} \op \sup_{\substack{p,q \le (1 + c) k \\ \phi_{p,q} \in \range E_k}} \mel{\mathrm{I}}{G_{BC}}{\mathrm{I}} \cp
	\end{equation}
\end{proposition}
\begin{proof}[Proof sketch of Proposition \ref{pr:phi-red-others-vanish}]
The proof is based on very similar considerations as in earlier parts of the paper. In order to not repeat many similar estimates, we summarize the overarching line of argument, but leave the details to the reader.

	We want to show that eight terms vanish (all but the one with $\mathrm{I}$ on both sides). We will view such terms as inner products of a state $\ket{x} \in \{\ket{\mathrm{I}}, \ket{\mathrm{II}}, \ket{\mathrm{III}}\}$ and the state $G_{BC} \ket{y}$ with $y \in \{\mathrm{II}, \mathrm{III} \}$. Up to complex conjugation (which does not affect the vanishing of the supremum), all eight terms can be written like this. The idea is that, from Cauchy-Schwarz, we know
	\begin{equation}
	|\mel{x}{G_{BC}}{y}| \le \| \ket{x} \| \cdot \| G_{BC} \ket{y} \| \le \| G_{BC} \ket{y} \|
	\end{equation}
	where the second inequality follows since $\| \ket{x} \| \le 1$ for all $x$. So it is sufficient to argue that the norms $\| G_{BC} \ket{\mathrm{II}} \|$ and $\| G_{BC} \ket{\mathrm{III}} \|$ will vanish in the limit of large $k$, given that $\| \ket{\mathrm{II}} \|$ and $\| \ket{\mathrm{III}} \|$ are always at most 1.
	
	This last property can be seen by considering several subcases for each state. For example, we know the states in $\ket{\mathrm{II}}$ reach zero height within region $C$, but not in $AB$. When acting with the $G_{BC}$ projector, we are implicitly taking their overlap with ground states on $BC$.
	
	Since our walks reach zero height in $C$ and terminate at height $q$, it is clear that, when only looking at their $BC$ portions, we still have $q$ unbalanced steps on the right. The number of unbalanced steps on the left (of the $BC$ portion only), will be the height that they reach at the border between $A$ and $B$; call this height $r$. It is clear that the only $BC$ ground state which will give nonzero overlap with such a walk is $\ket{GS_{r,q}^{BC}}$.
	
	For walks whose $r$ is small (e.g. below two-thirds of the length of $B$), we use the approximation Lemma \ref{lm:low-imbalance-approximation-trunc} to show that $\ket{GS_{r,q}}$ is approximated (up to exponentially small errors) by a sum of walks which reach zero height in both $B$ and $C$. That clearly cannot have any overlap with our walk, so the only nonzero contribution must come from the exponentially suppressed terms which were excluded in the approximation lemma.
	
	On the other hand, consider walks with large $r$ (e.g. above two-thirds the length of $B$). If they have a small $v$ value (height at the border between $B$ and $C$), then we use Lemma \ref{lm:low-imbalance-approximation-trunc} to argue that they must, from the beginning, have had an exponentially vanishing prefactor within the ground state $\ket{GS_{p,v}^{AB}}$. If on the other hand both $r$ and $v$ are large, we use the high-imbalance approximation Lemma \ref{lm:approx-high}: for the $AB$ ground state, due to the large $v$ it requires that the rightmost steps in $B$ are up; for the $BC$ ground state, it requires that the leftmost are down due to the large $r$. We can take our definitions of 'large' so that these two requirements are contradictory, so that any walk with $r$ and $v$ both large will either start out with an exponentially vanishing prefactor, or gain one from $G_{BC}$.
	
	Therefore, when we collect all the terms in $G_{BC} \ket{\mathrm{II}}$, we find that they either gained, or already had, an exponentially vanishing prefactor. Any combinatorial factors arising from $G_{BC}$ being written as a sum of projectors etc. will be at most polynomial in $k$, so they will not affect the conclusion that $\| G_{BC} \ket{\mathrm{II}} \|$ vanishes as $k \to \infty$. An entirely similar reasoning holds for $\ket{\mathrm{III}}$, and the argument is complete. 
\end{proof}

\subsection{Conclusion (Proof of Proposition \ref{pr:reduction-of-phi-pq})}
\begin{proof}[Proof of Proposition \ref{pr:reduction-of-phi-pq} ]
The state $\ket{\phi'_{p,q}}$ is not the normalized version
of the full term $\ket{I}$, but only of its low-intermediate-height part. We
therefore write
\[
    \ket{I}=\ket{I_{<}}+\ket{I_{\ge}},
\]
where
\[
    \ket{I_{<}}
    =
    \frac{1}{\sqrt{N^{\phi}_{p,q}}}
    \sum_{s<bk}
    \ket{GS^{AB}_{p,s}}\ket{\psi^{C}_{s,q}},
    \qquad
    \ket{I_{\ge}}
    =
    \frac{1}{\sqrt{N^{\phi}_{p,q}}}
    \sum_{s\ge bk}
    \ket{GS^{AB}_{p,s}}\ket{\psi^{C}_{s,q}} .
\]
By Definition, $\ket{\phi'_{p,q}}$ is the normalized version of $\ket{I_{<}}$, namely
\[
    \ket{\phi'_{p,q}}
    =
    \frac{\ket{I_{<}}}{\|\ket{I_{<}}\|}.
\]
Consequently,
\[
    \bra{I_{<}}G_{BC}\ket{I_{<}}
    =
    \|I_{<}\|^{2}
    \bra{\phi'_{p,q}}G_{BC}\ket{\phi'_{p,q}}
    \le
    \bra{\phi'_{p,q}}G_{BC}\ket{\phi'_{p,q}},
\]
since $\|I_{<}\|\le 1$.
It remains to compare $\ket{I}$ with $\ket{I_{<}}$. By the same low-imbalance truncation estimate used in Lemma \ref{lm:low-imbalance-approximation-trunc}, the discarded term $\ket{I_{\ge}}$ satisfies, uniformly in the low-imbalance regime under consideration,
\[
    \|I_{\ge}\|=o(1)
    \qquad (k\to\infty).
\]
Since $G_{BC}$ is a projection, $\|G_{BC}\|\le 1$, and hence
\[
\begin{aligned}
\left|
    \bra{I}G_{BC}\ket{I}
    -
    \bra{I_{<}}G_{BC}\ket{I_{<}}
\right|
&=
\left|
    \bra{I_{<}}G_{BC}\ket{I_{\ge}}
    +
    \bra{I_{\ge}}G_{BC}\ket{I_{<}}
    +
    \bra{I_{\ge}}G_{BC}\ket{I_{\ge}}
\right|  \\
&\le
    2\|I_{\ge}\|+\|I_{\ge}\|^{2}
    =
    o(1).
\end{aligned}
\]

Therefore,
\[
    \limsup_{k\to\infty}
    \sup
    \bra{I}G_{BC}\ket{I}
    \le
    \limsup_{k\to\infty}
    \sup
    \bra{\phi'_{p,q}}G_{BC}\ket{\phi'_{p,q}} .
\]
	Together with Proposition \ref{pr:phi-red-others-vanish}, and taking the $k \to \infty$ limit, we find
	\begin{align*}
	\limsup\limits_{k\to\infty} \op \sup_{\substack{p,q \le (1 + c) k \\ \phi_{p,q} \in \range E_k}} \mel{\phi_{p,q}}{G_{BC}}{\phi_{p,q}} \cp &\le \limsup\limits_{k\to\infty} \op \sum_{x,y \in \{\mathrm{I}, \mathrm{II}, \mathrm{III} \}} \sup_{\substack{p,q \le (1 + c) k \\ \phi_{p,q} \in \range E_k}} \mel{x}{G_{BC}}{y} \cp\\
	&= \limsup\limits_{k\to\infty}\op \sup_{\substack{p,q \le (1 + c) k \\ \phi_{p,q} \in \range E_k}} \mel{\mathrm{I}}{G_{BC}}{\mathrm{I}} \cp\\
	&\le \limsup\limits_{k\to\infty} \op \sup_{\substack{p,q \le (1 + c) k \\ \phi_{p,q} \in \range E_k}} \mel{\phi_{p,q}'}{G_{BC}}{\phi_{p,q}'} \cp 
	\end{align*}
	as desired.
\end{proof}

\bibliographystyle{alpha}
\newcommand{\etalchar}[1]{$^{#1}$}

\end{document}